\UseRawInputEncoding
\documentclass[11pt,reqno]{amsart}
\usepackage{amsthm,amsfonts,fancyhdr,amssymb,euscript,datetime,tikz,graphicx,float}
\usepackage[colorlinks,linkcolor=blue,citecolor=blue]{hyperref}

\newcommand{\bl}{\begin{lemma}}
\newcommand{\el}{\end{lemma}}
\def\beaa{\begin{eqnarray*}}
\def\eeaa{\end{eqnarray*}}
\def\ba{\begin{array}}
\def\ea{\end{array}}
\def\be#1{\begin{equation} \label{#1}}
\def \eeq{\end{equation}}

\newcommand{\gd}{{g \mkern-8mu /\ \mkern-5mu }}

\newcommand{\di}{\mbox{$d \mkern-9.2mu /$\,}}

\def\a{{\alpha}}

\def\be{{\beta}}

\def\ga{\gamma}
\def\Ga{\Gamma}
\def\de{\delta}

\def\la{\lambda}

\def\si{\sigma}

\def\om{\omega}
\def\Om{\Omega}
\def\th{\theta}
\def\ze{\zeta}

\def\varep{\varepsilon}

\def\pr{{\partial}}

\def\les{\lesssim}

\def\rh{{\rho}}

\def\ind{{\in \mkern-16mu /\ \mkern-4mu}}

\def\uQQ{{\underline{\mathcal{Q}}}}
\def\MMf{{\mathfrak{M}}}
\def\xdmf{{\dot{\mathfrak{X}}}}
\def\tth{{\tilde{\theta}}}
\def\XX{{\mathcal{X}}}
\def\YY{{\mathcal{Y}}}
\def\ZZ{{\mathcal{Z}}}
\def\xd{{\dot{x}}}
\def\mfq{{\dot{\mathfrak{c}}}}
\def\CCd{{\dot{\CC}}}

\def\gdcd{{\dot{\gd_c}}}
\def\Omd{{\dot{\Om}}}
\providecommand{\lrpar}[1]{\left( #1\right)}

\def\etabd{{\dot{\etab}}}
\def\etad{{\dot{\eta}}}
\def\rhd{{\dot{\rho}}}

\def\Kd{{\dot{K}}}
\def\Ldo{{\overset{\circ}{\Ld}}}
\def\phid{{\dot{\phi}}}
\def\omd{{\dot{\om}}}
\def\abd{{\dot{\ab}}}
\def\betad{{\dot{\beta}}}

\def\ombd{{\dot{\omb}}}
\def\ad{{\dot{\alpha}}}

\def\CC{{\mathcal C}}

\def\MM{{\mathcal M}}

\def\II{{\mathcal I}}
\def\FF{{\mathcal F}}

\def\HH{{\mathcal H}}

\def\GG{{\mathcal G}}

\def\OO{{\mathcal O}}

\def\DD{{\mathcal D}}
\def\PP{{\mathcal P}}

\def\QQ{{\mathcal Q}}

\def\HHb{\underline{\mathcal H}}

\def\D{{\bf D}}

\def\g{{\bf g}}

\def\SSS{{\mathbb{S}}}
\def\RRR{{\mathbb R}}

\def\f12{{\frac 1 2}}

\DeclareMathOperator{\Div}{\mathrm{div}}
\DeclareMathOperator*{\Curl}{\mathrm{curl}}
\def\half{\frac{1}{2}}

\newcommand{\pd}{\pd \mkern-9mu/\ \mkern-7mu}
\newcommand{\Lied}{\mathcal{L} \mkern-9mu/\ \mkern-7mu}

\newcommand{\DDd}{\DD \mkern-10mu /\ \mkern-5mu}

\newcommand{\Du}{\underline{D}}

\newcommand{\Divd}{\Div \mkern-17mu /\ }
\newcommand{\Divdo}{{\overset{\circ}{\Div \mkern-17mu /\ }}}
\newcommand{\Curld}{\Curl \mkern-17mu /\ }

\newcommand{\Nd}{\nabla \mkern-13mu /\ }

\newcommand{\Ld}{\triangle \mkern-12mu /\ }
\newcommand{\iin}{\in \mkern-16mu /\ \mkern-5mu}

\newcommand{\trchi}{{\tr \chi}}
\newcommand{\trchib}{{\tr \chib}}

\newcommand{\chihd}{{\dot{\chih}}}
\newcommand{\chibhd}{{\dot{\chibh}}}
\newcommand{\omtrchid}{\dot{(\Om \tr \chi)}}
\newcommand{\omtrchibd}{\dot{(\Om \tr \chib)}}

\def\ni{\noindent}

\def\Lb{{\,\underline{L}}}

\def\tr{\mathrm{tr}}

\def\chih{{\widehat \chi}}
\def\chib{{\underline \chi}}
\def\chibh{{\underline{\chih}}}

\def\etab{{\underline \eta}}
\def\omb{{\underline{\om}}}

\def\aa{{\underline{\a}}}

\def\th{\theta}

\def\f{\widetilde{f}}

\def\Rbf{{\mathbf{R}}}

\newcommand{\gac}{{\overset{\circ}{\ga}}}

\newcommand{\ab}{{\underline{\alpha}}}
\newcommand{\beb}{{\underline{\beta}}}

\newcommand{\DDdo}{{\overset{\circ}{\DDd}}}
\newcommand{\ilr}{{\int\limits_1^v}}

\newtheorem{theorem}{Theorem}[section]
\newtheorem{lemma}[theorem]{Lemma}
\newtheorem{proposition}[theorem]{Proposition}
\newtheorem{corollary}[theorem]{Corollary}
\newtheorem{definition}[theorem]{Definition}
\newtheorem{remark}[theorem]{Remark}

\numberwithin{equation}{section}

\makeatletter
\def\@setthanks{\vspace{-\baselineskip}\def\thanks##1{\@par##1\@addpunct.}\thankses}
\makeatother

\begin{document}

\title[The characteristic gluing problem. Linear and non-linear analysis]{The characteristic gluing problem  \\ for the Einstein vacuum equations. \\ Linear and non-linear analysis}
\author[S. Aretakis, S. Czimek and I. Rodnianski]{Stefanos Aretakis $^{(1)}$, Stefan Czimek $^{(2)}$, and Igor Rodnianski $^{(3)}$} 

\thanks{\noindent$^{(1)}$ Department of Mathematics, University of Toronto, 40 St George Street, Toronto, ON, Canada, \texttt{aretakis@math.toronto.edu}. \\
$^{(2)}$ Institute for Computational and Experimental Research in Mathematics, Brown University, 121 South Main Street, Providence, RI 02903, USA,  \texttt{stefan\_czimek@brown.edu}. \\
$^{(3)}$ Department of Mathematics, Princeton University, Fine Hall, Washington Road, Princeton, NJ 08544, USA, \texttt{irod@math.princeton.edu}. }

\begin{abstract} This is the second paper in a series of papers adressing the characteristic gluing problem for the Einstein vacuum equations. We solve the codimension-$10$ characteristic gluing problem for characteristic data which are close to the Minkowski data. We derive an infinite-dimensional space of gauge-dependent charges and a $10$-dimensional space of gauge-invariant charges that are conserved by the linearized null constraint equations and act as obstructions to the gluing problem. The gauge-dependent charges can be matched by applying angular and transversal gauge transformations of the characteristic data. By making use of a special hierarchy of radial weights of the null constraint equations, we construct the null lapse function and the conformal geometry of the characteristic hypersurface, and we show that the aforementioned charges are in fact the only obstructions to the gluing problem. Modulo the gauge-invariant charges, the resulting solution of the null constraint equations is $C^{m+2}$ for any specified integer $m\geq0$ in the tangential directions and $C^2$ in the transversal directions to the characteristic hypersurface. We also show that higher-order (in all directions) gluing is possible along bifurcated characteristic hypersurfaces (modulo the gauge-invariant charges). 
\end{abstract}

\maketitle
\setcounter{tocdepth}{3}
\tableofcontents
\section{Introduction} \label{SECintroduction} 

\ni The \emph{gluing problem} in general relativity asks to connect two given spacetimes across a gluing region. Technically speaking, one aims at solving the constraint equations with two prescribed initial data sets. The obstructions to gluing provide insights into the intrinsic rigidity of the Einstein equations. In their ground-breaking work, Corvino \cite{Corvino} and Corvino--Schoen \cite{CorvinoSchoen} pioneered the study of the Riemannian gluing problem (for spacelike initial data sets). In particular, their gluing construction shed light on the importance of the interplay between the rigidity and the flexibility of the geometric character of the Einstein equations.

In \cite{ACR3} we initiated the study of \emph{the gluing problem for characteristic initial data} for the Einstein vacuum equations. The characteristic gluing problem exhibits various novel features. For example, gluing along characteristic hypersurfaces is based on solving the so-called \emph{null constraint equations} which are of transport character, whereas the previously studied gluing problem for spacelike initial data requires to analyze the elliptic Riemannian constraint equations. Moreover, in the characteristic gluing construction, the null lapse function and the conformal geometry of the characteristic hypersurface can be freely prescribed.

The present paper provides the full details on the gluing of characteristic data which are close to the Minkowski data. Working close to Minkowski spacetime is natural in the sense that by rescaling, it corresponds to gluing near spacelike infinity in an asymptotically flat spacetime, see \cite{ACR3,ACR2}. We identify all obstructions to gluing (at the level of $C^2$-gluing for the metric components) and show that they are stemming from \emph{conservation laws of the linearized null constraint equations}. We show that these conservation laws determine a $10$-dimensional space of so-called \emph{gauge-invariant charges} and an infinite-dimensional space of so-called \emph{gauge-dependent charges}. We prove that the gauge-dependent charges can be matched by applying transversal perturbations and gauge transformations to the characteristic data. In particular, gauge transformations alone are not sufficient to match all gauge-dependent charges. In \cite{ACR3,ACR2} we geometrically interpret the remaining $10$-dimensional space of gauge-invariant charges by relating them to the ADM energy, linear momentum, angular momentum and center-of-mass, and we use this identification to glue asymptotically flat spacetimes to a member of the Kerr family.

This introduction is structured as follows: In Section \ref{SECnontechnicalINTRO} we introduce the characteristic gluing problem and present our main results in non-technical terms. In Section \ref{SEChistory9999} we provide an overview of the literature for the gluing problem. In Section \ref{SECnullGeometryintro9999} we set up the null geometry framework and the characteristic initial value problem for the Einstein vacuum equations and in Section \ref{SECintroStatementMainTheorem} we present a more formal version of our main theorem. In Sections \ref{SEClinearizedCHARgluing999intro} and \ref{SECintroNONLINEARgluing} we provide the main ideas of our methods and in Section \ref{SECintroTransversalDerivativesStatement} we discuss characteristic gluing along two null hypersurfaces bifurcating from an auxiliary sphere.

\subsection{Introduction to the characteristic gluing problem and overview of results} \label{SECnontechnicalINTRO}
In this section we introduce in a colloquial, non-technical way the characteristic gluing problem and the main results of this paper.

Consider the null hypersurfaces $(\HH_1, \HHb_1)$ and $(\HH_2, \HHb_2)$ emanating from two spheres $S_1$ and $S_2$, respectively, in two vacuum spacetimes $\MM_1$ and $\MM_2$. %
\begin{figure}[H]
\begin{center}
\includegraphics[width=9.5cm]{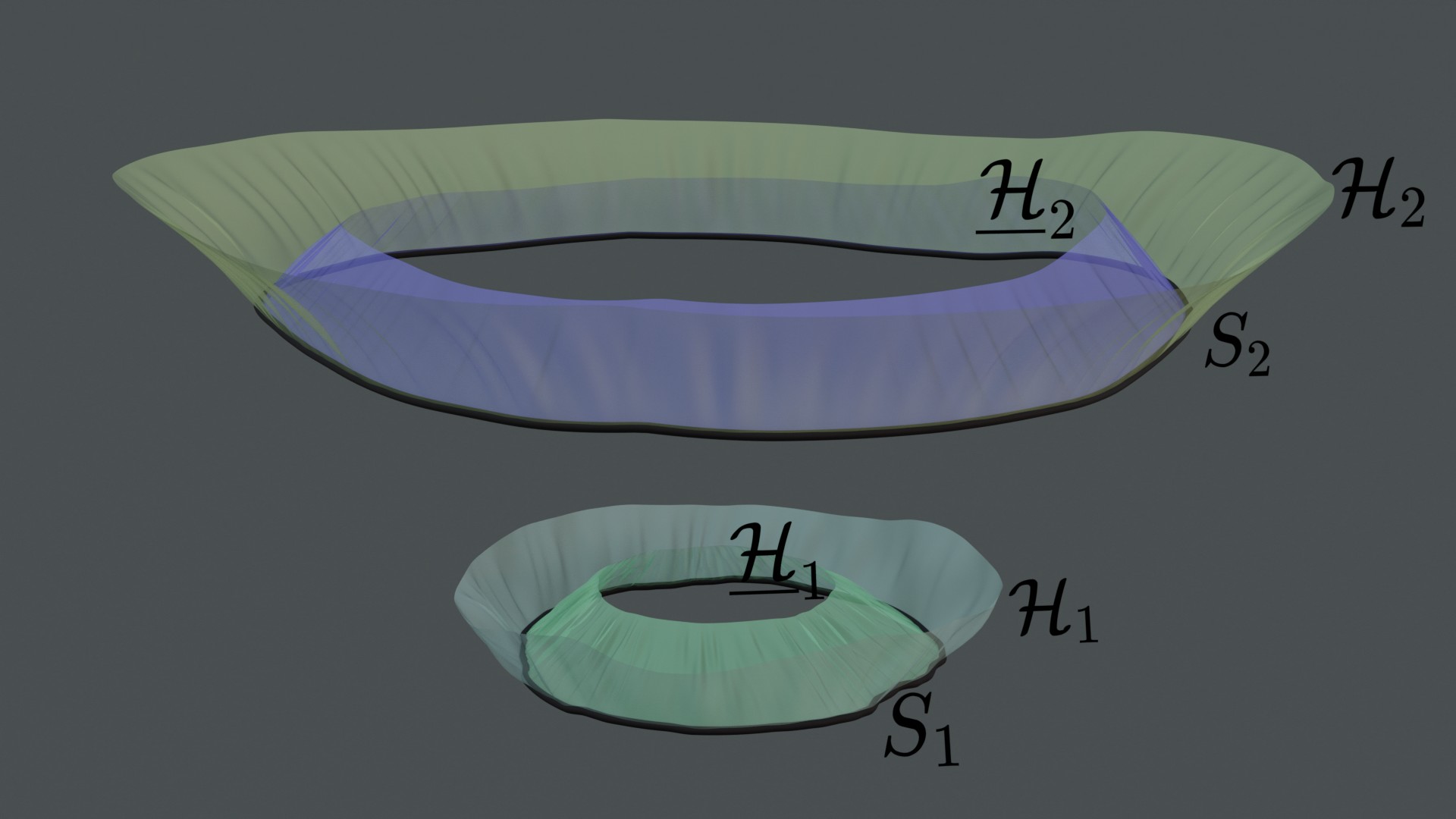} 
\vspace{0.4cm}
\caption{The null hypersurfaces $(\HH_1, \HHb_1)$ and $(\HH_2, \HHb_2)$ emanating from two spheres $S_1$ and $S_2$, in the vacuum spacetimes $\MM_1$ and $\MM_2$.}
\end{center}
\label{FIG1}
\end{figure}
\vspace{-0.8cm}

\ni A first formulation of the characteristic gluing problem asks if there exists a characteristic hypersurface $\HH$ that satisfies the null constraint equations whose characteristic data agrees on its past boundary with the data on $\HH_1$, and on its future boundary with the data on $\HH_2$? If not, what are the obstructions to the existence of such a hypersurface? 

\begin{figure}[H]
\begin{center}
\includegraphics[width=9.5cm]{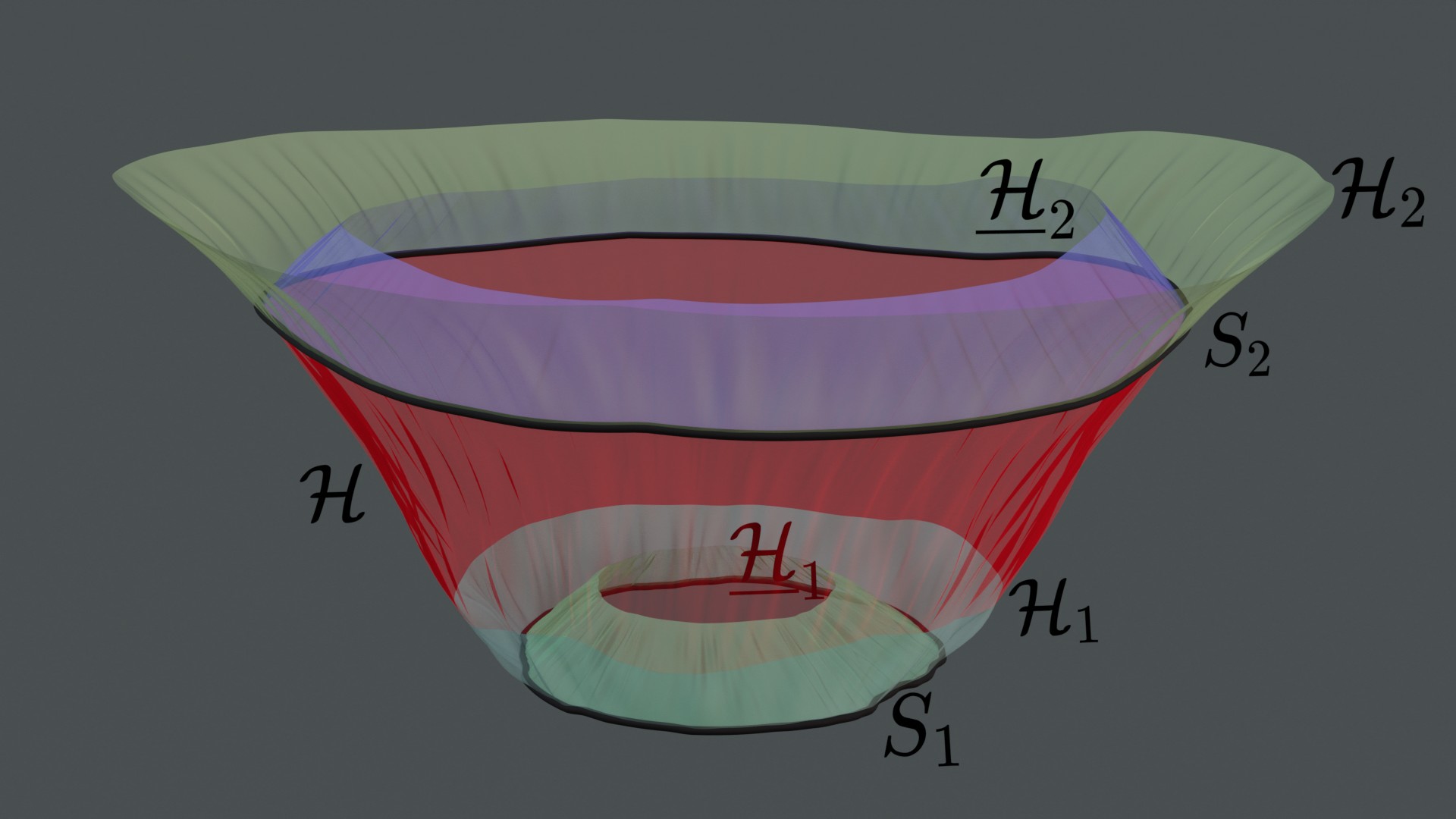} 
\vspace{0.4cm}
\caption{The red hypersurface glues the characteristic initial data for the two vacuum spacetimes $\MM_1$ and $\MM_2$.}
\end{center}
\label{FIG2}
\end{figure}
\vspace{-0.8cm}

\ni Let us denote by $x_1$ and $x_2$ the restriction of the metric components, the Christoffel symbols and the Riemann curvature components of the spacetime metrics of $\MM_1$ and $\MM_2$ to the spheres $S_1$ and $S_2$, respectively (with respect to local double null coordinate systems). We will refer to $x_1$ and $x_2$ as the \emph{sphere data} on $S_1$ and $S_2$ respectively.

\begin{figure}[H]
\begin{center}
\includegraphics[width=9.5cm]{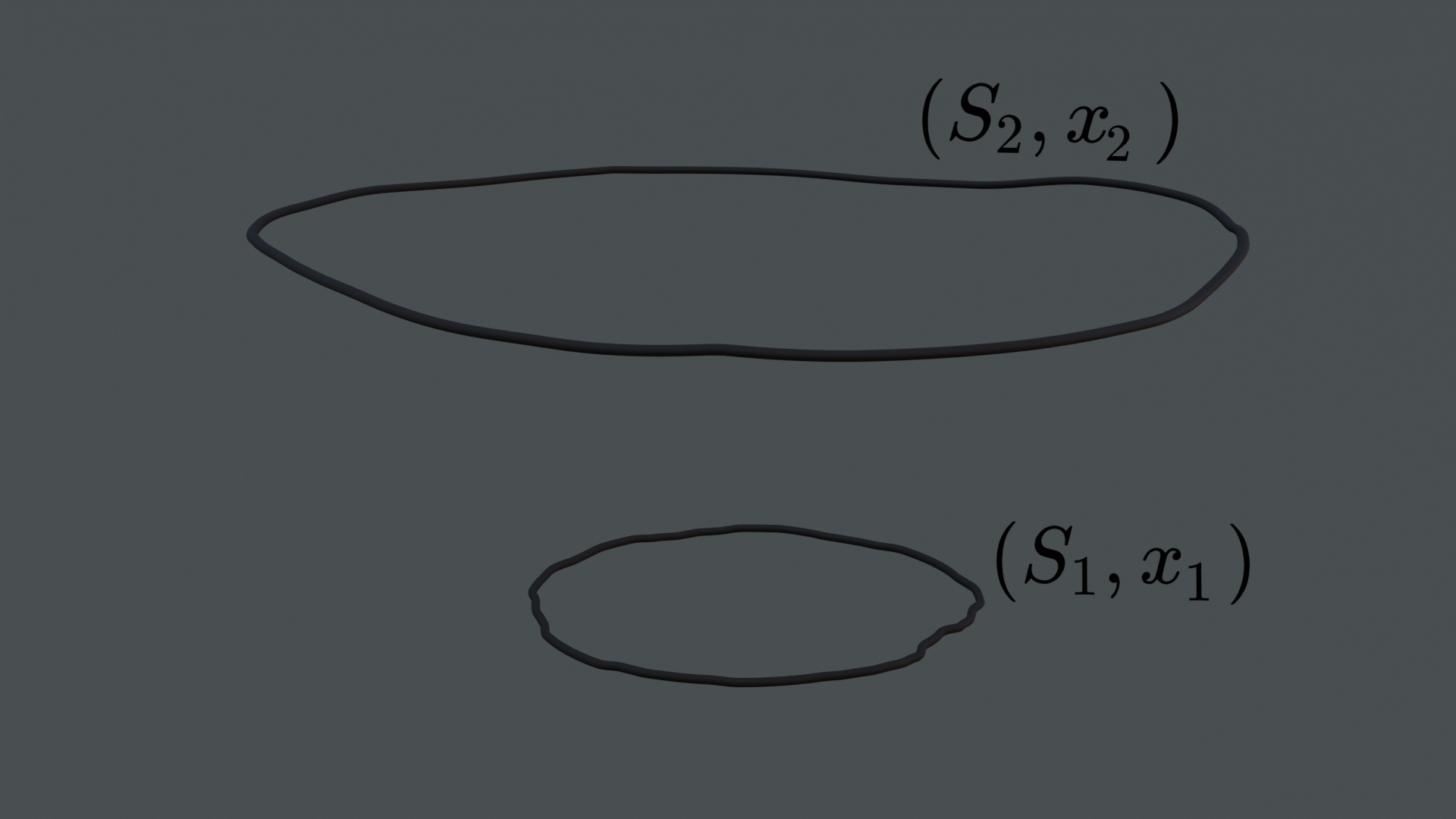} 
\vspace{0.4cm}
\caption{On each sphere $S_1$ and $S_2$ we consider the restrictions $x_1$ and $x_2$ of the metric components, the Christoffel symbols and the Riemann curvature components of the metrics of the ambient spacetimes $\MM_1$ and $\MM_2$.}
\end{center}
\label{FIG3}
\end{figure}
\vspace{-0.8cm}

\ni A reduction of the gluing problem can be formulated as follows: Given two spheres $S_1$ and $S_2$ equipped with sphere data $x_1$ and $x_2$, respectively, construct a solution to the null constraint equations along a null hypersurface $\HH_{[1,2]}$ whose boundary sections admit the sphere data $x_1$ and $x_2$. Such characteristic gluing of sphere data is at the level of $C^2$-gluing for the metric components.

\begin{figure}[H]
\begin{center}
\includegraphics[width=9.5cm]{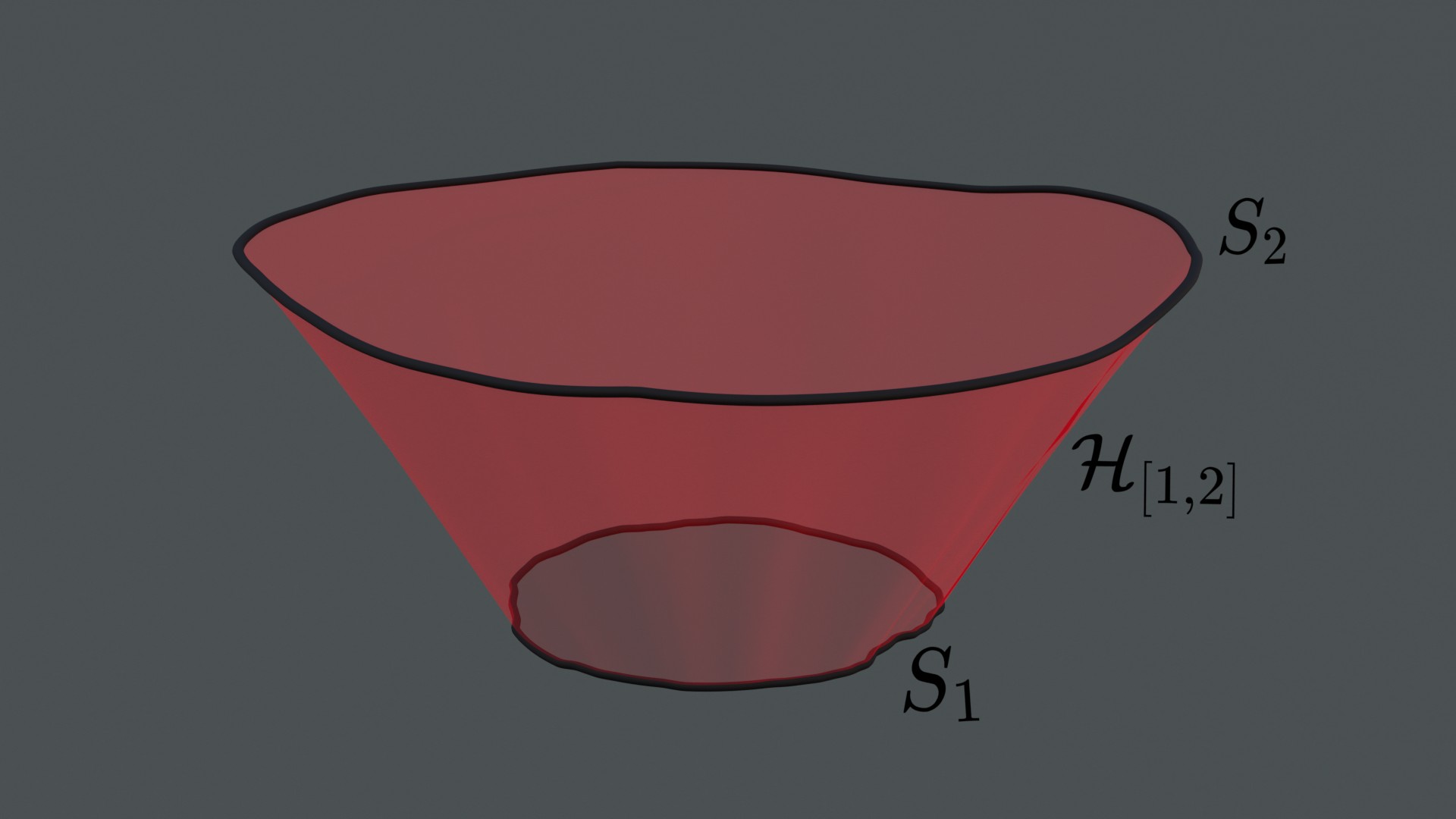} 
\vspace{0.2cm}
\caption{The characteristic gluing problem for sphere data $x_1$ and $x_2$.}
\end{center}
\label{FIG4}
\end{figure}
\vspace{-0.5cm}

\ni The characteristic gluing problem as stated above can not always be solved. To start, one obvious obstruction is imposed by the monotonicity property of the Raychaudhuri equation along null hypersurfaces. A more subtle obstruction is imposed by the existence of an infinite-dimensional space of conservation laws for the linearized null constraint equations at Minkowski spacetime.

To address this hindrance, we need to take into account the change of sphere data under \emph{sphere perturbations} and \emph{sphere diffeomorphisms}. We define \emph{sphere perturbations} of the sphere data $x_2$ on $S_2$ in the vacuum spacetime $\MM_2$ as follows. Consider the null hypersurface $\HHb_2$ in $\MM_2$ through $S_2$ that is conjugate (that is, transversal) to the null hypersurface $\HH_2$ which the gluing hypersurface should attach to. Then the sphere data $x'_2$ on a section $S_2'$ of $\HHb_2$ is called a sphere perturbation of $x_2$ on $S_2$. Sphere diffeomorphisms of sphere data is defined by pulling back the sphere data under a diffeomorphism of the sphere.

\begin{figure}[H]
\begin{center}
\includegraphics[width=9.5cm]{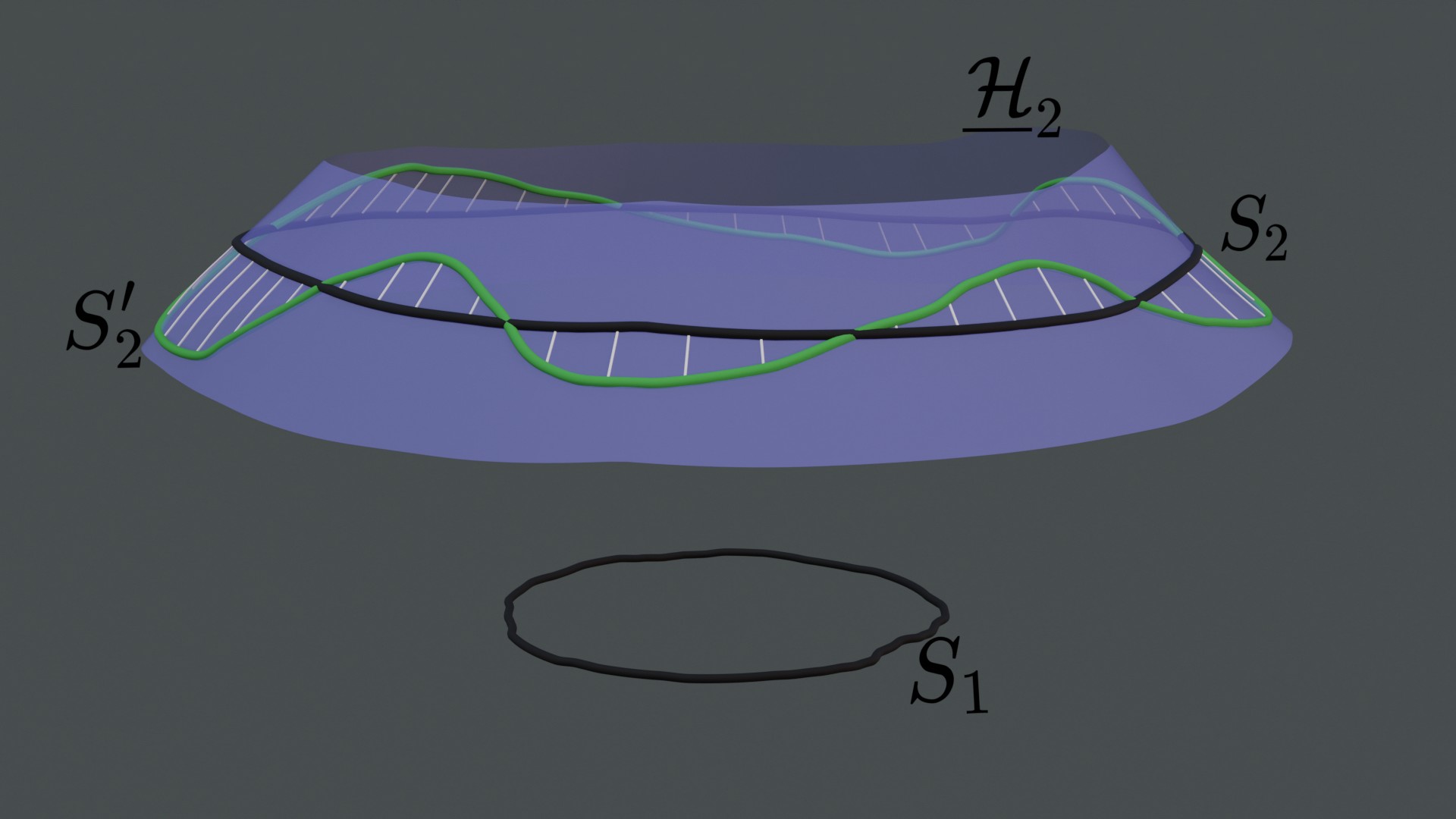} 
\vspace{0.2cm}
\caption{Transversal perturbations of the sphere data $x_2$ on $S_2$ in $\MM_2$.}
\end{center}
\label{FIG6}
\end{figure}
\vspace{-0.5cm}

\ni We arrive at the following reformulation of the characteristic gluing problem.\\

\ni Given sphere data $x_1$ and $x_2$ on two spheres $S_1$ and $S_2$ in vacuum spacetimes $\MM_1$ and $\MM_2$, respectively, construct:
\begin{enumerate}
\item a sphere perturbation $S'_2$ of $S_2$ with sphere data $x'_2$ (subject also to a sphere diffeomorphism),
\item a solution to the null constraint equations along a null hypersurface $\HH'_{[1,2]}$ whose boundary sections admit the sphere data $x_1$ and $x_2'$.
\end{enumerate}  

\begin{figure}[H]
\begin{center}
\includegraphics[width=9.5cm]{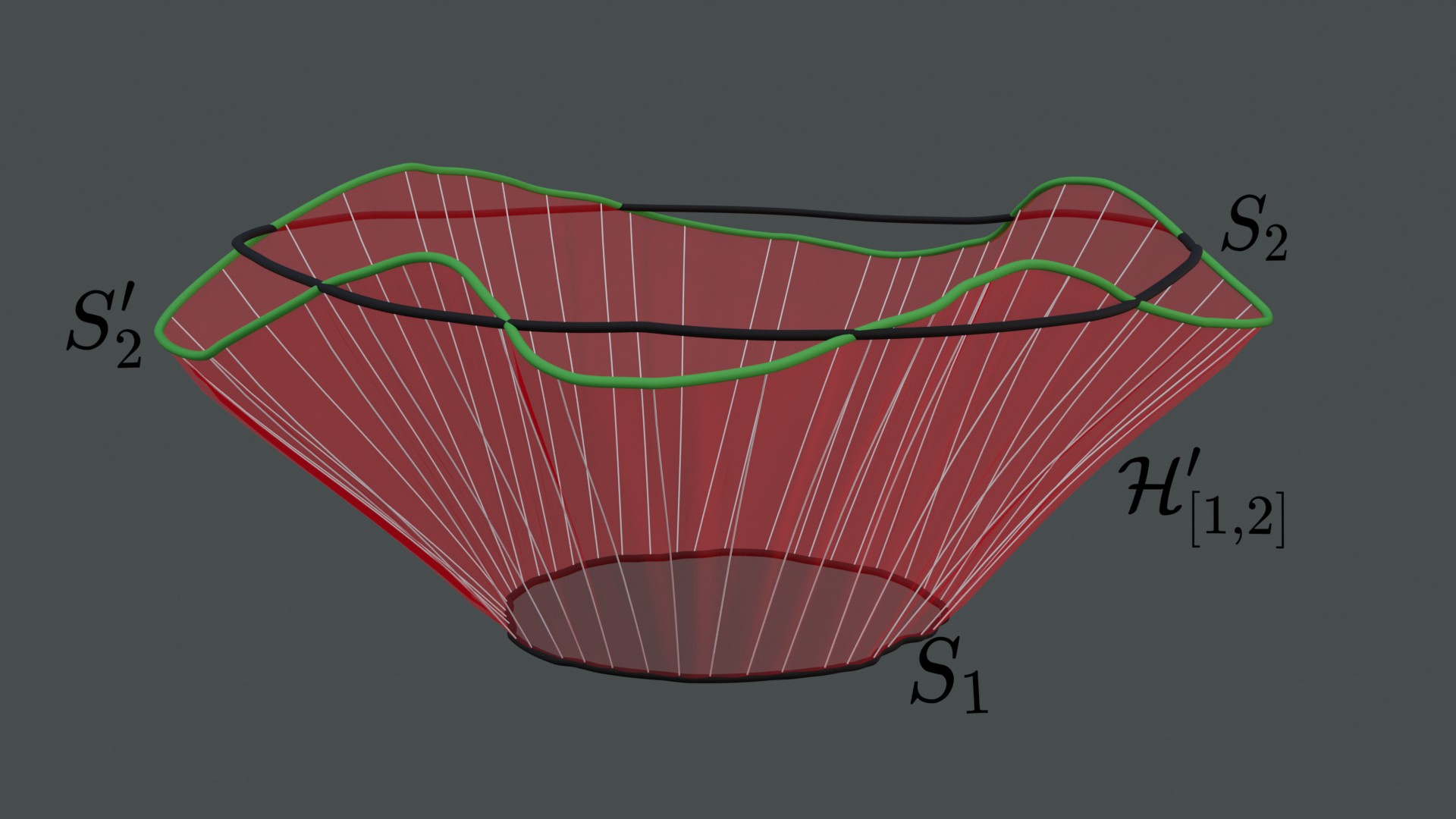} 
\vspace{0.2cm}
\caption{The characteristic gluing problem after taking into account transversal gauge transformations of the sphere data.}
\end{center}
\label{FIG7}
\end{figure}
\vspace{-0.8cm}

\ni We can now present a first version of the main result of this paper. \\

\ni \textbf{Theorem} (Codimension-$10$ perturbative characteristic gluing, version 1) \emph{Consider sphere data $x_1$ and $x_2$ on two spheres $S_1$ and $S_2$ which are sufficiently close to the sphere data on the round spheres of radius $1$ and $2$ in Minkowski spacetime, respectively. Then, modulo a $10$-dimensional space of charges, characteristic gluing in the above sense is possible. In other words, there is a null hypersurface $\HH'_{[1,2]}$, connecting the sphere data $x_1$ on $S_1$ and a sphere perturbation $S_2'$ of the sphere $S_2$ with sphere data $x_2'$ (subject also to a sphere diffeomorphism), satisfying the null constraint equations such that the sphere data $x_1$ and $x'_2$, excluding the $10$ charges explicitly defined at $S_2'$, are glued.}\\

\ni \textbf{Remark} \emph{In the above theorem, it is equivalently possible to perturb the sphere $S_1$ instead of the sphere $S_2$. Moreover, higher-order derivatives tangential to the gluing hypersurface can be glued without further obstructions.}\\

\ni It is important to underline that the 10 charges correspond to 10 constants and not to 10 functions  on the sphere $S_2'$. In fact, in this paper, we identify an infinite-dimensional space of charges determined by the sphere data $x\vert_S$ of the sections $S$ of the null hypersurface $\HH'_{[1,2]}$ and we show that it splits into a $10$-dimensional space of gauge invariant charges and an infinite-dimensional space of gauge-dependent charges.

\begin{figure}[H]
\begin{center}
\includegraphics[width=9.5cm]{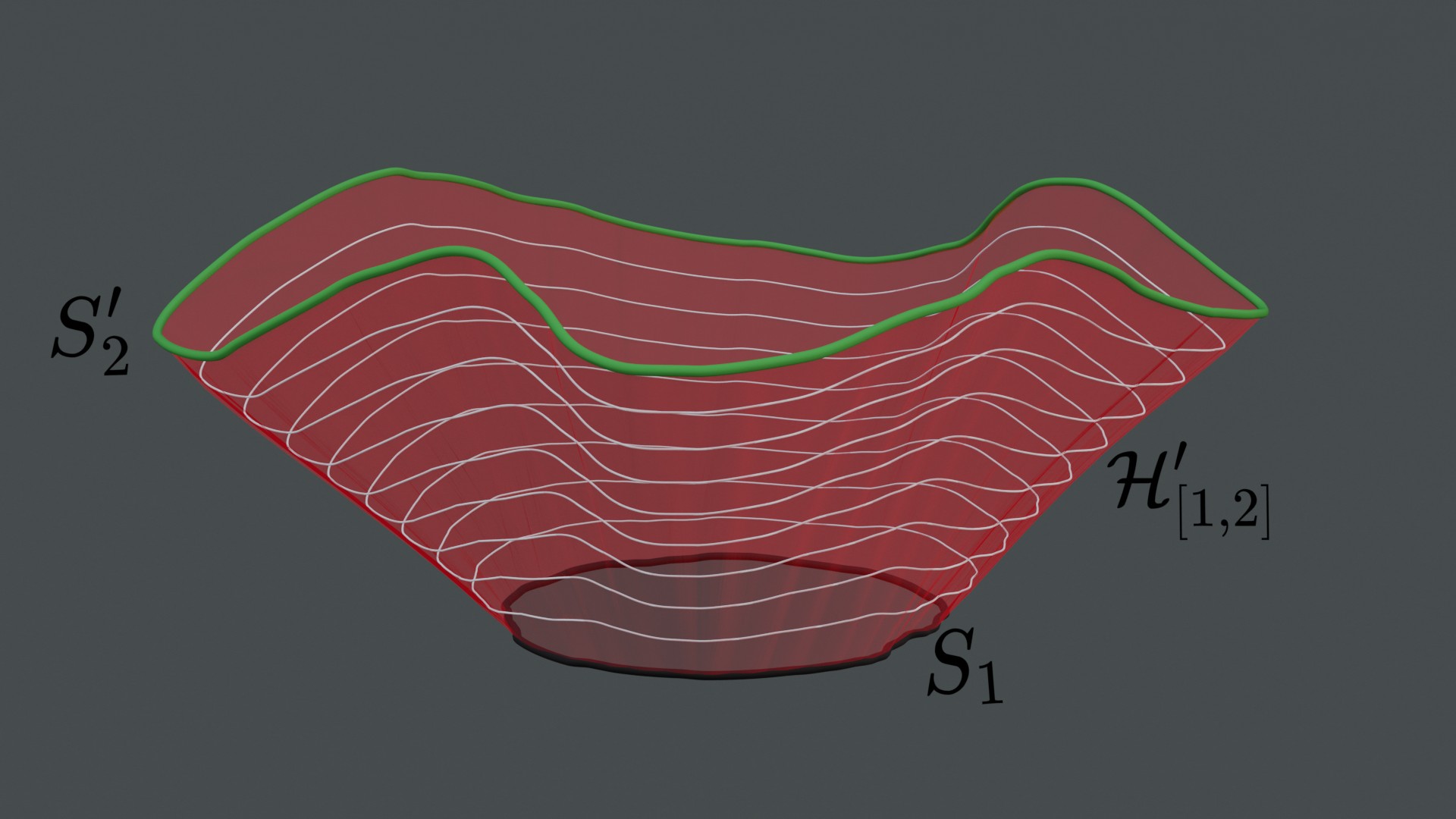} 
\vspace{0.2cm}
\caption{There is an infinite-dimensional space of charges defined on the sections of the characteristic hypersurface $\HH'_{[1,2]}$. These charges are conserved by the linearized null constraint equations at Minkowski.}
\end{center}
\label{FIG8}
\end{figure}

\ni We prove that the gauge-dependent charges can be matched by an appropriate choice of sphere perturbation along $\HHb_2$ and sphere diffeomorphism. The gauge-invariant charges can in general not be adjusted by such sphere data perturbations, as they change quadratically under perturbations while the gauge-dependent charges transform linearly. Gluing of the gauge-invariant charges is achieved in \cite{ACR3,ACR2}.
 
Two remarks regarding the above theorem are in order:

\begin{enumerate}
\item \textbf{Sphere perturbations and sphere diffeomorphisms.} We glue from $S_1$ to a transversal perturbation $S_2'$ of $S_2$, not to $S_2$ itself, and the sphere data on $S_2'$ is also subject to a sphere diffeomorphism. Sphere diffeomorphisms are gauge transformations intrinsic to the given sphere data. On the other hand, sphere perturbations are extrinsic to the given sphere data but they are intrinsic gauge transformations of the ambient spacetime; see, for example, the linearized pure gauge solutions in \cite{DHR}.
\item \textbf{Transversal regularity.} Higher-order derivatives of the sphere data which are transversal to the gluing null hypersurface are not glued at $S_2'$. This is due to the existence of additional higher-order conserved charges which involve these transversal derivatives. 
\end{enumerate}

\ni Our next theorem resolves both of these issues by gluing along two null hypersurfaces bifurcating from an auxiliary sphere $S_{\mathrm{aux}}$. The advantage of such an approach is that we can first glue $S_1$ to the auxiliary sphere $S_{\mathrm{aux}}$ in the "ingoing" direction and then glue $S_{\mathrm{aux}}$ to the sphere $S_2$ in the "outgoing" direction.

\begin{figure}[H]
\begin{center}
\includegraphics[width=9.5cm]{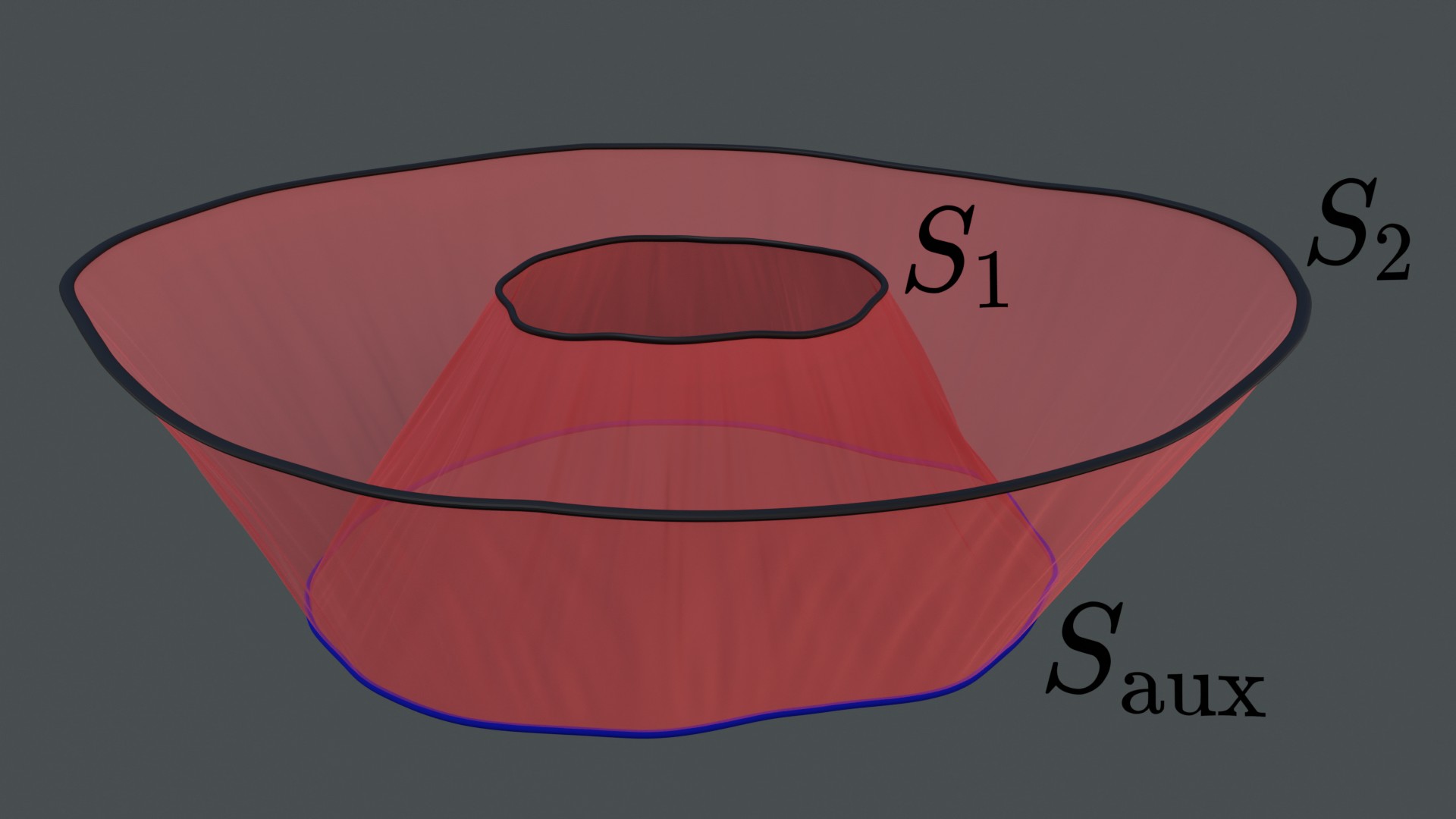} 
\caption{First illustration.}
\end{center}
\label{FIG9}
\end{figure}
\vspace{-0.5cm}

\ni \textbf{Theorem} (Codimension-$10$ bifurcate characteristic gluing, version 1) \emph{Let $m\geq0$ be an integer. The problem of characteristic gluing along two null hypersurfaces bifurcating from an auxiliary sphere can be solved close to Minkowski for $m^{\mathrm{th}}$-order derivatives in all directions (up to the $10$-dimensional space of gauge-invariant charges) without perturbing either of the spheres $S_1$ and $S_2$.}

\subsection{Previous gluing constructions} \label{SEChistory9999} 

\subsubsection{Gluing constructions in general relativity} \ni Gluing constructions in general relativity are, up to now, mainly focused on the gluing of \emph{spacelike} initial data satisfying the elliptic constraint equations. 

Gluing constructions based on the gluing of connected sums (see the works \cite{SchoenYauPSCM,GromovL} on codimension-$3$ surgery for manifolds of positive scalar curvature) were studied by Chru\'sciel--Isenberg--Pollack \cite{CIP1,CIP2}, Chru\'sciel--Mazzeo \cite{ChruscielMazzeo}, Isenberg--Maxwell--Pollack \cite{IMP3}, Isenberg--Mazzeo--Pollack \cite{IMP1,IMP2}.

On the other hand, in the ground breaking work of Corvino \cite{Corvino} and Corvino--Schoen \cite{CorvinoSchoen}, the \emph{geometric under-determinedness} of the spacelike constraint equations is used to study the (codimension-$1$) gluing problem. 
In particular, they showed that asymptotically flat spacelike initial data can be glued across a compact region to exactly Kerr spacelike initial data. Further constructions and refinements based on this approach were proved by Chru\'sciel--Delay \cite{ChruscielDelay1,ChruscielDelay2}, Chru\'sciel--Pollack \cite{ChruscielPollack}, Cortier \cite{Cortier}, Hintz \cite{Hintz}. Another milestone was the work by Carlotto--Schoen \cite{CarlottoSchoen} which showed that it is possible to glue spacelike initial data -- along a non-compact cone -- to spacelike initial data for Minkowski.

\subsubsection{Characteristic gluing for the wave equation} The characteristic gluing problem was studied before by the first author \cite{CiteGluing} in the much simpler setting of the linear homogeneous wave equation on general (but fixed) Lorentzian manifolds. Similarly to the present paper, \cite{CiteGluing} determined that the only obstructions to characteristic gluing are conservation laws along null hypersurfaces. Moreover, it was shown that a necessary and sufficient condition for the existence of such conservation laws is that the kernel of an elliptic operator defined on the null hypersurface \cite{CiteElliptic} is non-trivial. Hence, for the linear wave equation, \cite{CiteGluing} derived a geometric characterization of all obstructions to characteristic gluing along a general null hypersurface. Specific examples of null hypersurfaces which admit conservation laws for the wave equation are 
\begin{enumerate}
\item the standard cones in Minkowski spacetime, 
\item the null infinity of asymptotically flat spacetimes, and 
\item the event horizon of extremal black holes. 
\end{enumerate}
The conserved charges in cases (2) and (3) above have important applications in the study of the evolution of scalar perturbations on black hole spacetimes. Specifically the charges along null infinity (also known as the Newman--Penrose constants) are related to the leading-order coefficients of the late-time asymptotics of solution to the wave equation on Schwarzschild \cite{CiteSS,CitePriceLaw} and Kerr \cite{CiteKerr} spacetimes. Similar results were recently obtained for the Dirac equation on Schwarzschild in \cite{Ma}. On the other hand, the conservation laws on extremal horizons are the source of the horizon instability of extremal black holes \cite{CiteExtremal1,CiteExtremal2,CiteExtremal3}. It is worth noting that even though the latter charges are defined on sections of the extremal event horizon, they can be computed by far-away observers at null infinity and hence serve as potential observational signatures of extremal black holes \cite{Easymptotics,PRL,Khanna}. 

\subsection{Double null coordinates} \label{SECnullGeometryintro9999} 
In this section we outline the geometric framework of this paper to provide a first version of our main theorem in Section \ref{SECintroStatementMainTheorem}.

Let $S$ be a spacelike $2$-sphere in a spacetime $(\MM,\g)$, and let ${u}_0$ and ${v}_0$ with $v_0>u_0$ be two real numbers. Let $u$ and $v$ be two optical functions of $(\MM,\g)$ such that $S=\{ {u}={u}_0, {v}={v}_0\}$, and for real numbers ${u}_1$ and ${v}_1$, the hypersurfaces
\begin{align*} 
\begin{aligned} 
\HH_{{u}_1} := \{ {u}= {u}_1\}, \,\, \HHb_{{v}_1} := \{ {v}= {v}_1\},
\end{aligned} 
\end{align*}
are outgoing and ingoing null hypersurfaces, respectively. The union of these null hypersurfaces form a so-called \emph{double null foliation} of $(\MM,\g)$.

\begin{figure}[H]
\begin{center}
\includegraphics[width=9.5cm]{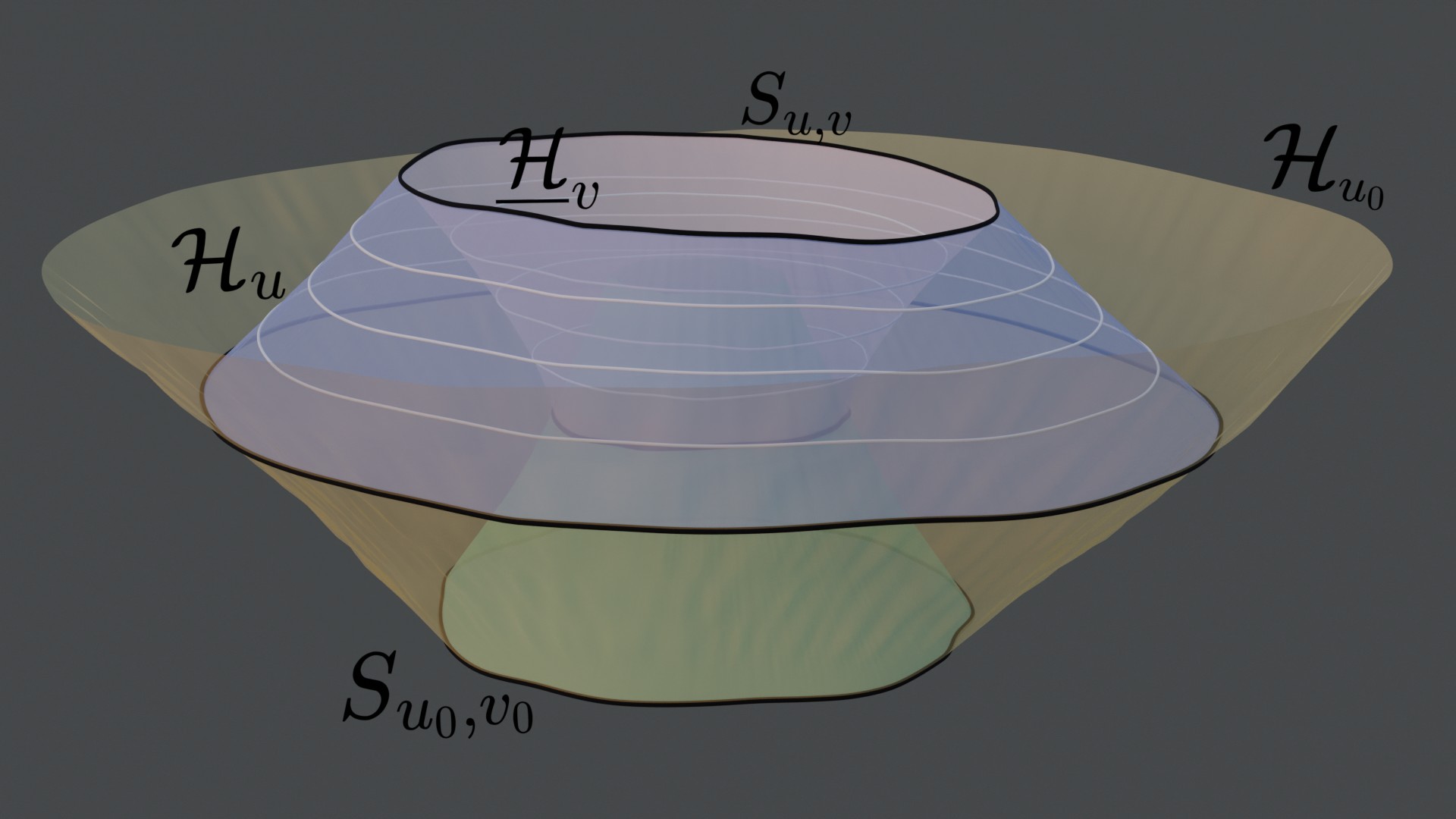} 
\vspace{0.4cm}
\caption{The double null foliation formed by the level sets of the optical functions $u$ and $v$.}
\end{center}
\label{FIGdoublenullS2pic1}
\end{figure}
\vspace{-0.2cm}

\ni On the sphere
\begin{align} 
\begin{aligned} 
S_{u_0,v_0} := \{ {u}={u}_0, {v}={v}_0\},
\end{aligned} \label{EQdefSphereinCoordinates}
\end{align}
we define local angular coordinates $(\th^1,\th^2)$, and extend them everywhere by propagating them first along the null generators of $\HH_{u_0}$ and then of $\HHb_v$ for all $v$, as in the Figure \ref{FIGdoublenullS2pic2} below. The resulting coordinate system $( u,  v, \th^1,\th^2)$ is called a \emph{double null coordinate system}.

\begin{figure}[H]
\begin{center}
\includegraphics[width=9.5cm]{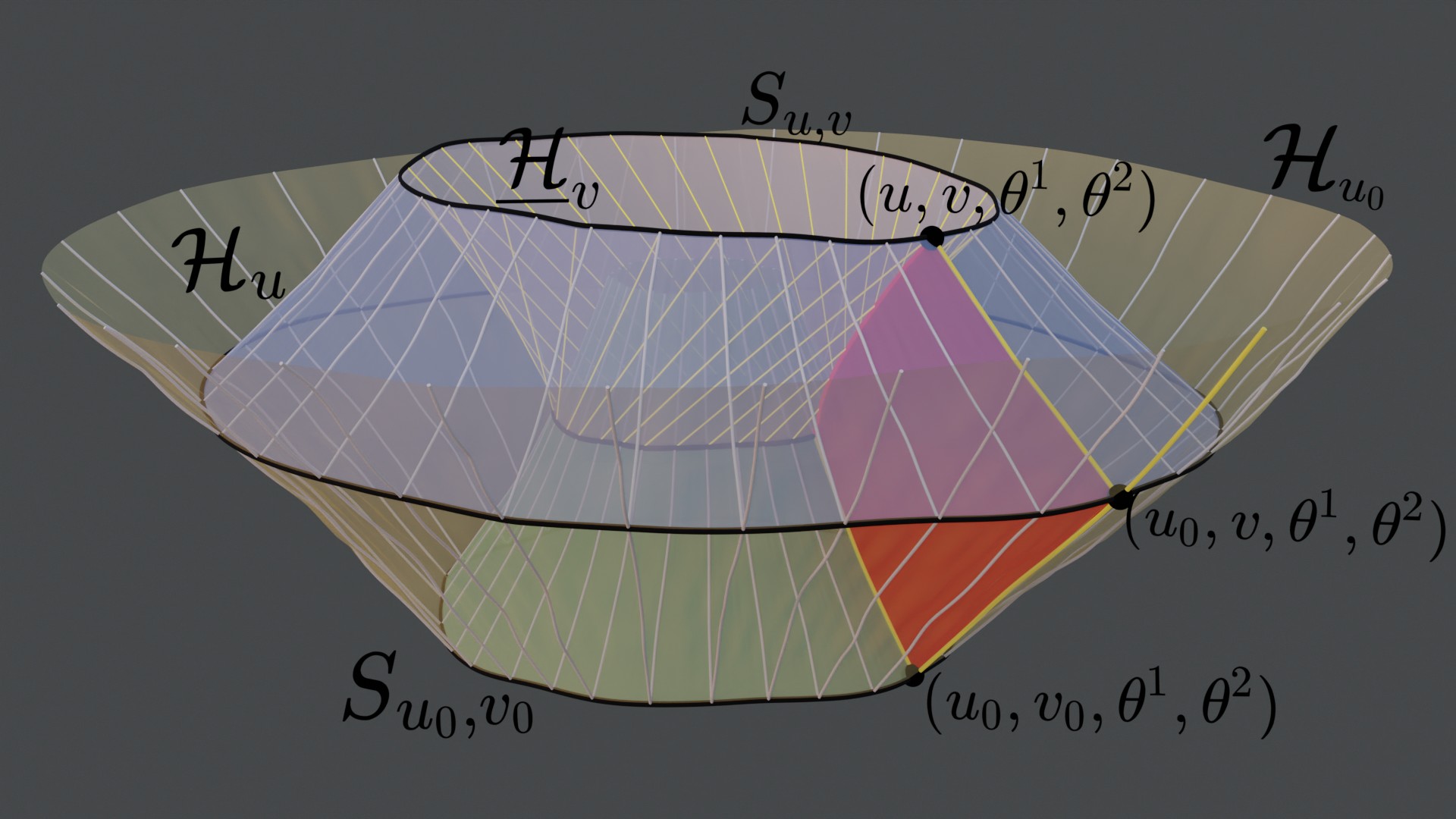} 
\vspace{0.4cm}
\caption{The construction of double null coordinates $(u,v,\th^1,\th^2)$ on $(\MM,\g)$.}
\end{center}
\label{FIGdoublenullS2pic2}
\end{figure}
\vspace{-0.5cm}

\ni With respect to double null coordinates, it holds that 
\begin{align} 
\begin{aligned} 
\g = -4 {\Om}^2 d u d v + \gd_{AB} \lrpar{d\th^A + {b}^A d v}\lrpar{d\th^B + {b}^B d v},
\end{aligned} \label{EQintroDoubleNull123}
\end{align}
where
\begin{itemize}
\item the scalar function $\Om$ is the so-called \emph{null lapse},
\item $\gd_{AB}$ is the induced Riemannian metric on the $2$-spheres $S_{{u},{v}}$ of constant $( u, v)$,
\item the $S_{ u,  v}$-tangent vectorfield ${b}$ is the so-called \emph{shift vector}. By construction it holds that $b$ vanishes on $u=u_0$.
\end{itemize}

\ni The induced metric $\gd$ can be expressed as
\begin{align*} 
\begin{aligned} 
\gd = \phi^2 \gd_c, 
\end{aligned} 
\end{align*}
where with respect to the coordinates $(\th^1,\th^2)$,
\begin{align*} 
\begin{aligned} 
\phi^2 :=\frac{\sqrt{\gd}}{\sqrt{\gac}}, \,\, \gd_c := \phi^{-2} \gd, \,\, \gac := (d\th^1)^2 + \sin^2\th^1 (d\th^2)^2.
\end{aligned} 
\end{align*} 

\ni Define the null vectors
\begin{align*} 
\begin{aligned} 
{L} := \pr_{ v} + {b}, \,\, {\Lb} := \pr_{ u}, \,\, \widehat{ L}:= {\Om}^{-1}  L, \,\, \widehat{ \Lb}:= {\Om}^{-1}  \Lb,
\end{aligned} 
\end{align*}
and let for $A=1,2$, $\pr_A := \pr_{\th^A}$. Then for $A,B=1,2$, the \emph{Ricci coefficients} are defined by
\begin{align} 
\begin{aligned} 
{\chi}_{AB} :=& \g(\D_A \widehat{ L},\pr_B), & {\chib}_{AB} :=& \g(\D_A \widehat{ \Lb},\pr_B), & {\zeta}_A :=& \half \g(\D_A \widehat{L}, \widehat{\Lb}), \\
\eta :=& \zeta + \di \log \Om, & \om :=& \D_L \log \Om, & \omb :=& \D_{\Lb} \log \Om,
\end{aligned} \label{EQdefRicciINTRO}
\end{align}
where $\D$ denotes the covariant derivative on $(\MM,\g)$, and define the \emph{null curvature components} by
\begin{align} 
\begin{aligned} 
\alpha_{AB} :=& \Rbf(\pr_A,\widehat{ L}, \pr_B, \widehat{ L}), & \beta_A :=& \half \Rbf(\pr_A, \widehat{ L},\widehat{\Lb},\widehat{ L}), \\
 \rh :=& \frac{1}{4} \Rbf(\widehat{\Lb}, \widehat{ L}, \widehat{\Lb}, \widehat{ L}), & \sigma \iin_{AB} :=& \half \Rbf(\pr_A,\pr_B,\widehat{\Lb}, \widehat{ L}), \\
\beb_A :=& \half \Rbf(\pr_A, \widehat{\Lb},\widehat{\Lb},\widehat{ L}), & \ab_{AB} :=& \Rbf(\pr_A,\widehat{\Lb}, \pr_B, \widehat{\Lb});
\end{aligned} \label{EQNCCintrodef}
\end{align}
see \eqref{DEFricciCoefficients} and \eqref{EQnullcurvatureCOMPDEF} for details. We split $\chi$ and $\chib$ into tracefree and trace parts as follows
\begin{align*} 
\begin{aligned} 
\chi = \chih +\half \trchi \gd, \,\, \chib= \chibh + \half \trchib \gd.
\end{aligned} 
\end{align*}
Moreover, for a $S_{u,v}$-tangential tensor $W$, denote by
\begin{align} 
\begin{aligned} 
DW = \Lied_L W, \,\, \Du W = \Lied_{\Lb} W,
\end{aligned} \label{EQdefintroLIED99901}
\end{align}
where $\Lied$ denotes the projection of the Lie-derivative onto $S_{u,v}$.

\subsection{Null constraint equations} \label{SECintroNullConstraintsANDcharIVP}

\ni A Lorentzian $4$-manifold $(\MM,\g)$ is called a \emph{vacuum spacetime} if it satisfies the Einstein vacuum equations
\begin{align} \label{EQeve}
\mathbf{Ric} =0,
\end{align}
where $\mathbf{Ric}$ denotes the Ricci tensor of $\g$. The Einstein equations \eqref{EQeve} together with the embedding equations for a double null foliation in a vacuum spacetime stipulate that the metric components, Ricci coefficients and null curvature components satisfy the so-called \emph{null structure equations}. These equations are of \emph{transport-elliptic} character, and they are either tangential to the ``outgoing" or the ``ingoing" null hypersurfaces. For example, they include the following (see Section \ref{SECnullstructureequations} for the full set of equations),
\begin{align} 
\begin{aligned} 
D \phi =& \frac{\Om \tr \chi \phi}{2}, \\
D \gd =& 2 \Om \chi, \\
D \trchi + \frac{\Om}{2} (\trchi)^2 - \om \trchi =& - \Om \vert \chih \vert^2_{\gd}, \\
D \eta =& \Om (\chi \cdot \etab - \beta), \\
\Divd \chih -\half \di \tr \chi + \chih \cdot \zeta - \half \trchi \zeta =& - \beta, 
\end{aligned} \label{EQnseINTRO0001}
\end{align}
where $(\Divd \chih)_A := \Nd^C \chih_{AC}$, and $\Nd$ denotes the covariant derivative and $\di$ the exterior derivative on $S_{u,v}$.

Moreover, the so-called \emph{null Bianchi equations} hold for the null curvature components. The null Bianchi equations for $D\be$ and $D\rh$ are as follows.
\begin{align} 
\begin{aligned} 
D \be + \frac{3}{2} \Om \tr \chi \be - \Om \chih \cdot \be - \om \be - \Om \left( \Divd \a + ( \etab + 2 \zeta) \cdot \a \right) =& 0, \\
D \rh + \frac{3}{2} \Om \tr \chi \rh - \Om \left( \Divd \be + (2 \etab + \zeta, \be) - \half (\chibh, \a) \right)=&0,
\end{aligned} \label{EQnseINTRO0002}
\end{align}
where $\Divd \be := \Nd_A \be^A$. We refer to Section \ref{SECnullstructureequations} for the complete set of null constraint equations and null Bianchi equations.

\subsection{The characteristic initial value problem and the gluing problem}

It is well-known that the Einstein equations \eqref{EQeve} are hyperbolic and admit a well-posed initial value formulation. In the context of this paper, in particular the \emph{characteristic initial value problem} where \emph{characteristic initial data} is posed on two transversely-intersecting null hypersurfaces is relevant.

\emph{Characteristic initial data} for the Einstein equations consists of a pair of hypersurfaces $\HH$ and $\HHb$ intersecting at a $2$-dimensional surface $S$ together with the (free) specification of 
\begin{align*} 
\begin{aligned} 
\lrpar{\Omega, \gd_c} \text{ on } \HH \cup \HHb \,\, \text{ and } \lrpar{\gd, \trchi, \trchib, \eta} \text{ on } S,
\end{aligned} 
\end{align*}
such that $\gd_c$ is compatible with $\gd$ on $S$.

The local well-posedness for the characteristic initial value problem was first obtained by Rendall \cite{Rendall}. Specifically, Rendall proved that for sufficiently regular characteristic initial data there exists a unique solution to the Einstein equations in a neighborhood of the surface $S$. Luk \cite{LukChar} subsequently extended the above result to appropriate neighborhoods of the initial hypersurfaces $\HH$ and $\HHb$. 

In particular, by virtue of the null structure equations, characteristic initial data determines on the sphere $S= \HH \cap \HHb$ the following tuple of $S$-tangential tensors on $S$,
\begin{align*} 
\begin{aligned} 
\Big( &\Om,\phi,\gd_c, \chi, \chib, \zeta, \eta, \om, \omb, \a, \be, \rh, \si, \beb, \ab, \\
& D \phi, D\gd_c, D \chi, D \chib, D \zeta, D \eta, D\om, D\omb,  \\
& \Du \phi, \Du\gd_c, \Du \chi, \Du \chib, \Du \zeta, \Du \eta, \Du\om, \Du\omb \Big) \big\vert_{S}.
\end{aligned} 
\end{align*}
We note that the above tuple specifies all derivatives of the spacetime metric up to order $2$ on $S$. By the null structure equations, some quantities are redundant, and we can reduce the above tuple to the following,
\begin{align*} 
\begin{aligned} 
x := \lrpar{\Om,\phi,\gd_c, \Om\trchi, \chih, \Om\trchib, \chibh, \eta, \om, D \om, \omb, \Du \omb, \a, \ab} \big\vert_{S}.
\end{aligned} 
\end{align*}
We call this tuple of tensors $C^2$-\emph{sphere data} (see also Definition \ref{DEFspheredata2}). Higher-order sphere data, suitable for the solution of a higher-regularity gluing problem, requires inclusion of the higher-order tangential 
and transversal derivatives (see also Section \ref{SECnotationHIGHER}). It is essential for this paper that sphere data are affected by gauge transformations (i.e. sphere perturbations and sphere diffeomorphisms).

We are now in position to state the problem of \emph{characteristic  gluing}.\\

\ni \textbf{Characteristic  gluing problem.} \emph{Given sphere data $x_{1}$ and $x_{2}$ on two spacelike $2$-spheres $S_{1}$ and $S_{2}$, respectively, does there exist a null hypersurface $\HH_{[1,2]} = \cup_{1\leq v \leq 2} S_v$ with a family of sphere data $(x'_v)_{1\leq v \leq 2}$ solving the null constraint equations on $\HH_{[1,2]}$ such that }
\begin{align*} 
\begin{aligned} 
x'_1 = x_{1} \text{ \emph{and} } x'_2 = x_{2}\text{?}
\end{aligned} 
\end{align*}
The degrees of freedom in the characteristic gluing problem are the free prescription of 
\begin{align*} 
\begin{aligned} 
(\Om, \gd_c) \text{ along } \HH_{[1,2]}.
\end{aligned} 
\end{align*}

\ni Characteristic gluing in the above generality is not always feasible. Indeed, for example, the monotonicity of $\trchi$ due to the Raychauduri equation \eqref{EQRaychauduri1} forms an obstacle to characteristic gluing. In this paper we analyse the characteristic gluing problem close to Minkowski. In fact, by applying the implicit function theorem, we can reduce to a study of the \emph{linearized characteristic gluing problem} at Minkowski.

\subsection{First statement of main theorem} \label{SECintroStatementMainTheorem}

\ni In the following we state a first version of our main theorem. First, on a sphere $S$ equipped with a round metric $\gac$ we define the projections of functions $f$ on $S$ onto the (normalized) spherical harmonics of mode $l=0$ and $l=1$ as follows,
\begin{align*} 
\begin{aligned} 
f^{(0)} :=  \int\limits_S  f \cdot Y^{(00)} d\mu_{\gac}, \,\, f^{(1m)} :=  \int\limits_S  f \cdot Y^{(1m)} d\mu_{\gac} \text{ for } m=-1,0,1.
\end{aligned} 
\end{align*}  
Moreover, we define the projections of vectorfields $X$ on $S$ onto the electric $E^{(1m)}$ and magnetic $H^{(1m)}$ vector spherical harmonics of mode $l=1$ (defined in Appendix \ref{SECfourierSpheres}) as follows,
\begin{align*} 
\begin{aligned} 
X_E^{(1m)} := \int\limits_{S} \gac\lrpar{X,E^{(1m)}} d\mu_\gac, \,\, X_H^{(1m)} := \int\limits_{S} \gac\lrpar{X,H^{(1m)}} d\mu_\gac \text{ for } m=-1,0,1.
\end{aligned} 
\end{align*}

\begin{definition}[Charges] \label{DEFchargesINTRO119911} Let $x$ be given sphere data on a sphere $S$. We additionally assume that $S$ is equipped with a round metric $\gac$. For $m=-1,0,1$, we define the charges
\begin{align*} 
\begin{aligned}  
\mathbf{E} :=& -\frac{1}{8\pi} \sqrt{4\pi} \lrpar{r^3 \lrpar{ \rho + r \Divd {\be}}}^{(0)}, \\
\mathbf{P}^m :=& -\frac{1}{8\pi} \sqrt{\frac{4\pi}{3}} \lrpar{r^3 \lrpar{\rho + r \Divd {\be}}}^{(1m)},\\
\mathbf{L}^m :=& \frac{1}{16\pi} \sqrt{\frac{8\pi}{3}} \lrpar{r^3 \lrpar{ \di \trchi + \trchi (\eta-\di\log\Om) }}_H^{(1m)},\\
\mathbf{G}^m :=&  \frac{1}{16\pi}\sqrt{\frac{8\pi}{3}} \lrpar{r^3 \lrpar{ \di \trchi + \trchi (\eta-\di\log\Om) }}^{(1m)}_E,
\end{aligned} 
\end{align*}
where $r=r(x)$ denotes the area radius of the sphere $S$ with sphere data $x$, and $\rh$ and $\beta$ are calculated from the sphere data by the null structure equations.
\end{definition}

\ni The following is the main result of this paper, see Theorem \ref{PROPNLgluingOrthA121} for a precise statement.
\begin{theorem}[Codimension-$10$ perturbative characteristic gluing, version 1] \label{THMmain} Let $x_{1}$ and $x_{0,2}$ be sphere data on two spheres $S_{1}$ and $S_{0,2}$, close to sphere data on the round spheres of radius $1$ and $2$ in Minkowski, respectively. For a real number $\de>0$, let ${\HHb}_{2}= \cup_{-\de \leq u \leq \de} S_{u,2}$ be an ingoing null hypersurface passing through $S_{0,2}$, equipped with a family of sphere data $(x_{u,2})_{-\de \leq u \leq \de}$ that is close to the respective sphere data in Minkowski and solves the null constraint equations. Then there are
\begin{itemize}
\item a null hypersurface $\HH'_{[1,2]}=\cup_{1\leq v \leq 2} S'_{v}$ with a family of sphere data $(x'_{v})_{1\leq v \leq 2}$ solving the null constraint equations,
\item sphere data $x'_{0,2}$ on a sphere $S'_{0,2}$ stemming from a perturbation of ${S}_{0,2}$ in ${\HHb}_{2}$ (and subject to a sphere diffeomorphism),
\end{itemize}
such that on $S_{1}$ we have the matching
\begin{align} 
\begin{aligned} 
x'_1 = x_{1},
\end{aligned} \label{EQmatchingMainTheoremVersion1}
\end{align}
and on $S'_{2}$ we have matching of $x'_2$ and $x'_{0,2}$ up to the charges $(\mathbf{E},\mathbf{P},\mathbf{L},\mathbf{G})$, that is, if it holds that
\begin{align*} 
\begin{aligned} 
(\mathbf{E},\mathbf{P},\mathbf{L},\mathbf{G})(x'_2) = (\mathbf{E},\mathbf{P},\mathbf{L},\mathbf{G})(x'_{0,2}),
\end{aligned} 
\end{align*}
then it holds that
\begin{align*} 
\begin{aligned} 
x'_2 = x'_{0,2}.
\end{aligned} 
\end{align*}
There is no additional obstruction to gluing higher-order tangential derivatives along $\HH'_{[1,2]}$.
\end{theorem}

\ni \emph{Remarks on Theorem \ref{THMmain}.}
\begin{enumerate}

\item The proof of Theorem \ref{THMmain} is based on the implicit function theorem and a solution of the \emph{linearized characteristic gluing problem} transversal to the obstruction space consisting of $(\mathbf{E},\mathbf{P},\mathbf{L}, \mathbf{G})$, see Section \ref{SEClinearizedProblem}.

\item The gluing of Theorem \ref{THMmain} is at the level of $C^2$ for the metric components. The characteristic gluing of higher-order \emph{tangential} derivatives is stated in more detail in Theorem \ref{THMHIGHERorderLgluingMAINTHEOREM}. For higher-order gluing of derivatives in \emph{all} directions, see the \emph{bifurcate characteristic gluing} in Section \ref{SECintroTransversalDerivativesStatement} below.

\end{enumerate}

\subsection{Linearized characteristic gluing} \label{SEClinearizedCHARgluing999intro} By the implicit function theorem, the study of the characteristic gluing problem in vicinity of Minkowski can be reduced to the study of the linearized characteristic gluing problem at Minkowski. In this section we discuss the linearized null constraint equations and null Bianchi equations at Minkowski, and the corresponding linearized gluing problem.

\subsubsection{Linearized equations, characteristic gluing and conserved charges} \label{SEClineqslincharglui} The linearized null constraint equations on $\HH_{[1,2]} = \cup_{1\leq v \leq 2} S_v$ can be derived by varying through a family of sphere data $(x^{\varep}_{v})_{1\leq v \leq 2}$ solving the null constraint equations around Minkowski. We formally denote its expansion in the parameter $\varep$ as follows,
\begin{align*} 
\begin{aligned} 
x^\varep_v =& \lrpar{1,v, \gac, \frac{2}{v}, 0, -\frac{2}{v}, 0, 0, 0, 0, 0,0, 0, 0} \Big\vert_{S_v} \\
&+ \varep \cdot \lrpar{\dot{\Om},\phid, \gdcd, \omtrchid, \chihd, \omtrchibd, \chibhd, \dot{\eta}, \dot{\om}, \dot{D \om}, \dot{\omb}, \dot{\Du \omb}, \dot{\a}, \dot{\ab}} \Big\vert_{S_v} + \OO(\varep^2).
\end{aligned} 
\end{align*}
The resulting linearized null constraint equations and linearized null Bianchi equations are equations for the \emph{linearized sphere data}
\begin{align*} 
\begin{aligned} 
\lrpar{\dot{\Om},\phid, \gdcd, \omtrchid, \chihd, \omtrchibd, \chibhd, \dot{\eta}, \dot{\om}, \dot{D \om}, \dot{\omb}, \dot{\Du \omb}, \dot{\a}, \dot{\ab}}.
\end{aligned} 
\end{align*}
For example, linearizing the null constraint equations \eqref{EQnseINTRO0001} and null Bianchi equations \eqref{EQnseINTRO0002} yields (where we have that $r=v$ and $D= \Lied_{\pr_v}$ in Minkowski space, and denote $\Divdo = \Divd_\gac$),
\begin{align} 
\begin{aligned} 
D\lrpar{\frac{\phid}{r}} = \frac{\omtrchid}{2}, \,\,  D\gdcd = \frac{2}{r^2} \chihd, \,\, D \lrpar{D\phid - 2 \Omd}=0,
\end{aligned} \label{EQlinIntro9991}
\end{align}
and
\begin{align} 
\begin{aligned} 
D \lrpar{r^2 \etad} + \frac{r^2}{2} \di \lrpar{\omtrchid-\frac{4}{r} \Omd} -2r \di\Omd -\Divdo \chihd =&0.
\end{aligned} \label{EQlinIntro9992}
\end{align}
We also note the linearized null Bianchi equations
\begin{align} 
\begin{aligned} 
D\lrpar{r^3 \dot{\beta}} -r \Divdo \ad =0, \,\, D\lrpar{r^3\dot{\rho}} -r \Divdo \betad =0,
\end{aligned} \label{EQlinIntro9992Bianchi1}
\end{align}

\ni The \emph{linearized characteristic gluing problem} can be stated as follows.\\

\ni \textbf{Linearized characteristic gluing problem.} \emph{Given linearized sphere data $\xd_{1}$ and $\xd_2$ on two spheres $S_1$ and $S_2$, respectively, does there exists a null hypersurface $\HH_{[1,2]} = \cup_{1\leq v \leq 2} S_v$ equipped with a family of linearized sphere data $(\xd'_v)_{1\leq v \leq 2}$ solving the linearized constraint equations such that}
\begin{align*} 
\begin{aligned} 
\xd'_1 = \xd_1 \text{ \emph{and} } \xd'_2 = \xd_2\text{?}
\end{aligned} 
\end{align*}
The degrees of freedom in the linearized characteristic gluing problem are given by prescribing $\Omd$ and $\gdcd$ on $\HH_{[1,2]}$. By the linearized first variation equation in \eqref{EQlinIntro9991}, that is,
\begin{align*} 
\begin{aligned} 
D\gdcd = \frac{2}{r^2} \chihd,
\end{aligned} 
\end{align*}
this is equivalent to the free prescription of $\Omd$ and $\chihd$ on $\HH_{[1,2]}$, which is the point-of-view we choose in this paper. 
 
In the following we analyse the obstacles to the linearized gluing problem. Combining the linearized constraint equations \eqref{EQlinIntro9991} and \eqref{EQlinIntro9992}, we get that
\begin{align} 
\begin{aligned} 
D\lrpar{r^2 \etad + \frac{r^3}{2} \di \lrpar{\omtrchid -\frac{4}{r} \Omd }} =&\Divdo \chihd, \\
D\lrpar{\frac{r}{2} \lrpar{\omtrchid-\frac{4}{r} \Omd} + \frac{\phid}{r}} =&0.
\end{aligned} \label{EQlinIntro9993}
\end{align}

\ni Importantly, projecting the second equation onto the vector spherical harmonics of mode $l=1$ (see Appendix \ref{SECfourierSpheres}) and using that, in general, the mode $l=1$ of the divergence of a symmetric tracefree $2$-tensor vanishes, we can read off \eqref{EQlinIntro9993} that the quantities
\begin{align} 
\begin{aligned} 
\QQ_0 :=& \lrpar{r^2 \etad + \frac{r^3}{2} \di \lrpar{\omtrchid -\frac{4}{r} \Omd }}^{[1]},\\
\QQ_1 :=& \frac{r}{2} \lrpar{\omtrchid-\frac{4}{r} \Omd} + \frac{\phid}{r},
\end{aligned} \label{EQrefINTROCHARGES0101999}
\end{align}
are \emph{conserved} along $\HH_{[1,2]}$ under the linearized null constraint equations. We note that $\QQ_0$ corresponds to $6$ numbers, while $\QQ_1$ accounts to one functional degree on the sphere.

The charges $\QQ_0$ and $\QQ_1$ are examples of the larger set of conserved \emph{charges} 
\begin{align*} 
\begin{aligned} 
\QQ_i \text{ for } 0\leq i \leq 7,
\end{aligned} 
\end{align*}
identified in Section \ref{SEClinearizedCHARGESMinkowskiSTEF89} of this paper. These charges are of fundamental importance for the characteristic gluing problem, as they form \emph{obstructions to gluing}. In particular, the linearizations of the charges $\mathbf{E}, \mathbf{P}, \mathbf{L}$ and $\mathbf{G}$ (see Definition \ref{DEFchargesINTRO119911}) form part of the set of conserved charges. In fact, $\QQ_0$ is directly related to the linearizations of $\mathbf{L}$ and $\mathbf{G}$, see \eqref{EQREMchargeslinatMinkowski}.

We remark that we alternatively could have derived $\QQ_0$ from the linearized null Bianchi equation \eqref{EQlinIntro9992Bianchi1} for $D\dot{\beta}$ by relating $\dot{\beta}$ to $\dot{\eta}$ via the linearized Gauss-Codazzi equation. Similarly, the conservation law for the linearizations of $\mathbf{E}$ and $\mathbf{P}$ can be derived by the linearized null Bianchi equation \eqref{EQlinIntro9992Bianchi1} for $D\dot{\rho}$ by means of the linearized Gauss and Gauss-Codazzi equations.

\subsubsection{Linearized perturbations of sphere data and matching of charges} \label{SEClinSpherepertChargeSplitIntro} In context of the linearized characteristic gluing problem, we also analyze the linearizations of sphere perturbations and sphere diffeomorphisms of sphere data. We remark that formally, the perturbation $S_2'$ of $S_2$ in the ingoing null hypersurface $\HHb_2=\cup_{-\de \leq u \leq \de} S_{u,2}$ is defined as the level set $$\{ f(u) = 0 \} \subset \HHb_2,$$ for a small \emph{perturbation function} $f$ defined on $\HHb_2$. We parametrize the sphere diffeomorphisms by two scalar functions $(j^1,j^2)$.

The linearization of the sphere data in the perturbation function $f$, evaluated at Minkowski and $f=0$, as well as the linearization of the sphere data in $(j^1,j^2)$ can be explicitly calculated by the transformation formulas for sphere data (see Appendix \ref{SECproofTEClemmasmoothness} and Lemmas \ref{LEMlinearizedTransversal} and \ref{LEMspherediffLIN}). For example, the linearization on $S_2$ of $\Om$, $\trchi$, $\phi$ and $\eta$ under the transversal perturbation function $f$, evaluated at Minkowski and $f=0$, is given by
\begin{align*} 
\begin{aligned} 
\Omd = \half \pr_u \dot{f}, \,\, \omtrchid = \half \lrpar{\Ldo+1} \dot{f}, \,\, \phid = -\dot{f}, \,\, \etad =  \di \lrpar{\pr_u \dot{f} +\frac{\dot{f}}{2}}.
\end{aligned} 
\end{align*}

\ni Plugging the above into the charge expressions \eqref{EQrefINTROCHARGES0101999}, we see that the charges $\QQ_0$ and $\QQ_1$ on $S_2$ change under linearized sphere perturbations as follows,
\begin{align} 
\begin{aligned} 
\QQ_0 =& 0, & \QQ_1 =& \frac{1}{2}\Ldo \dot{f}- \pr_{u} \dot{f}.
\end{aligned} \label{EQchargeIntroperturbation999012}
\end{align}
The identity \eqref{EQchargeIntroperturbation999012} reflects the following essential observation: The set of conserved charges $\QQ_i$, $0\leq i \leq 7$, splits into two categories:
\begin{enumerate}
\item \textbf{Gauge-invariant charges.} A $10$-dimensional space of \emph{gauge-invariant charges} which are not changing under linearized sphere perturbations. This space is spanned precisely by the linearizations of $\mathbf{E}, \mathbf{P}, \mathbf{L}$ and $\mathbf{G}$.
\item \textbf{Gauge-dependent charges.} An infinite-dimensional space of \emph{gauge-dependent charges} which can be adjusted in a surjective manner by a carefully chosen linearized sphere perturbations and sphere diffeomorphisms. All charges $\QQ_i$, $0\leq i \leq 7$ except for $\mathbf{E}, \mathbf{P}, \mathbf{L}$ and $\mathbf{G}$ fall into this category.
\end{enumerate}

\ni Therefore, in the linearized characteristic gluing problem, we can match all charges at $S_2$ -- except for $\mathbf{E}, \mathbf{P}, \mathbf{L}$ and $\mathbf{G}$ -- by adding a linearized sphere perturbation and a linearized sphere diffeomorphism at $S_2$.

In \cite{ACR3,ACR2} we show that the charges $\mathbf{E}, \mathbf{P}, \mathbf{L}$ and $\mathbf{G}$ are related to the ADM energy, linear momentum, angular momentum, and center-of-mass of an asymptotically flat spacetime.

\subsubsection{Hierarchical structure of radial weights in the characteristic gluing problem} In the previous section we showed that matching of all gauge-dependent charges is possible by adding a linearized sphere perturbation to $S_2$. In this section we explain how to prescribe, in addition to the matching of the gauge-dependent charges, the linearized free data $\Omd$ and $\chihd$ on $\HH_{[1,2]}$ such that on $S_2$ we have matching of the full linearized sphere data up to the $10$-dimensional space of gauge-invariant charges.

By integrating the linearized null constraint equations and using their nilpotent character (see Section \ref{SECconformalMethod1112}), we can derive representation formulas for each linearized quantity. For example, integrating \eqref{EQlinIntro9991} and \eqref{EQlinIntro9993} from $v'=1$ to $v'=v$, we get the following representation formulas for $\phid, \gdcd$ and $\etad$,
\begin{align*} 
\begin{aligned} 
\phid(v)-v\phid(1)-\frac{v-1}{2}\lrpar{\omtrchid(1)-4\Omd(1)} =& 2 \int\limits_1^v \Omd dv', & r\text{-weight for $\Omd$: }&1, \\
\dot{\gd_c}(v) -\gdcd(1)=&2 \int\limits_{1}^v \frac{1}{r^2} \chihd dv', & r\text{-weight for $\chihd$: }&\frac{1}{r'^2}\\
\left[ v'^2 \etad + \frac{v'^3}{2} \di\lrpar{\omtrchid - \frac{4}{v'} \Omd}\right]_1^v =& \Divdo\lrpar{ \int\limits_1^v \chihd dv' }, & r\text{-weight for $\chihd$: }&1.
\end{aligned} 
\end{align*}

\ni Importantly, the representation formulas display a special \emph{hierarchical structure of radial weights} where the integrals on the right-hand sides over the freely prescribed data $\Omd$ and $\chihd$ contain \emph{different $r$-weights}. Thereby the integrals are \emph{linearly independent} and it is possible, by prescribing the value of the weighted integrals of $\Omd$ and $\chihd$ over the interval $v=1$ to $v=2$, to choose $\Omd$ and $\chihd$ on $\HH_{[1,2]}$ such that gluing is achieved. 

The existence of conservation laws is connected to the presence of similar $r$-weights as follows: If the representation formulas for linearized quantities include only integrals of $\Omd$ and $\chihd$ of the same $r$-weight, then a conserved charge can be constructed from them.

Using the above principle, we can prescribe the free data along $\HH_{[1,2]}$ to glue transversely to the space of charges. As we matched the gauge-dependent charges by a linearized sphere perturbation in Section \ref{SEClinSpherepertChargeSplitIntro}, it follows that we glued the linearized sphere data on $S_2$ up to the $10$-dimensional space of gauge-invariant charges.

\subsection{Solution of the non-linear characteristic gluing problem} \label{SECintroNONLINEARgluing} The proof of Theorem \ref{THMmain} is based on the implicit function theorem and our analysis of the linearized characteristic gluing problem in Section \ref{SEClinearizedCHARgluing999intro}.

The setup for the implicit function theorem is as follows. Consider
\begin{itemize}
\item sphere data $x_1$ on a sphere $S_1$,
\item a family of sphere data $(x_{u,2})_{-\de \leq u \leq \de}$ on the ingoing null hypersurface ${\HHb}_{2}= \cup_{-\de \leq u \leq \de} S_{u,2}$,
\item a family of sphere data $(x'_v)_{1\leq v \leq2}$ on the outgoing null hypersurface $\HH_{[1,2]}=\cup_{1\leq v \leq 2} S_v$,
\item a sphere perturbation function $f$, and sphere diffeomorphism functions $(j^1,j^2)$.
\end{itemize}
Then we define the mapping $\FF$ as follows,
\begin{align*} 
\begin{aligned} 
\FF: \, (x_1, (x_{u,2})_{-\de \leq u \leq \de}, f, (j_1,j_2), (x'_v)_{1\leq v \leq2}) \mapsto \lrpar{x'_1 - x_1, \mathfrak{M}\lrpar{x_2'}-\mathfrak{M}\lrpar{\PP_{(j^1,j^2)}\PP_f((x_{u,2})_{-\de \leq u \leq \de})}, \CC\lrpar{x'_v}_{1\leq v \leq2} },
\end{aligned} 
\end{align*}
where 
\begin{itemize}
\item $\CC$ denotes the null constraint functions (as defined in Section \ref{SECconformalMethod1112})
\item $\mathfrak{M}(x)$ denotes the projection of sphere data $x$ on a space of codimension $10$ which accounts for the charges $\mathbf{E},\mathbf{P}, \mathbf{L}$ and $\mathbf{G}$ (see Definition \ref{DEFmatchingMAP} and Lemma \ref{LEMconditionalMATCHING}),
\item $\PP_f$ and $\PP_{(j^1,j^2)}$ denote the application of sphere perturbations and sphere diffeomorphisms, respectively.
\end{itemize}

\ni By the definition of $\FF$ it holds that if
\begin{align} 
\begin{aligned} 
\FF(x_1, (x_{u,2})_{-\de \leq u \leq \de}, f, (j_1,j_2), (x'_v)_{1\leq v \leq2}) =\lrpar{0,0,0}.
\end{aligned} \label{EQgluingINTROproblem123}
\end{align}
then the family of sphere data $(x'_v)_{1\leq v \leq2}$ solves the null constraint equations on $\HH_{[1,2]}$, agrees with $x_1$ on $S_1$, and matches -- up to $\mathbf{E},\mathbf{P},\mathbf{L}$ and $\mathbf{G}$ -- with a sphere perturbation and sphere diffeomorphism of $x_2$ on $S_2$. This corresponds to solving the characteristic gluing problem as outlined in Theorem \ref{THMmain}. In the following we use the implicit function theorem to construct $f$, $(j^1,j^2)$ and $(x'_v)_{1\leq v \leq2}$ for given $x_1$ and $(x_{u,2})_{-\de \leq u \leq \de}$ such that \eqref{EQgluingINTROproblem123} holds.

The implicit function theorem implies that if the linearization of $\FF$ in $f$, $(j^1,j^2)$ and $(x'_v)_{1\leq v \leq2}$, evaluated at $f=0$ and Minkowski reference data, is \emph{surjective}, then there exists a mapping $\GG$,
\begin{align*} 
\begin{aligned} 
\GG: (x_1,(x_{u,2})_{-\de \leq u \leq \de}) \mapsto (f, (j^1,j^2), (x'_v)_{1\leq v \leq 2}),
\end{aligned} 
\end{align*}
well-defined close to Minkowski values, such that $\GG(x_1,x_2)$ solves the gluing problem \eqref{EQgluingINTROproblem123}, that is, it holds that
\begin{align*} 
\begin{aligned} 
\FF(x_1,(x_{u,2})_{-\de \leq u \leq \de}, \GG(x_1,(x_{u,2})_{-\de \leq u \leq \de})) =(0,0,0).
\end{aligned} 
\end{align*}

\ni By construction, the surjectivity of the linearization of $\FF$ is equivalent to the solvability of the linearized characteristic gluing problem for the \emph{inhomogeneous} linearized null constraint equations. The latter can be shown by a slight generalization of the analysis of the homogeneous linearized equations in Section \ref{SEClinearizedCHARgluing999intro}. We remark that our derived estimates for solutions to the (inhomogeneous) linearized null constraint equations follow a specific \emph{regularity hierarchy} which is also reflected in our definition of function spaces. 

More generally, the above implicit function argument applies to the study of the characteristic gluing problem near Schwarzschild of small mass $M\geq0$. This is essential for our study of characteristic gluing to Kerr in \cite{ACR3,ACR2}. 

\subsection{Codimension-$10$ bifurcate characteristic gluing}
\label{SECintroTransversalDerivativesStatement} 

\ni In our solution to the characteristic gluing problem along $\HH_{[1,2]}$, the gluing of higher-order tangential derivatives is in fact without obstacles. However, higher-order \emph{transversal} derivatives cannot be glued in general as they are related to higher-order conserved charges along $\HH_{[1,2]}$ of the linearized null constraint equations. 

We show in this paper that it is possible to circumvent these conservation laws and glue derivatives of any direction and any order by gluing along two null hypersurfaces $\HHb_{[-1,0],1}$ and $\HH_{-1,[1,2]}$ bifurcating from an auxiliary spacelike sphere $S_{-1,1}$, see Figure \ref{FIG111} below.

\begin{figure}[H]
\begin{center}
\includegraphics[width=9.5cm]{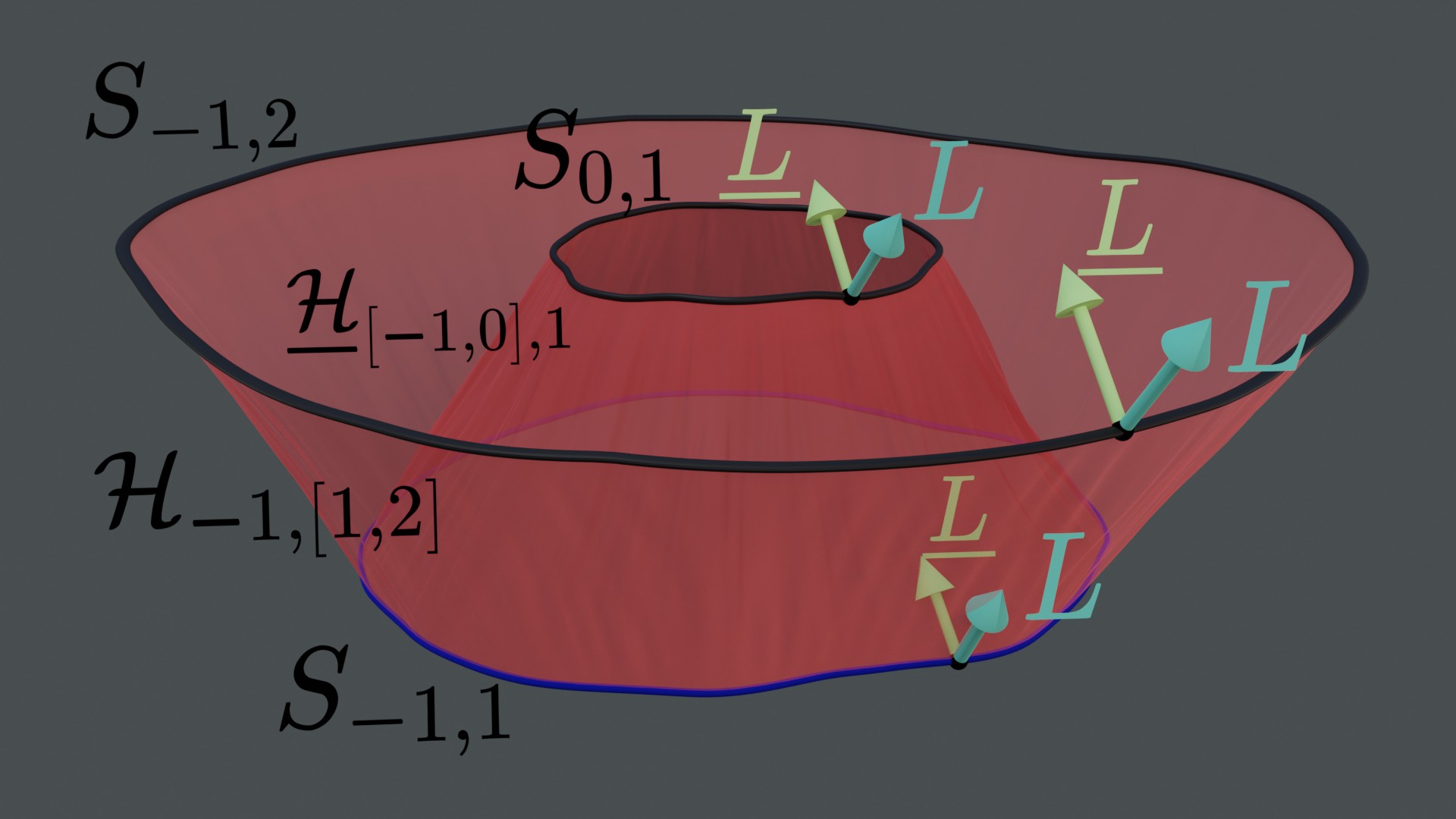} 
\caption{The null hypersurfaces $\HHb_{[-1,0],1}$ and $\HH_{-1,[1,2]}$ bifurcating from $S_{-1,1}$.}\label{FIG111}
\end{center}
\end{figure}
\vspace{-0.2cm}

\ni Our result can be summarized as follows, see Theorem \ref{THMtransversalHIGHERv2} for a precise statement.
\begin{theorem}[Codimension-$10$ bifurcate characteristic gluing, version 1] \label{THMtransversalHIGHER}
Let $m\geq0$ be an integer. Let $x_{0,1}$ and $x_{-1,2}$ be sphere data on two spheres $S_{0,1}$ and $S_{-1,2}$ together with derivatives in all directions up to order $m$. If the prescribed data on $S_{0,1}$ and $S_{-1,2}$ are sufficiently close to the respective Minkowski data, then there are families of sphere data 
\begin{align*} 
\begin{aligned} 
(\underline{x}'_{u,1})_{-1\leq u \leq 0} &\text{ on } \HHb_{[-1,0],1} = \bigcup\limits_{-1\leq u \leq 0} S_{u,0}, \\
(x'_{-1,v})_{1\leq v \leq 2} &\text{ on } \HH_{-1,[1,2]} = \bigcup\limits_{1\leq v \leq 2} S_{-1,v},
\end{aligned} 
\end{align*}
solving the null constraint equations and matching to $m^{\mathrm{th}}$-order at the bifurcate auxiliary sphere $S_{-1,1}$, such that we have 
\begin{itemize}
\item $m^{\mathrm{th}}$-order matching on $S_{0,1}$, 
\item $m^{\mathrm{th}}$-order matching up to the charges $\mathbf{E}, \mathbf{P}, \mathbf{L}$ and $\mathbf{G}$ on $S_{-1,2}$.
\end{itemize}

\end{theorem}

\ni The proof of Theorem \ref{THMtransversalHIGHER} is based on the implicit function theorem and a study of the \emph{linearized bifurcate characteristic gluing problem}. It is important to remark that in Theorem \ref{THMtransversalHIGHER} we are not applying any sphere perturbations to $S_{0,1}$ and $S_{-1,2}$. The key insight is that the gauge-dependent charges along $\HH_{-1,[1,2]}$ can be matched by adjusting the \emph{free data} on $\HHb_{[-1,0],1}$, and vice versa, the gauge-dependent charges on $\HHb_{[-1,0],1}$ can be matched by the \emph{free data} on $\HH_{-1,[1,2]}$. Moreover, the spaces of gauge-invariant charges along $\HH_{-1,[1,2]}$ and $\HHb_{[-1,0],1}$ agree. We refer to Section \ref{SECHIGHERfull} for a detailed discussion. 

Theorem \ref{THMtransversalHIGHER} is applied in \cite{ACR3,ACR2} to glue \emph{spacelike} initial data for the Einstein equations to spacelike initial data for a Kerr black hole spacetime.

\subsection{Overview of the paper} The paper is structured as follows.
\begin{itemize}
\item In Section \ref{sec:GeometryOfNullHypersurfaces} we introduce the notation and the geometric setup of this paper.
\item In Section \ref{SECpreciseStatementMainTheorem} we precisely state the main results of this paper.
\item In Section \ref{SEClinearizedProblem} we solve the linearized codimension-$10$ characteristic gluing problem.
\item In Section \ref{SECLinearizedNullConstraintsAroundMinkowski} we prove the codimension-$10$ perturbative characteristic gluing, Theorem \ref{THMmain}.
\item In Section \ref{SECHIGHERfull} we prove the codimension-$10$ bifurcate characteristic gluing, see Theorem \ref{THMtransversalHIGHER}.
\item In Appendix \ref{SECproofTEClemmasmoothness} we rigorously define and estimate non-linear perturbations of sphere data.
\item In Appendix \ref{SECnulltransportAPPENDIX} we derive and linearize null transport equations along $\HH$.
\item In Appendix \ref{SEClinearizedCHARGEequationsSSAPP} we study linearized null transport equations at Schwarzschild of small mass $M\geq0$.
\item In Appendix \ref{SECellEstimatesSpheres} we recall the theory of Hodge systems on $2$-spheres and tensor spherical harmonics, and provide a spectral analysis of differential operators studied in Section \ref{SEClinearizedProblem}.
\end{itemize}
\subsection{Acknowledgements} S.A. acknowledges support through the NSERC grant 502581 and the Ontario Early Researcher Award. S.C. acknowledges support through the NSF grant DMS-1439786 of the Institute for Computational and Experimental Research in Mathematics (ICERM). I.R. acknowledges support through NSF grants DMS-2005464, DMS-1709270 and a Simons Investigator Award.

\section{Notation, definitions and preliminaries} \label{sec:GeometryOfNullHypersurfaces}

\ni In this section we introduce the notation, definitions and preliminaries of this paper. We follow the notation of \cite{ChrFormationBlackHoles}. For two real numbers $A$ and $B$, the inequality $A \les B$ means that there is a universal constant $C>0$ such that $A \leq C \, B$. Greek indices range over $\a=0,1,2,3$, lowercase Latin indices over $a=1,2,3$ and uppercase Latin indices over $A=1,2$.

\subsection{Null geometry} \label{SECsetupdoublenull}
In this section we recapitulate the well-known construction of local double null coordinates in spacetimes, see, for example, Chapter 1 of \cite{ChrFormationBlackHoles}. 

Given a spacetime $(\MM,\g)$, let $\D$ denote the covariant derivative and $\Rbf$ the Riemann curvature tensor on $(\MM,\g)$. Let $S \subset \MM$ be a spacelike $2$-sphere and let $L'$ on $S$ be an outgoing future-pointing null vectorfield normal to $S$. Given a scalar function $\Om$ on $S$, the so-called \emph{null lapse}, let $\Lb'$ denote the unique ingoing future-pointing null vectorfield normal to $S$ such that 
\begin{align} 
\begin{aligned} 
\g\lrpar{L',\Lb'}=-2\Om^{-2}.
\end{aligned} \label{EQnormalisation}
\end{align}
Extend $L'$ and $\Lb'$ from $S$ as null geodesic vectorfields onto $(\MM,\g)$, and denote the resulting outgoing and ingoing null geodesic congruences by $\HH$ and $\HHb$, respectively.

Given a null lapse $\Om$ on $\HH$ and $\HHb$ which extends the null lapse $\Om$ on $S$, define the vectorfields
\begin{align} 
\begin{aligned} 
\widehat{L}:= \Om L', \,\, L := \Om^2 L' \text{ on } \HH, \,\, \widehat{\Lb}:=\Om \Lb', \,\, \Lb := \Om^2 \Lb' \text{ on } \HHb.
\end{aligned} \label{EQdefNullPairsINM}
\end{align} 
Define on $\HH$ the scalar function $v$ by 
\begin{align*} 
\begin{aligned} 
L(v)= 1 \text{ on } \HH, \,\, v \vert_S = 1,
\end{aligned} 
\end{align*}
and define on $\HHb$ the scalar function $u$ by
\begin{align*} 
\begin{aligned} 
\Lb(u)= 1 \text{ on } \HHb, \,\, u \vert_S = 0.
\end{aligned} 
\end{align*}
Let $S_{0,v}\subset \HH $ and $S_{u,1}\subset \HHb$ denote the level sets of $v$ and $u$, respectively. On each $S_{0,v}$, define $\Lb'$ as the unique ingoing future-pointing null vectorfield normal to $S_{0,v}$ such that \eqref{EQnormalisation} holds. Similarly, on each $S_{u,1}$, define $L'$ as the unique outgoing future-pointing null vectorfield normal to $S_{u,1}$ such that \eqref{EQnormalisation} holds. Extend $\Lb'$ from $\HH$ and $L'$ from $\HHb$, respectively, as null geodesic vectorfields onto $(\MM,\g)$. 

Subsequently, define $\Om$ in $(\MM,\g)$ by \eqref{EQnormalisation} and define $(\widehat{L}, \widehat{\Lb})$ and $(L,\Lb)$ in $\MM$ by \eqref{EQdefNullPairsINM}. Furthermore, define $v$ and $u$ in $\MM$ by 
\begin{align*} 
\begin{aligned} 
\Lb(v)=0, \,\, L(u)=0,
\end{aligned} 
\end{align*}
with initial values given by the construced $v$ on $\HH$ and $u$ on $\HHb$, respectively. Denote the level sets of $u$ and $v$ in $\MM$ by $\HH_u$ and $\HHb_v$, respectively, and let
\begin{align*} 
\begin{aligned} 
S_{u,v} := \HH_u \cap \HHb_v, \,\, \HH_{u,[v_1,v_2]} := \HH_u \cap \lrpar{\bigcup_{v_1\leq v \leq v_2} \HHb_v}, \,\, \HHb_{[u_1,u_2],v} := \lrpar{\bigcup_{u_1\leq u \leq u_2} \HH_u} \cap \HHb_v,
\end{aligned} 
\end{align*}
and let $\gd$ denote the induced Riemannian metric on $S_{u,v}$ and $\Nd$ the induced covariant dervative.

We are now in position to define the so-called \emph{double null coordinates} $(u,v,\th^1,\th^2)$ on $\MM$. First, define local coordinates $(\th^1,\th^2)$ on each $S_{0,v} \subset \HH$ by transporting local coordinates $(\th^1,\th^2)$ on $S=S_{0,1}$ along $\HH$ according to
\begin{align*} 
\begin{aligned} 
L(\th^1)=L(\th^2)=0 \text{ on } \HH,
\end{aligned} 
\end{align*}
and then define the local coordinates $(\th^1,\th^2)$ on $\MM$ by transporting $(\th^1,\th^2)$ according to
\begin{align*} 
\begin{aligned} 
\Lb(\th^1)=\Lb(\th^2)=0 \text{ on } \MM,
\end{aligned} 
\end{align*}
with given initial values on $\HH=\HH_0$.

The following is shown in Chapter 1 of \cite{ChrFormationBlackHoles}. 
\begin{itemize}
\item The functions $u$ and $v$ are local optical functions on $\MM$, that is, they satisfy the Eikonal equations
\begin{align*} 
\begin{aligned} 
\vert \D u \vert^2 =0, \,\, \vert \D v \vert^2 =0,
\end{aligned} 
\end{align*}
and it holds that
\begin{align*} 
\begin{aligned} 
L' = -2 \D u, \,\, \Lb' = -2 \D v.
\end{aligned} 
\end{align*}

\item In double null coordinates $(u,v,\th^1,\th^2)$, the Lorentzian metric $\g$ takes the form
\begin{align} 
\begin{aligned} 
\g = -4 \Om^2 du dv + \gd_{AB} \lrpar{d\th^A + b^A dv}\lrpar{d\th^B + b^B dv},
\end{aligned} \label{EQmetricexpressionDoubleNULLintro7788}
\end{align}
where the $S_{u,v}$-tangential vectorfield $b=b^A \pr_A$ is called \emph{shift vector} and satisfies by construction
\begin{align*} 
\begin{aligned} 
b=0 \text{ on } \HH_0.
\end{aligned} 
\end{align*}
\item It holds that
\begin{align} 
\begin{aligned} 
L= \pr_v + b, \,\, \Lb = \pr_u.
\end{aligned} \label{EQvectorexpr2488}
\end{align}
\end{itemize}

\ni Introduce furthermore the following notation (following \cite{ChrFormationBlackHoles}).

\begin{itemize}
\item On a sphere $(S_{u,v},\gd)$, denote by $r$ the \emph{area radius} defined by
\begin{align*} 
\begin{aligned} 
4\pi r^2 = \mathrm{area}_\gd\lrpar{S_{u,v}}.
\end{aligned} 
\end{align*}
\item By the coordinates $(\th^1,\th^2)$ on $S_{u,v,}$, we can equip each $S_{u,v}$ with the unit round metric 
\begin{align} 
\begin{aligned} 
\gac:= \lrpar{d\th^1}^2 + \sin^2 \th^1 \lrpar{d\th^2}^2,
\end{aligned} \label{EQdefRoundUnitMetric}
\end{align}
and define $\ga:= (v-u)^2\gac$ on $S_{u,v}$ for $v> u$. By construction, $\gac$ is invariant under the flow of $L$ on $\HH$ and under the flow of $\Lb$ on $\MM$. Denote the volume forms of $\gd$ and $\gac$ on $S_{v}$ by $\sqrt{\gd}$ and $\sqrt{\gac}$, respectively. 
\item We decompose the metric $\gd$ on $S_{u,v}$ into
\begin{align} \label{EQdefdecomposition1}
\gd = \phi^2 \gd_c,
\end{align}
where $\phi$ is a scalar function and $\gd_c$ is a Riemannian metric on $S_{u,v}$ given by
\begin{align*} 
\begin{aligned} 
\phi^2 := \frac{\sqrt{\gd}}{\sqrt{\gac}}, \,\, \gd_c := \phi^{-2} \gd.
\end{aligned} 
\end{align*}
By definition it holds that $\sqrt{\gd_c}=\sqrt{\gac}$.

\item On each $S_{u,v}$ define with respect to $(\th^1,\th^2)$ 

\begin{itemize}
\item the standard (real) spherical harmonics $Y^{(lm)}$ for $l\geq0$ and $-l \leq m \leq l$, 
\item the vector spherical harmonics $E^{(lm)}$ and $H^{(lm)}$ for $l\geq1$ and $-l \leq m \leq l$, 
\item the tensor spherical harmonics $\phi^{(lm)}$ and $\psi^{(lm)}$ for $l\geq2$ and $-l \leq m \leq l$. 
\end{itemize}
We refer to Appendix \ref{SECellEstimatesSpheres} for details and properties of spherical harmonics.

\item For a general $S_{u,v}$-tangent tensorfield $W$, introduce the notation
\begin{align}
DW:= \Lied_L W, \,\, \underline{D}W:= \Lied_\Lb W, \label{EQdefDnotation}
\end{align}
where $\Lied$ denotes the projection of the Lie derivative on $(\MM,\g)$ onto the tangent space of $S_{u,v}$.

\item For $S_{u,v}$-tangent vectorfields $X$ and $Y$, define the \emph{Ricci coefficients} by
\begin{align} \begin{aligned}
\chi(X,Y) :=& \g(\D_X \widehat{L},Y), & \chib(X,Y) :=& \g(\D_X \widehat{\Lb},Y), \\
\zeta(X) :=& \half \g(\D_X \widehat{L}, \widehat{\Lb}), & \underline{\zeta}(X) :=& \half \g(\D_X \widehat{\Lb}, \widehat{L}), \\
\eta :=& \zeta + \di \log \Om , & \etab :=& -\zeta + \di \log \Om, \\
\om :=& D \log \Om, & \omb :=& \Du \log \Om,
\end{aligned}\label{DEFricciCoefficients} \end{align}
where $\di$ denotes the extrinsic derivative of $S_{u,v}$. It holds that 
\begin{align} 
\begin{aligned} 
\underline{\zeta}=-\zeta, \,\, \etab=-\eta+2\di\log\Om.
\end{aligned} \label{EQriccirelationetabeta} 
\end{align}

\item For $S_{u,v}$-tangent vectorfields $X$ and $Y$, define the \emph{null curvature components} by
\begin{align} \begin{aligned}
\alpha(X,Y) :=& \Rbf(X,\widehat{L}, Y, \widehat{L}), & \beta(X) :=& \half \Rbf(X, \widehat{L},\widehat{\Lb},\widehat{L}), \\
\rh :=& \frac{1}{4} \Rbf(\widehat{\Lb}, \widehat{L}, \widehat{\Lb}, \widehat{L}), & \sigma \iin(X,Y) :=& \half \Rbf(X,Y,\widehat{\Lb}, \widehat{L}), \\
\beb(X) :=& \half \Rbf(X, \widehat{\Lb},\widehat{\Lb},\widehat{L}), & \ab(X,Y) :=& \Rbf(X,\widehat{\Lb}, Y, \widehat{\Lb}).
\end{aligned}\label{EQnullcurvatureCOMPDEF} \end{align}

\end{itemize}
\subsection{Null structure equations and null Bianchi equations} \label{SECnullstructureequations} By the null geometry setup in Section \ref{SECsetupdoublenull} and the Einstein equations, the metric components \eqref{EQdefRoundUnitMetric}, Ricci coefficients \eqref{DEFricciCoefficients} and null curvature components \eqref{EQnullcurvatureCOMPDEF} satisfy the so-called \emph{null structure equations}. Before stating them, we introduce the following notation, following Chapter 1 of \cite{ChrFormationBlackHoles}.
\begin{itemize}
\item For two $S_{u,v}$-tangential $1$-forms $X$ and $Y$,
\begin{align*} \begin{aligned}
(X,Y):=& \gd(X,Y), & ({}^\ast X)_A :=& \in_{AB}X^B, \\
(X \widehat{\otimes} Y)_{AB} :=& X_A Y_B + X_B Y_A - (X \cdot Y)\gd_{AB}, & \Divd X :=& \Nd^A X_A, \\
 (\Nd \widehat{\otimes} Y)_{AB} :=& \Nd_A Y_B + \Nd_B Y_A - (\Divd Y)\gd_{AB}, & \Curld X :=& \in^{AB}\Nd_A X_B,
\end{aligned}\end{align*}
where $\in$ denotes the area $2$-form of $S_{u,v}$.
\item For two symmetric $S_{u,v}$-tangential $2$-tensors $V$ and $W$, 
\begin{align*}
\tr V := \gd^{AB} V_{AB}, \,\, \widehat{V} := V - \half \tr V \gd, \,\, V \wedge W := \ind^{AB} V_{AC}W^C_{\,\,\,B}.
\end{align*}
\item For a symmetric $S_{u,v}$-tangential $2$-tensor $V$ and a $1$-form $X$,
\begin{align*} 
\begin{aligned} 
(V \cdot X)_A := V_{AB}X^B.
\end{aligned} 
\end{align*}
\item For a symmetric $S_{u,v}$-tangential  $2$-tensor $V$,
\begin{align*} 
\begin{aligned} 
\Divd V_A := \Nd^B V_{BA}.
\end{aligned} 
\end{align*}
\item For a symmetric $S_{u,v}$-tangential tensor $W$, let $\widehat{D}W$ and $\widehat{\Du}W$ denote the tracefree parts of $DW$ and $\Du W$, respectively, with respect to $\gd$.
\end{itemize}

\ni In this paper we also use the operators $\DDd_1, \DDd_2, \DDd_1^{\ast}$ and $\DDd_2^\ast$ which are introduced in Appendix \ref{SECellEstimatesSpheres}. 

We are now in position to state the \emph{null structure equations} of a spacetime. We have the first variation equations,
\begin{align} \begin{aligned}
D \gd =& 2 \Om \chi, & \Du \gd =& 2 \Om \chib,
\end{aligned} \label{EQfirstvariation1}\end{align}
which imply specifically that
\begin{align} 
\begin{aligned} 
D \phi = \frac{\Om \tr \chi \phi}{2}, \,\, \Du \phi = \frac{\Om\tr\chib \phi}{2},
\end{aligned} \label{EQusefulDphiRELATION}
\end{align}
the Raychauduri equations, 
\begin{align} \begin{aligned}
D \trchi + \frac{\Om}{2} (\trchi)^2 - \om \trchi =& - \Om \vert \chih \vert^2_{\gd},& \Du \trchib + \frac{\Om}{2} (\trchib)^2 - \omb \trchib =& - \Om \vert \chibh \vert^2_{\gd},
\end{aligned}\label{EQRaychauduri1}\end{align}
and
\begin{align} \begin{aligned}
D\chih =& \Om \vert \chih \vert^2 \gd + \om \chih - \Om \a, &\Du \chibh =& \Om \vert \chibh \vert^2 \gd + \omb \chibh - \Om \ab,\\
D \eta =& \Om (\chi \cdot \etab - \beta), & \Du \etab =& \Om (\chib \cdot \eta + \beb), \\
D \omb =& \Om^2(2 (\eta, \etab) - \vert \eta \vert^2 -\rh), & \Du \om=& \Om^2(2 (\eta,\etab) - \vert \etab \vert^2 - \rh), \\
\Curld \eta=& - \half \chih \wedge \chibh - \si, & \Curld \etab =& - \Curld \eta = - \Curld \zeta, \\
D\etab =& - \Om (\chi \cdot \etab -\be) + 2 \di \om, & \Du\eta =& - \Om (\chib \cdot \eta + \beb) + 2 \di \omb.
\end{aligned} \label{EQtransportEQLnullstructurenonlinear}\end{align}
Further we have the Gauss equation, 
\begin{align} \label{EQGaussEquation}
K + \frac{1}{4} \tr \chi \tr \chib - \half (\chih,\chibh) = - \rh,
\end{align}
where $K$ denotes the Gauss curvature of $S_{u,v}$, the Gauss-Codazzi equations
\begin{align} \begin{aligned}
\Divd \chih -\half \di \tr \chi + \chih \cdot \zeta - \half \trchi \zeta =& - \beta,\\
\Divd \chibh - \half \di \trchib -\chibh \cdot \zeta +\half \trchib \zeta =& \beb,
\end{aligned}\label{EQgausscodazzinonlinear1}\end{align}
and 
\begin{align}\begin{aligned}
D (\Om \trchib) =& 2 \Om^2 \Divd \etab + 2 \Om^2 \vert \etab \vert^2 - \Om^2 (\chih, \chibh) - \half \Om^2 \trchi \trchib + 2 \Om^2 \rh,\\ 
\Du (\Om \trchi) =& 2 \Om^2 \Divd \eta + 2 \Om^2 \vert \eta \vert^2 - \Om^2 (\chih, \chibh) - \half \Om^2 \trchi \trchib + 2 \Om^2 \rh,
\end{aligned} \label{EQtransporttrchitrchib1}\end{align}
as well as
\begin{align} \begin{aligned}
D(\Om \chibh) =& \Om^2\lrpar{(\chih, \chibh) \gd + \half \trchi \chibh + \Nd \widehat{\otimes}\etab + \etab \widehat{\otimes}\etab - \half \trchib \chih}, \\
\Du(\Om \chih) =& \Om^2 \lrpar{(\chih, \chibh)\gd + \half \trchib \chih + \Nd \widehat{\otimes} \eta + \eta \widehat{\otimes} \eta - \half \trchi \chibh}.
\end{aligned} \label{EQchihequations1} \end{align} 

\ni By Proposition 1.2 in \cite{ChrFormationBlackHoles}, the following \emph{null Bianchi equations} hold,
\begin{align} \begin{aligned}
\widehat{\Du} \a - \half \Om \tr \chib \a + 2 \omb \a + \Om \left( -\Nd \widehat{\otimes} \be - (4 \eta + \zeta) \widehat{\otimes} \be + 3 \chih \rh + 3 {}^\ast \chih \si \right) =&0, \\
\widehat{D} \aa - \half \Om \tr \chi \aa + 2 \om \aa + \Om \left( \Nd \widehat{\otimes} \beb + (4 \etab - \zeta) \widehat{\otimes} \beb + 3 \chibh \rh -3 {}^\ast \chibh \si \right) =&0, \\
D \be + \frac{3}{2} \Om \tr \chi \be - \Om \chih \cdot \be - \om \be - \Om \left( \Divd \a + ( \etab + 2 \zeta) \cdot \a \right) =& 0, \\
\Du \beb + \frac{3}{2} \Om \tr \chib \beb - \Om \chibh \cdot \beb - \omb \beb + \Om \left( \Divd \aa + (\eta-2 \zeta) \cdot \aa \right)=&0, \\
\Du \be + \half \Om \tr \chib \be - \Om \chibh \cdot \be + \omb \be - \Om \left(\di \rh + {}^\ast \di \si +  3 \eta \rh +3 {}^\ast \eta \si + 2 \chibh \cdot \beb \right)=&0, \\
D \beb + \half \Om \tr \chi \beb - \Om \chih \cdot \beb + \om \beb + \Om \left( \di \rh- {}^\ast \di \si + 3 \etab \rh - 3 {}^\ast \etab \si - 2 \chibh \cdot \be \right) =&0, \\
D \rh + \frac{3}{2} \Om \tr \chi \rh - \Om \left( \Divd \be + (2 \etab + \zeta, \be) - \half (\chibh, \a) \right)=&0, \\
\Du \rh + \frac{3}{2} \Om \tr \chib \rh+ \Om \left( \Divd \beb + (2 \eta- \zeta, \beb)+ \half (\chih,\aa) \right) =&0, \\
D \si + \frac{3}{2} \Om \tr \chi \si + \Om \left( \Curld \be + (2 \etab+\zeta, {}^\ast \be)- \half \chibh  \wedge \a \right)=&0, \\
\Du \si + \frac{3}{2} \Om \tr \chib \si + \Om \left( \Curld \beb + (2 \eta-\zeta, {}^\ast \beb) + \half \chih \wedge \aa \right)=&0.
\end{aligned} \label{EQnullBianchiEquations}\end{align}

\ni In addition to the above null structure equations, the following transport equation for $\Du\omb$ is derived in Appendix \ref{SECappendixDerivationTransveralEquations}, 
\begin{align} 
\begin{aligned} 
D\Du\omb =&12 \Om^3 (\di\log\Om-\eta)\omb +2 \Om^2 \omb \lrpar{2(\eta,\etab)- \vert \eta\vert^2} + \rh \lrpar{\frac{3}{2}\Om^3\trchib -2\Om^2 \omb} \\
&+12 \Om^3 \chib(\eta, \eta-\di\log\Om) + \Om^3 (\beb, 7\eta-3\di\log\Om)+ \Om^3 \Divd \beb + \half \Om^3 (\chih, \ab).
\end{aligned} \label{EQDUOMU1}
\end{align}

\ni Similar equations for higher derivatives can be derived by commuting the above equations with $D, \Du, \Nd$.

\subsection{Null geometry of Minkowski and Schwarzschild spacetimes} \label{SECminkowskiDEF} 

\ni In this section we discuss the null geometry of the Minkowski spacetime and the Schwarzschild family of spacetimes.\\

\ni \textbf{Minkowski spacetime.} Minkowski spacetime, the trivial solution to the Einstein vacuum equations \eqref{EQeve}, is given by
\begin{align*} 
\begin{aligned} 
\MM= \RRR^{1+3}, \,\, \g = \mathbf{m}:=\mathrm{diag}(-1,1,1,1).
\end{aligned} 
\end{align*}
From the Cartesian coordinates $(t,x^1,x^2,x^3)$ on $\RRR^{1+3}$, the double null coordinates $(u,v,\th^1,\th^2)$ in Minkowski are defined by
\begin{align} 
\begin{aligned} 
(u,v,\th^1,\th^2) := \lrpar{\half \lrpar{t-r},\half \lrpar{t+r}, \th^1, \th^2},
\end{aligned}\label{EQdoubleNullMinkowski}
\end{align}
where $r := \sqrt{\sum\limits_{i=1}^3 (x^i)^2}$, and with respect to which
\begin{align*} 
\begin{aligned} 
\mathbf{m} = -4 du dv + (v-u)^2 \gac_{AB} d\th^A d\th^B.
\end{aligned} 
\end{align*}
We note that the area radius of the sphere $S_{u,v}$ is given by $r=v-u.$

In coordinates \eqref{EQdoubleNullMinkowski}, the metric components, Ricci coefficients and null curvature components are given on $S_{u,v}$ by, with $r=v-u$ and $\gac$ as in \eqref{EQdefRoundUnitMetric},
\begin{align} 
\begin{aligned} 
\Om =& 1, & \gd=& r^2\gac, & & && && \\
\trchi=&\frac{2}{r}, & \trchib=&-\frac{2}{r}, & \chih=&0, &\chibh=&0, &&&& \\
\eta=&0, & \etab=& 0, & \zeta=&0, &\underline{\zeta}=&0, &&&& \\
\om=& 0, & D\om=& 0, & \omb=& 0, & \Du\omb=& 0, &&&& \\
\a=&0, &\be=&0, &\beb=&0, & \ab=&0, & \rh=& 0, & \si=&0.
\end{aligned} \label{EQspheredataMINKOWKSI001}
\end{align}

\ni \textbf{Schwarzschild family of spacetimes.} For real numbers $M \in \RRR$, the family of Schwarzschild metrics is given in Schwarzschild coordinates $(t,r,\th^1,\th^2)$ by (see for example \cite{HawkingEllis})
\begin{align} 
\begin{aligned} 
\g= -\lrpar{1-\frac{2M}{r}} dt^2 + \lrpar{1-\frac{2M}{r}}^{-1} dr^2 + r^2 \lrpar{d\th^2 + \sin^2\th d\phi^2}.
\end{aligned} \label{EQdefSSmetric}
\end{align}
Setting $M=0$ leads back to Minkowski spacetime, while $M>0$ is interpreted as a black hole solution with event horizon at $\{r=2M\}$. The so-called exterior region $\{r>2M\}$ is covered by Eddington-Finkelstein double null coordinates $(u,v,\th^1,\th^2)$ in which the metric takes the form
\begin{align} 
\begin{aligned} 
\g = -4 \lrpar{1-\frac{2M}{r_M(u,v)}} du dv + r_M(u,v)^2\gac_{CD} d\th^C d\th^D,
\end{aligned} \label{EQSSformulaMETRIC}
\end{align}
where the area radius $r_M(u,v)$ is defined by (see for example (98) in \cite{DHR})
\begin{align} 
\begin{aligned} 
\frac{v-u}{2M} = \frac{r_M(u,v)}{2M} + \log \lrpar{\frac{r_M(u,v)}{2M}-1}.
\end{aligned} \label{EQdefRbyUV}
\end{align}

\ni The area radius function $r_M(u,v)$ is smooth in $M$ away from $M=0$, and continuous in $M$ at $M=0$. The corresponding null lapse $\Om_M$ is determined by \eqref{EQSSformulaMETRIC} to be
\begin{align} 
\begin{aligned} 
\Om_M^2 =1-\frac{2M}{r_M},
\end{aligned} \label{EQSSnullLapseFormula}
\end{align}
where we note the following standard identities,
\begin{align} 
\begin{aligned} 
\pr_v \Om_M = \frac{\Om_M M}{r_M^2}, \,\, \pr_u \Om_M = -\frac{\Om_M M}{r_M^2}, \,\, \pr_v r_M = \Om_M^2, \,\, \pr_u r_M =-\Om_M^2.
\end{aligned} \label{EQSSdoublenullEFstandardRelations}
\end{align}

\ni Using \eqref{EQSSnullLapseFormula}, \eqref{EQSSdoublenullEFstandardRelations} and that by \eqref{EQdefNullPairsINM} and \eqref{EQvectorexpr2488},
\begin{align*} 
\begin{aligned} 
\widehat{L}= \Om^{-1} \pr_v, \,\, \widehat{L}= \Om^{-1} \pr_u,
\end{aligned} 
\end{align*}
we calculate that in Eddington-Finkelstein double null coordinates $(u,v,\th^1,\th^2)$, the metric components, Ricci coefficients and null curvature components are given on $S_{u,v}$ by
\begin{align} 
\begin{aligned} 
\Om_M =& \sqrt{1-\frac{2M}{r_M}}, & \gd=& r_M^2\gac, & & && \\
\trchi=&\frac{2\Om_M}{r_M}, & \trchib=&-\frac{2\Om_M}{r_M}, & \chih=&0, &\chibh=&0, \\
\eta=&0, & \etab=& 0, & \zeta=&0, &\underline{\zeta}=&0, \\
\om=& \frac{M}{r_M^2}, & D\om=& -\frac{2M}{r_M^3}\Om_M^2, & \omb=& -\frac{M}{r_M^2}, & \Du\omb=& -\frac{2M}{r_M^3}\Om_M^2, \\
\a=&0, &\be=&0, &\beb=&0, & \ab=&0, \\
\rh=& -\frac{2M}{r_M^3}, & \si=&0, && &&
\end{aligned} \label{EQspheredataSSM111222}
\end{align}
for $v>u$ such that $r_M=r_M(u,v)>2M$.

\subsection{Tensor spaces and calculus estimate} \label{SECfunctionSpacesPreliminary} In this section we define the basic function spaces of this paper.

\begin{definition}[Tensor spaces on Riemannian $2$-spheres $S_{u,v}$] \label{DEFsphereSPACES} For two real numbers $v\geq u$, let $S_{u,v}$ be a $2$-sphere equipped with a round unit metric $\gac$. For integers $m\geq0$ and tensors $T$ on $S_{u,v}$, define
\begin{align*} 
\begin{aligned} 
\Vert T \Vert^2_{H^m(S_{u,v})} :=   \sum\limits_{i=1}^m  \left\Vert \Nd^i T \right\Vert^2_{L^2(S_{u,v})},
\end{aligned} 
\end{align*}
where the covariant derivative $\Nd$ and the measure in $L^2(S_{u,v})$ are with respect to the round metric $\ga=(v-u)^2\gac$. Moreover, let 
\begin{align*} 
\begin{aligned} 
H^m(S_{u,v}) := \{ T: \Vert T \Vert_{H^m(S_{u,v})} < \infty \}.
\end{aligned} 
\end{align*}
\end{definition}

\begin{definition}[Tensor spaces on null hypersurfaces] \label{DEFnullHHspaces} Let $m\geq0$ and $l\geq0$ be two integers. In the following let $D$ and $\Du$ be defined as in \eqref{EQdefDnotation} for null hypersurfaces in Minkowski. 
\begin{enumerate}

\item For real numbers $u_0 < v_1 < v_2$ and $S_{u_0,v}$-tangential tensors $T$ on $\HH_{u_0,[v_1,v_2]}$, define
\begin{align*} 
\begin{aligned} 
\Vert T \Vert^2_{H^m_l(\HH_{u_0,[v_1,v_2]})} :=  \int\limits_{v_1}^{v_2} \,\, \sum\limits_{1\leq j\leq l} \left\Vert D^j T \right\Vert^2_{H^m(S_{u_0,v})} dv,
\end{aligned} 
\end{align*}
and let
\begin{align*} 
\begin{aligned} 
H^m_l(\HH_{u_0,[v_1,v_2]}) := \{ F: \Vert F \Vert_{H^m_l(\HH_{u_0,[v_1,v_2]})} < \infty\}.
\end{aligned} 
\end{align*}

\item For real numbers $u_1 < u_2 < v_0$ and $S_{u,v_0}$-tangential tensors $T$ on $\HHb_{[u_1,u_2],v_0}$, define
\begin{align*} 
\begin{aligned} 
\Vert T \Vert^2_{H^m_l(\HHb_{[u_1,u_2],v_0})} :=  \int\limits_{u_1}^{u_2} \,\, \sum\limits_{1\leq j\leq l} \left\Vert \Du^j T \right\Vert^2_{H^m(S_{u,v_0})} du,
\end{aligned} 
\end{align*}
and let
\begin{align*} 
\begin{aligned} 
H^m_l(\HHb_{[u_1,u_2],v_0}) := \{ F: \Vert F \Vert_{H^m_l(\HHb_{[u_1,u_2],v_0})} < \infty\}.
\end{aligned} 
\end{align*}

\end{enumerate}
\end{definition}

\ni The following standard calculus estimates are applied tacitly throughout this paper. They follow, for example, from Corollaries 3.3 and 3.4, and Lemma 3.20 in \cite{J3}, and the results of Chapter 13 in \cite{TaylorPDE3}. 
\begin{lemma}[Calculus estimates] \label{LEMstandardTRACE} Let $u_0<v_1<v_2$ be real numbers. The following holds.

\begin{enumerate}

\item \textbf{Trace estimate.} For any $S_{u_0,v}$-tangent tensor $T$ on $\HH_{u_0,[v_1,v_2]}$, we have that for $v_1\leq v \leq v_2$,
\begin{align*} 
\begin{aligned} 
\Vert T \Vert_{H^0(S_{u_0,v})} \leq& C_{u_0,v_1,v_2} \cdot \Vert T \Vert_{H^0_1(\HH_{u_0,[v_1,v_2]})},
\end{aligned} 
\end{align*}
where the constant $C_{u_0,v_1,v_2}>0$ depends on $u_0, v_1$ and $v_2$. 

\item \textbf{$L^\infty$-estimate.} For any $S_{u_0,v}$-tangent tensor $T$ on $\HH_{u_0,[v_1,v_2]}$, we have that
\begin{align*} 
\begin{aligned} 
\Vert T \Vert_{L^\infty(\HH_{u_0,[v_1,v_2]})} \leq& C_{u_0,v_1,v_2} \cdot \Vert T \Vert_{H^2_1(\HH_{u_0,[v_1,v_2]})},
\end{aligned} 
\end{align*}
where the constant $C_{u_0,v_1,v_2}>0$ depends on $u_0, v_1$ and $v_2$. 

\item \textbf{Product estimate.} Let $m_1, m_2 \geq 2$ and $l_1, l_2 \geq 1$ be integers, and further let $T \in H^{m_1}_{l_1}(\HH_{u_0,[v_1,v_2]})$ and $T' \in H^{m_2}_{l_2}(\HH_{u_0,[v_1,v_2]})$ be two $S_{u_0,v}$-tangent tensors. Then it holds that for integers $0\leq m\leq \mathrm{min}(m_1,m_2)$ and $0\leq l \leq \mathrm{min}(l_1,l_2)$,
\begin{align*} 
\begin{aligned} 
\Vert T \cdot T' \Vert_{H^m_l(\HH_{u_0,[v_1,v_2]})} \leq& C \cdot \Vert T\Vert_{H^m_l(\HH_{u_0,[v_1,v_2]})} \cdot \Vert T' \Vert_{H^{m_2}_{l_2}(\HH_{u_0,[v_1,v_2]})} \\
&+ C \cdot \Vert T\Vert_{H^{m_1}_{l_1}(\HH_{u_0,[v_1,v_2]})} \cdot \Vert T' \Vert_{H^{m}_{l}(\HH_{u_0,[v_1,v_2]})},
\end{aligned} 
\end{align*}
where the constant $C>0$ depends on $m,m_1,m_2,l,l_1$ and $l_2$.
\end{enumerate}
\end{lemma}

\subsection{Sphere data, null data and norms} \label{SECdefFirstOrderSphereData} In this section we set up the essential definitions for the characteristic gluing problem.

\begin{definition}[$C^2$-sphere data] \label{DEFspheredata2} For two real numbers $v\geq u$, $C^2$-sphere data $x_{u,v}$ consists of a $2$-sphere $S_{u,v}$ equipped with a round metric $\gac$, see \eqref{EQdefRoundUnitMetric}, and the following tuple of tensors on $S_{u,v}$,
\begin{align*} 
\begin{aligned} 
x_{u,v} = (\Om,\gd, \Om\trchi, \chih, \Om\trchib, \chibh, \eta, \om, D\om, \omb, \Du\omb, \a, \ab),
\end{aligned} 
\end{align*}
where
\begin{itemize}
\item $\Om>0$ is a positive scalar function and $\gd$ is a Riemannian metric, 
\item $\Om\trchi, \Om\trchib,\om, D\Om, \omb, \Du\omb$ are scalar functions,
\item $\eta$ is a vectorfield,
\item $\chih$, $\chibh$, $\a$ and $\ab$ are symmetric $\gd$-tracefree $2$-tensors.
\end{itemize}
\end{definition}

\ni \emph{Remarks on Definition \ref{DEFspheredata2}.}
\begin{enumerate}

\item By the specification of $\gac$, it follows that sphere data is coordinate-dependent. More generally, the sphere data induced on a spacelike $2$-sphere in a spacetime is gauge-dependent; see also Section \ref{SECdefEquivalenceFirstOrderSphereData}.

\item By the null structure equations and null Bianchi equations of Section \ref{SECnullstructureequations}, sphere data $x_{u,v}$ fully determines on $S_{u,v}$ the Ricci coefficients and null curvature components
\begin{align*} 
\begin{aligned} 
\lrpar{\etab, \zeta, \underline{\zeta}}, \,\, (\be,\rh,\si,\beb),
\end{aligned} 
\end{align*}
as well as the derivatives
\begin{align*} 
\begin{aligned} 
&\lrpar{D\eta, D\etab, D\zeta, D\chi, D\chib, D\omb}, \,\, \lrpar{ \Du \eta, \Du\etab, \Du\zeta, \Du\chi, \Du\chib, \Du\om }, \\
&\lrpar{D\beta, D\rh, D\si, D\beb, D\ab}, \lrpar{\Du\beb, \Du\si,\Du\rh,\Du\be,\Du\a}.
\end{aligned} 
\end{align*}
\end{enumerate}

\textbf{Notation.} In this paper, we denote the \emph{reference Minkowski sphere data} by
\begin{align*} 
\begin{aligned} 
\mathfrak{m}_{u,v} = \lrpar{1,r^2 \gac, \frac{2}{r}, 0, -\frac{2}{r}, 0, 0, 0, 0, 0, 0, 0, 0},
\end{aligned} 
\end{align*}
where $r=v-u$, see \eqref{EQspheredataMINKOWKSI001}. For real numbers $M$, we denote the \emph{reference Schwarzschild sphere data} by
\begin{align*} 
\begin{aligned} 
\mathfrak{m}^M_{u,v} = \lrpar{\Om_M,r_M^2 \gac, \frac{2\Om_M}{r_M}, 0, -\frac{2\Om_M}{r_M}, 0, 0, \frac{M}{r_M^2}, -\frac{2M\Om_M^2}{r_M^3}, -\frac{M}{r_M^2}, -\frac{2M\Om_M^2}{r_M^3}, 0, 0},
\end{aligned} 
\end{align*}
where $r_M=r_M(u,v)$ is defined in \eqref{EQdefRbyUV} and $\Om_M=(1-\frac{2M}{r_M})^{1/2}$ in \eqref{EQSSnullLapseFormula}, see \eqref{EQspheredataSSM111222}.

\begin{definition}[Norm for sphere data] \label{DEFnormFirstOrderDATA} Let $x_{u,v}$ be sphere data. Define 
\begin{align*} 
\begin{aligned} 
\Vert x_{u,v} \Vert_{\XX(S_{u,v})} :=& \Vert \Om \Vert_{H^{6}(S_{u,v})} +\Vert \gd \Vert_{H^{6}(S_{u,v})} + \Vert \Om\trchi \Vert_{H^{6}(S_{u,v})} + \Vert \chih \Vert_{H^{6}(S_{u,v})}\\
& + \Vert \Om\trchib \Vert_{H^{4}(S_{u,v})} + \Vert \chibh \Vert_{H^{4}(S_{u,v})} + \Vert \eta \Vert_{H^{5}(S_{u,v})} \\
&+ \Vert \om \Vert_{H^{6}(S_{u,v})}+ \Vert D\om \Vert_{H^{6}(S_{u,v})}+\Vert \omb \Vert_{H^{4}(S_{u,v})} + \Vert \Du\omb \Vert_{H^{2}(S_{u,v})} \\
&+ \Vert \a \Vert_{H^{6}(S_{u,v})} +\Vert \ab \Vert_{H^{2}(S_{u,v})},
\end{aligned} 
\end{align*}
where the norms are with respect to the round metric $\ga= (v-u)^2 \gac$ on $S_{u,v}$. Let
\begin{align} 
\begin{aligned} 
\XX(S_{u,v}) := \{ x_{u,v} : \Vert x_{u,v} \Vert_{\XX(S_{u,v})} < \infty\}.
\end{aligned} \label{EQdefinitionSpacesProjection}
\end{align}
\end{definition}
\emph{Remarks on Definition \ref{DEFnormFirstOrderDATA}.}
\begin{itemize}
\item Definition \ref{DEFnormFirstOrderDATA} reflects the regularity hierarchy of the null structure equations along the $L$-direction.
\end{itemize}

\begin{definition}[Null data] \label{DEFnulldata111} We define the following.
\begin{enumerate}
\item For real numbers $u_0 < v_1 <v_2$, \emph{outgoing null data} on $\HH_{u_0,[v_1,v_2]}$ is given by a tuple of $S_{u_0,v}$-tangent tensors 
\begin{align} 
\begin{aligned} 
x=(\Om,\gd,\Om\trchi, \chih, \Om\trchib, \chibh, \eta, \om, D\om, \omb, \Du\omb, \a, \ab),
\end{aligned} \label{EQcollectionofTensors11122}
\end{align}
such that $x_{u_0,v} := x \vert_{S_{u_0,v}}$ is sphere data on each $S_{u_0,v} \subset \HH_{u_0,[v_1,v_2]}$.

\item For real numbers $u_1 < u_2 <v_0$, \emph{ingoing null data} on $\HHb_{[u_1,u_2],v_0}$ is given by a tuple of $S_{u,v_0}$-tangent tensors 
\begin{align} 
\begin{aligned} 
\underline{x}=(\Om,\gd,\Om\trchi, \chih, \Om\trchib, \chibh, \etab, \om, D\om, \omb, \Du\omb, \a, \ab),
\end{aligned} \label{EQcollectionofTensors111223}
\end{align}
such that $x_{u,v_0} := x \vert_{S_{u,v_0}}$ is sphere data on each $S_{u,v_0} \subset \HHb_{[u_1,u_2],v_0}$.
\end{enumerate}

\end{definition}

\ni \textbf{Notation.} The reference outgoing and ingoing null data of Minkowski is denoted by $\mathfrak{m}$ and $\underline{\mathfrak{m}}$, respectively; see \eqref{EQspheredataMINKOWKSI001}. The reference outgoing and ingoing null data of Schwarzschild of mass $M$ is denoted by $\mathfrak{m}^M$ and $\underline{\mathfrak{m}}^M$, respectively; see \eqref{EQspheredataSSM111222}. \\

\ni The following norm for null data respects the \emph{regularity hierarchy} of the null structure equations.
\begin{definition}[Norm for null data] \label{DEFnormHH} Let $x$ be null data on $\HH:=\HH_{u_0,[v_1,v_2]}$. Define
\begin{align*} 
\begin{aligned}
\Vert x \Vert_{\XX(\HH)} :=& \Vert \Om \Vert_{H^6_3(\HH)} +\Vert \gd \Vert_{H^6_3(\HH)}+ \Vert \Om\trchi \Vert_{H^6_3(\HH)}+\Vert \chih \Vert_{H^6_2(\HH)}\\
& + \Vert \Om\trchib \Vert_{H^4_2(\HH)}+ \Vert \chibh \Vert_{H^4_3(\HH)}+ \Vert \eta \Vert_{H^5_2(\HH)}\\
&+\Vert \om \Vert_{H^6_2(\HH)} +\Vert D\om \Vert_{H^6_1(\HH)}+ \Vert \omb \Vert_{H^4_3(\HH)}+ \Vert \Du\omb \Vert_{H^2_3(\HH)}\\
&+  \Vert \a \Vert_{H^{6}_1(\HH)} +\Vert \ab \Vert_{H^{2}_3(\HH)}.
\end{aligned} 
\end{align*}
Let $\underline{x}$ be null data on $\HHb:=\HHb_{[u_0,u_1],v_0}$. Define
\begin{align*} 
\begin{aligned}
\Vert \underline{x} \Vert_{\XX(\HHb)} :=& \Vert \Om \Vert_{H^6_3(\HHb)} +\Vert \gd \Vert_{H^6_3(\HHb)}+ \Vert \Om\trchib \Vert_{H^6_3(\HHb)}+\Vert \chibh \Vert_{H^6_2(\HHb)}\\
& + \Vert \Om\trchi \Vert_{H^4_2(\HHb)}+ \Vert \chih \Vert_{H^4_3(\HHb)}+ \Vert \etab \Vert_{H^5_2(\HHb)}\\
&+\Vert \omb \Vert_{H^6_2(\HHb)} +\Vert \Du\omb \Vert_{H^6_1(\HHb)}+ \Vert \om \Vert_{H^4_3(\HHb)}+ \Vert D\om \Vert_{H^2_3(\HHb)}\\
&+  \Vert \ab \Vert_{H^{6}_1(\HHb)} +\Vert \a \Vert_{H^{2}_3(\HHb)}.
\end{aligned} 
\end{align*}
Moreover, let
\begin{align*} 
\begin{aligned} 
\XX(\HH) := \{ x : \Vert x \Vert_{\XX(\HH)} < \infty \}, \,\, \XX(\HHb) := \{ \underline{x} : \Vert \underline{x} \Vert_{\XX(\HHb)} < \infty \}.
\end{aligned} 
\end{align*}
\end{definition}

\begin{remark} For given null data $x$ on $\HH_{u_0,[v_1,v_2]}$, the null structure equations \eqref{EQtransportEQLnullstructurenonlinear}, \eqref{EQGaussEquation} and \eqref{EQgausscodazzinonlinear1} determine the null curvature components $(\be,\rh,\si,\beb)$. By standard calculus estimates on $S_{u_0,v}$ (see, for example, Lemma \ref{LEMstandardTRACE}), it follows that they are bounded by
\begin{align*} 
\begin{aligned} 
&\Vert \be \Vert_{H^5_2(\HH_{u_0,[v_1,v_2]})}+\left\Vert \rho + \frac{2M}{r_M^3}  \right\Vert_{H^4_2(\HH_{u_0,[v_1,v_2]})}+\Vert \si \Vert_{H^4_2(\HH_{u_0,[v_1,v_2]})}+\Vert \beb \Vert_{H^3_2(\HH_{u_0,[v_1,v_2]})}\\
\les& \Vert x -\mathfrak{m}^M \Vert_{\XX(\HH_{u_0,[v_1,v_2]})}.
\end{aligned} 
\end{align*}
Analogously, for null data $\underline{x}$ on $\HHb_{[u_1,u_2],v_0}$, 
\begin{align*} 
\begin{aligned} 
&\Vert \beb \Vert_{H^5_2(\HHb_{[u_1,u_2],v_0})}+\left\Vert \rho + \frac{2M}{r_M^3}  \right\Vert_{H^4_2(\HHb_{[u_1,u_2],v_0})}+\Vert \si \Vert_{H^4_2(\HH_{u_0,[v_1,v_2]})}+\Vert \be \Vert_{H^3_2(\HHb_{[u_1,u_2],v_0})} \\
\les& \Vert \underline{x}-\mathfrak{m}^M \Vert_{\XX(\HHb_{[u_1,u_2],v_0})}.
\end{aligned} 
\end{align*}
\end{remark}

\ni In the context of sphere perturbations we work with ingoing null data of higher regularity, see Section \ref{SECdefEquivalenceFirstOrderSphereData} and specifically Propositions \ref{PropositionSmoothnessF} and \ref{PropositionSmoothnessF2}. The corresponding norm for the higher regularity ingoing null data is denoted by $\XX^+$. Similarly to $\XX$ above, $\XX^+$ respects the regularity hierarchy of the null structure equations.

\begin{definition}[Higher regularity norm for ingoing null data] \label{DEFspacetimeNORM} For ingoing null data $\underline{x}$ on $\HHb:= \HHb_{[u_1,u_2],v_0}$ define
\begin{align*} 
\begin{aligned} 
\Vert \underline{x} \Vert_{\XX^+\lrpar{\HHb}} :=& \Vert \Om \Vert_{H^{12}_9(\HHb)} +\Vert \gd \Vert_{H^{12}_9(\HHb)}+ \Vert \Om\trchib \Vert_{H^{12}_9(\HHb)}+\Vert \chibh \Vert_{H^{12}_8(\HHb)}\\
& + \Vert \Om\trchi \Vert_{H^{10}_8(\HHb)}+ \Vert \chih \Vert_{H^{10}_9(\HHb)}+ \Vert \etab \Vert_{H^{11}_8(\HHb)}\\
&+\Vert \omb \Vert_{H^{12}_8(\HHb)} +\Vert \Du\omb \Vert_{H^{12}_7(\HHb)}+ \Vert \om \Vert_{H^{10}_9(\HHb)}+ \Vert D\om \Vert_{H^8_9(\HHb)}\\
&+  \Vert \ab \Vert_{H^{12}_7(\HHb)} +\Vert \a \Vert_{H^{8}_9(\HHb)}.
\end{aligned} 
\end{align*}

\ni Further, let
\begin{align*} 
\begin{aligned} 
\XX^+\lrpar{\HHb} := \left\{ \underline{x}: \Vert \underline{x} \Vert_{\XX^+\lrpar{\HHb}} < \infty\right\}.
\end{aligned} 
\end{align*}
\end{definition}

\subsection{Charges $(\mathbf{E},\mathbf{P},\mathbf{L},\mathbf{G})$ and matching map $\mathfrak{M}$} \label{SECMatchingMapdefinition} 

\ni In this section we define the charges $(\mathbf{E},\mathbf{P},\mathbf{L},\mathbf{G})$ which are of fundamental importance for the characteristic gluing problem, see Theorem \ref{PROPNLgluingOrthA121}, and the matching map $\MMf$ which is used to solve the characteristic gluing problem transversally to the charges.

\begin{definition}[Charges] \label{DEFnonlinearcharges6} Let $x_{u,v}$ be sphere data. For $m=-1,0,1$, define
\begin{align*} 
\begin{aligned}  
\mathbf{E} :=& -\frac{1}{8\pi} \sqrt{4\pi} \lrpar{r^3 \lrpar{ \rho + r \Divd {\be}}}^{(0)}, \\
\mathbf{P}^m :=& -\frac{1}{8\pi} \sqrt{\frac{4\pi}{3}} \lrpar{r^3 \lrpar{\rho + r \Divd {\be}}}^{(1m)},\\
\mathbf{L}^m :=& \frac{1}{16\pi} \sqrt{\frac{8\pi}{3}} \lrpar{r^3 \lrpar{ \di \trchi + \trchi (\eta-\di\log\Om) }}_H^{(1m)},\\
\mathbf{G}^m :=&  \frac{1}{16\pi}\sqrt{\frac{8\pi}{3}} \lrpar{r^3 \lrpar{ \di \trchi + \trchi (\eta-\di\log\Om) }}^{(1m)}_E,
\end{aligned} 
\end{align*}
where $r=r(x_{u,v})$ denotes the area radius of $(S_{u,v},\gd)$ and the spherical harmonics projections are defined with respect to the unit round metric $\gac$ on $S_{u,v}$, see Appendix \ref{SECellEstimatesSpheres}. Here, the null curvature components $\rh$ and $\be$ are calculated from $x_{u,v}$ by \eqref{EQGaussEquation} and \eqref{EQgausscodazzinonlinear1}. 
\end{definition}

\emph{Remarks on Definition \ref{DEFnonlinearcharges6}.}
\begin{enumerate}

\item The linearizations of $(\mathbf{E},\mathbf{P},\mathbf{L},\mathbf{G})$ at Minkowski satisfy conservation laws along $\HH$, see Section \ref{SEClinearizedProblem} and \eqref{EQREMchargeslinatMinkowski}. In \cite{ACR3,ACR2} we show that in asymptotically flat spacetimes, these conservation laws are related to the \emph{conservation of energy}, \emph{linear momentum}, \emph{angular momentum}, and the \emph{equation of motion for the center-of-mass}.

\item It holds that on the sphere $S_{u,v}$,
\begin{align*} 
\begin{aligned} 
(\mathbf{E},\mathbf{P},\mathbf{L},\mathbf{G})(\mathfrak{m}^M) = (M,0,0,0).
\end{aligned} 
\end{align*}

\item The charges $(\mathbf{E},\mathbf{P},\mathbf{L},\mathbf{G})$ play a major role in the characteristic gluing problem because they \emph{cannot} be glued by our methods, see the statement of Theorem \ref{PROPNLgluingOrthA121}. This stems from the fact that at the linear level, they satisfy conservation laws and are \emph{invariant} under the linearized sphere perturbations introduced in Section \ref{SECdefEquivalenceFirstOrderSphereData}.

\item For sphere data $x_{u,v} \in \XX(S_{u,v})$, the charges are well-defined. Indeed, first, from \eqref{EQGaussEquation} and \eqref{EQgausscodazzinonlinear1}, it is straight-forward to show that for sufficiently small real numbers $\varep>0$ and $M$, and sphere data $x_{u,v}$ with
\begin{align*} 
\begin{aligned} 
\Vert x_{u,v} - \mathfrak{m}^M \Vert_{\XX(S_{u,v})} \leq \varep,
\end{aligned} 
\end{align*}
we have that
\begin{align*} 
\begin{aligned} 
&\Vert \be \Vert_{H^{5}(S_{u,v})} + \left\Vert \rh+ \frac{2M}{r_M^3} \right\Vert_{H^{4}(S_{u,v})} + \Vert \si \Vert_{H^{4}(S_{u,v})} + \Vert \beb \Vert_{H^{3}(S_{u,v})}\\
 \les& C_{u,v} \Vert x_{u,v} - \mathfrak{m}^M \Vert_{\XX(S_{u,v})},
\end{aligned} 
\end{align*}
where the constant $C_{u,v}>0$ depends on $u$ and $v$. Consequently, by Definition \ref{DEFnonlinearcharges6} together with standard estimates (see, for example, Lemma \ref{LEMstandardTRACE}), the charges are bounded by
\begin{align*} 
\begin{aligned} 
\vert \mathbf{E} - M\vert + \vert \mathbf{P}\vert + \vert \mathbf{L} \vert+ \vert \mathbf{G} \vert \les C_{u,v} \Vert x_{u,v} - \mathfrak{m}^M \Vert_{\XX(S_{u,v})}.
\end{aligned} 
\end{align*}
\end{enumerate}

\ni As remarked above, the charges $(\mathbf{E},\mathbf{P},\mathbf{L},\mathbf{G})$ cannot be glued with our methods. To study the characteristic gluing problem modulo the charges, we introduce the following matching map.

\begin{definition}[Matching map $\mathfrak{M}$] \label{DEFmatchingMAP} Let $x_{u,v}$ be sphere data on $S_{u,v}$. Define 
\begin{align*} 
\begin{aligned} 
\mathfrak{M}(x_{u,v}):= \lrpar{\Om, \phi, \gd_c, \Om\trchi, \chih, (\Om\trchib)^{[\geq2]}, \chibh, \eta^{[\geq2]}, \om, D\om, \omb^{[\geq2]}, \Du\omb^{[\geq2]}, \tilde{\QQ}_5, \tilde{\QQ}_6, \a, \ab},
\end{aligned} 
\end{align*}
where $\phi$ and $\gd_c$ are defined by \eqref{EQdefdecomposition1}, and $\tilde{\QQ}_5$ and $\tilde{\QQ}_6$ are defined with $r=v-u$ by
\begin{align*} 
\begin{aligned} 
\tilde{\QQ}_5 :=& \omb^{[\leq1]} + \frac{1}{4}\lrpar{\Om\trchib}^{[\leq1]}- \frac{1}{6r} \Divdo \eta^{[1]} \\
&- \frac{1}{12r^3} (\Ldo+3) \lrpar{\Om\trchi-\frac{4}{r}\Om}^{[\leq1]} - \frac{1}{2r^2}(\Ldo+2)\phi^{[\leq1]}, \\
\tilde{\QQ}_6 :=& \lrpar{\Du\omb}^{[\leq1]} -\frac{1}{6} (\Ldo-3) \lrpar{\frac{1}{r}\Om\trchib- \frac{2}{r^3}(\Ldo+2)\phi}^{[\leq1]} \\
&+ \frac{1}{6r} \lrpar{\Ldo\Ldo + \Ldo -3} \lrpar{\Om\trchi-\frac{4}{r}\Om}^{[\leq1]} -\frac{2}{3r^2} \Divdo \eta^{[1]},
\end{aligned} 
\end{align*}
where $\Divdo$ and $\Ldo$ are the divergence and Laplace-Beltrami operator with respect to the standard unit round metric $\gac$ on $S_{u,v}$. We also call $\MMf_{u,v}:= \mathfrak{M}(x_{u,v})$ the matching data at $S_{u,v}$.
\end{definition}

\emph{Remarks on Definition \ref{DEFmatchingMAP}.}
\begin{enumerate}
\item In the proof of our main theorem, we show that we are able to glue the matching data on $S_{0,2}$.
\item The linearizations of $\tilde{\QQ}_5$ and $\tilde{\QQ}_6$ at Minkowski equal the gauge-dependent charges $\QQ_5$ and $\QQ_6$ of the linearized null structure equations at Minkowski, see \eqref{EQdefChargesMinkowski8891} and Lemma \ref{CORtransportEQS3333} in Section \ref{SEClinPRELIM3}.
\end{enumerate}

\ni The following lemma shows that the range of the matching map $\MMf$ is the complement to the charges $(\mathbf{E},\mathbf{P},\mathbf{L},\mathbf{G})$.

\begin{lemma}[Matching map and charges] \label{LEMconditionalMATCHING} Let $x_{u,v}$ and $x'_{u,v}$ be sphere data on $S_{u,v}$ such that for a real number $\varep>0$,
\begin{align} 
\begin{aligned} 
\Vert x_{u,v} -\mathfrak{m} \Vert_{\XX(S_{u,v})} + \Vert x'_{u,v} -\mathfrak{m} \Vert_{\XX(S_{u,v})} \leq \varep,
\end{aligned} \label{EQmatchingMap157778}
\end{align}
and satisfying
\begin{align} 
\begin{aligned} 
\MMf\lrpar{x_{u,v}} =\MMf(x'_{u,v}).
\end{aligned} \label{EQmatchingMap17778}
\end{align}
For $\varep>0$ sufficiently small, the following holds: If, in addition to \eqref{EQmatchingMap17778},
\begin{align} 
\begin{aligned} 
\lrpar{\mathbf{E},\mathbf{P}, \mathbf{L}, \mathbf{G}}(x_{u,v}) = \lrpar{\mathbf{E},\mathbf{P}, \mathbf{L}, \mathbf{G}}(x'_{u,v}),
\end{aligned} \label{EQmatchingMap27778}
\end{align}
then
\begin{align} 
\begin{aligned} 
x_{u,v}= x'_{u,v}.
\end{aligned} \label{EQmatchingMap37778}
\end{align}
\end{lemma}

\begin{proof}[Proof of Lemma \ref{LEMconditionalMATCHING}] First, by Definition \ref{DEFnonlinearcharges6} we rewrite the matching of $\mathbf{L}$ and $\mathbf{G}$ in \eqref{EQmatchingMap27778} as
\begin{align} 
\begin{aligned} 
&\lrpar{ \di \trchi(x_{u,v}) + \trchi(x_{u,v}) (\eta(x_{u,v})-\di\log\Om(x_{u,v})) }^{[1]} \\
=& \lrpar{ \di \trchi(x_{u,v}') + \trchi(x_{u,v}') (\eta(x_{u,v}')-\di\log\Om(x_{u,v}') )}^{[1]}.
\end{aligned} \label{EQmatchingnonlinear11111}
\end{align}
By the matching of $\trchi$, $\Om$ and $\eta^{[\geq2]}$ in \eqref{EQmatchingMap17778}, we can rewrite \eqref{EQmatchingnonlinear11111} as
\begin{align} 
\begin{aligned} 
0=& \lrpar{\trchi(x_{u,v}') \lrpar{\eta(x_{u,v}) - \eta(x_{u,v}')} }^{[1]} \\
=& \lrpar{\trchi(x_{u,v}') \lrpar{\eta(x_{u,v}) - \eta(x_{u,v}')}^{[1]} }^{[1]}\\
=& \frac{2}{r}\lrpar{\eta(x_{u,v}) - \eta(x_{u,v}')}^{[1]} \\
&+ \lrpar{\lrpar{\trchi(x_{u,v}')-\frac{2}{r}} \lrpar{\eta(x_{u,v}) - \eta(x_{u,v}')}^{[1]} }^{[1]}\\
=& \frac{2}{r}X + \lrpar{\lrpar{\trchi(x_{u,v}')-\frac{2}{r}} X }^{[1]},
\end{aligned} \label{EQrewritingetanonlinearmatching234}
\end{align}
where we denoted $X= \lrpar{\eta(x_{u,v}) - \eta(x_{u,v}')}^{[1]}$. Using that by \eqref{EQmatchingMap157778},
\begin{align*} 
\begin{aligned} 
\left\Vert \trchi(x_{u,v}')- \frac{2}{r} \right\Vert_{H^{6}(S_{u,v})} \les \varep, 
\end{aligned} 
\end{align*}
the relation \eqref{EQrewritingetanonlinearmatching234} implies for $\varep>0$ sufficiently small that $X=0$, that is,
\begin{align} 
\begin{aligned} 
\eta(x_{u,v})^{[1]} = \eta(x_{u,v}')^{[1]}.
\end{aligned} \label{EQetamatching11110}
\end{align}
By the matching of $\eta^{[\geq2]}$ in \eqref{EQmatchingMap17778}, this implies the matching of $\eta$ at $S_{u,v}$. By \eqref{EQmatchingMap17778} and the Gauss-Codazzi equation \eqref{EQgausscodazzinonlinear1} this further implies the matching of $\be$ at $S_{u,v}$.

Second, by \eqref{EQmatchingMap17778}, the Gauss equation \eqref{EQGaussEquation} and the above matching of $\beta$, the matching of $\mathbf{E}$ and $\mathbf{P}$ in \eqref{EQmatchingMap27778} 
can be written as
\begin{align} 
\begin{aligned} 
&\lrpar{K(x_{u,v})+ \frac{1}{4} \trchi(x_{u,v}) \trchib(x_{u,v}) - \half (\chih(x_{u,v}),\chibh(x_{u,v}))}^{[\leq1]} \\
=& \lrpar{K(x_{u,v}')+ \frac{1}{4} \trchi(x_{u,v}') \trchib(x_{u,v}') - \half (\chih(x_{u,v}'),\chibh(x_{u,v}'))}^{[\leq1]}.
\end{aligned} \label{EQtrchib1gluing555}
\end{align}
By the gluing of $\gd$, $\chih$, $\chibh$, $\trchi$ and $\trchib^{[\geq2]}$ in \eqref{EQmatchingMap17778}, we can rewrite \eqref{EQtrchib1gluing555} as
\begin{align} 
\begin{aligned} 
0=& \lrpar{ \trchi(x_{u,v}') \lrpar{\trchib(x_{u,v})-\trchib(x_{u,v}')}}^{[\leq1]}\\
=&  \lrpar{ \trchi(x_{u,v}') \lrpar{\trchib(x_{u,v})-\trchib(x_{u,v}')}^{[\leq1]}}^{[\leq1]} \\
=& \frac{2}{r} \lrpar{\trchib(x_{u,v})^{[\leq1]}-\trchib(x_{u,v}')^{[\leq1]}}\\
&+ \lrpar{ \lrpar{\trchi(x_{u,v}')-\frac{2}{r}} \lrpar{\trchib(x_{u,v})-\trchib(x_{u,v}')}^{[\leq1]}}^{[\leq1]}\\
=& \frac{2}{r} Y+ \lrpar{ \lrpar{\trchi(x_{u,v}')-\frac{2}{r}} Y}^{[\leq1]},
\end{aligned} \label{EQtrchib1gluing5552}
\end{align}
where we denoted $Y=\lrpar{\trchib(x_{u,v})^{[\leq1]}-\trchib(x_{u,v}')^{[\leq1]}}$. Using that by \eqref{EQmatchingMap157778},
\begin{align*} 
\begin{aligned} 
\left\Vert \trchi(x_{u,v}')- \frac{2}{r} \right\Vert_{H^{6}(S_{u,v})} \les \varep, 
\end{aligned} 
\end{align*}
the relation \eqref{EQtrchib1gluing5552} implies for $\varep>0$ sufficiently small that $Y=0$, that is,
\begin{align} 
\begin{aligned} 
\trchib(x_{u,v})^{[\leq1]} = \trchib(x_{u,v}')^{[\leq1]}.
\end{aligned} \label{EQtrchib01matching1111110}
\end{align}
From \eqref{EQtrchib01matching1111110} and \eqref{EQmatchingMap17778}, we deduce the matching of $\trchib$ at $S_{u,v}$.

Third, it remains to show that $\omb$ and $\Du\omb$ are matching. The gluing of $\omb^{[\geq2]}$ and $\Du\omb^{[\geq2]}$ is contained in \eqref{EQmatchingMap17778}, see Definition \ref{DEFmatchingMAP}. Moreover, \eqref{EQmatchingMap17778} includes the matching of 
\begin{align*} 
\begin{aligned} 
\QQ_5 :=& \omb^{[\leq1]} + \frac{1}{4}\lrpar{\Om\trchib}^{[\leq1]}- \frac{1}{6r} \Divdo \eta^{[1]} \\
&- \frac{1}{12r^3} (\Ldo+3) \lrpar{\Om\trchi-\frac{4}{r}\Om}^{[\leq1]} - \frac{1}{2r^2}(\Ldo+2)\phi^{[\leq1]}, \\
\QQ_6 :=& \lrpar{\Du\omb}^{[\leq1]} -\frac{1}{6} (\Ldo-3) \lrpar{\frac{1}{r}\Om\trchib- \frac{2}{r^3}(\Ldo+2)\phi}^{[\leq1]} \\
&+ \frac{1}{6r} \lrpar{\Ldo\Ldo + \Ldo -3} \lrpar{\Om\trchi-\frac{4}{r}\Om}^{[\leq1]} -\frac{2}{3r^2} \Divdo \eta^{[1]},
\end{aligned} 
\end{align*}
at $S_{u,v}$. Thus from the matching of $\Om, \phi$ and $\trchi$ in \eqref{EQmatchingMap17778} and the above matching of $\eta$ and $\trchib$, it follows that $\omb^{[\leq1]}$ and $\Du\omb^{[\leq1]}$ agree on $S_{u,v}$. This finishes the proof of Lemma \ref{LEMconditionalMATCHING}. \end{proof}

\ni The following lemma follows directly by Definition \ref{DEFmatchingMAP}. Its proof is omitted.
\begin{lemma}[Smoothness of $\mathfrak{M}$] Let $v>u$ be two real numbers. The matching map $\mathfrak{M}$ is a smooth mapping from an open neighbourhood of $\mathfrak{m}$ in $\XX(S_{u,v})$ to $\mathcal{Z}_{\mathfrak{M}}(S_{u,v})$, where (with all spaces over the sphere $S_{u,v}$)
\begin{align*} 
\begin{aligned} 
\mathcal{Z}_{\mathfrak{M}}(S_{u,v}) =& H^6 \times H^6 \times H^6 \times H^6 \times H^6 \times H^4 \times H^4 \times H^5 \\
&\times H^6 \times H^6 \times H^4 \times H^2 \times H^4 \times H^2 \times H^6 \times H^2, 
\end{aligned} 
\end{align*}
and we have the estimate
\begin{align*} 
\begin{aligned} 
\Vert \mathfrak{M}(x_{u,v}) - \mathfrak{M}(\mathfrak{m}) \Vert_{\mathcal{Z}_{\mathfrak{M}}(S_{u,v})} \les C_{u,v} \Vert x_{u,v} -\mathfrak{m}\Vert_{\XX(S_{u,v})},
\end{aligned} 
\end{align*}
where the constant $C_{u,v}>0$ depends on $u$ and $v$.
\end{lemma}

\subsection{Nilpotent character of null structure equations} \label{SECconformalMethod1112}

\ni It is well-known that solutions to the null structure equations can be constructed from \emph{free data} which is not subject to any constraint equations, see for example \cite{ChrFormationBlackHoles} and \cite{GravImpulsesLukRod1}. This is due to the nilpotent character of the null structure equations which allows to rewrite them into a \emph{hierarchy of null transport equations} which can be solved subsequently from the free data. We proceed as follows.
\begin{itemize}
\item In Section \ref{SECderivationConstraintFunctions1} we define the free data and derive the hierarchy of transport equations, denoted by
\begin{align*} 
\begin{aligned} 
\CC_i= 0 \text{ for } 1\leq i \leq 10,
\end{aligned} 
\end{align*}
where the maps $(\CC_i)_{1\leq i \leq 10}$ are called \emph{constraint functions}.
\item In Section \ref{SECderivationConstraintFunctions2} we calculate the linearization of the constraint functions at Schwarzschild of mass $M$.
\end{itemize}

\subsubsection{Definition of free data and derivation of hierarchy} \label{SECderivationConstraintFunctions1} The following definition of free data is the starting point for the construction of solutions to the null structure equations. For explicitness, we define free data on the null hypersurface $\HH_{0,[1,\infty)}$.
\begin{definition}[Free data] \label{DEFconformalSEED}  On $ \HH_{0,[1,\infty)}$ prescribe
\begin{itemize} 
\item the conformal class $\mathrm{conf}(\gd)$ of induced Riemannian metrics $\gd$ on $S_{0,v}$,
\item a scalar function $\Om$, called the null lapse.
\end{itemize}
On $S_{0,1}$ prescribe
\begin{itemize}
\item the induced Riemannian metric $\gd$ (compatible with the conformal class on $S_{0,1}$), 
\item the scalar functions $\trchi$, $\trchib$, $\omb$, $\Du\omb$, 
\item an $S_{0,1}$-tangential vectorfields $\eta$,  
\item a $\gd$-tracefree $S_{0,1}$-tangential symmetric $2$-tensors $\chibh$ and $\ab$.
\end{itemize}
\end{definition}
 
\ni Before constructing a solution to the null structure equations with the above free data, we introduce the following objects.
\begin{enumerate}
\item Given the conformal class $\mathrm{conf}(\gd)$ on $S_v$, let $\gd_c$ be the unique representative such that
\begin{align*}
\sqrt{\det {\gd}_c}(v,\th^1,\th^2) = \sqrt{\det \gac}(\th^1,\th^2).
\end{align*}

\item Let $\gd$ denote the induced metric on $S_{0,v}$ of the solution of the null constraint equations to be constructed. Define $\phi>0$ to be the conformal factor such that
\begin{align}
\gd = \phi^2 \gd_c, \label{EQrelationgdtildegd}
\end{align}
that is,
\begin{align*}
\phi^2 := \frac{\sqrt{\det \gd}}{\sqrt{\det \gac}}.
\end{align*}

\item It is straight-forward to verify that the \emph{shear} $e$, defined by 
$$e:= \vert \chih \vert_\gd^2,$$ 
is conformally invariant and can therefore be explictly calculated from $\gd_c$ on $\HH$.

\end{enumerate}

\ni We are now in position to construct a solution to the null structure equations from the free data. In the following we derive a hierarchy of null transport equations, called the \emph{constraint functions} $\CC_i$ for $1\leq i \leq 10$, which can be solved based on the free data. \\

\ni \textbf{Equation for $\phi$.} By combining \eqref{EQusefulDphiRELATION} and \eqref{EQRaychauduri1}, that is,
\begin{align*} 
\begin{aligned} 
D \phi = \frac{\Om \trchi \phi}{2}, \,\, D \trchi + \frac{\Om}{2} (\trchi)^2 - \om \trchi = - \Om \vert \chih \vert^2,
\end{aligned} 
\end{align*}
we get that $\phi$ satisfies the following \emph{linear} transport equation,
\begin{align*} 
\begin{aligned} 
\CC_1 :=& D^2\phi -\om \Om\trchi \phi + \frac{1}{2} \Om^2\vert \chih \vert^2 \phi =0.
\end{aligned} 
\end{align*}
We note that $\phi$ together with $\gd_c$ fully determines $\gd$ on each sphere.\\

\ni \textbf{Equation for $\chi$.} By \eqref{EQfirstvariation1}, $\chi$ satisfies 
\begin{align*} 
\begin{aligned} 
D\gd -2\Om \chi =0.
\end{aligned} 
\end{align*}
Splitting \eqref{EQfirstvariation1} into a trace and a tracefree part and using the decomposition \eqref{EQrelationgdtildegd}, we get the constraint equations
\begin{align*} 
\begin{aligned} 
\CC_2:=& 2\phi D\phi + \frac{\phi^2}{2} \tr_{\gd_c} D\gd_c - \Om \trchi \phi^2 =0, \\
\CC_3 :=& -2 \Om \chih + \phi^2 \lrpar{D\gd_c -\half (\tr_{\gd_c} D\gd_c) \gd_c} =0.
\end{aligned} 
\end{align*}

\ni \textbf{Equation for $\eta$.} By combining \eqref{EQtransportEQLnullstructurenonlinear} and \eqref{EQgausscodazzinonlinear1}, that is,
\begin{align*}
D\eta = \Om (\chi \cdot \etab - \be), \,\, - \beta=\Divd \chih -\half \di \tr \chi + \chih \cdot \zeta - \half \trchi \zeta,
\end{align*}
and using that by \eqref{DEFricciCoefficients}, 
\begin{align*}
\etab = - \eta + 2 \di \log \Om, \,\, \ze = \eta - \di \log \Om,
\end{align*}
we get that $\eta$ satisfies the following transport equation,
\begin{align*} 
\begin{aligned} 
\CC_4 :=& D\eta + \Om \trchi \eta - \Om \lrpar{\Divd \chih - \half\di \trchi+\chih \di \log \Om + \frac{3}{2}\trchi \di \log \Om} =0.
\end{aligned} 
\end{align*}

\ni \textbf{Equation for $\Om\trchib$.} By combining \eqref{DEFricciCoefficients}, \eqref{EQGaussEquation} and \eqref{EQtransporttrchitrchib1}, we get that
\begin{align*} 
\begin{aligned} 
\CC_5 :=& D(\Om \trchib) + \Om\trchi (\Om\trchib) +2\Om^2\Divd (\eta-2\di \log\Om) \\
&- 2 \Om^2 \vert \eta-2\di \log\Om \vert^2 +2 \Om^2 K \\
=&0,
\end{aligned} 
\end{align*}
where $K$ denotes the Gauss curvature of $(S_{0,v},\gd)$.\\

\ni \textbf{Equation for $\chibh$.} By \eqref{DEFricciCoefficients} and \eqref{EQchihequations1}, it follows that $\omb$ satisfies
\begin{align*} 
\begin{aligned} 
\CC_6 :=&D\lrpar{\Om \chibh} - (\Om \chih, \Om \chibh) \gd - \half \Om \trchi \Om \chibh \\
&- \Om^2 \lrpar{\Nd \widehat{\otimes}(2\di \log\Om-\eta) + (2\di \log\Om-\eta) \widehat{\otimes} (2\di \log\Om-\eta) - \half \trchib \chih} \\
=&0.
\end{aligned} 
\end{align*}

\ni \textbf{Equation for $\omb$.} By \eqref{DEFricciCoefficients}, \eqref{EQtransportEQLnullstructurenonlinear} and \eqref{EQGaussEquation}, it follows that
\begin{align*} 
\begin{aligned} 
\CC_7 := D\omb - \Om^2 \lrpar{4(\eta, \di \log \Om)- 3 \vert \eta \vert^2 + K +\frac{1}{4} \trchi \trchib -\half (\chih, \chibh)}=0.
\end{aligned} 
\end{align*}

\ni \textbf{Equation for $\ab$.} By \eqref{DEFricciCoefficients}, \eqref{EQGaussEquation}, \eqref{EQgausscodazzinonlinear1} and \eqref{EQnullBianchiEquations} it follows that $\ab$ satisfies the following transport equation,
\begin{align*} 
\begin{aligned} 
\CC_8 :=& \widehat{D}\ab - \half \Om \trchi \ab + 2 \om \ab \\
&+ \Om \Nd \widehat{\otimes} \lrpar{\Divd \chibh - \half \di \trchib - \chibh \cdot (\eta-\di\log\Om) + \half \trchib (\eta-\di \log\Om)}\\
&+ \Om \lrpar{9\di \log \Om -5 \eta}\widehat{\otimes} \lrpar{\Divd \chibh - \half \di \trchib - \chibh \cdot (\eta-\di\log\Om) + \half \trchib (\eta-\di \log\Om)}\\
&-3\Om \chibh \lrpar{K + \frac{1}{4} \trchi \trchib - \half (\chih,\chibh)} + 3\Om {}^*\chibh \lrpar{\Curld \eta + \half \chih \wedge \chibh}\\
=&0.
\end{aligned} 
\end{align*}

\ni \textbf{Equation for $\Du\omb$.} By \eqref{DEFricciCoefficients}, \eqref{EQriccirelationetabeta}, \eqref{EQGaussEquation}, \eqref{EQgausscodazzinonlinear1} and \eqref{EQDUOMU1}, it follows that 
\begin{align*} 
\begin{aligned} 
\CC_9 :=& D\Du\omb -12 \Om^3 (\di\log\Om-\eta)\omb -2 \Om^2 \omb \lrpar{2(\eta,-\eta+2\di\log\Om)- \vert \eta\vert^2}\\
& +\lrpar{K+ \frac{1}{4} \trchi\trchib-\half (\chih,\chibh)} \lrpar{\frac{3}{2}\Om^3\trchib -2\Om^2 \omb} \\
&-12 \Om^3 \chib(\eta, \eta-\di\log\Om) - \half \Om^3 (\chih, \ab)\\
& - \Om^3 \lrpar{\Divd \chibh-\half \di\trchib-\chibh \cdot \lrpar{\eta- \di \log \Om}+ \half \trchib \lrpar{\eta-\di\log\Om}, 7\eta-3\di\log\Om}\\
&- \Om^3 \Divd \lrpar{\Divd \chibh-\half \di\trchib-\chibh \cdot \lrpar{\eta- \di \log \Om}+ \half \trchib \lrpar{\eta-\di\log\Om}}  \\
=&0.
\end{aligned} 
\end{align*}

\ni \textbf{Equation for $\a$.} By \eqref{EQtransportEQLnullstructurenonlinear}, $\a$ satisfies
\begin{align*} 
\begin{aligned} 
\CC_{10} :=& \Om \a + D\chih-\Om \vert \chih \vert^2 \gd - \om \chih=0.
\end{aligned} 
\end{align*}

\ni The following lemma shows that the constraint functions are a smooth mapping. Its proof is straight-forward and omitted.
\begin{lemma}[Smoothness of constraint functions] \label{LEMstandardEstimates} Consider null data on $\HH_{0,[1,2]}$,
\begin{align*} 
\begin{aligned} 
x=(\Om, \gd, \Om\trchi, \chih, \Om\trchib, \chibh, \eta, \om, D\om, \omb, \Du\omb, \a, \ab).
\end{aligned} 
\end{align*}
The constraints map $\mathcal{C}$,
\begin{align*} 
\begin{aligned} 
\mathcal{C}: x \mapsto (\CC_i(x))_{1\leq i \leq 10},
\end{aligned} 
\end{align*}
is a smooth mapping from an open neighbourhood of $\mathfrak{m}$ in $\mathcal{X}(\HH_{0,[1,2]})$ to $\mathcal{Z}_{\CC}$, where
\begin{align*} 
\begin{aligned} 
\mathcal{Z}_{\CC}:=&  H^6_2 \times  H^6_3 \times H^6_2 \times H^{5}_1 \times  H^{4}_1 \times  H^{5}_2 \times  H^{4}_2\times H^{2}_2 \times  H^{2}_2 \times  H^6_1,
\end{aligned} 
\end{align*}
where each space is over $\HH_{0,[1,2]}$. Moreover, we have the estimates
\begin{align*} 
\begin{aligned} 
\Vert (\CC_i(x))_{1\leq i \leq 10} \Vert_{\mathcal{Z}_{\CC}} \les \Vert x - \mathfrak{m} \Vert_{\XX(\HH_{0,[1,2]})}.
\end{aligned} 
\end{align*}
\end{lemma}

\subsubsection{Linearized constraint functions at Minkowski}\label{SECderivationConstraintFunctions2}
In this section we linearize the constraint functions $(\CC_i(x))_{1\leq i \leq 10}$ at Minkowski, that is, at $x=\mathfrak{m}$. The linearization procedure is adapted from \cite{DHR}: We expand the sphere data
\begin{align*} 
\begin{aligned} 
x =& \lrpar{\Om, \gd, \Om\trchi, \chih, \Om\trchib, \chibh, \eta, \om, D\om, \omb, \Du\omb, \a,  \ab} \\
=& \lrpar{1, r^2 \gac, \frac{2}{r}, 0, -\frac{2}{r}, 0, 0, 0, 0, 0, 0, 0,0}\\
&+ \varep \cdot \lrpar{\dot\Om, \dot{\gd}, \omtrchid, \dot\chih, \omtrchibd, \dot\chibh, \dot\eta, \omd, D\omd, \ombd, \Du\ombd, \ad, \abd} + \mathcal{O}(\varep^2),
\end{aligned} 
\end{align*}
and differentiate in $\varep$ at $\varep=0$. Here, we recall that the Minkowski value for $r$ is given by $r=v-u$. The proof of the next lemma follows by explicit calculation.
\begin{lemma}[Linearization of constraint functions at Minkowski] \label{LEMlinearizedConstraints} Let $(\dot{\CC}_i)_{1\leq i \leq 10}$ denote the linearization of the constraint functions $({\CC}_i)_{1\leq i \leq 10}$ at Minkowski. Then it holds that
\begin{align*} 
\begin{aligned} 
\dot{\CC}_1=& D^2 \phid - 2 \omd = D(D\phid-2\Omd), &
\dot{\CC}_2=& r^2 \lrpar{2D\lrpar{\frac{\phid}{r}}-\omtrchid} ,\\
\dot{\CC}_3=& r^2 D\gdcd - 2 \chihd, &
\dot{\CC}_4=& \frac{1}{r^2} D\lrpar{r^2 \etad} - \frac{4}{r} \di \Omd - \frac{1}{r^2} \Divdo \chihd + \frac{1}{2} \di \omtrchid, 
\end{aligned} 
\end{align*}
and moreover,
\begin{align*} 
\begin{aligned} 
\dot{\CC}_5=& \frac{1}{r^2} D\lrpar{r^2 \omtrchibd} - \frac{2}{r} \omtrchid + \frac{2}{r^2} \Divdo \lrpar{\etad-2\di \Omd} +2 \Kd + \frac{4}{r^2} \Omd,\\
\dot{\CC}_6=& r D\lrpar{\frac{\chibhd}{r}} - 2 \DDd_2^\ast \lrpar{\etad-2\di\Omd} -\frac{1}{r} \chihd ,\\
\dot{\CC}_7=& D\ombd - \Kd -\frac{1}{2r}\omtrchibd + \frac{1}{2r} \omtrchid - \frac{2}{r^2} \Omd, \\
\dot{\CC}_8=& r D\lrpar{\frac{\abd}{r}} - 2 \DDd_2^\ast \lrpar{\frac{1}{r^2} \Divdo \chibhd - \half \di \omtrchibd - \frac{1}{r} \etad}, \\
\dot{\CC}_9=& D \lrpar{\Du \ombd} -\frac{3}{r} \lrpar{\Kd +\frac{1}{2r}\omtrchibd - \frac{1}{2r} \omtrchid+ \frac{2}{r^2} \Omd} \\
&- \frac{1}{r^2} \Divdo \lrpar{\frac{1}{r^2} \Divdo \chibhd - \frac{1}{2} \di \omtrchibd - \frac{1}{r} \etad}.\\
\CCd_{10}=& \ad + D\chihd.
\end{aligned} 
\end{align*}
\end{lemma}
\begin{remark} In addition to the above, we have by \eqref{EQdefdecomposition1} and \eqref{DEFricciCoefficients} that
\begin{align} 
\begin{aligned} 
\dot{\gd} = 2r \phid \gac + r^2 \gdcd, \,\, \omd = D\Omd, \,\, \ombd = \Du\Omd, \,\, \etabd =& - \etad + 2 \di \Omd.
\end{aligned} \label{EQssLinearizationRelations7778}
\end{align}
Moreover, by (242) in \cite{DHR} the linearization of the Gauss curvature $\Kd$ is given by
\begin{align} \label{EQgaussDHR}
\Kd =& \frac{1}{2r^2} \Divdo \Divdo \gdcd - \frac{1}{r^3} (\Ldo +2) \phid.
\end{align}
Moreover, using that the area radius $r$ is defined by
\begin{align*} 
\begin{aligned} 
r^2 = \frac{1}{4\pi} \int\limits_{S_{u,v}} \phi^2 d\mu_{\gac},
\end{aligned} 
\end{align*}
and that for Minkowski sphere data, $\phi= r$, we have that $\dot{r}^{[\geq1]}=0$ and
\begin{align} 
\begin{aligned} 
\dot{r}^{(0)} = \phid^{(0)}.
\end{aligned} \label{EQareaRADIUSlinearization}
\end{align}
\end{remark}

\subsection{Perturbations of sphere data} \label{SECdefEquivalenceFirstOrderSphereData} In this section we introduce two types of perturbations of sphere data, \emph{transversal perturbations} and \emph{angular perturbations}. They respectively come from transversal perturbations of the sphere and sphere diffeomorphisms. We proceed as follows.

\begin{itemize}
\item In Section \ref{SECdefPerturbations1111} we introduce the transversal perturbation mapping $\PP_f$ and a norm for the perturbation function $f$.
\item In Section \ref{SECdefPerturbations1222} we introduce the angular perturbation mapping $\PP_{(j^1,j^2)}$ and a norm for the perturbation functions ${(j^1,j^2)}$.
\item In Section \ref{SECPerturbationslinearizesd222111} we discuss the linearizations of $\PP_f$ and $\PP_{(j^1,j^2)}$ at Schwarzschild of mass $M\geq0$.

\end{itemize}

\ni We note that at the linear level (see Section \ref{SECPerturbationslinearizesd222111}), transversal perturbations and angular perturbations correspond directly to specific \emph{linear gauge solutions} in \cite{DHR}. However, the regularity control of the perturbed sphere data at the nonlinear level loses regularity compared to the linear level, and thus needs separate discussion, see Propositions \ref{PropositionSmoothnessF} and \ref{PropositionSmoothnessF2}, and Appendix \ref{SECproofTEClemmasmoothness}.
\subsubsection{Transversal perturbations $\PP_f$} \label{SECdefPerturbations1111}

In this section we introduce transversal perturbations. In words, the idea is as follows. Given a spacelike $2$-sphere $\tilde{S}$ in a vacuum spacetime $\lrpar{\MM,\g}$ and a scalar function $f$ on $\tilde{S}$, we perturb $\tilde{S}$ along the ingoing null direction by an amount $f$. The resulting sphere is denoted by $S$ and its sphere data by $x$.

In the following we sketch the formal definition of transversal perturbations; we refer to Appendix \ref{SECproofTEClemmasmoothness} for full details and estimates. Let $(\tilde{u},\tilde{v},\tilde{\th}^1,\tilde{\th}^2)$ be a local double null coordinate system around $\tilde{S}$ such that
\begin{align*} 
\begin{aligned} 
\g =& - 4 \tilde{\Om}^2 d\tilde{u} d\tilde{v} + \tilde{\gd}_{CD} (d\tilde{\th}^C - \tilde{b}^C d\tilde{v})(d\tilde{\th}^D - \tilde{b}^D d\tilde{v}),\end{aligned} 
\end{align*}
and
\begin{align*} 
\begin{aligned} 
\tilde{S} = \tilde{S}_{0,2} := \{ \tilde{u}=0, \tilde{v}=2\}.
\end{aligned} 
\end{align*}
Denote by $\tilde{x}_{0,2}$ the sphere data on $\tilde{S}_{0,2}$ with respect to $(\tilde{u},\tilde{v},\tilde{\th}^1,\tilde{\th}^2)$.

We define new coordinates $(u,\th^1,\th^2)$ on $\tilde{\HHb}_2:= \{\tilde{v}=2\}$ as follows. For a given smooth scalar function $f(u,\tth^1,\tth^2)$, define $(u,\th^1,\th^2)$ on $\tilde{\HHb}_2$ by
\begin{align} 
\begin{aligned} 
\tilde{u}=u+f(u,\th^1,\th^2), \,\, \tth^1= \th^1, \,\, \tth^2=\th^2.
\end{aligned} \label{EQdefinitionuHHB2}
\end{align}
For $f$ sufficiently small, $(u,\th^1,\th^2)$ indeed forms a local coordinate system on $\tilde{\HHb}_2$. Define the sphere $S' \subset \HHb_2$ by  
\begin{align*} 
\begin{aligned} 
S':= \{ u=0\} = \{\tilde u = f(0,\th^1,\th^2), \tilde{v}=2 \},
\end{aligned} 
\end{align*}
where we used \eqref{EQdefinitionuHHB2}. Let $(u,v,\th^1,\th^2)$ be the local double null coordinate system on $\MM$ such that $v=\tilde{v}$ on $\MM$ and $(u,\th^1,\th^2)$ agree with the constructed $(u,\th^1,\th^2)$ on $\tilde{\HHb}_2$. Let $x_{0,2}$ be the sphere data of $S_{0,2}=S'$ with respect to $(u,v,\th^1,\th^2)$. An explicit calculation of $x_{0,2}$ is provided in Appendix \ref{SECproofTEClemmasmoothness}.

The sphere data $x_{0,2}$ depends not only on $f$ and $\tilde{x}_{0,2}$ but also on the ingoing null data $\tilde{\underline{x}}$ of $(\MM,\g)$ on $\tilde{\HHb}_2$ (with respect to $(\tilde u, \tilde v, \tilde{\th}^1, \tilde{\th}^2)$). Hence we denote the transversal perturbation mapping $\PP_f$ by
\begin{align*} 
\begin{aligned} 
x_{0,2} := \PP_f(\tilde{\underline{x}}).
\end{aligned} 
\end{align*}

\begin{remark} \label{REMARKfperturbationREALuse} In Appendix \ref{SECproofTEClemmasmoothness} it is shown that the sphere data $x_{0,2}$ on $S_{0,2}$ depends on $f$ only via the four scalar functions
\begin{align*} 
\begin{aligned} 
\lrpar{f(0,\th^1,\th^2), \pr_uf(0,\th^1,\th^2), \pr_u^2 f (0,\th^1,\th^2), \pr_u^3 f (0,\th^1,\th^2)}.
\end{aligned}
\end{align*}
In the rest of the paper we abuse notation and denote this tuple of scalar functions simply by
\begin{align} 
\begin{aligned} 
f := \lrpar{f(0), \pr_uf(0), \pr_u^2 f (0), \pr_u^3 f (0)}.
\end{aligned}  \label{EQdefPerturbationFunctionfdefi}
\end{align}
\end{remark}

\ni We introduce the following norm for the perturbation function $f$ in \eqref{EQdefPerturbationFunctionfdefi}.
\begin{definition}[Norm for perturbation function $f$] \label{DEFnormfperturbations} 
For a perturbation function $f$ on $\SSS^2$ as in \eqref{EQdefPerturbationFunctionfdefi} given by
\begin{align*} 
\begin{aligned} 
f:= \lrpar{f(0), \pr_uf(0), \pr_u^2 f (0), \pr_u^3 f (0)},
\end{aligned} 
\end{align*}
define
\begin{align*} 
\begin{aligned} 
\Vert f \Vert_{\YY_f} := \Vert f(0) \Vert_{H^8(\SSS^2)} + \Vert \pr_u f(0) \Vert_{H^6(\SSS^2)}+\Vert \pr_u^2 f(0) \Vert_{H^4(\SSS^2)}+\Vert \pr_u^3 f(0) \Vert_{H^2(\SSS^2)},
\end{aligned} 
\end{align*}
where the norms are with respect to the round unit metric $\gac$ on $\SSS^2$. Moreover, let
\begin{align*} 
\begin{aligned} 
\YY_f := \{ f: \Vert f \Vert_{\YY_f} < \infty\}.
\end{aligned} 
\end{align*}
\end{definition}

\ni The following proposition is proved in Appendix \ref{SECproofTEClemmasmoothness}. We note that the regularity analysis of $\PP_f$ and the proof of \eqref{EQestimatePPfsmoothness1stder} is different than the regularity analysis of its linearization $\dot{\PP}_f$ (see Section \ref{SECPerturbationslinearizesd222111}).
\begin{proposition}[Smoothness of $\PP_f$] \label{PropositionSmoothnessF}
Let $\de>0$ be a real number. The mapping 
\begin{align*} 
\begin{aligned} 
\PP_f: \, \XX^+(\tilde{\HHb}_{[-\de,\de],2}) \times \YY_f &\to \XX(S_{0,2}), \\
(\tilde{\underline{x}},f) &\mapsto x_{0,2}:= \PP_f(\tilde{\underline{x}})
\end{aligned} 
\end{align*}
is well-defined and smooth in an open neighbourhood of $(\tilde{\underline{x}},f)=(\underline{\mathfrak{m}},0)$. Moreover, it holds that
\begin{align} 
\begin{aligned} 
\Vert \PP_{f}(\tilde{\underline{x}}) - \tilde{\underline{x}}_{0,2} \Vert_{\XX(S_{0,2})} \les \Vert f \Vert_{\YY_{f}} + \Vert \tilde{\underline{x}}-\underline{\mathfrak{m}} \Vert_{\XX^+(\tilde{\HHb}_{[-\de,\de],2})},
\end{aligned} \label{EQestimatePPfsmoothness1stder}
\end{align}
where we denoted $\tilde{\underline{x}}_{0,2} := \tilde{\underline{x}}\vert_{S_{0,2}}$.
\end{proposition}

\subsubsection{Angular perturbations $\PP_{(j^1,j^2)}$} \label{SECdefPerturbations1222} In this section we introduce angular perturbations. Consider a $2$-sphere $S$ with sphere data $\tilde{x}$ expressed with respect to local coordinates $(\tth^1,\tth^2)$ on $S$. For two smooth scalar functions $j_1(\th^1,\th^2)$ and $j_2(\th^1,\th^2)$ on $S$, define new local coordinates $(\th^1,\th^2)$ by
\begin{align*} 
\begin{aligned} 
\tth^1 = \th^1 + j^1(\th^1,\th^2), \,\, \tth^2= \th^2+j^2(\th^1,\th^2).
\end{aligned} 
\end{align*}
For $j^1$ and $j^2$ sufficiently small, $(\th^1, \th^2)$ is indeed a local coordinate system on $S$. Let 
\begin{align*} 
\begin{aligned} 
x := \PP_{(j^1,j^2)}(\tilde{x})
\end{aligned} 
\end{align*}
denote the sphere data on $S$ expressed with respect to $(\th^1,\th^2)$. We refer to Appendix \ref{SECproofTEClemmasmoothness} for explicit formulas for $x$.

We introduce the following norm for the perturbation functions $(j^1,j^2)$.
\begin{definition}[Norm for perturbation functions $(j^1,j^2)$] \label{DEFnormj1j2} Given two perturbation functions $(j^1,j^2)$ on $\SSS^2$, let
\begin{align*} 
\begin{aligned} 
\Vert (j^1,j^2) \Vert_{\mathcal{Y}_{(j^1,j^2)}} :=& \Vert j_1 \Vert_{H^{7}(\SSS^2)}+ \Vert j_2 \Vert_{H^{7}(\SSS^2)},
\end{aligned} 
\end{align*}
where the norms are with respect to the round unit metric $\gac$. Moreover, let
\begin{align*} 
\begin{aligned} 
\mathcal{Y}_{(j^1,j^2)} := \{(j^1,j^2) : \Vert (j^1,j^2) \Vert_{\YY_{(j^1,j^2)}}<\infty \}.
\end{aligned} 
\end{align*}
\end{definition}

\ni The following proposition is proved in Appendix \ref{SECproofTEClemmasmoothness}.
\begin{proposition}[Smoothness of $\PP_{(j^1,j^2)}$] \label{PropositionSmoothnessF2}
The mapping 
\begin{align*} 
\begin{aligned}  
\PP_{(j_1,j_2)}: \, \XX(S_{0,2}) \times \YY_{(j_1,j_2)} &\to \XX(S_{0,2}), \\
(x_{0,2},(j_1,j_2)) &\mapsto \PP_{(j_1,j_2)}(x_{0,2})
\end{aligned} 
\end{align*}
is well-defined and smooth in an open neighbourhood of $(x_{0,2}, {(j_1,j_2)})=(\mathfrak{m}_{0,2},(0,0))$. Moreover, it holds that
\begin{align} 
\begin{aligned} 
\Vert \PP_{(j_1,j_2)}(x_{0,2}) - x_{0,2} \Vert_{\XX(S_{0,2})} \les \Vert (j_1,j_2) \Vert_{\YY_{(j_1,j_2)}} + \Vert x_{0,2}-\mathfrak{m}_{0,2} \Vert_{\XX(S_{0,2})}.
\end{aligned} \label{EQestimatePPJJsmoothness1stder}
\end{align}
\end{proposition}

\subsubsection{Linearization of $\PP_f$ and $\PP_{(j^1,j^2)}$ at Minkowski} \label{SECPerturbationslinearizesd222111} 

In this section we state the linearization of $\PP_f$ and $\PP_{(j^1,j^2)}$ at Minkowski. For a proof we refer to Lemmas \ref{LEMlinearizedTransversalSCHWARZSCHILD} and \ref{LEMspherediffLINSCHWARZSCHILD} in Appendix \ref{SEClinearizedCHARGEequationsSSAPP} where, more generally, their linearization at Schwarzschild of mass $M\geq0$ is calculated. The linearizations are, by construction, closely related to the \emph{linearized pure gauge solutions} of \cite{DHR}. 

First, we have the following lemma for $\PP_f$.
\begin{lemma}[Linearization of $\PP_f$] \label{LEMlinearizedTransversal} 

Let $\dot{\PP}_{f}$ denote the linearization of $\PP_{f}$ in $f$ at $f=0$ and Minkowski. For a given linearized perturbation function $\dot f$,
\begin{align*} 
\begin{aligned} 
\dot f := \lrpar{\dot{f}(0,\th^1,\th^2), \pr_u \dot{f}(0,\th^1,\th^2), \,\, \pr_u^2 \dot{f}(0,\th^1,\th^2), \,\, \pr_u^3 \dot{f}(0,\th^1,\th^2) },
\end{aligned} 
\end{align*}
the non-trivial components of $\dot{\PP}_{f}\lrpar{\dot{f}}$ are given by
\begin{align*} \begin{aligned} 
\Omd =& \frac{1}{2} \pr_{u} \lrpar{\dot{f}}, & \phid=&- \dot{f}, & \etad=&  r \di \lrpar{\pr_u \lrpar{\frac{f}{r}}},\\
\chihd =& - 2 \DDd_2^\ast \di \dot{f}, & \omtrchibd =& -2\pr_u \lrpar{\frac{f}{r}}, & \omtrchid=& \frac{2}{r^2} \lrpar{\Ldo+1} \dot{f},
\end{aligned} \end{align*}
and
\begin{align*} 
\begin{aligned} 
\ombd =\pr_u\lrpar{\frac{1}{2} \pr_{u}\dot{f} }, \,\,
\Du\ombd = \pr_u^2 \lrpar{\frac{1}{2} \pr_{u} \lrpar{\dot{f}}},
\end{aligned} 
\end{align*}
where we tacitly evaluated at $u=0$.
\end{lemma}

\ni Second, we have the following lemma for $\PP_{(j^1,j^2)}$. It is a corollary of Lemma 6.1.3 in \cite{DHR}, where we note that the proof in \cite{DHR} at Schwarzschild (see also Lemma \ref{LEMspherediffLINSCHWARZSCHILD} in Appendix \ref{SEClinearizedCHARGEequationsSSAPP}) also goes through at Minkowski, and our notation connects to \cite{DHR} as follows,
\begin{align*} 
\begin{aligned} 
\widehat{\dot{\gd}} = r_M^2 \gdcd, \,\, \frac{\dot{\sqrt{\det\gd}}}{\sqrt{\det\gd}} = \frac{2\phid}{r_M}.
\end{aligned} 
\end{align*}
\begin{lemma}[Linearized angular perturbations] \label{LEMspherediffLIN} Let $\dot{\PP}_{(j_1,j_2)}$ denote the linearization of $\PP_{(j_1,j_2)}$ in $(j_1,j_2)$ at $(j_1,j_2)=(0,0)$ and Minkowski. The nontrivial components of $\dot{\PP}_{(j_1,j_2)} (\dot{j}_1,\dot{j}_2)$ are given by
\begin{align*}
\phid= \frac{r}{2} \Ldo \dot{q}_1, \,\, \gdcd = 2 \DDd_2^\ast \DDd_1^\ast (\dot{q}_1,\dot{q}_2),
\end{align*}
where the scalar functions $q_1$ and $q_2$ on $S_2$ are related to $\dot{j}^1$ and $\dot{j}^2$ by
\begin{align*} 
\begin{aligned} 
\dot{j}^1 d\th^1 + \dot{j}^2 d\th^2 = - r^2 \DDd_1^\ast(\dot{q}_1,\dot{q}_2).
\end{aligned} 
\end{align*}
\end{lemma}

\ni From Lemmas \ref{LEMlinearizedTransversal} and \ref{LEMspherediffLIN} we directly conclude the following.
\begin{lemma}[Boundedness of linearized perturbations of sphere data] \label{LEMestimatesSpherePerturbations} For real numbers $M\geq0$ sufficiently small, it holds that the linearizations $\dot{\PP}^M_f$ and $\dot{\PP}^M_{(j^1,j^2)}$ are bounded,
\begin{align*} 
\begin{aligned} 
\Vert \dot{\PP}^M_f(\dot f) \Vert_{\XX(S_{0,2})} \les \Vert \dot f \Vert_{\YY_f}, \,\, \Vert \dot{\PP}^M_{(j^1,j^2)}(\dot{j}^1,\dot{j}^2) \Vert_{\XX(S_{0,2})} \les \Vert (\dot{j}^1,\dot{j}^2) \Vert_{\YY_{({j}^1,{j}^2)}}.
\end{aligned} 
\end{align*}
\end{lemma}

\subsection{Implicit function theorem} \label{SECiftEstimates} The proof of the main result of this paper is based on the standard implicit function theorem. In the following we recall its statement and provide further estimates which are applied in Section \ref{SECconclusion}.

In the following, for Hilbert spaces $X$ and $Z$ and integers $r\geq0$, let $C^r(X;Z)$ denote the space of $r$-times continuously differentiable maps from $X$ to $Z$. Denote the standard norm of this space by $\Vert \cdot \Vert_{C^r(X;Z)}$.

The standard implicit function theorem is as follows, see, for example, Theorem 2.5.7 in \cite{MarsdenImplicit}.
\begin{theorem}[Implicit function theorem] \label{thm:InverseMars14}
Let $X,Y$ and $Z$ be Hilbert spaces. Let $U \subset X$ and $V \subset Y$ be open subsets and let $\FF: U \times V \to Z$ be a $C^r$-mapping for some integer $r\geq1$. Assume that for some $x_0 \in U$ and $y_0 \in V$, the linearization $$D_x\FF\vert_{(x_0,y_0)}: X \to Z$$ is an isomorphism. Then there exist open neighbourhoods $V_0 \subset V$ of $y_0$ and $W_0 \subset Z$ of $\FF(x_0,y_0)$ as well as a unique $C^r$-mapping $\GG: V_0 \times W_0 \to U$ such that for $(y,z) \in V_0 \times W_0$,
\begin{align*}
\FF(\GG(y,z),y)=z.
\end{align*}
\end{theorem}

\ni We further state the following standard calculus estimate.
\begin{lemma}[Calculus estimate] \label{LEMoperatorEstimates} Let $X,Y$ and $Z$ be Hilbert spaces. Let $U \subset X$ and $V \subset Y$ be open subsets around $x_0 \in X$ and $y_0 \in Y$, respectively. Let $\FF: U \times V \to Z$ be a $C^r$-mapping for an integer $r\geq1$. Then
\begin{align*} 
\begin{aligned} 
\Vert \FF(x,y) - \FF(x_0,y_0) \Vert_{Z} \les \Vert x-x_0 \Vert_X+ \Vert y-y_0 \Vert_Y,
\end{aligned} 
\end{align*}
where the constant depends on $\FF$.
\end{lemma}

\begin{proof}[Proof of Lemma \ref{LEMoperatorEstimates}] By the fundamental theorem of calculus,
\begin{align*} 
\begin{aligned} 
\FF(x,y) - \FF(x_0,y_0) =& \lrpar{\FF(x,y)- \FF(x_0,y)} -\lrpar{ \FF(x_0,y_0)-\FF(x_0,y) } \\
=& \int\limits_0^1 D_x \FF\vert_{(xt+(1-t)x_0,y)} \lrpar{x-x_0} dt + \int\limits_0^1 D_y \FF\vert_{(x_0,yt+(1-t)y_0)} \lrpar{y-y_0} dt.
\end{aligned} 
\end{align*}
Taking the norm of the above shows that
\begin{align*} 
\begin{aligned} 
\Vert \FF(x,y) - \FF(x_0,y_0) \Vert_Z \les& \Vert D_x \FF \Vert_{C^0(U\times V;Z)} \Vert x-x_0 \Vert_X + \Vert D_y \FF \Vert_{C^0(U\times V; Z)} \Vert y-y_0 \Vert_Y \\
\les& \Vert x-x_0 \Vert_X + \Vert y-y_0 \Vert_Y.
\end{aligned} 
\end{align*}
This finishes the proof of Lemma \ref{LEMoperatorEstimates}. \end{proof}

\subsection{Notation for characteristic gluing of higher-order derivatives} \label{SECnotationHIGHER} As remarked in the introduction, our main characteristic gluing result, Theorem \ref{PROPNLgluingOrthA121}, can be generalized to glue higher-order derivatives, see Theorems \ref{THMHIGHERorderLgluingMAINTHEOREM} and \ref{THMtransversalHIGHERv2} in Section \ref{SECpreciseStatementMainTheorem}. In this section we introduce the necessary notation to precisely state these results.

First, we define higher-order sphere data (for $C^{m+2}$-gluing of metric components) as generalization of Definition \ref{DEFspheredata2}.
\begin{definition}[Higher-order sphere data] \label{DEFhigherORDERsphereDATA7778} Let $m\geq1$ be an integer, and let $x_{u,v}$ be $C^2$-sphere data on a $2$-sphere $S_{u,v}$. We define 
\begin{itemize}
\item \emph{higher-order $L$-derivatives sphere data of order $m\geq1$} to be the pair
\begin{align*} 
\begin{aligned} 
(x_{u,v}, \DD^{L,m}_{u,v}),
\end{aligned} 
\end{align*}
where $\DD^{L,m}_{u,v}$ is the following tuple of tensors,
\begin{align} 
\begin{aligned} 
\DD^{L,m}_{u,v} = \lrpar{\widehat{D}\a, \dots, \widehat{D}^m\a,D^2\om, \dots, D^{m+1}\om},
\end{aligned} \label{EQdefDDLtuple}
\end{align}
where $\widehat{D}^j\a$, $1\leq j \leq m$, are $\gd$-tracefree symmetric $2$-tensors on $S_{u,v}$ and $D^j\om$, $2\leq j \leq m+1$, are scalar functions on $S_{u,v}$.

\item \emph{higher-order $\Lb$-derivatives sphere data of order $m\geq1$} to be the pair
\begin{align*} 
\begin{aligned} 
(x_{u,v}, \DD^{\Lb,m}_{u,v}),
\end{aligned} 
\end{align*}
where $\DD^{\Lb,m}_{u,v}$ is a tuple of tensors,
\begin{align} 
\begin{aligned} 
\DD^{\Lb,m}_{u,v} = \lrpar{\widehat{\Du}\ab, \dots, \widehat{\Du}^m\ab,\Du^2\omb, \dots, \Du^{m+1}\omb},
\end{aligned} \label{EQdefDDLbtuple}
\end{align}
where $\widehat{\Du}^j\ab$, $1\leq j \leq m$, are $\gd$-tracefree symmetric $2$-tensors on $S_{u,v}$ and $\Du^j\omb$, $2\leq j \leq m+1$, are scalar functions on $S_{u,v}$.

\item \emph{higher-order sphere data of order $m\geq1$} to be the triplet
\begin{align*} 
\begin{aligned} 
(x_{u,v}, \DD^{L,m}_{u,v}, \DD^{\Lb,m}_{u,v}),
\end{aligned} 
\end{align*}
where $\DD^{L,m}_{u,v}$ and $\DD^{\Lb,m}_{u,v}$ are tuples as in \eqref{EQdefDDLtuple} and \eqref{EQdefDDLbtuple}.

\end{itemize}
\end{definition}

\begin{remark} By the Einstein vacuum equations (see also the \emph{higher-order null structure equations} below) higher-order $L$-derivatives ($\Lb$-derivatives) sphere data determines higher-order $L$-derivatives ($\Lb$-derivatives) of the metric components, Ricci coefficients and null curvature components on the sphere. In this sense, higher-order sphere data as defined in Definition \ref{DEFhigherORDERsphereDATA7778} is appropriate for the higher-order characteristic gluing problem (at the level of $C^{m+2}$ for metric components).
\end{remark}

\ni Consider outgoing null data $x$ on $\HH_{0,[1,2]}$ (see Definition \ref{DEFnulldata111}). By applying $D$-derivatives, the null data $x$ determines the following tuple,
\begin{align*} 
\begin{aligned} 
(x,\DD^{L,m}) \text{ on } \HH_{0,[1,2]},
\end{aligned} 
\end{align*}
where $\DD^{L,m}$ is as in \eqref{EQdefDDLbtuple}. Similarly, ingoing null data $\underline{x}$ on $\HHb_{ [-\de,\de],2}$ (see Definition \ref{DEFnulldata111}) determines the tuple
\begin{align*} 
\begin{aligned} 
(\underline{x},\underline{\DD}^{\Lb,m}) \text{ on } \HHb_{ [-\de,\de],2},
\end{aligned} 
\end{align*}
where $\underline{\DD}^{\Lb,m}$ denotes the tuple of derivatives on the right-hand side of \eqref{EQdefDDLbtuple}.

In the following we define higher-order null data on $\HH_{0,[1,2]}$ and $\HHb_{ [-\de,\de],2}$.
\begin{definition}[Higher-order null data] \label{DEFhigherORDERoutgoingingoingnullDATA7778}
Let $m\geq1$ be an integer. We define 
\begin{itemize}
\item \emph{higher-order outgoing null data of order $m$} to be the triple
\begin{align*} 
\begin{aligned} 
(x,\DD^{L,m}, \DD^{\Lb,m}) \text{ on } \HH_{0,[1,2]}
\end{aligned} 
\end{align*}
such that for each $S_{0,v} \subset \HH_{0,[1,2]}$, $(x,\DD^{L,m}, \DD^{\Lb,m})_{0,v}$ is higher-order sphere data of order $m$. 
\item \emph{higher-order ingoing null data of order $m$} to be the triple
\begin{align*} 
\begin{aligned} 
(\underline{x},\underline{\DD}^{L,m}, \underline\DD^{\Lb,m}) \text{ on } \HHb_{[-\de,\de],2}
\end{aligned} 
\end{align*}
such that for each $S_{u,2} \subset \HHb_{ [-\de,\de],2}$, $(\underline{x},\underline{\DD}^{L,m}, \underline\DD^{\Lb,m})_{u,2}$ is higher-order sphere data of order $m$. Here $\underline{\DD}^{L,m}$ denotes the tuple of tensors on the right-hand side of \eqref{EQdefDDLtuple}.
\end{itemize}
\end{definition}

\ni In addition to the constraints equations 
\begin{align*} 
\begin{aligned} 
\CC_i(x) =0 \text{ for } 1\leq i \leq 10,
\end{aligned} 
\end{align*}
for $x$ on $\HH_{0,[1,2]}$ and their direct implications for $\DD^{L,m}$, the Einstein equations \eqref{EQeve} also imply null transport equations for $\DD^{\Lb,m}$ along $\HH_{0,[1,2]}$; see, for example, the null transport equation \eqref{EQDUOMU1} for $\Du\omb$ along $\HH_{0,[1,2]}$. We call these null transport equations the \emph{higher-order null structure equations}, and formally denote them by
\begin{align*} 
\begin{aligned} 
\CC\lrpar{x,\DD^{L,m},\DD^{\Lb,m}} =0 \text{ on } \HH_{0,[1,2]}.
\end{aligned} 
\end{align*}
Similarly, we denote the higher-order null structure equations on $\HHb_{ [-\de,\de],2}$ by 
\begin{align*} 
\begin{aligned} 
\underline{\CC}\lrpar{\underline{x},\underline{\DD}^{L,m},\underline{\DD}^{\Lb,m}} =0 \text{ on } \HHb_{ [-\de,\de],2}.
\end{aligned} 
\end{align*}
The higher-order null structure equations are relevant for higher-order conservation laws which act as obstructions to higher-order characteristic gluing, see also Theorem \ref{THMtransversalHIGHERv2}.

\begin{remark}[Higher-order sphere data perturbations] \label{REMhigherOrderSpherePerturbations} In the context of higher-order sphere data, the sphere perturbations ${\PP}_{{f}}$ and ${\PP}_{({j}^1,{j}^2)}$ can be straight-forward generalized to smooth mappings of higher-order incoming null data $(\underline{x}, \underline{\DD}^{L,m}, \underline{\DD}^{\Lb,m})$ on $\HHb_{[-\de,\de],2}$ to higher-order $L$-derivatives sphere data $(x_{0,2}, \DD^{L,m}_{0,2})$ on $S_{0,2}$. It is worthwile to note that linearly at Minkowski, $\DD^{L,m}$ is invariant under sphere variations, that is, for all $\dot{f}$ and $(\dot{j}^1,\dot{j}^2)$,
\begin{align*} 
\begin{aligned} 
\dot{\DD}_L \lrpar{\dot{\PP}_f(\dot{f})+\dot{\PP}_{(j^1,j^2)}(\dot{j}^1,\dot{j}^2)} =0.
\end{aligned} 
\end{align*}
Indeed, this follows by the direct relation between our $\dot{\PP}_{{f}}$ and $\dot{\PP}_{({j}^1,{j}^2)}$ and Lemmas 6.1.2 and 6.1.3 in \cite{DHR} (as in the proof of Lemma \ref{LEMlinearizedTransversal}), where in the latter it is shown that
\begin{align*} 
\begin{aligned} 
\ad \equiv 0, \,\, \omd \equiv 0,
\end{aligned} 
\end{align*}
which implies that all higher $D$-derivatives vanish on $S_{0,2}$.
\end{remark}

\section{Statement of main results} \label{SECpreciseStatementMainTheorem}

\ni The following is the main theorem of this paper.
\begin{theorem}[Codimension-$10$ perturbative characteristic gluing, version 2] \label{PROPNLgluingOrthA121} Let $\de>0$ be a real number. Consider sphere data $x_{0,1}$ on $S_{0,1}$, and sphere data $\tilde{x}_{0,2}$ on $\tilde{S}_{0,2}$ contained in ingoing null data $\tilde{\underline{x}}$ on $\tilde{\HHb}_{[-\de,\de],2}$ satisfying the null structure equations. Assume that for some real number $\varep>0$, 
\begin{align} 
\begin{aligned} 
\Vert x_{0,1}-\mathfrak{m}^M \Vert_{\mathcal{X}(S_{0,1})} + \Vert \tilde{\underline{x}}-\underline{\mathfrak{m}}^M \Vert_{\mathcal{X}^+(\tilde{\HHb}_{[-\de,\de],2}) } \leq \varep.
\end{aligned} \label{EQsmallnessAssumptionMainTheoremMAIN222}
\end{align}
There exist universal real numbers $M_0>0$ and $\varep_0>0$ such that for all real numbers $0\leq M < M_0$ and $0<\varep<\varep_0$ sufficiently small, there are
\begin{itemize}
\item a solution $x$ to the null structure equations on on $\HH_{0,[1,2]}$,
\item sphere data $x_{0,2}$ on a sphere $S_{0,2} \subset \tilde{\HHb}_{[-\de,\de],2}$ stemming from a perturbation of $\tilde{S}_{0,2}$, that is, there are perturbation functions $f$ and $(j^1,j^2)$ such that
\begin{align*} 
\begin{aligned} 
x_{0,2} = \PP_{(j^1,j^2)} \PP_f(\tilde{\underline{x}}),
\end{aligned} 
\end{align*}
\end{itemize}
such that on $S_{0,1}$ we have matching of sphere data,
\begin{align} 
\begin{aligned} 
 &x \vert_{S_{0,1}} = x_{0,1}, 
\end{aligned} \label{EQmaintheoremMATCHINGcond4}
\end{align}
and on $S_{0,2}$ we have matching up to the charges $(\mathbf{E},\mathbf{P}, \mathbf{L}, \mathbf{G})$, that is, if 
\begin{align} 
\begin{aligned} 
\lrpar{\mathbf{E},\mathbf{P}, \mathbf{L}, \mathbf{G}}(x \vert_{S_{0,2}}) =\lrpar{\mathbf{E},\mathbf{P}, \mathbf{L}, \mathbf{G}}\lrpar{x_{0,2}},
\end{aligned} \label{EQmainTheoremmatchingcondition122}
\end{align}
then it holds that
\begin{align} 
\begin{aligned} 
x \vert_{S_{0,2}} = x_{0,2}.
\end{aligned} \label{EQspheredatamatchingMAINTHEOREMS02}
\end{align}
Moreover, the following bounds hold, 
\begin{align} 
\begin{aligned} 
\Vert x -\mathfrak{m}^M \Vert_{\XX(\HH_{0,[1,2]})}+ \Vert x_{0,2} -\tilde{\underline{x}}_{0,2} \Vert_{\XX(S_{0,2})}\les& \varep, \\
\Vert f \Vert_{\YY_f} + \Vert (j^1,j^2) \Vert_{\YY_{(j^1,j^2)}}  \les& \varep, \label{EQboundsFullRESULTMAINTHM1}
\end{aligned} 
\end{align}
where we denoted $\tilde{\underline{x}}_{0,2} := \tilde{\underline{x}} \vert_{S_{0,2}}$. Furthermore, we have the perturbation estimate
\begin{align} 
\begin{aligned} 
\left\vert \lrpar{\mathbf{E},\mathbf{P},\mathbf{L},\mathbf{G}}\lrpar{x_{0,2}} - \lrpar{\mathbf{E},\mathbf{P},\mathbf{L},\mathbf{G}}\lrpar{\tilde{\underline{x}}_{0,2}} \right\vert \les \varep M+ \varep^2,
\end{aligned} \label{EQChargeEstimatesMainTheorem0}
\end{align}
and the transport estimate
\begin{align} 
\begin{aligned}
\left\vert \lrpar{\mathbf{E},\mathbf{P},\mathbf{L},\mathbf{G}}\lrpar{x \vert_{S_{0,2}}} - \lrpar{\mathbf{E},\mathbf{P},\mathbf{L},\mathbf{G}}\lrpar{x \vert_{S_{0,1}}}\right\vert \les \varep M+ \varep^2.
\end{aligned} \label{EQChargeEstimatesMainTheorem}
\end{align}

\end{theorem}

\ni \emph{Remarks on Theorem \ref{PROPNLgluingOrthA121}.}
\begin{enumerate}

\item The matching on $S_{0,2}$ can be described more precisely as follows. The solution $x$ constructed in Theorem \ref{PROPNLgluingOrthA121} is such that on $S_{0,2}$,
\begin{align} 
\begin{aligned} 
\MMf\lrpar{x \vert_{S_{0,2}}} =\MMf(x_{0,2}),
\end{aligned} \label{EQmatchingRemark77781}
\end{align}
where $\MMf$ denotes the matching map introduced in Definition \ref{DEFmatchingMAP}. Lemma \ref{LEMconditionalMATCHING} implies then that if in addition to \eqref{EQmatchingRemark77781} we have the matching of charges \eqref{EQmainTheoremmatchingcondition122}, then we have the sphere data matching \eqref{EQspheredatamatchingMAINTHEOREMS02}.

\item A straight-forward inspection of the proof of Theorem \ref{PROPNLgluingOrthA121} shows that the angular regularity of the characteristic gluing can be increased without change to the proof.

\item Theorem \ref{PROPNLgluingOrthA121} can be equivalently stated with ingoing null data on $\tilde{\HHb}_{[-\de,\de],0}$ and sphere data on $S_{0,2}$. This alternative formulation of Theorem \ref{PROPNLgluingOrthA121} is used in \cite{ACR3,ACR2}.

\item The constructed solution $x \in \XX(\HH_{0,[1,2]})$ is sufficiently regular for the application of local existence results for the characteristic initial value problem; see \cite{ACR3,ACR2} for further discussion. The gluing of Theorem \ref{PROPNLgluingOrthA121} is at the level of $C^2$-gluing for metric components.

\item The construction of the solution $x$ in Theorem \ref{PROPNLgluingOrthA121} is based on the implicit function theorem and solving the \emph{linearized characteristic gluing problem at Minkowski} in Section \ref{SEClinearizedProblem}. The perturbation and transport estimates \eqref{EQChargeEstimatesMainTheorem0} and \eqref{EQChargeEstimatesMainTheorem} for the charges $(\mathbf{E},\mathbf{P},\mathbf{L},\mathbf{G})$ require further an analysis of the linearized sphere perturbations, angular perturbations and null transport equations for $(\mathbf{E},\mathbf{P},\mathbf{L},\mathbf{G})$ \emph{at Schwarzschild of mass $M\geq0$} provided in Appendix \ref{SEClinearizedCHARGEequationsSSAPP}.

\end{enumerate}

\ni In addition to the gluing of higher-order angular derivatives, the proof of Theorem \ref{PROPNLgluingOrthA121} accomodates also the characteristic gluing of higher-order $L$-derivatives in a straight-forward way. The precise statement is as follows; for ease of presentation, we state it with smooth sphere data.
\begin{theorem}[Codimension-$10$ perturbative characteristic gluing of higher-order $L$-derivatives] \label{THMHIGHERorderLgluingMAINTHEOREM} Let $\de>0$ be a real number and let $m\geq1$ be an integer. 
Consider smooth higher-order $L$-derivatives sphere data $(x_{0,1},\DD^{L,m}_{0,1})$ on $S_{0,1}$ and smooth higher-order ingoing null data $(\tilde{\underline{x}},\tilde{\underline{\DD}}^{\tilde{L},m}, \tilde{\underline{\DD}}^{\tilde{\Lb},m})$ solving the higher-order null structure equations on $\tilde{\HHb}_{[-\de,\de],2}$. For $(x_{0,1},\DD^{L,m}_{0,1})$ and $(\tilde{\underline{x}},\tilde{\underline{\DD}}^{\tilde{L},m}, \tilde{\underline{\DD}}^{\tilde{\Lb},m})$ sufficiently close to the their respective reference values in a Schwarzschild spacetime of sufficiently small mass $M\geq0$, there are
\begin{itemize}
\item a smooth solution $(x,\DD^{L,m})$ to the null structure equations on $\HH_{0,[1,2]}$, 
\item smooth higher-order $L$-derivatives sphere data $(x_{0,2},\DD^{L,m}_{0,2})$ stemming from a perturbation of $\tilde{S}_{0,2}$ in $\tilde{\HHb}_{[-\de,\de],2}$, that is, there are perturbation functions $f$ and $(j^1,j^2)$ such that
\begin{align*} 
\begin{aligned} 
(x_{0,2},\DD^{L,m}) = \PP_{(j^1,j^2)} \PP_f\lrpar{\tilde{\underline{x}},\tilde{\underline{\DD}}^{\tilde{L},m}, \tilde{\underline{\DD}}^{\tilde{\Lb},m}},
\end{aligned} 
\end{align*}
\end{itemize}
such that on $S_{0,1}$ we have the matching
\begin{align} 
\begin{aligned} 
(x,\DD^{L,m}) \vert_{S_{0,1}} = (x_{0,1},\DD^{L,m}_{0,1}),
\end{aligned} \label{EQmaintheoremMATCHINGcond4777}
\end{align}
and on $S_{0,2}$ we have matching up to the charges $(\mathbf{E},\mathbf{P}, \mathbf{L}, \mathbf{G})$, that is, if 
\begin{align*} 
\begin{aligned} 
\lrpar{\mathbf{E},\mathbf{P}, \mathbf{L}, \mathbf{G}}(x \vert_{S_{0,2}}) =\lrpar{\mathbf{E},\mathbf{P}, \mathbf{L}, \mathbf{G}}\lrpar{x_{0,2}},
\end{aligned} 
\end{align*}
then it holds that
\begin{align} 
\begin{aligned} 
(x,\DD^{L,m}) \vert_{S_{0,2}} = (x_{0,2},\DD^{L,m}_{0,2}).
\end{aligned} \label{EQspheredatamatchingMAINTHEOREMS02677747}
\end{align}
Moreover, we have charge estimates analogous to Theorem \ref{PROPNLgluingOrthA121}.
\end{theorem}

\ni \emph{Remarks on Theorem \ref{THMHIGHERorderLgluingMAINTHEOREM}.}
\begin{enumerate}
\item The proof of Theorem \ref{THMHIGHERorderLgluingMAINTHEOREM} is a direct generalization of the proof of Theorem \ref{PROPNLgluingOrthA121}. We indicate the necessary generalizations in Section \ref{SECconclusion4} and Section \ref{SEClinearizedProblem}.

\item The matching on $S_{0,2}$ in Theorem \ref{THMHIGHERorderLgluingMAINTHEOREM} is more precisely given by
\begin{align} 
\begin{aligned} 
\MMf\lrpar{x \vert_{S_{0,2}}} =\MMf(x_{0,2}), \DD^{L,m} \vert_{S_{0,2}} = \DD^{L,m}_{0,2}.
\end{aligned} \label{EQmatchingRemark77781888}
\end{align}
In particular, as mentioned before, the gluing of 
\begin{align*} 
\begin{aligned} 
\DD^{L,m} := \lrpar{\widehat{D}\a, \dots, \widehat{D}^m\a,D^2\om, \dots, D^{m+1}\om},
\end{aligned} 
\end{align*}
is without obstacles. 

\item To deduce from \eqref{EQmatchingRemark77781888} the gluing of higher-order $L$-derivatives of metric components, Ricci coefficients and null curvature components, the full matching \eqref{EQspheredatamatchingMAINTHEOREMS02677747} is needed to apply the higher-order null structure equations on $S_{0,2}$.
\end{enumerate}

\ni Theorem \ref{THMHIGHERorderLgluingMAINTHEOREM} shows that $\DD^{L,m}$ can be characteristically glued without obstruction, while $\DD^{\Lb,m}$ is subject to higher-order conservation laws along $\HH_{0,[1,2]}$. In the following theorem we show that by gluing along \emph{two null hypersurfaces bifurcating from an auxiliary sphere}, namely along $\HH_{[-1,0],1}$ and $\HH_{0,[1,2]}$, it is possible to glue higher-order $L$- and $\Lb$-derivatives.

\begin{theorem}[Codimension-$10$ bifurcate characteristic gluing, version 2] \label{THMtransversalHIGHERv2} Let $m\geq0$ be an integer. Consider smooth higher-order sphere data 
\begin{align*} 
\begin{aligned} 
(x_{0,1},\DD^{L,m}_{0,1}, \DD^{\Lb,m}_{0,1}) \text{ on } S_{0,1} \text{ and } (x_{-1,2},\DD^{L,m}_{-1,2},\DD^{\Lb,m}_{-1,2}) \text{ on } S_{-1,2}.
\end{aligned} 
\end{align*}
For $(x_{0,1},\DD^{L,m}_{0,1}, \DD^{\Lb,m}_{0,1})$ and $(x_{-1,2},\DD^{L,m}_{-1,2},\DD^{\Lb,m}_{-1,2})$ sufficiently close to the their respective reference values in a Schwarzschild spacetime of sufficiently small mass $M\geq0$, there are

\begin{itemize}
\item a smooth solution $(\underline{x},\underline{\DD}^{L,m},\underline{\DD}^{\Lb,m})$ to the higher-order null structure equations on $\HHb_{[-1,0],1}$, satisfying higher-order sphere data matching on $S_{0,1}$,
\begin{align*} 
\begin{aligned} 
(\underline{x},\underline{\DD}^{L,m},\underline{\DD}^{\Lb,m}) \Big\vert_{S_{0,1}}=& (x_{0,1}, \DD^{L,m}_{0,1},\DD^{\Lb,m}_{0,1}),
\end{aligned} 
\end{align*} 
\item a smooth solution $({x},{\DD}^{L,m},{\DD}^{\Lb,m})$ to the higher-order null structure equations on $\HH_{-1,[1,2]}$, matching with $(\underline{x},\underline{\DD})$ on $S_{-1,1}$,
\begin{align*} 
\begin{aligned} 
(\underline{x},\underline{\DD}^{L,m},\underline{\DD}^{\Lb,m}) \Big\vert_{S_{-1,1}} =& ({x},{\DD}^{L,m},{\DD}^{\Lb,m}) \Big\vert_{S_{-1,1}}
\end{aligned} 
\end{align*}
\end{itemize}
such that $({x},{\DD}^{L,m},{\DD}^{\Lb,m})$ matches $(x_{-1,2},\DD^{L,m}_{-1,2},\DD^{\Lb,m}_{-1,2})$ up to the charges $(\mathbf{E},\mathbf{P},\mathbf{L},\mathbf{G})$ on $S_{-1,2}$, that is, if it holds that
\begin{align*} 
\begin{aligned} 
\lrpar{\mathbf{E}, \mathbf{P}, \mathbf{L}, \mathbf{G}}\lrpar{x\vert_{S_{-1,2}}} = \lrpar{\mathbf{E}, \mathbf{P}, \mathbf{L}, \mathbf{G}}\lrpar{x_{-1,2}},
\end{aligned} 
\end{align*}
then 
\begin{align*} 
\begin{aligned} 
({x},{\DD}^{L,m},{\DD}^{\Lb,m}) \Big\vert_{S_{-1,2}} =& (x_{-1,2},\DD^{L,m}_{-1,2},\DD^{\Lb,m}_{-1,2}).
\end{aligned}
\end{align*}
Moreover, we have charge estimates analogous to \eqref{EQChargeEstimatesMainTheorem} in Theorem \ref{PROPNLgluingOrthA121} for
\begin{align*} 
\begin{aligned} 
\left\vert (\mathbf{E},\mathbf{P},\mathbf{L},\mathbf{G})(x \vert_{S_{-1,2}})- \lrpar{\mathbf{E},\mathbf{P},\mathbf{L},\mathbf{G}-2 \mathbf{P}}(x_{0,1})\right\vert.
\end{aligned} 
\end{align*}
\end{theorem}

\ni \emph{Remarks on Theorem \ref{THMtransversalHIGHERv2}.}
\begin{enumerate}
\item Theorem \ref{THMtransversalHIGHERv2} shows that for bifurcate characteristic gluing, the obstruction space consists entirely of the $10$-dimensional space $(\mathbf{E},\mathbf{P},\mathbf{L},\mathbf{G})$.
\item The proof of Theorem \ref{THMtransversalHIGHERv2} is based on the methods of Theorems \ref{PROPNLgluingOrthA121} and \ref{THMHIGHERorderLgluingMAINTHEOREM}, and given in Section \ref{SECHIGHERfull}.

\item In \cite{ACR3,ACR2}, Theorem \ref{THMtransversalHIGHERv2} is applied to glue spacelike initial data to a Kerr black hole spacetime.

\item It is possible to prove that, with respect to appropriately defined norms, the necessary closeness to Schwarzschild can be uniformly bounded away from zero for all $m\geq0$. We omit the statement and proof of $C^\infty$-gluing results.

\end{enumerate}

\ni The above results concern codimension-$10$ characteristic gluing. However, by adding to the sphere data on $S_{-1,2}$ a sphere data perturbation $W$ which adjusts the charges $(\mathbf{E},\mathbf{P},\mathbf{L},\mathbf{G})$ (this implies that $W$ is \emph{not} coming from a sphere perturbation or sphere diffeomorphism), it is straight-forward to extend the above results to full characteristic gluing of (higher-order) sphere data. The sphere data perturbation $W$ can be chosen to be supported on an arbitrary angular region $K$ on $S_{-1,2}$. An explicit such $W$ with advantageous properties is used in the gluing problems studied in \cite{ACR3}. The following result is an extension of Theorem \ref{THMtransversalHIGHERv2}; a variant can also be stated for characteristic gluing along one null hypersurface of Theorem \ref{PROPNLgluingOrthA121}.

\begin{proposition}[Bifurcate characteristic gluing with localized sphere data perturbation $W$] \label{PROPchargluingW} Let $m\geq0$ be an integer. Let $K \subset S_{-1,2}$ be an angular region. Consider smooth higher-order sphere data 
\begin{align*} 
\begin{aligned} 
(x_{0,1},\DD^{L,m}_{0,1}, \DD^{\Lb,m}_{0,1}) \text{ on } S_{0,1} \text{ and } (x_{-1,2},\DD^{L,m}_{-1,2},\DD^{\Lb,m}_{-1,2}) \text{ on } S_{-1,2}.
\end{aligned} 
\end{align*}
For $(x_{0,1},\DD^{L,m}_{0,1}, \DD^{\Lb,m}_{0,1})$ and $(x_{-1,2},\DD^{L,m}_{-1,2},\DD^{\Lb,m}_{-1,2})$ sufficiently close to the their respective reference values in a Schwarzschild spacetime of sufficiently small mass $M\geq0$, there are
\begin{itemize}
\item a smooth solution $(\underline{x},\underline{\DD}^{L,m},\underline{\DD}^{\Lb,m})$ to the higher-order null structure equations on $\HHb_{[-1,0],1}$, satisfying higher-order sphere data matching on $S_{0,1}$,
\begin{align*} 
\begin{aligned} 
(\underline{x},\underline{\DD}^{L,m},\underline{\DD}^{\Lb,m}) \Big\vert_{S_{0,1}}=& (x_{0,1}, \DD^{L,m}_{0,1},\DD^{\Lb,m}_{0,1}),
\end{aligned} 
\end{align*} 
\item a smooth solution $({x},{\DD}^{L,m},{\DD}^{\Lb,m})$ to the higher-order null structure equations on $\HH_{-1,[1,2]}$, matching with $(\underline{x},\underline{\DD})$ on $S_{-1,1}$,
\begin{align*} 
\begin{aligned} 
(\underline{x},\underline{\DD}^{L,m},\underline{\DD}^{\Lb,m}) \Big\vert_{S_{-1,1}} =& ({x},{\DD}^{L,m},{\DD}^{\Lb,m}) \Big\vert_{S_{-1,1}},
\end{aligned} 
\end{align*}

\item a smooth higher-order sphere data perturbation $(W,0,0)$, compactly supported in $K$,
\end{itemize}
such that 
\begin{align*} 
\begin{aligned} 
({x},{\DD}^{L,m},{\DD}^{\Lb,m}) \vert_{S_{-1,2}} = (x_{-1,2},\DD^{L,m}_{-1,2},\DD^{\Lb,m}_{-1,2}) + (W,0,0),
\end{aligned} 
\end{align*}
Moreover, we have appropriate bounds for $W$ on $S_{-1,2}$, $({x},{\DD}^{L,m},{\DD}^{\Lb,m})$ on $\HH_{-1,[1,2]}$, and $({x},{\DD}^{L,m},{\DD}^{\Lb,m})$ on $\HH_{-1,[1,2]}$.
\end{proposition}
\ni The proof of Proposition \ref{PROPchargluingW} is a slight generalization of the proof of Theorem \ref{THMtransversalHIGHERv2} based on the additional localized sphere perturbation $W$, see Section \ref{SECproofWgluing}.

\section{Linearized characteristic gluing at Minkowski} \label{SEClinearizedProblem} 

\ni The following is the main result of this section. It shows that the linearized characteristic gluing problem at Minkowski is solvable up to a $10$-dimensional space, and forms the basis for the proof of Theorem \ref{PROPNLgluingOrthA121} in Section \ref{SECLinearizedNullConstraintsAroundMinkowski}.

\begin{theorem}[Codimension-$10$ linearized characteristic gluing at Minkowski] \label{PROPlingluing22} Given 
\begin{itemize} 
\item linearized sphere data $\dot{\mathfrak{X}}_{0,1} \in \mathcal{X}(S_{0,1})$ at $S_{0,1}$,
\item linearized matching data $\dot{\MMf}_{0,2} \in \ZZ_\MMf(S_{0,2})$ at $S_{0,2}$,
\item linearized source terms $(\mfq_i)_{1\leq i \leq 10} \in \mathcal{Z}_\CC$ on $\HH_{0,[1,2]}$, 
\end{itemize} 
there exist 
\begin{itemize}
\item linearized null data $\xd \in \mathcal{X}(\HH_{0,[1,2]})$ on $\HH_{0,[1,2]}$,
\item linearized perturbation functions $\dot f$ and $(\dot{j}^1,\dot{j}^2)$ at $S_{0,2}$ 
\end{itemize}
such that
\begin{align*} 
\begin{aligned} 
\dot{\FF}^0(\xd, \dot{f}, (\dot{j}^1,\dot{j}^2)) = \lrpar{\dot{\mathfrak{X}}_{0,1}, \dot{\MMf}_{0,2}, (\mfq_i)_{1\leq i \leq 10}},
\end{aligned} 
\end{align*}
that is,
\begin{align} 
&\dot{\CC}_i(\xd)=\, \mfq_i \,\, \text{ on } \HH_{0,[1,2]} \text{ for } 1\leq i \leq 10, \label{EQlinearizedOPsystem2} \\
&\xd\vert_{S_{0,1}} =\, \dot{\mathfrak{X}}_{0,1}, \,\,  \MMf\lrpar{\xd \vert_{S_{0,2}}-\dot{\PP}^0_{(j^1,j^2)}(\dot{j}^1,\dot{j}^2) -\dot{\PP}^0_f\lrpar{\dot{f}}}=\, {\dot{\MMf}_{0,2}}. \label{EQmatchingLinSectionGOAL}
\end{align}
Moreover, the following estimate holds,
\begin{align} 
\begin{aligned} 
&\Vert \xd \Vert_{\mathcal{X}(\HH_{0,[1,2]})} + \Vert \dot f \Vert_{\YY_f} + \Vert (\dot{j}^1,\dot{j}^2) \Vert_{\YY_{({j}^1,{j}^2)}}+\left\Vert \dot{\PP}^0_{(j^1,j^2)}(\dot{j}_1,\dot{j}_2) \right\Vert_{\mathcal{X}(S_{0,2})}+\left\Vert \dot{\PP}^0_f\lrpar{\dot{f}} \right\Vert_{\mathcal{X}(S_{0,2})}\\
 \les& \Vert \dot{\mathfrak{X}}_{0,1} \Vert_{\mathcal{X}(S_{0,1})}+\Vert \dot{\MMf}_{0,2} \Vert_{\ZZ_\MMf(S_{0,2})}+ \Vert (\mfq_i)_{1\leq i \leq 10} \Vert_{\mathcal{Z}_\CC}.
\end{aligned} \label{EQlinearEstimateTOPROVE1233}
\end{align}
\end{theorem}

\ni We proceed as follows. 
\begin{enumerate} 

\item In Section \ref{SEClinearizedCHARGESMinkowskiSTEF89} we define the \emph{charges} $\QQ_i$ for $0\leq i \leq7$ which satisfy \emph{conservation laws} along $\HH$ by the linearized null structure equations at Minkowski.

\item In Section \ref{SECgaugeTheorySTEF89} we analyze how the charges $\QQ_i$ split into \emph{gauge-invariant} and \emph{gauge-dependent} charges.

\item In Section \ref{SECprelimAnalysis} we analyze the system \eqref{EQlinearizedOPsystem2}. We integrate the transport equations \eqref{EQlinearizedOPsystem2} to get representation formulas for metric coefficients, Ricci coefficients and null curvature components. 

\item In Section \ref{SECsolutionLinearized99902} we prove Theorem \ref{PROPlingluing22} by solving the linearized characteristic gluing problem.

\end{enumerate}

\begin{remark} \label{REMARKfreedatalinear} While in the non-linear setting we interpreted $\Om$ and $\gd_c$ as free conformal data, see Definition \ref{DEFconformalSEED}, in the linearized setting we choose $\Omd$ and $\chihd$ as degrees of freedom. Indeed, by the linearized equation
$$D\gdcd = \frac{2}{r^2} \chihd + \frac{1}{r^2} {\mfq_3}$$
the two approaches are equivalent at the linear level. Consequently, in our approach to the linearized gluing problem, the gluing of the following quantities is trivial,
$$\lrpar{\Omd,  \omd,  D\omd, \chihd,  D\chihd}.$$ \end{remark}

\begin{remark}[Linearized characteristic gluing of higher-order $L$-derivatives I] \label{REMhigherLgluing1} Let $m\geq1$ be an integer. In the following we also show that if one additionally prescribes the tuple
\begin{align*} 
\begin{aligned} 
\dot{\DD}^{L,m} :=& \lrpar{\widehat{D}\ad, \dots, \widehat{D}^m \ad, D^2\omd, \dots, D^{m+2} \omd}
\end{aligned} 
\end{align*}
on $S_{0,1}$ and $S_{0,2}$, denoted by $\dot{\DD}^{L,m}_{0,1}$ and $\dot{\DD}^{L,m}_{0,2}$, respectively, then we can construct $\xd$ such that it satisfies, in addition to \eqref{EQmatchingLinSectionGOAL},
\begin{align*} 
\begin{aligned} 
\dot{\DD}^{L,m}(\xd) \vert_{S_{0,1}} = \dot{\DD}^{L,m}_{0,1}, \,\, \dot{\DD}^{L,m}(\xd) \vert_{S_{0,2}} = \dot{\DD}^{L,m}_{0,2}.
\end{aligned} 
\end{align*}
In this setting, the right-hand side of \eqref{EQlinearEstimateTOPROVE1233} is replaced by
\begin{align*} 
\begin{aligned} 
&\Vert \dot{\mathfrak{X}}_1 \Vert_{\XX(S_{0,1})}+\Vert \dot{\MMf}_2 \Vert_{\ZZ_\MMf(S_{0,2})}+ \Vert (\mfq_i)_{1\leq i \leq 10} \Vert_{\mathcal{Z}_\CC} \\
&+ \sum\limits_{0\leq i \leq m} \lrpar{ \Vert \widehat{D}^i \ad \Vert_{H^6(S_{0,1})} + \Vert D^i \omd \Vert_{H^6(S_{0,1})} + \Vert \widehat{D}^i \ad \Vert_{H^6(S_{0,2})} + \Vert D^i \omd \Vert_{H^6(S_{0,2})} }.
\end{aligned} 
\end{align*}
\end{remark}

\subsection{Conserved charges $\QQ_i$ for the linearized equations} \label{SEClinearizedCHARGESMinkowskiSTEF89} The following \emph{charges} $\QQ_i$, $0\leq i \leq 7$, play an essential role in the characteristic gluing problem. In Section \ref{SECprelimAnalysis} we prove that the linearized null structure equations at Minkowski \eqref{EQlinearizedOPsystem2} (see also Lemma \ref{LEMlinearizedConstraints}) imply \emph{conservation laws} for the following \emph{charges}, see Lemmas \ref{CORtransport1EQS}, \ref{LEMgaugeconservationlaws} and \ref{CORtransportEQS3333}.
\begin{align} 
\begin{aligned} 
\QQ_0 :=& r^2 \etad^{[1]} + \frac{r^3}{2}\di \lrpar{\omtrchid^{[1]}-\frac{4}{r}\Omd^{[1]}}, \\
\QQ_1 :=&\frac{r}{2} \lrpar{\omtrchid-\frac{4}{r}\Omd} + \frac{\phid}{r}, \\
\QQ_2 :=& r^2 \omtrchibd -\frac{2}{r}\Divdo \lrpar{r^2\etad+\frac{r^3}{2}\di\lrpar{\omtrchid-\frac{4}{r}\Omd}} \\
&-r^2 \lrpar{\omtrchid-\frac{4}{r}\Omd} +2r^3 \Kd, \\
\QQ_3:=& \frac{\chibhd}{r} -\half \lrpar{ \DDd_2^\ast \Divdo +1} \gdcd + \DDd_2^\ast \lrpar{ \etad + \frac{r}{2}\di \lrpar{\omtrchid-\frac{4}{r}\Omd}} \\
&- r \DDd_2^\ast \di \lrpar{\omtrchid-\frac{4}{r}\Omd}, \\
\QQ_4 :=& \frac{\abd_\psi}{r} + 2 \DDd_2^\ast \lrpar{\frac{1}{r^2} \Divdo \chibhd - \frac{1}{r} \etad- \half \di \omtrchid + \DDd_1^\ast \lrpar{\ombd,0}}_{\psi},\\
\QQ_5 :=& \ombd^{[\leq1]} +\frac{1}{4r^2}\QQ_2^{[\leq1]} +\frac{1}{3r^3} \Divdo\QQ_0, \\
\QQ_6 :=& \Du\ombd^{[\leq1]} - \frac{1}{6r^3} (\Ldo-3) \QQ_2^{[\leq1]} +\frac{1}{r^4} \Divdo \QQ_0, \\
\QQ_7 :=& \Du\ombd^{[2]} +\frac{3}{2r^3} \QQ_2^{[2]}+ \frac{1}{2r^2} \Divdo \Divdo \QQ_3^{[2]} -\frac{12}{r^2} \QQ_1^{[2]} \\
&+{ \frac{3}{2r^2} \Divdo \lrpar{\eta+ \frac{r}{2}\di\lrpar{\omtrchid-\frac{4}{r}\Omd}} }^{[2]} - {\frac{3}{4r^2} \Divdo \Divdo \gdcd }^{[2]},
\end{aligned} \label{EQdefChargesMinkowski8891}
\end{align}
where $\psi$ denotes the electric part of a tracefree symmetric $2$-tensor, see Appendix \ref{SECellEstimatesSpheres}. \\

\ni \emph{Remarks on the charges $\QQ_i$, $0\leq i \leq 7$, in \eqref{EQdefChargesMinkowski8891}.}
\begin{enumerate}

\item By explicit calculation (see \eqref{EQlinEPLGKerr1} and \eqref{EQlinEPLGKerr2}), the linearizations $(\dot{\mathbf{E}}, \dot{\mathbf{P}},\dot{\mathbf{L}},\dot{\mathbf{G}})$ of $(\mathbf{E},\mathbf{P},\mathbf{L},\mathbf{G})$ at Minkowski are related to the above charges for $m=-1,0,1$ by
\begin{align} 
\begin{aligned} 
-\frac{8\pi}{\sqrt{4\pi}} \dot{\mathbf{E}} =& -\frac{1}{2} \QQ_2^{(0)}, & - \frac{8\pi}{\sqrt{\frac{4\pi}{3}}}\dot{\mathbf{P}}^m=& -\frac{1}{2} \QQ_2^{(1m)}, \\ 
\frac{16\pi}{\sqrt{\frac{8\pi}{3}}}\dot{\mathbf{L}}^m =& 2 (\QQ_0)_H^{(1m)}, & \frac{16\pi}{\sqrt{\frac{8\pi}{3}}}\dot{\mathbf{G}}^m =& 2 (\QQ_0)_E^{(1m)}.
\end{aligned} \label{EQREMchargeslinatMinkowski}
\end{align}

\item The $\QQ_5$ and $\QQ_6$ in \eqref{EQdefChargesMinkowski8891} are equal to the linearizations at Minkowski of $\tilde{\QQ}_5$ and $\tilde{\QQ}_6$ defined in context of the matching map $\MMf$, see Definition \ref{DEFmatchingMAP}.

\end{enumerate}

\subsection{Gauge dependence of the conserved charges $\QQ_i$} \label{SECgaugeTheorySTEF89}

In this section we show that the charges $\QQ_i$, $0\leq i \leq 7$, of Section \ref{SEClinearizedCHARGESMinkowskiSTEF89} split into \emph{gauge-invariant} and \emph{gauge-dependent} charges. The following is the main result of this section.

\begin{proposition} \label{PROPadjustmentCharges} The following holds. 
\begin{enumerate}
\item \textbf{Gauge-invariant charges.} For any linearized perturbation functions $\dot f$ and $(\dot{j}^1, \dot{j}^2)$ on $S_2$ it holds that
\begin{align*} 
\begin{aligned} 
\QQ_0\lrpar{\dot{\PP}_f(\dot f)+\dot{\PP}_{({j}^1,{j}^2)}(\dot{j}^1,\dot{j}^2)}=& 0, \\
\QQ_2^{[\leq 1]}\lrpar{\dot{\PP}_f(\dot f)+\dot{\PP}_{({j}^1,{j}^2)}(\dot{j}^1,\dot{j}^2)}=& 0,
\end{aligned} 
\end{align*}
that is, the $10$-dimensional space of charges $\QQ_0$ and $\QQ_2^{[\leq1]}$ is invariant under linearized perturbations $\dot{\PP}_f$ and $\dot{\PP}_{(j^1,j^2)}$ of sphere data.

\item \textbf{Gauge-dependent charges.} Let $(\QQ_1)_0$, $(\QQ_2)_0$, $(\QQ_5)_0$, $(\QQ_6)_0$ and $(\QQ_7)_0$ be scalar functions on $S_2$ such that
\begin{align} 
\begin{aligned} 
(\QQ_2)_0^{[\leq1]} = 0, \,\, (\QQ_5)_0^{[\geq2]} = 0, \,\,(\QQ_6)_0^{[\geq2]} = 0, \,\, (\QQ_7)_0=(\QQ_7)_0^{[2]},
\end{aligned} \label{EQchargematchingmodeconditions555}
\end{align}
and let $(\QQ_3)_0$ and $(\QQ_4)_0$ be symmetric tracefree $2$-tensors on $S_2$ such that
\begin{align} 
\begin{aligned} 
(\QQ_4)_0 =\lrpar{(\QQ_4)_0}_{\psi},
\end{aligned} \label{EQQ4electricAssumption}
\end{align}
where $\psi$ denotes the electric part of a symmetric tracefree $2$-tensor. Assume that
\begin{align*} 
\begin{aligned} 
\Vert (\QQ_1)_0 \Vert_{H^6(S_2)} + \Vert (\QQ_2)_0 \Vert_{H^4(S_2)} + \Vert (\QQ_3)_0 \Vert_{H^4(S_2)}+ \Vert (\QQ_4)_0 \Vert_{H^2(S_2)} <& \infty, \\
\Vert (\QQ_5)_0 \Vert_{H^4(S_2)} + \Vert (\QQ_6)_0 \Vert_{H^2(S_2)}+ \Vert (\QQ_7)_0 \Vert_{H^2(S_2)} <& \infty.
\end{aligned} 
\end{align*}
Then there exist linearized perturbation functions $\dot f$ and $(\dot{j}^1,\dot{j}^2)$ at $S_2$ such that
\begin{align} \begin{aligned} \label{EQchargeMATCH}
\QQ_i\lrpar{\dot{\PP}_f(\dot f)+\dot{\PP}_{({j}^1,{j}^2)}(\dot{j}^1,\dot{j}^2)}=&(\QQ_i)_0 \,\, \text{ for } 1\leq i \leq 7,
\end{aligned} \end{align}
and satisfying
\begin{align} \begin{aligned} 
&\Vert \dot f \Vert_{\YY_f} + \Vert (\dot{j}^1,\dot{j}^2) \Vert_{\YY_{({j}^1,{j}^2)}}+ \Vert \dot{\PP}_f(\dot f) \Vert_{\mathcal{X}(S_2)} +\Vert \dot{\PP}_{({j}^1,{j}^2)}(\dot{j}^1,\dot{j}^2) \Vert_{\mathcal{X}(S_2)} \\
\les& \Vert (\QQ_1)_0 \Vert_{H^6(S_2)} + \Vert (\QQ_2)_0 \Vert_{H^4(S_2)} + \Vert (\QQ_3)_0 \Vert_{H^4(S_2)}+ \Vert (\QQ_4)_0 \Vert_{H^2(S_2)}\\
&+ \Vert (\QQ_5)_0 \Vert_{H^4(S_2)}+ \Vert (\QQ_6)_0 \Vert_{H^2(S_2)}+ \Vert (\QQ_7)_0 \Vert_{H^2(S_2)}.
\end{aligned} \label{EQquantEstimChargeLin}\end{align}
\end{enumerate}
\end{proposition}

\ni The rest of this section is concerned with the proof of Proposition \ref{PROPadjustmentCharges}. First, for linearized transversal perturbations, see Lemma \ref{LEMlinearizedTransversal}, the charges on $S_2$ are calculated to be
\begin{align*} \begin{aligned} 
\QQ_0 =& 0, & \QQ_1 =& \frac{1}{2}\Ldo \dot{f}- \pr_{u} \dot{f}, \\
 \QQ_2 =& -2 \Ldo (\Ldo+2) \dot{f}, & \QQ_3 =& \DDd_2^\ast \di \lrpar{2 \pr_{u} \dot{f}}  - \frac{1}{2} \DDd_2^\ast \di \lrpar{ \Ldo  \dot f} , \\
  \QQ_4 =&\DD_2^\ast \DDd_1^\ast (\pr_u^2 \dot f,0), & \QQ_5 =& \half \pr_u^2 \dot{f}^{[\leq 1]} , \\
   \QQ_6 =&\half \pr_u^3 \dot{f}^{[\leq1]}, & \QQ_7 =&\half \pr_u^3 \dot{f}^{[2]}-3 \dot{f}^{[2]}.
\end{aligned} \end{align*}

\ni Second, for linearized angular perturbations, see Lemma \ref{LEMspherediffLIN}, the charges on $S_2$ are calculated to be
\begin{align*} \begin{aligned} 
\QQ_0 =&0, & \QQ_1 =& \frac{1}{2} \Ldo \dot{q}_1, \\
 \QQ_2 =& 0, & \QQ_3 =&  - \lrpar{\DDd_2^\ast \Divdo +1}\lrpar{\DDd_2^\ast \DDd_1^\ast(\dot{q}_1,\dot{q}_2)} , \\ 
  \QQ_4 =& 0, & \QQ_5 =& 0 ,\\
    \QQ_6 =& 0, & \QQ_7 =&0,
\end{aligned} \end{align*}
where we recall that the pairs of scalar functions $(\dot{q}_1,\dot{q}_2)$ and $(\dot{j}^1,\dot{j}^2)$ are related by
\begin{align} 
\begin{aligned} 
\dot{j}^1 d\th^1 + \dot{j}^2 d\th^2 = - \DDd_1^\ast(\dot{q}_1,\dot{q}_2).
\end{aligned} \label{EQrelationjQ444}
\end{align}

\ni Summing up the above two yields
\begin{subequations}
\begin{align}
\QQ_0 =& 0, \label{EQQ0999}\\
\QQ_1 =& \frac{1}{2}\Ldo \dot{f}- \pr_{u} \dot{f} +\frac{1}{2} \Ldo \dot{q}_1, \label{EQQ1999}\\
 \QQ_2 =& -2 \Ldo (\Ldo+2) \dot{f}, \label{EQQ2999}\\
 \QQ_3 =& \DDd_2^\ast \di \lrpar{2 \pr_{u} \dot{f}}  - \frac{1}{2} \DDd_2^\ast \di \lrpar{ \Ldo  \dot{f}} - \lrpar{\DDd_2^\ast \Divdo +1}\lrpar{\DDd_2^\ast \DDd_1^\ast(\dot{q}_1,\dot{q}_2)} ,\label{EQQ3999}\\
  \QQ_4 =& \DD_2^\ast \DDd_1^\ast (\pr_u^2 \dot f,0), \label{EQQ4999}\\
   \QQ_5 =& \half \pr_u^2 \dot{f}^{[\leq 1]} , \label{EQQ5999}\\
    \QQ_6 =&\half \pr_u^3 \dot{f}^{[\leq2]}-3 \dot{f}^{[2]}.\label{EQQ6999}
\end{align}
\end{subequations}

\ni The right-hand side of \eqref{EQQ2999} has vanishing projection on the modes $l=0$ and $l=1$, hence \eqref{EQQ0999} and \eqref{EQQ2999} imply (1) of Proposition \ref{PROPadjustmentCharges}.

In the following we prove (2) of Proposition \ref{PROPadjustmentCharges} by determining $ \dot{f}, \pr_{u}  \dot{f}, \pr_u^2 \dot{f},  \pr_u^3 \dot{f}$ and $\dot{q}_1$ and $\dot{q}_2$ from \eqref{EQQ1999}-\eqref{EQQ6999} such that \eqref{EQchargeMATCH} is satisfied.\\

\textbf{(1) Definition of $ \dot f$ on $S_2$.} To solve \eqref{EQQ2999}, define the scalar function $\dot f=\dot{f}(0)$ on $S_2$ as solution to
\begin{align} \begin{aligned} \label{DEFtildeFS2matching}
 -2 \Ldo (\Ldo+2) \dot{f} = (\QQ_2)_0.
\end{aligned} \end{align}
with the additional condition that $\dot{f}^{[\leq 1]}=0$. In Fourier space, \eqref{DEFtildeFS2matching} is equivalent for $l\geq2$, $m=-l,\dots,l$ to
\begin{align*} \begin{aligned} 
-2 (-l(l+1))(-l(l+1)+2) \dot{f}^{(lm)} = {(\QQ_2)_0}^{(lm)},
\end{aligned} \end{align*}
which yields
\begin{align*} \begin{aligned} 
\dot{f}^{(lm)} = \frac{1}{2l(l+1)(-l(l+1)+2)} {(\QQ_2)_0}^{(lm)}.
\end{aligned} \end{align*}
Hence $\dot f$ is well-defined and bounded by
\begin{align} \begin{aligned} 
\Vert \dot f \Vert_{H^8(S_2)} \les \Vert (\QQ_2)_0 \Vert_{H^4(S_2)}.
\end{aligned} \label{EQESTIMQmatching1}\end{align}

\ni \textbf{(2) Definition of $\pr_{u}  \dot f$, $\dot{q}_1$ and $\dot{q}_2$ on $S_2$.} To solve \eqref{EQQ1999} and \eqref{EQQ3999}, the scalar functions $\pr_{u} \dot f$, $\dot{q}_1$ and $\dot{q}_2$ on $S_2$ have to solve
\begin{align} \begin{aligned} 
- \pr_{u} \dot{f} +\frac{1}{2} \Ldo \dot{q}_1 =& (\QQ_1)_0- \frac{1}{2}\Ldo \dot{f},\\
\DDd_2^\ast \di \lrpar{2 \pr_{u} \dot{f}}  - \lrpar{\DDd_2^\ast \Divdo +1}\lrpar{\DDd_2^\ast \DDd_1^\ast(\dot{q}_1,\dot{q}_2)} =& (\QQ_3)_0+ \frac{1}{2} \DDd_2^\ast \di \lrpar{ \Ldo  \dot f}.
\end{aligned} \label{EQq1q3matchingmatrixeq} \end{align}

\ni First, while the second of \eqref{EQq1q3matchingmatrixeq} has no $l=0$ and $l=1$ mode, the $l=0$ and $l=1$ mode of the first of \eqref{EQq1q3matchingmatrixeq} can be solved by prescribing, for $l=0,1$,
\begin{align*} 
\begin{aligned} 
\dot{q}_1^{[l]}=0, \,\, \pr_{u} \dot{f}^{[l]} = - \lrpar{(\QQ_1)_0- \frac{1}{2}\Ldo \dot{f}}^{[l]} = - (\QQ_1)_0^{[l]},
\end{aligned} 
\end{align*}
where we used that $\dot{f}=\dot{f}^{[\geq2]}$. This implies that
\begin{align} 
\begin{aligned} 
\Vert \pr_{u} \dot{f}^{[\leq1]} \Vert_{H^6(S_2)} \les \Vert (\QQ_1)_0 \Vert_{H^6(S_2)}.
\end{aligned} \label{EQESTIMQmatching2}
\end{align}

\ni Second, considering the electric part of \eqref{EQq1q3matchingmatrixeq} in Fourier space, we get that for $l\geq2$,
\begin{align} \begin{aligned} 
- (\pr_{u} \dot{f})^{[lm]} - \frac{l(l+1)}{2} \dot{q}_1^{[lm]} =& \lrpar{(\QQ_1)_0- \frac{1}{2}\Ldo \dot{f}}^{[lm]}, \\
- 2 {\pr_{u}  \dot{f}}^{[lm]} -  \frac{ l(l+1)}{2} \dot{q}_1^{[lm]} =&\frac{ \lrpar{(\QQ_3)_0+ \frac{1}{2} \DDd_2^\ast \di \lrpar{ \Ldo  \dot{f}}}_\psi^{[lm]}}{\sqrt{l(l+1)}\sqrt{\half l(l+1)-1}}.
\end{aligned}\label{EQelectric112} \end{align}
The coefficient matrix on the left-hand side of \eqref{EQelectric112} is 
\begin{align*}
\begin{bmatrix}
-1 & -\frac{l(l+1)}{2}  \\
-2 &  -\frac{l(l+1)}{2},
\end{bmatrix}.
\end{align*}
with determinant
\begin{align*} \begin{aligned} 
\mathrm{det} = - \frac{l(l+1)}{2} \neq 0,
\end{aligned} \end{align*}
and matrix inverse
\begin{align*} \begin{aligned} 
\frac{1}{\mathrm{det}} \begin{bmatrix}
-\frac{l(l+1)}{2} & \frac{l(l+1)}{2}  \\
2 &  -1
\end{bmatrix} = \begin{bmatrix}
1 & -1  \\
-\frac{4}{l(l+1)} &  \frac{2}{l(l+1)}
\end{bmatrix}.
\end{aligned} \end{align*}
Therefore the solution to \eqref{EQelectric112} is given, for $l\geq2$, by
\begin{align*} \begin{aligned} 
 {\pr_{u}  \dot{f}}^{[lm]} =&\lrpar{(\QQ_1)_0- \frac{1}{2}\Ldo \dot{f}}^{[lm]} - \frac{ \lrpar{(\QQ_3)_0+ \frac{1}{2} \DDd_2^\ast \di \lrpar{ \Ldo  \dot{f}}}^{[lm]}_\psi}{\sqrt{l(l+1)}\sqrt{\half l(l+1)-1}}, \\
 \dot{q}_1^{[lm]} =& -\frac{4}{l(l+1)} \lrpar{(\QQ_1)_0- \frac{1}{2}\Ldo \dot{f}}^{[lm]} +\frac{2}{l(l+1)}\frac{ \lrpar{(\QQ_3)_0+ \frac{1}{2} \DDd_2^\ast \di \lrpar{ \Ldo  \dot{f}}}_\psi^{[lm]}}{\sqrt{l(l+1)}\sqrt{\half l(l+1)-1}},
\end{aligned} \end{align*}
from which we can derive with \eqref{EQESTIMQmatching1} the estimates
\begin{align} \begin{aligned} 
\Vert \pr_{u}  \dot{f}^{[\geq2]} \Vert_{H^6(S_2)} \les& \Vert (\QQ_1)_0 \Vert_{H^6(S_2)} + \Vert \dot f \Vert_{H^8(S_2)}+ \Vert (\QQ_3)_0 \Vert_{H^4(S_2)}, \\
\les& \Vert (\QQ_1)_0 \Vert_{H^6(S_2)} + \Vert (\QQ_2)_0 \Vert_{H^4(S_2)}+ \Vert (\QQ_3)_0 \Vert_{H^4(S_2)}, \\
\Vert \dot{q}_1^{[\geq2]} \Vert_{H^8(S_2)} \les& \Vert (\QQ_1)_0 \Vert_{H^6(S_2)} + \Vert \dot f \Vert_{H^8(S_2)}+ \Vert (\QQ_3)_0 \Vert_{H^4(S_2)}, \\
\les& \Vert (\QQ_1)_0 \Vert_{H^6(S_2)} + \Vert (\QQ_2)_0 \Vert_{H^4(S_2)}+ \Vert (\QQ_3)_0 \Vert_{H^4(S_2)}.
\end{aligned}\label{EQESTIMQmatching3} \end{align}

\ni Third, we consider the magnetic part of \eqref{EQq1q3matchingmatrixeq}, that is,
\begin{align} \begin{aligned} 
 - \lrpar{\DDd_2^\ast \Divdo +1}\lrpar{\DDd_2^\ast \DDd_1^\ast(0,\dot{q}_2)} =& (\QQ_3)_0.
\end{aligned} \label{EQmagneticpartmatchingcharges44}\end{align}
Going into Fourier space, \eqref{EQmagneticpartmatchingcharges44} is equivalent to, for $l\geq2$,
\begin{align*} 
\begin{aligned} 
-\half l(l+1) \sqrt{\half l(l+1)-1} \sqrt{l(l+1)} \dot{q}_2^{[lm]} =& \left((\QQ_3)_0\right)_\phi^{[lm]}.
\end{aligned} 
\end{align*}
Hence the solution $q_2=q_2^{[\geq2]}$ is well-defined and 
\begin{align} \begin{aligned} 
\Vert \dot{q}_2^{[\geq2]} \Vert_{H^8(S_2)} \les \Vert (\QQ_3)_0 \Vert_{H^4(S_2)}.
\end{aligned} \label{EQESTIMQmatching4}\end{align}

\ni Using \eqref{EQrelationjQ444} and elliptic estimates for Hodge systems of Appendix \ref{SECellEstimatesSpheres}, \eqref{EQESTIMQmatching3} and \eqref{EQESTIMQmatching4} imply that
\begin{align} 
\begin{aligned} 
\Vert (\dot{j}^1,\dot{j}^2) \Vert_{\YY_{({j}^1,{j}^2)}} \les  \Vert (\QQ_1)_0 \Vert_{H^6(S_2)} + \Vert (\QQ_2)_0 \Vert_{H^4(S_2)}+ \Vert (\QQ_3)_0 \Vert_{H^4(S_2)}.
\end{aligned} \label{EQjcontrolMATCHING4334}
\end{align}

\ni \textbf{(3) Definition of $\pr_u^2 \dot f$ on $S_2$.} To solve \eqref{EQQ4999}, define the modes $l\geq2$ of ${\pr_u^2 \dot{f}}$ by
\begin{align} 
\begin{aligned} 
(\QQ_4)_0 = \DD_2^\ast \DDd_1^\ast (\pr_u^2 \dot f,0).
\end{aligned} \label{EQQQ4matchingfirstl2equations}
\end{align}
By \eqref{EQQ4electricAssumption}, \eqref{EQQQ4matchingfirstl2equations} is well-defined and is in Fourier space given by, for $l\geq2$,
\begin{align} 
\begin{aligned} 
\lrpar{\pr_u^2 \dot{f}}^{(lm)} = \frac{1}{\sqrt{\half l(l+1)-1} \sqrt{l(l+1)}} \lrpar{(\QQ_4)_0}_\psi^{(lm)}.
\end{aligned} \label{EQpru2fFourierExpression}
\end{align}
To solve \eqref{EQQ5999}, define the modes $l\leq 1$ of $\pr_u^2 \dot f$ by 
\begin{align*} 
\begin{aligned} 
 (\QQ_5)_0 =& \half \pr_u^2 \dot{f}^{[\leq 1]},
\end{aligned} 
\end{align*}
which in Fourier modes equals, for $l\leq1$,
\begin{align} 
\begin{aligned} 
\pr_u^2 \dot{f}^{(lm)} = 2 (\QQ_5)_0^{(lm)}.
\end{aligned} \label{EQpru2fFourierExpression222}
\end{align}

\ni From \eqref{EQpru2fFourierExpression} and \eqref{EQpru2fFourierExpression222} we directly get that
\begin{align} 
\begin{aligned} 
\left\Vert \pr_u^2 \dot{f} \right\Vert_{H^4(S_2)} \les \Vert (\QQ_4)_0 \Vert_{H^2(S_2)}.
\end{aligned} \label{EQcontrolpruf2111}
\end{align}

\ni \textbf{(4) Definition of $\pr_u^3 \dot f$ on $S_2$.} To solve \eqref{EQQ6999}, define $\pr_u^3 \dot f^{[\leq1]}$ by
\begin{align*} 
\begin{aligned} 
(\QQ_6)_0 =&\half \pr_u^3 \dot{f}^{[\leq1]}.
\end{aligned} 
\end{align*}
and $\pr_u^3 \dot f^{[2]}$ by
\begin{align*} 
\begin{aligned} 
(\QQ_7)_0 =&\half \pr_u^3 \dot{f}^{[2]}-3\dot{f}^{[2]},
\end{aligned} 
\end{align*}
and let $\pr_u^3 \dot f^{[\geq3]}=0$. By \eqref{EQchargematchingmodeconditions555}, $\pr_u^3 \dot f$ is well-defined and it holds that
\begin{align} 
\begin{aligned} 
\Vert \pr_u^3 \dot f \Vert_{H^2(S_2)} \les \Vert (\QQ_6)_0 \Vert_{H^2(S_2)}+ \Vert \dot{f} \Vert_{H^8(S_2)}.
\end{aligned} \label{EQpru3estimatematching}
\end{align}

\ni To summarise the above, we constructed linearized perturbation functions 
\begin{align*} 
\begin{aligned} 
(\dot{j}^1,\dot{j}^2) \text{ and } \dot f = (\dot f,\pr_{u} \dot f, \pr_u^2 \dot f, \pr_u^3 \dot f)
\end{aligned} 
\end{align*}
such that \eqref{EQchargeMATCH} is satisfied. Further, from \eqref{EQESTIMQmatching1}, \eqref{EQESTIMQmatching2}, \eqref{EQESTIMQmatching3}, \eqref{EQjcontrolMATCHING4334}, \eqref{EQcontrolpruf2111} and \eqref{EQpru3estimatematching} it follows that
\begin{align*} 
\begin{aligned} 
\Vert \dot f \Vert_{\YY_f} + \Vert (\dot{j}^1,\dot{j}^2)\Vert_{\YY_{({j}^1,{j}^2)}} \les& \Vert (\QQ_1)_0 \Vert_{H^6(S_2)}+\Vert (\QQ_2)_0 \Vert_{H^4(S_2)}+\Vert (\QQ_3)_0 \Vert_{H^4(S_2)}+\Vert (\QQ_4)_0 \Vert_{H^2(S_2)} \\
& +\Vert (\QQ_5)_0 \Vert_{H^4(S_2)}+\Vert (\QQ_6)_0 \Vert_{H^2(S_2)}.
\end{aligned} 
\end{align*}

\ni The estimate \eqref{EQquantEstimChargeLin} follows from the above by Lemma \ref{LEMestimatesSpherePerturbations}. This finishes the proof of Proposition \ref{PROPadjustmentCharges}.

\subsection{Representation formulas and estimates} \label{SECprelimAnalysis} In this section, we rewrite the linearized null structure equations \eqref{EQlinearizedOPsystem2} into a set of transport equations and integrate them to derive representation formulas and estimates. 

\begin{itemize}
\item In Section \ref{SEClinPRELIM1} we consider $\phid$ and $\gdcd$.
\item In Section \ref{SEClinPRELIM20} we consider $\omtrchid$ and $\etad$.
\item In Section \ref{SEClinPRELIM2} we consider $\omtrchibd$ and $\chibhd$.
\item In Section \ref{SEClinPRELIM3} we consider $\abd$, $\ombd$ and $\Du\ombd$.
\end{itemize}

\ni \textbf{Notation.} In this section we ease presentation by leaving away the trivial $u$ index on spheres and sphere data, denoting $\HH= \HH_{0,[1,2]}$, and writing $v$ instead of $r$ on $\HH$.

\subsubsection{Analysis of $\phid$ and $\gdcd$} \label{SEClinPRELIM1}
The linearized null structure equations for $\phid$ and $\gdcd$ in \eqref{EQlinearizedOPsystem2}, that is,
\begin{align*} 
\begin{aligned} 
D\lrpar{D\phid - 2 \Omd} = \mfq_1, \,\, r^2 D\gdcd -2 \chihd =& \mfq_3,
\end{aligned} 
\end{align*}
are equivalent to
\begin{align} 
\begin{aligned} 
DD\phid = 2D\Omd + \mfq_1, \,\, 
D\gdcd = \frac{2}{r^2} \chihd + \frac{1}{r^2} {\mfq_3}. 
\end{aligned}\label{EQlinNSE122222}
\end{align}
By integration of \eqref{EQlinNSE122222}, we directly get the following lemma.
\begin{lemma}[Representation formulas and estimates for $\phid$ and $\gdcd$] \label{LEMenergyestimatesPHID} Consider sphere data $\xdmf_1$ on $S_1$ and linearized conformal data $\chihd$ and $\Omd$ on $\HH$. Integrating the transport equations \eqref{EQlinNSE122222} and using $\dot{\CC}_2=\mfq_2$ (see Lemma \ref{LEMlinearizedConstraints}) yields the following representation formulas.
\begin{enumerate}
\item It holds that
\begin{align} 
\begin{aligned} 
\phid =& 2 \int\limits_1^v \Omd dv' + \ilr \int\limits_1^{v'} \mfq_1 dv'' dv'+v\phid(1)+\frac{v-1}{2}\lrpar{\omtrchid(1)-4\Omd(1)+\mfq_2(1)}, \label{EQreps11}
\end{aligned} 
\end{align}
which yields the estimate
\begin{align*} 
\begin{aligned} 
\Vert \phid \Vert_{H^6_4(\HH)} \les& \Vert \xdmf_1 \Vert_{\XX(S_1)}+ \Vert \Omd \Vert_{{H^6_3(\HH)}}+ \Vert (\mfq_i)_{1\leq i \leq 10} \Vert_{\mathcal{Z}_\CC}.
\end{aligned} 
\end{align*} 

\item It holds that
\begin{align} 
\begin{aligned} 
\dot{\gd_c} =&2 \int\limits_{1}^v \frac{1}{v'^2} \chihd dv' + \ilr \frac{1}{v'^2}\mfq_3 dv'+ \gdcd(1), \label{EQreps12}
\end{aligned} 
\end{align}
which yields the estimate
\begin{align*} 
\begin{aligned} 
\Vert \gdcd \Vert_{H^6_3(\HH)} \les& \Vert \xdmf_1 \Vert_{\XX(S_1)}+ \Vert \chihd \Vert_{{H^6_2(\HH)}}+ \Vert (\mfq_i)_{1\leq i \leq 10} \Vert_{\mathcal{Z}_\CC}.
\end{aligned} 
\end{align*}

\end{enumerate}
\end{lemma}

\subsubsection{Analysis of $\omtrchid$ and $\etad$} \label{SEClinPRELIM20} Recall that the linearized null structure equations for $\omtrchid$ and $\etad$ in \eqref{EQlinearizedOPsystem2} are given by
\begin{subequations}
\begin{align} 
D\lrpar{\frac{\phid}{r}}-\frac{\omtrchid}{2}=&  \frac{1}{2r^2} \mfq_2, \label{EQlinNSE122222111}\\  
D\lrpar{r^2\etad} +\frac{r^2}{2} \di \lrpar{\omtrchid-\frac{4}{r}\Omd} -2r\di\Omd -\Divdo \chihd  =& r^2 \mfq_4. \label{EQlinNSE122222222}
\end{align}
\end{subequations}
In the following, we rewrite \eqref{EQlinNSE122222111} and \eqref{EQlinNSE122222222} to get useful bounds and representation formulas for $\omtrchid$ and $\etad$.

On the one hand, using \eqref{EQlinNSE122222}, \eqref{EQlinNSE122222111} can be rewritten as
\begin{align} 
\begin{aligned} 
D\lrpar{r^2 \lrpar{\omtrchid-\frac{4}{r}\Omd}+ \mfq_2} =&-4\Omd + 2 r\mfq_1 \label{EQrewritten15}.
\end{aligned}
\end{align}
Straight-forward integration of \eqref{EQlinNSE122222222} and \eqref{EQrewritten15} yields the following lemma.
\begin{lemma}[Bounds for $\omtrchid$ and $\etad$] \label{LEMenergyestimatesETAD} Consider given sphere data $\xdmf_1$ on $S_1$ and given $\chihd$ and $\Omd$ on $\HH$. Integrating subsequently the transport equations \eqref{EQrewritten15} for $\omtrchid$ and \eqref{EQlinNSE122222222} for $\etad$ yields
\begin{align*} 
\begin{aligned} 
\Vert \omtrchid \Vert_{H^6_3(\HH)} \les& \Vert \xdmf_1 \Vert_{\XX(S_1)}+ \Vert \Omd \Vert_{{H^6_3(\HH)}}+ \Vert (\mfq_i)_{1\leq i \leq 10} \Vert_{\mathcal{Z}_\CC}, \\
\Vert \etad \Vert_{H^5_2(\HH)} \les& \Vert \xdmf_1 \Vert_{\XX(S_1)}+ \Vert \Omd \Vert_{{H^6_3(\HH)}}+ \Vert \chihd \Vert_{{H^6_2(\HH)}}+ \Vert (\mfq_i)_{1\leq i \leq 10} \Vert_{\mathcal{Z}_\CC}.
\end{aligned} 
\end{align*}

\end{lemma}

\ni On the other hand, using \eqref{EQlinNSE122222}, \eqref{EQlinNSE122222222} and \eqref{EQrewritten15} can be rewritten as
\begin{subequations}
\begin{align} 
D\lrpar{r^2 \etad + \frac{r^3}{2}\di \lrpar{\omtrchid-\frac{4}{r}\Omd}+ \frac{r}{2}\di\mfq_2}=& \Divdo \chihd + r^2 \mfq_4 +r^2 \di \mfq_1 + \frac{1}{2}\di \mfq_2, \label{EQrewritten10} \\
D\lrpar{\frac{r}{2} \lrpar{\omtrchid-\frac{4}{r}\Omd}+\frac{\phid}{r}+\frac{1}{2r} \mfq_2}=&\mfq_1.\label{EQrewritten16666}
\end{align}
\end{subequations}

\ni Integrating \eqref{EQrewritten10}, we get the following representation formula for $\etad$.
\begin{lemma}[Representation formulas for $\etad$] \label{LEMrepphigdeta} \label{CORenergyEST12111} Consider given sphere data $\xdmf_1$ on $S_1$ and given $\chihd$ and $\Omd$ on $\HH$. Integrating \eqref{EQrewritten10} yields that
\begin{align*} 
\begin{aligned} 
&\left[ v^2 \etad + \frac{v^3}{2} \di\lrpar{\omtrchid - \frac{4}{v} \Omd}+\frac{v}{2}\di \mfq_2\right]_1^v\\
 =& \Divdo\lrpar{ \int\limits_1^v \chihd dv' } +\ilr \lrpar{v'^2 \mfq_4 +v'^2 \di \mfq_1 +\frac{1}{2}\di \mfq_2} dv'.
\end{aligned}
\end{align*}
\end{lemma}

\ni Recall from \eqref{EQdefChargesMinkowski8891} that $\QQ_0$ and $\QQ_1$ are defined by
\begin{align*} 
\begin{aligned} 
\QQ_0 := r^2 \etad^{[1]} + \frac{r^3}{2}\di \lrpar{\omtrchid^{[1]}-\frac{4}{r}\Omd^{[1]}}, \,\,
\QQ_1 :=\frac{r}{2} \lrpar{\omtrchid-\frac{4}{r}\Omd} + \frac{\phid}{r}.
\end{aligned} 
\end{align*}
The next lemma follows directly from \eqref{EQrewritten10} and \eqref{EQrewritten16666}.
\begin{lemma}[Conservation laws I] \label{CORtransport1EQS} It holds that
\begin{align*} 
\begin{aligned} 
D\lrpar{ \QQ_0+ \frac{r}{2}\di\mfq_2^{[1]} } =& r^2 \mfq_4^{[1]} +r^2 \di \mfq_1^{[1]} + \frac{1}{2}\di \mfq_2^{[1]}, \\
D\lrpar{ \QQ_1 +\frac{1}{2r} \mfq_2}  =&\mfq_1.
\end{aligned} 
\end{align*}
\end{lemma}

\ni The proof of the following lemma is omitted.
\begin{lemma}[Properties of charges I] \label{EQQestimatesatspheres00111111} \label{CORtransportEstimatesCHARGES111} The following holds.
\begin{enumerate}
\item  Let $\xd_v$ be sphere data on a sphere $S_v \subset \HH$. Then
\begin{align*} 
\begin{aligned} 
\Vert \QQ_0 \Vert_{H^5(S_v)} + \Vert \QQ_1 \Vert_{H^6(S_v)} \les \Vert \xd_v \Vert_{\mathcal{X}(S_v)}.
\end{aligned} 
\end{align*}
\item For given sphere data $\xdmf_1$ on $S_1$ and source terms $(\mfq_i)_{1\leq i \leq 10}$ on $\HH$, define ${}^{(1)}\QQ_0$ and ${}^{(1)}\QQ_1$ as solution to the transport equations of Lemma \ref{CORtransport1EQS} on $\HH$ with initial values given by $\QQ_0$ and $\QQ_1$ calculated from $\xdmf_1$. Then it holds that
\begin{align*} 
\begin{aligned} 
\Vert {}^{(1)}\QQ_0 \Vert_{H^5_2(\HH)}+\Vert {}^{(1)}\QQ_1 \Vert_{H^6_3(\HH)} \les& \Vert \QQ_0(\xdmf_1) \Vert_{H^5(S_1)}+\Vert \QQ_1(\xdmf_1) \Vert_{H^6(S_1)} + \Vert (\mfq_i)_{1\leq i \leq 10} \Vert_{\mathcal{Z}_\CC}.
\end{aligned} 
\end{align*}
\end{enumerate}
\end{lemma}

\subsubsection{Analysis of $\protect\dot{\protect\Omega\protect\tr\protect\underline{\protect\chi}}$ and $\protect\dot{\protect\widehat{\protect\underline{\protect\chi}}}$} \label{SEClinPRELIM2} The linearized null structure equations \eqref{EQlinearizedOPsystem2} for $\omtrchibd$ and $\chibhd$ are given by
\begin{align} 
\begin{aligned} 
D\lrpar{v^2 \omtrchibd} -2v \omtrchid +2 \Divdo \lrpar{\etad-2\di\Omd}+ 2v^2 \Kd +4 \Omd =& v^2 \mfq_5,\\
D\lrpar{\frac{\chibhd}{v}}-\frac{2}{v} \DDd_2^\ast \lrpar{\etad -2 \di\Omd}-\frac{1}{v^2} \chihd=& \frac{1}{v} \mfq_6. 
\end{aligned} \label{EQnullstructureRECALLomtrchibdchibhd}
\end{align}

\ni The following lemma follows from \eqref{EQnullstructureRECALLomtrchibdchibhd} and Lemmas \ref{LEMenergyestimatesPHID} and \ref{LEMenergyestimatesETAD}.
\begin{lemma}[Bounds for $\omtrchibd$ and $\chibhd$] \label{LEMenergyestimatesCHIBHD} Consider given sphere data $\xdmf_1$ on $S_1$ and given $\chihd$ and $\Omd$ on $\HH$. Integrating the transport equations \eqref{EQnullstructureRECALLomtrchibdchibhd} for $\omtrchibd$ and $\chibhd$ yields
\begin{align*} 
\begin{aligned} 
\Vert \omtrchibd \Vert_{H^4_2(\HH)} \les& \Vert \xdmf_1 \Vert_{\XX(S_1)}+ \Vert \Omd \Vert_{{H^6_3(\HH)}}+ \Vert \chihd \Vert_{{H^6_2(\HH)}}+ \Vert (\mfq_i)_{1\leq i \leq 10} \Vert_{\mathcal{Z}_\CC}, \\
\Vert \chibhd \Vert_{H^4_3(\HH)} \les& \Vert \xdmf_1 \Vert_{\XX(S_1)}+ \Vert \Omd \Vert_{{H^6_3(\HH)}}+ \Vert \chihd \Vert_{{H^6_2(\HH)}}+ \Vert (\mfq_i)_{1\leq i \leq 10} \Vert_{\mathcal{Z}_\CC}.
\end{aligned} 
\end{align*}
\end{lemma}

\ni Recall from \eqref{EQdefChargesMinkowski8891} that $\QQ_2$ and $\QQ_3$ are defined as
\begin{align*} 
\begin{aligned} 
\QQ_2 :=& r^2 \omtrchibd -\frac{2}{r}\Divdo \lrpar{r^2\etad+\frac{r^3}{2}\di\lrpar{\omtrchid-\frac{4}{r}\Omd}} -r^2 \lrpar{\omtrchid-\frac{4}{r}\Omd} +2r^3 \Kd, \\
\QQ_3:=& \frac{\chibhd}{r} -\half \lrpar{ \DDd_2^\ast \Divdo +1} \gdcd + \DDd_2^\ast \lrpar{ \etad + \frac{r}{2}\di \lrpar{\omtrchid-\frac{4}{r}\Omd}} - r \DDd_2^\ast \di \lrpar{\omtrchid-\frac{4}{r}\Omd}.
\end{aligned} 
\end{align*}

\ni The following lemma shows that they are indeed subject to conservation laws.
\begin{lemma}[Conservation laws II] \label{LEMgaugeconservationlaws} The linearized null structure equations \eqref{EQlinearizedOPsystem2} imply the following transport equations for $\QQ_2$ and $\QQ_3$,
\begin{align*} \begin{aligned} 
D\lrpar{ \QQ_2 - (\Ldo+1) \mfq_2}=&v^2 \mfq_5 -2v \Divdo \mfq_4 -2v(\Ldo+1)\mfq_1 -\frac{1}{v}(\Ldo+2) \mfq_2+ \frac{1}{v}\Divdo\Divdo \mfq_3, \\
D\lrpar{ \QQ_3 - \frac{1}{2v}\DDd_2^\ast \di \mfq_2} =&\frac{1}{v}\mfq_6 -\frac{1}{2v^2}\lrpar{\DDd_2^\ast \Divdo +1}\mfq_3+ \DDd_2^\ast \mfq_4-\DDd_2^\ast \di \mfq_1 +\frac{1}{2v^2} \DDd_2^\ast\di \mfq_2.
\end{aligned} \end{align*}
\end{lemma}

\begin{proof}[Proof of Lemma \ref{LEMgaugeconservationlaws}] 

First, by \eqref{EQgaussDHR}, \eqref{EQdefChargesMinkowski8891}, \eqref{EQrewritten10}, \eqref{EQrewritten16666} and \eqref{EQnullstructureRECALLomtrchibdchibhd}, we have that
\begin{align*} \begin{aligned} 
&D\QQ_2\\
 =& v^2 \mfq_5 + 2v \omtrchid-2\Divdo \lrpar{\etad-2\di\Omd}-2v^2\Kd-4\Omd\\
&-2 D\lrpar{\frac{1}{v}\Divdo \lrpar{v^2\etad+\frac{v^3}{2}\di\lrpar{\omtrchid-\frac{4}{v}\Omd}}}\\
&-D\lrpar{v^2\lrpar{\omtrchid-\frac{4}{v}\Omd}} +2v^2 \Kd +2v D\lrpar{v^2 \Kd}\\
=& v^2 \mfq_5 + 2v \omtrchid-2\Divdo \lrpar{\etad-2\di\Omd}-4\Omd+\frac{2}{v^2} \Divdo \lrpar{v^2\etad+\frac{v^3}{2}\di\lrpar{\omtrchid-\frac{4}{v}\Omd}} \\
&-\frac{2}{v} \Divdo \lrpar{\Divdo \chihd + v^2 \mfq_4 + v^2 \di \mfq_1 + \half \di \mfq_2 - D\lrpar{\frac{v}{2}\di \mfq_2}}\\
&+4\Omd-2 v \mfq_1+D\mfq_2+ v \Divdo \Divdo \lrpar{\frac{2}{v^2}\chihd +\frac{1}{v^2} \mfq_3}-2v(\Ldo+2)\lrpar{\frac{\omtrchid}{2}+\frac{1}{2v^2} \mfq_2}\\
=& v^2 \mfq_5 -2v \Divdo \mfq_4 -2v(\Ldo+1)\mfq_1+ D\lrpar{(\Ldo+1)\mfq_2}-\frac{1}{v}(\Ldo+2) \mfq_2+ \frac{1}{v}\Divdo\Divdo \mfq_3.
\end{aligned} \end{align*}

\ni Second, by \eqref{EQdefChargesMinkowski8891}, \eqref{EQrewritten10}, \eqref{EQrewritten16666} and \eqref{EQnullstructureRECALLomtrchibdchibhd}, we have that
\begin{align*} \begin{aligned} 
D\QQ_3=& \frac{1}{v}\mfq_6 +\frac{2}{v}\DDd_2^\ast \lrpar{\etad-2\di\Omd} +\frac{1}{v^2}\chihd-\half \lrpar{\DDd_2^\ast \di +1}\lrpar{\frac{2}{v^2}\chihd +\frac{1}{v^2}\mfq_3}\\
&-\frac{2}{v}\DDd_2^\ast \lrpar{\etad+\frac{v}{2}\di\lrpar{\omtrchid-\frac{4}{v}\Omd}}\\
&+ \frac{1}{v^2} \DDd_2^\ast \lrpar{\Divdo \chihd +v^2\mfq_4+v^2\di\mfq_1 +\half \di \mfq_2 -D\lrpar{\frac{v}{2}\di\mfq_2}}\\
&+ \DDd_2^\ast \di \lrpar{\omtrchid-\frac{4}{v}\Omd} -\frac{1}{v}\DDd_2^\ast \di \lrpar{-4\Omd +2v\mfq_1-D\mfq_2}\\
=& \frac{1}{v}\mfq_6 -\frac{1}{2v^2}\lrpar{\DDd_2^\ast \Divdo +1}{\mfq_3}+ \DDd_2^\ast \mfq_4-\DDd_2^\ast \di \mfq_1 +D\lrpar{\frac{1}{2v}\DDd_2^\ast \di \mfq_2}+ \frac{1}{2v^2}\DDd_2^\ast \di \mfq_2.
\end{aligned} \end{align*}

\ni This finishes the proof of Lemma \ref{LEMgaugeconservationlaws}. \end{proof}

\ni The proof of the following lemma is omitted.
\begin{lemma}[Properties of charges II] \label{EQQestimatesatspheres0011} \label{CORtransportEstimatesCHARGES} The following holds.
\begin{enumerate}
\item  Let $\xd_v$ be sphere data on a sphere $S_v \subset \HH$. Then
\begin{align*} 
\begin{aligned} 
 \Vert \QQ_2 \Vert_{H^4(S_v)}+ \Vert \QQ_3 \Vert_{H^4(S_v)} \les \Vert \xd_v \Vert_{\mathcal{X}(S_v)}.
\end{aligned} 
\end{align*}
\item For given sphere data $\xdmf_1$ on $S_1$ and source terms $(\mfq_i)_{1\leq i \leq 10}$ on $\HH$, define ${}^{(1)}\QQ_i$, $2\leq i \leq 3$, as solution to the transport equations of Lemma \ref{LEMgaugeconservationlaws} on $\HH$ with initial values given by $\QQ_i$ calculated from $\xdmf_1$. Then it holds that
\begin{align*} 
\begin{aligned} 
 \Vert {}^{(1)}\QQ_2 \Vert_{H^4_2(\HH)}+ \Vert {}^{(1)}\QQ_3 \Vert_{H^4_2(\HH)} 
\les  \Vert \QQ_2(\xdmf_1) \Vert_{H^4(S_1)}+ \Vert \QQ_3(\xdmf_1) \Vert_{H^6(S_1)}+ \Vert (\mfq_i)_{1\leq i \leq 10} \Vert_{\mathcal{Z}_\CC}.
\end{aligned} 
\end{align*}
\end{enumerate}
\end{lemma}

\subsubsection{Analysis of $\protect\ombd$, $\protect\abd$ and $\protect\Du \protect\ombd$} \label{SEClinPRELIM3} The linearized null structure equations \eqref{EQlinearizedOPsystem2} for $\ombd$, $\abd$ and $\Du \ombd$ are
\begin{align} 
\begin{aligned} 
D\ombd -\Kd - \frac{1}{2v}\omtrchibd + \frac{1}{2v} \omtrchid - \frac{2}{v^2} \Omd =& \mfq_7, \\
D\lrpar{\frac{\abd}{v}} - \frac{2}{v} \DDd_2^\ast \lrpar{\frac{1}{v^2}\Divdo \chibhd - \half \di \omtrchibd -\frac{1}{v} \etad} =& \frac{1}{v}\mfq_8, 
\end{aligned} \label{EQlinearizedRECALL4111}
\end{align}
and
\begin{align} 
\begin{aligned} 
&D\Du\ombd -\frac{3}{v}\lrpar{\Kd + \frac{1}{2v}\omtrchibd - \frac{1}{2v}\omtrchid + \frac{2}{v^2}\Omd}\\
=& \frac{1}{v^2} \Divdo \lrpar{\frac{1}{v^2} \Divdo \chibhd - \frac{1}{v}\etad - \half \di \omtrchibd}+ \mfq_9.
\end{aligned} \label{EQlinearizedRECALL4222}
\end{align}

\ni Integrating \eqref{EQlinearizedRECALL4111} and \eqref{EQlinearizedRECALL4222} and applying Lemmas \ref{LEMenergyestimatesPHID}, \ref{LEMenergyestimatesETAD} and \ref{LEMenergyestimatesCHIBHD} yields the following lemma.
\begin{lemma}[Bounds for $\ombd$, $\abd$ and $\Du\ombd$] \label{LEMenergyestimatesDUOMBD} Consider given sphere data $\xdmf_1$ on $S_1$ and given $\chihd$ and $\Omd$ on $\HH$. Integrating the transport equations \eqref{EQnullstructureRECALLomtrchibdchibhd} for $\omtrchibd$ and $\chibhd$ yields
\begin{align*} 
\begin{aligned} 
&\Vert \ombd \Vert_{H^4_3(\HH)}+\Vert \abd \Vert_{H^2_3(\HH)} +\Vert \Du\ombd \Vert_{H^2_3(\HH)}  \\
\les& \Vert \xdmf_1 \Vert_{\XX(S_1)}+ \Vert \Omd \Vert_{{H^6_3(\HH)}}+ \Vert \chihd \Vert_{{H^6_2(\HH)}}+ \Vert (\mfq_i)_{1\leq i \leq 10} \Vert_{\mathcal{Z}_\CC}.
\end{aligned} 
\end{align*}
\end{lemma}

\ni In Appendix \ref{SECDerivation1} it is proved that the linearized null structure equations \eqref{EQlinearizedRECALL4111} and \eqref{EQlinearizedRECALL4222} can be rewritten as follows.

\begin{lemma} \label{LEMtransportEQUStrans} The linearized null structure equations \eqref{EQlinearizedOPsystem2} imply the following transport equations.
\begin{align} 
\begin{aligned} 
&D\lrpar{ \ombd +\frac{1}{4v^2}\QQ_2 +\frac{1}{3v} \Divdo\lrpar{\etad+ \frac{v}{2}\di\lrpar{\omtrchid-\frac{4}{v}\Omd}} -\frac{1}{12v^2}(\Ldo+3) \mfq_2}\\
=&  \frac{1}{3v^3} \Divdo \Divdo \chihd+h_{\ombd} ,
\end{aligned} \label{eqombd444}
\end{align}
with source term
\begin{align} 
\begin{aligned} 
h_{\ombd} := \mfq_7+ \frac{1}{4}\mfq_5 + \frac{1}{4v^3} \Divdo \Divdo \mfq_3 - \frac{1}{6v} \Divdo \mfq_4-\frac{1}{6v}(\Ldo+3) \mfq_1 -\frac{1}{12v^3} \Ldo \mfq_2,
\end{aligned} \label{EQHOMBD}
\end{align}

\ni and
\begin{align} \begin{aligned}
&D\left( \frac{\abd}{v} +\frac{2}{v}\DDd_2^\ast \Divdo \QQ_3 - \frac{1}{2v^2} \DDd_2^\ast \di \QQ_2 - \frac{2}{v} \DDd_2^\ast \di \lrpar{\Ldo+2}\QQ_1 \right) \\
&-D\left( \frac{2}{3v} \DDd_2^\ast \lrpar{\Divdo \DDd_2^\ast + 1 + \di\Divdo} \lrpar{\etad+ \frac{v}{2}\di\lrpar{\omtrchid-\frac{4}{v}\Omd}} \right)\\
&+D\left( \frac{1}{v} \DDd_2^\ast \lrpar{\Divdo \DDd_2^\ast + 1 + \di\Divdo} \Divdo\gdcd -\frac{1}{3v^2} \DDd_2^\ast \lrpar{\Divdo \DDd_2^\ast + 1 + \di\Divdo} \Divdo\di \mfq_2\right)\\
=& \frac{4}{3v^3} \DDd_2^\ast \lrpar{\Divdo \DDd_2^\ast + 1 + \di\Divdo} \Divdo \chihd +h_{\abd},
\end{aligned} \label{EQabd54444}\end{align}
with source term
\begin{align*} 
\begin{aligned} 
h_{\abd} :=& \frac{1}{v}\mfq_8 + \frac{2}{v}\DDd_2^\ast \Divdo \lrpar{D\QQ_3} -\frac{1}{2v^2} \DDd_2^\ast \di \lrpar{D\QQ_2}\\
&-\frac{2}{v} \DDd_2^\ast \di\lrpar{\Ldo+2} \lrpar{D\QQ_1} +\frac{1}{v^3}\DDd_2^\ast \lrpar{\Divdo \DDd_2^\ast + 1 + \di\Divdo} \Divdo \mfq_3\\
&-\frac{2}{3v^3}\DDd_2^\ast \lrpar{\Divdo \DDd_2^\ast + 1 + \di\Divdo} \Divdo\lrpar{v^2 \mfq_4 + v^2 \di \mfq_1 - \di \mfq_2},
\end{aligned} 
\end{align*}

\ni and 

\begin{align} 
\begin{aligned} 
&D\lrpar{ \Du\ombd -\frac{1}{6v^3} \lrpar{\Ldo-3}\QQ_2+ \frac{1}{2v^2} \Divdo \Divdo \QQ_3 +\frac{1}{v^2} \Divdo \Divdo \DDd_2^\ast \di \QQ_1}\\
&-D\lrpar{ \frac{1}{4v^2} \Divdo \lrpar{\di \Divdo -2 + \Divdo\DD_2^\ast} \lrpar{\eta+ \frac{v}{2}\di\lrpar{\omtrchid-\frac{4}{v}\Omd}} }\\
&+ D\lrpar{\frac{1}{8v^2} \Divdo \di \Divdo \Divdo \gdcd+ \frac{1}{2v^3}\lrpar{\frac{1}{12}\Ldo\Ldo-\frac{1}{6}\Ldo+\frac{1}{4}\Divdo \Divdo \DDd_2^\ast \di - 1}\mfq_2}\\
=&\frac{1}{4v^4} \Divdo \lrpar{2-\Divdo \DDd_2^\ast} \Divdo \chihd + h_{\Du\ombd},
\end{aligned} \label{EQequationDuombd444}
\end{align}
with source term
\begin{align} 
\begin{aligned} 
h_{\Du \ombd} :=& \mfq_8+ \frac{1}{2v^3} \Divdo \Divdo \mfq_6 - \frac{1}{6v} (\Ldo-3)\mfq_5  \\
&+ \frac{1}{v^2}\lrpar{\frac{1}{12} \Ldo \Ldo -\frac{1}{6} \Ldo -1 + \frac{1}{4}\Divdo \Divdo \DDd_2^\ast \di} \mfq_1\\
&+ \frac{1}{v^4} \lrpar{-\frac{1}{8} \Divdo \Divdo \DDd_2^\ast \di + \frac{1}{24}\Ldo\Ldo +\frac{1}{12} \Ldo +\half }\mfq_2 \\
&+ \frac{1}{v^4} \Divdo \lrpar{-\frac{1}{24} \di \Divdo + \frac{1}{4} - \frac{1}{4} \Divdo \DDd_2^\ast} \Divdo \mfq_3 \\
&+ \frac{1}{v^2}\lrpar{\frac{1}{12} (\Ldo-6)\Divdo +\frac{1}{4} \Divdo \Divdo \DDd_2^\ast } \mfq_4.
\end{aligned} \label{EQdefhDUOMBD}
\end{align}
\end{lemma} 

\ni As consequence of Lemma \ref{LEMtransportEQUStrans}, we get the following useful representation formulas for $\ombd, \abd$ and $\Du\ombd$.
\begin{lemma}[Representation formulas for $\ombd, \abd$ and $\Du\ombd$]  \label{COR2ndenergyEstimatesOMBDABDDUOMBD} Consider sphere data $\xdmf_1$ on $S_1$ and $\chihd$ and $\Omd$ on $\HH$. Integrating the transport equations of Lemma \ref{LEMtransportEQUStrans} yields the following representation formulas for $\ombd, \abd$ and $\Du\ombd$.

\begin{enumerate}

\item It holds that
\begin{align*} 
\begin{aligned} 
&\left[ \ombd +\frac{1}{4v^2}\QQ_2 +\frac{1}{3v} \Divdo\lrpar{\etad+ \frac{v}{2}\di\lrpar{\omtrchid-\frac{4}{v}\Omd}} -\frac{1}{12v^2}(\Ldo+3) \mfq_2\right]_1^v\\
=&  \frac{1}{3} \Divdo \Divdo \lrpar{\int\limits_1^v \frac{1}{v'^3} \chihd dv'} +\int\limits_1^v h_{\ombd} dv',
\end{aligned} 
\end{align*}

\item It holds that
\begin{align} \begin{aligned}
&\left[ \frac{\abd}{v} +\frac{2}{v}\DDd_2^\ast \Divdo \QQ_3 - \frac{1}{2v^2} \DDd_2^\ast \di \QQ_2 - \frac{2}{v} \DDd_2^\ast \di \lrpar{\Ldo+2}\QQ_1 \right]^v_1 \\
&-\left[ \frac{2}{3v} \DDd_2^\ast \lrpar{\Divdo \DDd_2^\ast + 1 + \di\Divdo} \lrpar{\etad+ \frac{v}{2}\di\lrpar{\omtrchid-\frac{4}{v}\Omd}} \right]^v_1\\
&+\left[ \frac{1}{v} \DDd_2^\ast \lrpar{\Divdo \DDd_2^\ast + 1 + \di\Divdo} \Divdo\gdcd- \frac{1}{3v^2} \DDd_2^\ast \lrpar{\Divdo \DDd_2^\ast + 1 + \di\Divdo} \Divdo\di \mfq_2 \right]^v_1\\
=& \frac{4}{3} \DDd_2^\ast \lrpar{\Divdo \DDd_2^\ast + 1 + \di\Divdo} \Divdo \lrpar{\ilr \frac{1}{v'^3}\chihd dv'} +\int\limits_1^v h_{\abd} dv',
\end{aligned} \label{EQbigrepformulaABD111} \end{align}

\item It holds that
\begin{align} 
\begin{aligned} 
&\left[ \Du\ombd -\frac{1}{6v^3} \lrpar{\Ldo-3}\QQ_2+ \frac{1}{2v^2} \Divdo \Divdo \QQ_3 +\frac{1}{v^2} \Divdo \Divdo \DDd_2^\ast \di \QQ_1\right]_1^v\\
&-\left[ \frac{1}{4v^2} \Divdo \lrpar{\di \Divdo -2 + \Divdo\DD_2^\ast} \lrpar{\eta+ \frac{v}{2}\di\lrpar{\omtrchid-\frac{4}{v}\Omd}} \right]_1^v\\
&+ \left[\frac{1}{8v^2} \Divdo \di \Divdo \Divdo \gdcd+ \frac{1}{2v^3}\lrpar{\frac{1}{12}\Ldo\Ldo-\frac{1}{6}\Ldo+\frac{1}{4}\Divdo \Divdo \DDd_2^\ast \di - 1}\mfq_2\right]_1^v\\
=& \frac{1}{4} \Divdo \lrpar{2-\Divdo \DDd_2^\ast} \Divdo \lrpar{\ilr \frac{1}{v'^4}\chihd dv'}+\int\limits_1^v h_{\Du\ombd} dv'.
\end{aligned} \label{EQbigrepformulaDUombd111}
\end{align}
\end{enumerate}
\end{lemma}

\ni In the transport equations of Lemma \ref{LEMtransportEQUStrans}, we observe that the $\chihd$-terms appearing on the right-hand sides of \eqref{eqombd444} and \eqref{EQabd54444} have the same weight in $v$, which indicates a conservation law involving $\ombd^{[\geq2]}$ and $\abd$. Moreover, the modes $l\leq 1$ of the right-hand side of \eqref{eqombd444} and the modes $l \leq 2$ of the right-hand side of \eqref{EQequationDuombd444} do not contain $\Omd$ or $\chihd$ (see also \eqref{EQkernelofDUomboperator} in Appendix \ref{SECspecAnalysis}), which also indicates conservation laws.

Recall that $\QQ_4, \QQ_5, \QQ_6$ and $\QQ_7$ are defined in \eqref{EQdefChargesMinkowski8891} as
\begin{align*} 
\begin{aligned} 
\QQ_4 :=& \frac{\abd_\psi}{r} + 2 \DDd_2^\ast \lrpar{\frac{1}{r^2} \Divdo \chibhd - \frac{1}{r} \etad- \half \di \omtrchid + \DDd_1^\ast \lrpar{\ombd,0}}_{\psi},\\
\QQ_5 :=& \ombd^{[\leq1]} +\frac{1}{4r^2}\QQ_2^{[\leq1]} +\frac{1}{3r^3} \Divdo\QQ_0, \\
\QQ_6 :=& \Du\ombd^{[\leq1]} - \frac{1}{6r^3} (\Ldo-3) \QQ_2^{[\leq1]} +\frac{1}{r^4} \Divdo \QQ_0, \\
\QQ_7 :=& \Du\ombd^{[2]} +\frac{3}{2r^3} \QQ_2^{[2]}+ \frac{1}{2r^2} \Divdo \Divdo \QQ_3^{[2]} -\frac{12}{r^2} \QQ_1^{[2]} \\
&+{ \frac{3}{2r^2} \Divdo \lrpar{\eta+ \frac{r}{2}\di\lrpar{\omtrchid-\frac{4}{r}\Omd}} }^{[2]} - {\frac{3}{4r^2} \Divdo \Divdo \gdcd }^{[2]}.
\end{aligned} 
\end{align*}
From Lemma \ref{LEMtransportEQUStrans}, we get the following conservation laws.
\begin{lemma}[Conservation laws III] \label{CORtransportEQS3333} It holds that
\begin{align*} 
\begin{aligned} 
D\QQ_4 =& \frac{1}{v} \lrpar{\mfq_8}_\psi + \frac{2}{v^2} \lrpar{\DDd_2^\ast \Divdo\mfq_6}_\psi -\frac{2}{v} \lrpar{\DDd_2^\ast \mfq_4}_\psi - \lrpar{\DDd_2^\ast \di \mfq_5}_\psi - 2 \lrpar{\DDd_2^\ast \di \mfq_7}_\psi,
\end{aligned} 
\end{align*}
and
\begin{align*} 
\begin{aligned} 
&D\lrpar{\QQ_5 -\frac{1}{12v^2}(\Ldo+3) \mfq_2^{[\leq1]}} = h_{\ombd}^{[\leq1]},
\end{aligned} 
\end{align*}
where $h_{\ombd}$ is defined in \eqref{EQHOMBD} and
\begin{align*} 
\begin{aligned} 
D\lrpar{\QQ_6 + \frac{1}{2v^3}\lrpar{\frac{1}{12}\Ldo\Ldo-\frac{1}{6}\Ldo - 1}\mfq_2^{[\leq1]}}
=h_{\Du\ombd}^{[\leq1]},
\end{aligned} 
\end{align*}
and
\begin{align*} 
\begin{aligned} 
D\lrpar{\QQ_7 - \frac{1}{v^3}\mfq_2^{[2]}}
=h_{\Du\ombd}^{[2]},
\end{aligned} 
\end{align*}
where $h_{\Du\ombd}$ is defined in \eqref{EQdefhDUOMBD}.
\end{lemma}

\begin{proof}[Proof of Lemma \ref{CORtransportEQS3333}] By the linearized null structure equations \eqref{EQlinearizedOPsystem2}, we have that
\begin{align} 
\begin{aligned} 
D \QQ_4 =& D \lrpar{\frac{\abd_\psi}{v} + 2 \DDd_2^\ast \lrpar{ \frac{1}{v} \Divdo \lrpar{\frac{\chibhd}{v}}- \frac{1}{v^3} \lrpar{v^2 \etad} -\frac{1}{2v^2} \di \lrpar{v^2 \omtrchibd} - \di \ombd}_{\psi}}\\
=& \frac{1}{v} \lrpar{\mfq_8}_\psi +\frac{2}{v}\DDd_2^\ast \lrpar{\frac{1}{v^2} \Divdo \chibhd - \half \di \omtrchibd -\frac{1}{v} \etad}_\psi \\
&+ 2\DDd_2^\ast \lrpar{-\frac{1}{v^2} \Divdo\lrpar{\frac{\chibhd}{v}}+\frac{3}{v^4} \lrpar{v^2\etad}+ \frac{1}{v^3} \di \lrpar{v^2 \omtrchibd}}_\psi \\
&+\frac{2}{v} \DDd_2^\ast \Divdo \lrpar{ \frac{1}{v} \mfq_6 + \frac{2}{v} \DDd_2^\ast \lrpar{\etad-2\di\Omd}+\frac{1}{v^2} \chihd}_\psi \\
&-\frac{2}{v^3} \DDd_2^\ast \lrpar{v^2 \mfq_4 + 4v\di\Omd + \Divdo \chihd -\frac{v^2}{2} \di\omtrchid}_\psi \\
&-\frac{1}{v^2} \DDd_2^\ast \di \lrpar{v^2 \mfq_5 + 2v \omtrchid -2 \Divdo \lrpar{\etad-2\di\Omd}-2v^2\Kd-4\Omd} \\
&-2 \DDd_2^\ast \di \lrpar{\Kd + \frac{1}{2v} \omtrchibd -\frac{1}{2v}\omtrchid + \frac{2}{v^2}\Omd}.
\end{aligned} \label{EQQ4derivation1212}
\end{align}
Summing up the terms on the right-hand side of \eqref{EQQ4derivation1212} and using that by \eqref{EQidentity5app} in Appendix \ref{SECellEstimatesSpheres}, for all $S_v$-tangential vectorfields $X$,
\begin{align*} 
\begin{aligned} 
\lrpar{1+ \Divdo \DDd_2^\ast + \half \di \Divdo }X = 0,
\end{aligned} 
\end{align*}
it follows that
\begin{align*} 
\begin{aligned} 
D \QQ_4 =&  \frac{1}{v} \lrpar{\mfq_8}_\psi + \frac{2}{v^2} \lrpar{\DDd_2^\ast \Divdo\mfq_6}_\psi -\frac{2}{v} \lrpar{\DDd_2^\ast \mfq_4}_\psi - \lrpar{\DDd_2^\ast \di \mfq_5}_\psi - 2 \lrpar{\DDd_2^\ast \di \mfq_7}_\psi.
\end{aligned} 
\end{align*}
The conservation laws for $\QQ_5$, $\QQ_6$ and $\QQ_7$ follow directly by projecting the transport equations for $\ombd$ and $\Du\ombd$ of Lemma \ref{LEMtransportEQUStrans} onto the modes $l\leq1$ and $l= 2$. This finishes the proof of Lemma \ref{CORtransportEQS3333}. \end{proof}

\ni The proof of the following lemma is omitted.
\begin{lemma}[Properties of charges III] \label{EQQestimatesatspheres00113333} The following holds.
\begin{enumerate}
\item  Let $\xd_v$ be sphere data on a sphere $S_v \subset \HH$. Then
\begin{align*} 
\begin{aligned} 
\Vert \QQ_4 \Vert_{H^4(S_v)} + \Vert \QQ_5 \Vert_{H^4(S_v)}+ \Vert \QQ_6 \Vert_{H^2(S_v)}+ \Vert \QQ_7 \Vert_{H^2(S_v)} \les \Vert \xd_v \Vert_{\mathcal{X}(S_v)}.
\end{aligned} 
\end{align*}
\item For given sphere data $\xdmf_1$ on $S_1$, define ${}^{(1)}\QQ_i$, $4\leq i \leq 6$, as solution to the transport equation of Lemma \ref{CORtransportEQS3333} on $\HH$ with initial values given by $\QQ_i$ calculated from $\xdmf_1$. Then it holds that
\begin{align*} 
\begin{aligned} 
&\Vert {}^{(1)}\QQ_4 \Vert_{H^4_2(\HH)}+\Vert {}^{(1)}\QQ_5 \Vert_{H^4_2(\HH)}+\Vert {}^{(1)}\QQ_6 \Vert_{H^2_2(\HH)}+\Vert {}^{(1)}\QQ_7 \Vert_{H^2_2(\HH)}\\
 \les& \Vert \QQ_4(\xdmf_1) \Vert_{H^4(S_1)}+\Vert \QQ_5(\xdmf_1) \Vert_{H^4(S_1)}+\Vert \QQ_6(\xdmf_1) \Vert_{H^2(S_1)}+\Vert \QQ_7(\xdmf_1) \Vert_{H^2(S_1)} + \Vert (\mfq_i)_{1\leq i \leq 10} \Vert_{\mathcal{Z}_\CC}.
\end{aligned} 
\end{align*}
\end{enumerate}
\end{lemma}

\subsection{Solution of the linearized characteristic gluing problem} \label{SECsolutionLinearized99902} In this section we prove Theorem \ref{PROPlingluing22}, that is, we solve the linearized characteristic gluing problem. Consider given 
\begin{itemize}
\item source terms $(\mfq_i)_{1\leq i \leq 10} \in \mathcal{Z}_{\CC}$,
\item sphere data $\dot{\mathfrak{X}}_{1}\in \XX(S_{1})$,
\item matching data $\dot{\MMf}_{2} \in \ZZ_\MMf(S_{2})$.
\end{itemize}

\ni In the following we use the charges, representation formulas and estimates of Sections \ref{SEClinearizedCHARGESMinkowskiSTEF89}, \ref{SECgaugeTheorySTEF89} and \ref{SECprelimAnalysis} to construct a solution $\xd$ on $\HH$ satisfying \eqref{EQlinearizedOPsystem2} with matching conditions \eqref{EQmatchingLinSectionGOAL} and bounds \eqref{EQlinearEstimateTOPROVE1233}. We proceed as follows.
\begin{itemize}
\item In Section \ref{SECmatchingCharges11}, we apply Proposition \ref{PROPadjustmentCharges} to add linearized perturbations of sphere data on $S_{2}$ to match the gauge-dependent charges on $S_2$ with the gauge-dependent charges coming from $S_1$.

\item In Section \ref{SECsecondchoice12}, we derive conditions on the free conformal data $\Omd$ and $\chihd$ on $\HH$ such that the constructed solution $\xd$, given through the representation formulas in Section \ref{SECprelimAnalysis}, satisfies \eqref{EQmatchingLinSectionGOAL}.

\item In Section \ref{SECboundslinearizedSolution112345} we prove the estimate \eqref{EQlinearEstimateTOPROVE1233} for the constructed solution and the linearized perturbation functions.
\end{itemize}

\subsubsection{Matching of gauge-dependent charges} \label{SECmatchingCharges11}

In this section we apply Proposition \ref{PROPadjustmentCharges} to add linearized perturbations of sphere data to $S_2$ to match the gauge-dependent charges
\begin{align*} 
\begin{aligned} 
\lrpar{\QQ_1, \QQ_2^{[\geq2]}, \QQ_3, \QQ_4, \QQ_5, \QQ_6, \QQ_7}.
\end{aligned} 
\end{align*}

\ni On the one hand, for given sphere data $\xdmf_1 \in \XX(S_1)$ on $S_1$ and matching condition on $S_1$
\begin{align*} 
\begin{aligned} 
\xd \vert_{S_{0,1}} = \xdmf_1,
\end{aligned} 
\end{align*}
define ${}^{(1)}{\QQ}_i$, $0\leq i \leq 7$ on $\HH$ to be the solutions to the transport equations of Lemmas \ref{CORtransport1EQS}, \ref{LEMgaugeconservationlaws} and \ref{CORtransportEQS3333} with initial values ${}^{(1)}\QQ_i$ on $S_1$ calculated from $\xdmf_1$. We underline that the charges ${}^{(1)}\QQ_i$ on $\HH$ depend only on $\xdmf_1$ and $(\mfq_i)_{1\leq i \leq 10}$, and are independent of the solution $\xd$ to the linearized null structure equations to be constructed on $\HH$. Lemmas \ref{EQQestimatesatspheres00111111}, \ref{EQQestimatesatspheres0011} and \ref{EQQestimatesatspheres00113333} imply that
\begin{align} 
\begin{aligned} 
&\Vert {}^{(1)}\QQ_0 \Vert_{H^5(S_2)}+\Vert {}^{(1)}\QQ_1 \Vert_{H^6(S_2)} + \Vert {}^{(1)}\QQ_2 \Vert_{H^4(S_2)}+ \Vert {}^{(1)}\QQ_3 \Vert_{H^4(S_2)} \\
&+\Vert {}^{(1)}\QQ_4 \Vert_{H^4(S_2)} + \Vert {}^{(1)}\QQ_5 \Vert_{H^4(S_2)}+ \Vert {}^{(1)}\QQ_6 \Vert_{H^2(S_2)}+ \Vert {}^{(1)}\QQ_7 \Vert_{H^2(S_2)} \\
 \les& \Vert \xdmf_1 \Vert_{\XX(S_1)} + \Vert (\mfq_i)_{1\leq i \leq 10} \Vert_{\mathcal{Z}_\CC}.
\end{aligned} \label{EQQiEstimates2223}
\end{align}

\ni On the other hand, for given linearized matching data $\dot{\MMf}_2 \in \ZZ_{\MMf}(S_2)$ at $S_2$,
\begin{align*} 
\begin{aligned} 
\dot{\MMf}_2 = \lrpar{\Omd, \phid, \gdcd, \omtrchid, \chihd, \omtrchibd^{[\geq2]}, \chibhd, \etad^{[\geq2]}, \omd, D\omd, \ombd^{[\geq2]}, \Du\ombd^{[\geq2]}, \QQ_5, \QQ_6, \ad, \abd},
\end{aligned} 
\end{align*}
 define on $S_2$ the charges 
\begin{align*} 
\begin{aligned} 
\lrpar{{}^{(2)}\QQ_1, {}^{(2)}\QQ_2^{[\geq2]}, {}^{(2)}\QQ_3, {}^{(2)}\QQ_4, {}^{(2)}\QQ_5, {}^{(2)}\QQ_6, {}^{(2)}\QQ_7}.
\end{aligned} 
\end{align*}
By definition of $\MMf$ in Definition \ref{DEFmatchingMAP}, all charges are well-defined, and we have the bounds
\begin{align} 
\begin{aligned} 
&\Vert {}^{(2)}\QQ_1 \Vert_{H^6(S_2)} + \Vert {}^{(2)}\QQ_2^{[\geq2]} \Vert_{H^4(S_2)}+ \Vert {}^{(2)}\QQ_3 \Vert_{H^4(S_2)} \\
&+\Vert {}^{(2)}\QQ_4 \Vert_{H^4(S_2)} + \Vert {}^{(2)}\QQ_5 \Vert_{H^4(S_2)}+ \Vert {}^{(2)}\QQ_6 \Vert_{H^2(S_2)}+ \Vert {}^{(2)}\QQ_7 \Vert_{H^2(S_2)} \\
 \les& \Vert \dot{\MMf}_2 \Vert_{\ZZ_\MMf}.
\end{aligned} \label{EQmatchingdatacharges1}
\end{align}

\ni Applying Proposition \ref{PROPadjustmentCharges} with
\begin{align} 
\begin{aligned} 
(\QQ_1)_0:=& {}^{(1)}\QQ_1-{}^{(2)}\QQ_1, & (\QQ_2)_0:=&{}^{(1)}\QQ_2^{[\geq2]}-{}^{(2)}\QQ_2^{[\geq2]}, &(\QQ_3)_0:=&{}^{(1)}\QQ_3-{}^{(2)}\QQ_3,\\
(\QQ_4)_0:=&{}^{(1)}\QQ_4-{}^{(2)}\QQ_4, & (\QQ_5)_0:=&{}^{(1)}\QQ_5-{}^{(2)}\QQ_5, &(\QQ_6)_0:=&{}^{(1)}\QQ_6-{}^{(2)}\QQ_6,\\
(\QQ_7)_0:=& {}^{(1)}\QQ_7-{}^{(2)}\QQ_7, & & & &
\end{aligned} \label{EQapplicationofChargeMatching4321}
\end{align}
it follows that there exist linearized perturbation functions $\dot f$ and $(\dot{j}^1, \dot{j}^2)$ at $S_2$ such that the gauge-dependent charges of 
\begin{align*} 
\begin{aligned} 
\xd\vert_{S_2}+ \dot{\PP}_f(\dot f)+\dot{\PP}_{({j}^1,{j}^2)}(\dot{j}^1,\dot{j}^2)
\end{aligned} 
\end{align*}
match with the gauge-dependent charges ${}^{(1)}{\QQ}_i$, that is, for $i=1,3,4,5,6,7$,
\begin{align} 
\begin{aligned} 
\QQ_i\lrpar{\xd\vert_{S_2}+ \dot{\PP}_f(\dot f)+\dot{\PP}_{({j}^1,{j}^2)}(\dot{j}^1,\dot{j}^2)} =& {}^{(2)}\QQ_i+\QQ_i\lrpar{ \dot{\PP}_f(\dot f)+\dot{\PP}_{({j}^1,{j}^2)}(\dot{j}^1,\dot{j}^2)} \\
= &{}^{(2)}\QQ_i+ (\QQ_i)_0 \\
=& {}^{(1)}\QQ_i,
\end{aligned} \label{EQchargematching314222}
\end{align}
and
\begin{align} 
\begin{aligned} 
\QQ_2^{[\geq2]}\lrpar{\xd\vert_{S_2}+ \dot{\PP}_f(\dot f)+\dot{\PP}_{({j}^1,{j}^2)}(\dot{j}^1,\dot{j}^2)} =&
{}^{(2)}\QQ_2^{[\geq2]}+\QQ_2^{[\geq2]}\lrpar{ \dot{\PP}_f(\dot f)+\dot{\PP}_{({j}^1,{j}^2)}(\dot{j}^1,\dot{j}^2)}\\
=&{}^{(2)}\QQ_2^{[\geq2]}+ (\QQ_2)_0\\
 =& {}^{(1)}\QQ_2^{[\geq2]}.
\end{aligned} \label{EQchargematching3142223}
\end{align}

\ni Moreover, by Proposition \ref{PROPadjustmentCharges}, \eqref{EQQiEstimates2223}, \eqref{EQmatchingdatacharges1} and \eqref{EQapplicationofChargeMatching4321}, we have the estimate
\begin{align} 
\begin{aligned} 
&\Vert \dot f \Vert_{\YY_f}+ \Vert (\dot{j}^1,\dot{j}^2) \Vert_{\YY_{({j}^1,{j}^2)}}+ \Vert \dot{\PP}_f(\dot f) \Vert_{\XX(S_2)} + \Vert \dot{\PP}_{({j}^1,{j}^2)}(\dot{j}^1,\dot{j}^2) \Vert_{\XX(S_2)}\\
\les& \Vert (\QQ_1)_0 \Vert_{H^6(S_2)}+\Vert (\QQ_2)_0 \Vert_{H^4(S_2)}+\Vert (\QQ_3)_0 \Vert_{H^4(S_2)}+\Vert (\QQ_4)_0 \Vert_{H^2(S_2)} \\
& +\Vert (\QQ_5)_0 \Vert_{H^4(S_2)}+\Vert (\QQ_6)_0 \Vert_{H^2(S_2)}+\Vert (\QQ_7)_0 \Vert_{H^2(S_2)}\\
\les& \Vert \xdmf_1 \Vert_{\mathcal{X}(S_1)} + \Vert \dot{\MMf}_2 \Vert_{\ZZ_\MMf(S_2)} + \Vert (\mfq_i)_{1\leq i \leq 10} \Vert_{\mathcal{Z}_\CC}.
\end{aligned} \label{EQQmatchingEstimatefullRic11}
\end{align}
This finishes our matching of the gauge-dependent charges on $S_2$.

\subsubsection{Integral conditions on $\Omd$ and $\chihd$} 
\label{SECsecondchoice12} In the previous section we constructed linearized perturbation functions $\dot{f}$ and $(\dot{j}^1,\dot{j}^2)$ such that the gauge-dependent charges match on $S_2$, see \eqref{EQchargematching314222} and \eqref{EQchargematching3142223}. In this section we use the matching of gauge-dependent charges together with the representation formulas for $\phid, \gdcd, \etad^{[\geq2]}, \abd$ and $\Du\ombd$ and our freedom of prescribing $\Omd$ and $\chihd$ along $\HH$ to construct a solution $\xd$ of the linearized null structure equations \eqref{EQlinearizedOPsystem2} satisfying
\begin{align} 
\begin{aligned} 
\xd\vert_{S_1} =& \xdmf_1, \\
\MMf\lrpar{\xd\vert_{S_2}+ \dot{\PP}_f(\dot f)+\dot{\PP}_{({j}^1,{j}^2)}(\dot{j}^1,\dot{j}^2)} =& \dot{\MMf}_2 \text{ on } S_2,
\end{aligned} \label{EQS2matchingdata12333}
\end{align}
where we recall from Definition \ref{DEFmatchingMAP} that
\begin{align*} 
\begin{aligned} 
\mathfrak{M}(x_2):= \lrpar{\Omd, \phid, \gdcd, \omtrchid, \chihd, \omtrchibd^{[\geq2]}, \chibhd, \etad^{[\geq2]}, \omd, D\omd, \ombd^{[\geq2]}, \Du\ombd^{[\geq2]}, \QQ_5, \QQ_6, \ad, \abd}.
\end{aligned} 
\end{align*}
Specifically, the additional conditions derived below on $\Omd$ and $\chihd$ are independent of the boundary values of $\Omd$ and $\chihd$ on $S_1$ and $S_2$ which are already determined by \eqref{EQS2matchingdata12333}. \\

\ni \textbf{(1) Gluing of $\phid$.} By the representation formula \eqref{EQreps11} for $\phid$, we have
\begin{align*} \begin{aligned} 
\phid =& 2 \int\limits_1^v \Omd dv' + \ilr \int\limits_1^{v'} \mfq_1 dv'' dv'+v\phid(1)+\frac{v-1}{2}\lrpar{\omtrchid(1)-4\Omd(1)+\mfq_2(1)}.
\end{aligned} \end{align*}

\ni To match $\phid$ according to \eqref{EQS2matchingdata12333}, prescribe $\Omd$ on $\HH$ such that
\begin{align} \begin{aligned} 
2\int\limits_1^2 \Omd dv' = \phid(2)-2\phid(1)- \half \lrpar{\omtrchid(1)-4\Omd(1)+\mfq_2(1)}- \int\limits_1^2 \int\limits_1^{v'} \mfq_1 dv'' dv'.
\end{aligned}\label{EQprescriptionOMEGA} \end{align}
In particular it holds that
\begin{align} 
\begin{aligned} 
\left\Vert \int\limits_1^2 \Omd dv' \right\Vert_{H^6(S_1)} \les& \Vert \xdmf_1 \Vert_{\mathcal{X}(S_1)} + \Vert \dot{\MMf}_2 \Vert_{\ZZ_\MMf(S_2)} + \Vert (\mfq_i)_{1\leq i \leq 10} \Vert_{\mathcal{Z}_\CC}.
\end{aligned} \label{EQprescriptionOMEGAestimate1}
\end{align}

\ni \textbf{(2) Gluing of $\gdcd$.} By the representation formula \eqref{EQreps12} for $\gdcd$, we have that
\begin{align*} \begin{aligned} 
\dot{\gd_c}(v) =&2 \int\limits_{1}^v \frac{1}{v'^2} \chihd dv' + \gdcd(1)+ \ilr \frac{1}{v'^2}{\mfq_3}dv'.
\end{aligned} \end{align*}

\ni To match $\gdcd$ according to \eqref{EQS2matchingdata12333}, prescribe $\chihd$ on $\HH$ such that
\begin{align} \begin{aligned} 
\int\limits_{1}^2 \frac{1}{v'^2} \chihd dv' = \half \lrpar{\gdcd(2)-\gdcd(1)-\int\limits_1^2 \frac{1}{v'^2}{\mfq_3}dv'}.
\end{aligned} \label{EQprescription1v2chihd} \end{align}

\ni In particular it holds that
\begin{align} 
\begin{aligned} 
\left\Vert \int\limits_{1}^2 \frac{1}{v'^2} \chihd dv' \right\Vert_{H^{5}(S_1)} \les& \Vert \xdmf_1 \Vert_{\mathcal{X}(S_1)} + \Vert \dot{\MMf}_2 \Vert_{\ZZ_\MMf(S_2)} + \Vert (\mfq_i)_{1\leq i \leq 10} \Vert_{\mathcal{Z}_\CC}.
\end{aligned} \label{EQprescription1v2chihdestimate1}
\end{align}

\ni \textbf{(3) Gluing of $\omtrchid$.} By the above gluing of $\Omd$, $\phid$ and the matching of
\begin{align*} \begin{aligned} 
\QQ_1 := \frac{v}{2} \lrpar{\omtrchid-\frac{4}{v} \Omd} +\frac{\phid}{v},
\end{aligned} \end{align*}
at $S_2$ in \eqref{EQchargematching314222}, it follows that $\omtrchid^{[\geq0]}$ is matched according to \eqref{EQS2matchingdata12333}.\\

\ni \textbf{(4) Gluing of $\etad^{[\geq2]}$.} By the representation formula for $\etad$ in Lemma \ref{LEMrepphigdeta}, we have that
\begin{align*} \begin{aligned} 
\left[ v'^2 \etad + \frac{v'^3}{2} \di\lrpar{\omtrchid - \frac{4}{v'} \Omd}+\frac{v'}{2}\di \mfq_2 \right]_1^v
 =& \Divdo\lrpar{ \int\limits_1^v \chihd dv' } +\ilr \lrpar{v'^2 \mfq_4 +v'^2 \di \mfq_1 +\frac{1}{2}\di \mfq_2} dv'.
\end{aligned} \end{align*}

\ni By the above gluing of $\omtrchid$ and $\Omd$, to glue $\etad^{[\geq2]}$ according to \eqref{EQS2matchingdata12333} it suffices to choose $\chihd$ such that
\begin{align} \begin{aligned} 
\Divdo\lrpar{ \int\limits_1^2 \chihd dv' } =& \left[ v^2 \etad + \frac{v^3}{2} \di\lrpar{\omtrchid - \frac{4}{v} \Omd}+\frac{v}{2}\di \mfq_2 \right]_1^2\\
&-\int\limits_1^2 \lrpar{v'^2 \mfq_4 +v'^2 \di \mfq_1 +\frac{1}{2}\di \mfq_2} dv'
\end{aligned} \label{EQgluingConditionCHIHD2} \end{align}
By elliptic estimates for the operator $\Divdo$, see Appendix \ref{SECellEstimatesSpheres}, we have that
\begin{align} 
\begin{aligned} 
\left\Vert \int\limits_1^2 \chihd dv'  \right\Vert_{H^{5}(S_1)} \les& \Vert \xdmf_1 \Vert_{\mathcal{X}(S_1)} + \Vert \dot{\MMf}_2 \Vert_{\ZZ_\MMf(S_2)}+ \Vert (\mfq_i)_{1\leq i \leq 10} \Vert_{\mathcal{Z}_\CC}.
\end{aligned} \label{EQgluingConditionCHIHD2estimate1}
\end{align}
\ni \textbf{(5) Gluing of $\omtrchibd^{[\geq2]}$.} By the matching of 
\begin{align*} \begin{aligned} 
\QQ_2 :=& v^2 \omtrchibd -\frac{2}{v}\Divdo \lrpar{v^2\etad+\frac{v^3}{2}\di\lrpar{\omtrchid-\frac{4}{v}\Omd}} -v^2 \lrpar{\omtrchid-\frac{4}{v}\Omd} +2v^3 \Kd, 
\end{aligned} \end{align*}
for modes $l\geq2$ in \eqref{EQchargematching3142223} and the relation \eqref{EQgaussDHR}, that is,
\begin{align*}
\Kd =& \frac{1}{2r^2} \Divdo \Divdo \gdcd - \frac{1}{r^3} (\Ldo +2) \phid,
\end{align*}
it follows with the above gluing of $\Omd, \gdcd, \phid, \omtrchid$ and $\etad^{[\geq2]}$ that $\omtrchibd^{[\geq2]}$ is glued at $S_2$ according to \eqref{EQS2matchingdata12333}. \\

\ni \textbf{(6) Gluing of $\chibhd$.} By the matching of 
\begin{align*} \begin{aligned} 
\QQ_3:=& \frac{\chibhd}{v} -\half \lrpar{ \DDd_2^\ast \Divdo +1} \gdcd + \DDd_2^\ast \lrpar{ \etad + \frac{v}{2}\di \lrpar{\omtrchid-\frac{4}{v}\Omd}} - v \DDd_2^\ast \di \lrpar{\omtrchid-\frac{4}{v}\Omd},
\end{aligned} \end{align*}
at $S_2$ in \eqref{EQchargematching314222}, together with the above gluing of $\Omd, \gdcd, \omtrchid$ and $\etad^{[\geq2]}$, it follows that $\chibhd$ is glued at $S_2$ according to \eqref{EQS2matchingdata12333}.\\

\ni \textbf{(7) Gluing of $\ad$.} We have by the linearized null structure equations \eqref{EQlinearizedOPsystem2} that on $\HH$
\begin{align} 
\begin{aligned} 
\ad + D\chihd=\mfq_{10}.
\end{aligned} \label{EQtransporteqchihd1234}
\end{align}
Hence we glue $\ad$ at $S_2$ according to \eqref{EQS2matchingdata12333} by prescribing,
\begin{align} 
\begin{aligned} 
D\chihd(1) = \mfq_{10}(1) - \ad(1), \,\, D\chihd(2) = \mfq_{10}(2) - \ad(2).
\end{aligned} \label{EQadgluing}
\end{align}
This implies that
\begin{align} 
\begin{aligned} 
&\Vert D\chihd(1) \Vert_{H^6(S_1)}+\Vert D\chihd(2) \Vert_{H^6(S_2)}\\
 \les& \Vert \mfq_{10} \Vert_{H^6(S_1)}+\Vert \mfq_{10} \Vert_{H^6(S_2)} + \Vert \xdmf_{1}\Vert_{\XX(S_1)} + \Vert \dot{\MMf}_2 \Vert_{\ZZ_\MMf(S_2)}.
\end{aligned} \label{EQadgluingESTIMATES}
\end{align}

\ni \textbf{(8) Gluing of $\omd$ and $D\omd$.} By the relation $\omd = D\Omd$, the gluing of $\omd$ and $D\omd$ at $S_2$ according to \eqref{EQS2matchingdata12333} is satisfied if
\begin{align} 
\begin{aligned} 
\omd(1)=& D\Omd(1), & \omd(2)=& D\Omd(2), \\
D\omd(1)=& D^2\Omd(1), & D\omd(2)=& D^2\Omd(2).
\end{aligned} \label{EQgluingConditionsDOMDDOM}
\end{align}
In particular, we have the bound
 \begin{align} 
\begin{aligned} 
&\Vert D\Omd(1) \Vert_{H^6(S_1)}+\Vert D\Omd(2) \Vert_{H^6(S_2)}+ \Vert D^2\Omd(1) \Vert_{H^6(S_1)}+\Vert D^2\Omd(2) \Vert_{H^6(S_2)} \\
\les& \Vert \xdmf_1 \Vert_{\XX(S_1)} + \Vert \dot{\MMf}_2 \Vert_{\ZZ_{\MMf}}.
\end{aligned} \label{EQgluingConditionsDOMDDOMBOUND}
\end{align}
 
\ni \textbf{(9) Gluing of $\abd$.} Using the representation formula \eqref{EQbigrepformulaABD111} for $\abd$, to glue $\abd$ at $S_2$ according to \eqref{EQS2matchingdata12333}, we can pick $\chihd$ such that
\begin{align} \begin{aligned}
&\left[ \frac{\abd}{v} +\frac{2}{v}\DDd_2^\ast \Divdo \QQ_3 - \frac{1}{2v^2} \DDd_2^\ast \di \QQ_2 - \frac{2}{v} \DDd_2^\ast \di \lrpar{\Ldo+2}\QQ_1 \right]^2_1 \\
&-\left[ \frac{2}{3v} \DDd_2^\ast \lrpar{\Divdo \DDd_2^\ast + 1 + \di\Divdo} \lrpar{\etad+ \frac{v}{2}\di\lrpar{\omtrchid-\frac{4}{v}\Omd}} \right]^2_1\\
&+\left[ \frac{1}{v} \DDd_2^\ast \lrpar{\Divdo \DDd_2^\ast + 1 + \di\Divdo} \Divdo\gdcd- \frac{1}{3v^2} \DDd_2^\ast \lrpar{\Divdo \DDd_2^\ast + 1 + \di\Divdo} \Divdo\di \mfq_2 \right]^2_1\\
=& \frac{4}{3} \DDd_2^\ast \lrpar{\Divdo \DDd_2^\ast + 1 + \di\Divdo} \Divdo \lrpar{\int\limits_1^2 \frac{1}{v'^3}\chihd dv'} +\int\limits_1^2 h_{\abd} dv'.
\end{aligned}\label{EQgluingConditionCHIHD4} \end{align}

\ni By the elliptic estimate  \eqref{EQestimate42app} in Appendix \ref{SECellEstimatesSpheres} for the operator
\begin{align*} 
\begin{aligned} 
\DDd_2^\ast \lrpar{\Divdo \DDd_2^\ast + 1 + \di\Divdo} \Divdo
\end{aligned} 
\end{align*}
it follows that the integral is well-defined and
\begin{align} 
\begin{aligned} 
\left\Vert \int\limits_1^2 \frac{1}{v'^3}\chihd dv' \right\Vert_{H^6(S_1)} \les& \Vert \xdmf_1 \Vert_{\mathcal{X}(S_1)} + \Vert \dot{\MMf}_2 \Vert_{\ZZ_\MMf(S_2)} + \Vert (\mfq_i)_{1\leq i \leq 10} \Vert_{\mathcal{Z}_\CC}. 
\end{aligned} \label{EQchihcond3estimate2}
\end{align}

\ni \textbf{(10) Gluing of $\ombd^{[\geq2]}$.} By the matching of 
\begin{align*} 
\begin{aligned} 
\QQ_4 :=& \,  \frac{\abd_\psi}{v} + 2 \DDd_2^\ast \lrpar{\frac{1}{v^2} \Divdo \chibhd - \frac{1}{v} \etad- \half \di \omtrchid + \DDd_1^\ast \lrpar{\ombd,0}}_{\psi},
\end{aligned} 
\end{align*}
in \eqref{EQchargematching314222}, the above gluing of $\abd$, $\chibhd$, $\etad^{[\geq2]}$ and $\omtrchid$, and the fact that the operator
 \begin{align*} 
\begin{aligned} 
\DDd_2^\ast \lrpar{\DDd_1^\ast \lrpar{\ombd,0}}_{\psi}
\end{aligned} 
\end{align*}
has trivial kernel and is elliptic (see Appendix \ref{SECellEstimatesSpheres}), it follows that $\ombd^{[l\geq2]}$ is glued at $S_2$ according to \eqref{EQS2matchingdata12333}. \\

\ni \textbf{(11) Gluing of $\Du \ombd^{[\geq2]}$.} On the one hand, by the representation formula for $\Du\ombd$ of Lemma \ref{COR2ndenergyEstimatesOMBDABDDUOMBD}, we have that

\begin{align} 
\begin{aligned} 
&\left[ \Du\ombd -\frac{1}{6v^3} \lrpar{\Ldo-3}\QQ_2+ \frac{1}{2v^2} \Divdo \Divdo \QQ_3 +\frac{1}{v^2} \Divdo \Divdo \DDd_2^\ast \di \QQ_1\right]_1^2\\
&-\left[ \frac{1}{4v^2} \Divdo \lrpar{\di \Divdo -2 + \Divdo\DD_2^\ast} \lrpar{\eta+ \frac{v}{2}\di\lrpar{\omtrchid-\frac{4}{v}\Omd}} \right]_1^2\\
&+ \left[\frac{1}{8v^2} \Divdo \di \Divdo \Divdo \gdcd+ \frac{1}{2v^3}\lrpar{\frac{1}{12}\Ldo\Ldo-\frac{1}{6}\Ldo+\frac{1}{4}\Divdo \Divdo \DDd_2^\ast \di - 1}\mfq_2\right]_1^2\\
=& \frac{1}{4} \Divdo \lrpar{2-\Divdo \DDd_2^\ast} \Divdo \lrpar{\int\limits_1^2 \frac{1}{v'^4}\chihd dv'}+\int\limits_1^2 h_{\Du\ombd} dv'.
\end{aligned} \label{EQbigrepformulaDUombd111MATCHING}
\end{align}

\ni On modes $l\geq3$, the operator on the right-hand side of \eqref{EQbigrepformulaDUombd111MATCHING}
\begin{align*} 
\begin{aligned} 
\Divdo \lrpar{2-\Divdo \DDd_2^\ast} \Divdo,
\end{aligned} 
\end{align*}
has trivial kernel and is elliptic, see Appendix \ref{SECellEstimatesSpheres}. Hence, projecting the above representation formula onto modes $l\geq3$, we can pick the $\chihd$-integral such that $\Du\ombd^{[\geq3]}$ is glued at $S_2$ according to \eqref{EQS2matchingdata12333}. Picking the projection of the $\chihd$-integral onto the mode $l=2$ to vanish, we get the estimate
\begin{align} 
\begin{aligned} 
\left\Vert \int\limits_1^2 \frac{1}{v'^4}\chihd dv' \right\Vert_{H^6(S_1)} \les& \Vert \xdmf_1 \Vert_{\mathcal{X}(S_1)} + \Vert \dot{\MMf}_2 \Vert_{\ZZ_\MMf(S_2)} + \Vert (\mfq_i)_{1\leq i \leq 10} \Vert_{\mathcal{Z}_\CC}. 
\end{aligned} \label{EQchihcond3estimate}
\end{align}

\ni On the other hand, by the matching of $\QQ_1$, $\QQ_2^{[2]}$, $\QQ_3$ and
\begin{align*} 
\begin{aligned} 
\QQ_7 :=& \Du\ombd^{[2]} +\frac{3}{2v^3} \QQ_2^{[2]}+ \frac{1}{2v^2} \Divdo \Divdo \QQ_3^{[2]} -\frac{12}{v^2} \QQ_1^{[2]} \\
&+{ \frac{3}{2v^2} \Divdo \lrpar{\eta+ \frac{v}{2}\di\lrpar{\omtrchid-\frac{4}{v}\Omd}} }^{[2]} - {\frac{3}{4v^2} \Divdo \Divdo \gdcd }^{[2]}, 
\end{aligned} 
\end{align*}
in \eqref{EQchargematching3142223} and the above gluing of $\Omd, \gdcd, \omtrchid,\etad^{[\geq2]}$ it follows that $\Du\ombd^{[2]}$ is glued at $S_2$ according to \eqref{EQS2matchingdata12333}.\\ 

\ni To summarise the above, we derived the integral conditions \eqref{EQprescriptionOMEGA}, \eqref{EQprescription1v2chihd}, \eqref{EQgluingConditionCHIHD2}, \eqref{EQadgluing}, \eqref{EQgluingConditionsDOMDDOM}, \eqref{EQgluingConditionCHIHD4} and \eqref{EQbigrepformulaDUombd111MATCHING} on $\Omd$ and $\chihd$ along $\HH$ which, if satisfied, imply the matching of \emph{matching data} \eqref{EQS2matchingdata12333} on $S_2$, that is,
\begin{align*} 
\begin{aligned} 
\MMf\lrpar{\xd\vert_{S_2}+ \dot{\PP}_f(\dot f)+\dot{\PP}_{({j}^1,{j}^2)}(\dot{j}^1,\dot{j}^2)} =& \dot{\MMf}_2 \text{ on } S_2,
\end{aligned} 
\end{align*}
with
\begin{align*} 
\begin{aligned} 
\mathfrak{M}(x_2):= \lrpar{\Omd, \phid, \gdcd, \omtrchid, \chihd, \omtrchibd^{[\geq2]}, \chibhd, \etad^{[\geq2]}, \omd, D\omd, \ombd^{[\geq2]}, \Du\ombd^{[\geq2]}, \QQ_5, \QQ_6, \ad, \abd}.
\end{aligned} 
\end{align*}
In the next section we show that $\Omd$ and $\chihd$ satisfying these conditions can be constructed in a regular fashion and prove estimates for the constructed solution $\xd$.

\subsubsection{Construction of solution and estimates} \label{SECboundslinearizedSolution112345} In this section we pick $\Omd$ and $\chihd$ subject to the conditions \eqref{EQS2matchingdata12333}, \eqref{EQprescriptionOMEGA}, \eqref{EQprescription1v2chihd}, \eqref{EQgluingConditionCHIHD2}, \eqref{EQadgluing}, \eqref{EQgluingConditionsDOMDDOM}, \eqref{EQgluingConditionCHIHD4} and \eqref{EQbigrepformulaDUombd111MATCHING}, and subsequently prove the estimate \eqref{EQlinearEstimateTOPROVE1233} for the constructed solution $\xd$.\\

\ni \textbf{Choice of $\chihd$ and $\Omd$ and estimates.} The proof of the following technical lemma follows from a straight-forward orthogonality construction and is omitted.

\begin{lemma}[Technical lemma] \label{LEMtechnicalEXPLICITFORMULA} Consider scalar functions $h_i$, $1\leq i \leq 7$, on the round unit sphere $S_1$. Then there exists a scalar function $\Omd$ on $\HH$ such that
\begin{align} 
\begin{aligned} 
\Omd(1) =& h_1, & D\Omd(1) =& h_2, & D^2\Omd(1) =& h_3, \\
\Omd(2) =& h_4, & D\Omd(2) =& h_5, & D^2\Omd(2) =& h_6,
\end{aligned} \label{EQomegadboundarycond}
\end{align}
and 
\begin{align*} 
\begin{aligned} 
\int\limits_1^2 \Omd dv' = h_7,
\end{aligned} 
\end{align*}
and the following bound holds,
\begin{align} 
\begin{aligned} 
\Vert \Omd \Vert_{H^6_2(\HH)} \les& \sum\limits_{1\leq i \leq 7}\Vert h_i \Vert_{H^6(S_1)}.
\end{aligned} \label{EQOmdconstructionbounds}
\end{align}
Further, consider tracefree symmetric $2$-tensors $W_i$, $1\leq i \leq8$ on $S_1$. Then there exists a tracefree symmetric $S_v$-tangent $2$-tensor $\chihd$ on $\HH$ such that 
\begin{align} 
\begin{aligned} 
 \chihd(1) =& W_1, &D\chihd(1) =& W_3, \\
 \chihd(2) =& W_2, & D\chihd(2) =& W_4,
\end{aligned} \label{EQchihdboundarycond1}
\end{align}
and
\begin{align*} 
\begin{aligned} 
\int\limits_1^2 \frac{1}{v'^2} \chihd dv' = W_5, \,\,\int\limits_1^2 \chihd dv'= W_6, \,\, \int\limits_1^2 \frac{1}{v'^3} \chihd dv'  = W_7, \,\,\int\limits_1^2 \frac{1}{v'^4} \chihd dv' = W_8.
\end{aligned} 
\end{align*}
and the following estimates hold,
\begin{align} 
\begin{aligned} 
\Vert \chihd \Vert_{H^6_2(\HH)} \les& \sum\limits_{1\leq i \leq 8} \Vert W_i \Vert_{H^6(S_1)}.
\end{aligned} \label{EQchihdconstructionbounds}
\end{align}

\end{lemma}

\begin{remark}[Linearized characteristic gluing of higher-order $L$-derivatives II] \label{REMhigherLgluing2} Let $m\geq0$ be an integer. By the linearized null structure equations
\begin{align*} 
\begin{aligned} 
D\chihd = -\ad + \dot{\CC}_{10}, \,\, \omd = D\Omd,
\end{aligned} 
\end{align*}
Lemma \ref{LEMtechnicalEXPLICITFORMULA} extends in a straight-forward way to the higher-order boundary conditions given by \eqref{EQomegadboundarycond}, \eqref{EQchihdboundarycond1} and the additional
\begin{align*} 
\begin{aligned} 
D^i \omd(1) = V_{1,i}, \,\, D^i \omd(2) = V_{2,i} \,\, \widehat{D}^i \ad(1) = V_{3,i}, \,\, \widehat{D}^i \ad(2) = V_{4,i}  \text{ for } 1\leq i \leq m.
\end{aligned} 
\end{align*}
In this setting, the right-hand sides of \eqref{EQOmdconstructionbounds} and \eqref{EQchihdconstructionbounds} get the following additional terms, respectively,
\begin{align*} 
\begin{aligned} 
&\sum\limits_{0\leq i \leq m} \lrpar{\Vert h_i \Vert_{H^6(S_1)} + \Vert V_{1,i}\Vert_{H^6(S_2)}+\Vert V_{2,i}\Vert_{H^6(S_2)}}  \\ 
\text{ and }&\sum\limits_{0\leq i \leq m} \lrpar{\Vert W_i \Vert_{H^6(S_1)} + \Vert V_{3,i} \Vert_{H^6(S_2)}+\Vert V_{4,i} \Vert_{H^6(S_2)}}.
\end{aligned} 
\end{align*}
\end{remark}

\ni Let $\Omd$ and $\chihd$ be the quantities constructed in Lemma \ref{LEMtechnicalEXPLICITFORMULA} subject to the gluing conditions
\begin{itemize}
\item \eqref{EQprescriptionOMEGA} for $\Omd$,
\item \eqref{EQprescription1v2chihd}, \eqref{EQgluingConditionCHIHD2},  \eqref{EQgluingConditionCHIHD4} and  \eqref{EQbigrepformulaDUombd111MATCHING} for $\chihd$,
\end{itemize}
and the prescribed boundary values given by $\xdmf_1$ on $S_1$, $\dot{\MMf}_2$ on $S_2$ and \eqref{EQadgluing} and \eqref{EQgluingConditionsDOMDDOM}.

By Lemma \ref{LEMtechnicalEXPLICITFORMULA} together with the estimates \eqref{EQprescriptionOMEGAestimate1}, \eqref{EQprescription1v2chihdestimate1}, \eqref{EQgluingConditionCHIHD2estimate1}, \eqref{EQadgluingESTIMATES}, \eqref{EQgluingConditionsDOMDDOMBOUND}, \eqref{EQchihcond3estimate2} and \eqref{EQchihcond3estimate}, the constructed $\Omd$ and $\chihd$ satisfy
\begin{align} 
\begin{aligned} 
\Vert \Omd \Vert_{H^6_3(\HH)}+\Vert \chihd \Vert_{H^6_2(\HH)} \les& \Vert \dot{\mathfrak{X}}_1 \Vert_{\XX(S_1)} + \Vert \dot{\MMf}_2 \Vert_{\ZZ_\MMf(S_2)} + \Vert (\mfq_i)_{1\leq i \leq 10} \Vert_{\mathcal{Z}_\CC}.
\end{aligned} \label{EQestimateFinalForOMDCHIHD}
\end{align}

\ni \textbf{Estimates for remaining quantities.} In the following we prove the next bound,
\begin{align} 
\begin{aligned} 
&\Vert \xd \Vert_{\XX(\HH)} + \Vert \dot f \Vert_{\YY_f} + \Vert (\dot{j}^1,\dot{j}^2)\Vert_{\YY_{({j}^1,{j}^2)}}+\Vert \PP_f(\dot f)\Vert_{\XX(S_2)} + \Vert \PP_{({j}^1,{j}^2)}(\dot{j}^1,\dot{j}^2) \Vert_{\XX(S_2)} \\
 \les& \Vert \dot{\mathfrak{X}}_1 \Vert_{\XX(S_1)} +  \Vert \dot{\MMf}_2 \Vert_{\ZZ_\MMf(S_2)}+ \Vert (\mfq_i)_{1\leq i \leq 10} \Vert_{\mathcal{Z}_\CC}.
\end{aligned} \label{EQestimateTOPROVE23}
\end{align}

\ni First, by Lemmas \ref{LEMenergyestimatesPHID}, \ref{LEMenergyestimatesETAD}, \ref{LEMenergyestimatesCHIBHD} and \ref{LEMenergyestimatesDUOMBD}, and \eqref{EQestimateFinalForOMDCHIHD}, we have that
\begin{align*} 
\begin{aligned} 
\Vert \xd \Vert_{\XX(\HH)} \les& \Vert \xdmf_1 \Vert_{\XX(S_1)}+ \Vert \Omd \Vert_{H^6_3(\HH)} + \Vert \chihd \Vert_{H^6_2(\HH)} + \Vert (\mfq_i)_{1\leq i \leq 10} \Vert_{\mathcal{Z}_\CC}.
\end{aligned} 
\end{align*}

\ni Second, we have by \eqref{EQQmatchingEstimatefullRic11} in Section \ref{SECmatchingCharges11} that
\begin{align*} 
\begin{aligned} 
& \Vert \dot f \Vert_{\YY_f} + \Vert (\dot{j}^1,\dot{j}^2)\Vert_{\YY_{({j}^1,{j}^2)}}+\Vert \PP_f(\dot f)\Vert_{\XX(S_2)} + \Vert \PP_{({j}^1,{j}^2)}(\dot{j}^1,\dot{j}^2) \Vert_{\XX(S_2)} \\
  \les& \Vert \dot{\mathfrak{X}}_1 \Vert_{\XX(S_1)} +  \Vert \dot{\MMf}_2 \Vert_{\ZZ_\MMf(S_2)} + \Vert (\mfq_i)_{1\leq i \leq 10} \Vert_{\mathcal{Z}_\CC}.
\end{aligned} 
\end{align*}
 
\ni This finishes the proof of \eqref{EQestimateTOPROVE23} and hence of Theorem \ref{PROPlingluing22}.

\section{Proof of main theorem} \label{SECLinearizedNullConstraintsAroundMinkowski} 
In this section we prove Theorem \ref{PROPNLgluingOrthA121}. We proceed as follows.
\begin{itemize}
\item In Section \ref{SECimplicitA12setup} we set up the framework for the implicit function theorem.

\item In Section \ref{SECconclusion} we use the implicit function theorem and the solution to the linearized characteristic gluing problem at Minkowski (see Section \ref{SEClinearizedProblem}) to construct the solution $x$ to the null structure equations on $\HH_{0,[1,2]}$. 

\item In Sections \ref{SECconclusion2} and \ref{SECconclusion27778} we prove the additional charge estimates \eqref{EQChargeEstimatesMainTheorem0} and \eqref{EQChargeEstimatesMainTheorem}, respectively. These estimates for $(\mathbf{E},\mathbf{P},\mathbf{L},\mathbf{G})$ are based on the construction in Section \ref{SECconclusion} as well as the analysis of the linearizations of the sphere perturbations, angular perturbations and null transport equations for charges \emph{at Schwarzschild of mass $M\geq0$} provided in Appendix \ref{SEClinearizedCHARGEequationsSSAPP}.

\item In Section \ref{SECconclusion4} we give an outline of the proof of Theorem \ref{THMHIGHERorderLgluingMAINTHEOREM}, that is, the characteristic gluing of higher-order $L$-derivatives.
\end{itemize}

\subsection{Setup of framework for the proof} \label{SECimplicitA12setup} The proof of Theorem \ref{PROPNLgluingOrthA121} is based on the application of the implicit function theorem to the following mapping $\FF$.

\begin{definition}[Definition of $\FF$] \label{DEFmappingFF0} Let 
\begin{itemize}
\item $x_{0,1} \in \mathcal{X}(S_{0,1})$ be sphere data, 
\item $\tilde{\underline{x}} \in \XX^+(\tilde{\HHb}_{[-\de,\de],2})$ be ingoing null data on $\tilde{\HHb}_{[-\de,\de],2}$,
\item $x \in \XX(\HH_{0,[1,2]})$ be null data on $\HH_{0,[1,2]}$, 
\item $f \in \YY_f$ and $(j^1,j^2) \in \YY_{(j^1,j^2)}$ be perturbation functions,
\end{itemize}
where the spaces $\mathcal{X}(S_{0,1})$, $\XX^+(\tilde{\HHb}_{[-\de,\de],2})$, $\XX(\HH_{0,[1,2]})$, $\YY_f$ and $\YY_f$ are introduced in Definitions \ref{DEFnormFirstOrderDATA}, \ref{DEFspacetimeNORM}, \ref{DEFnormHH}, \ref{DEFnormfperturbations}, and \ref{DEFnormj1j2}, respectively. Define the mapping $\FF$ by
\begin{align} 
\begin{aligned} 
\FF(x_{0,1},\tilde{\underline{x}},x,f,(j^1,j^2)) := \lrpar{x\vert_{S_{0,1}}-x_{0,1},\mathfrak{M}\lrpar{x\vert_{S_{0,2}}} -\mathfrak{M}\lrpar{\PP_{(j^1,j^2)} {\PP_f\lrpar{\tilde{\underline{x}}}}}, \lrpar{\CC_i(x)}_{1\leq i \leq 10}},
\end{aligned} \label{EQdefinitiontildeCC}
\end{align}
where 
\begin{itemize}
\item $\mathfrak{M}$ is the matching map of Definition \ref{DEFmatchingMAP},
\item $\lrpar{\CC_i}_{1\leq i \leq 10}$ are the constraint functions defined in Section \ref{SECconformalMethod1112},
\item $\PP_{(j^1,j^2)}$ and $\PP_{f}$ are the perturbations of sphere data defined in Section \ref{SECdefEquivalenceFirstOrderSphereData}.
\end{itemize}
\end{definition}

\ni From Definition \ref{DEFmappingFF0}, Lemma \ref{LEMstandardEstimates} and Propositions \ref{PropositionSmoothnessF} and \ref{PropositionSmoothnessF2}, we make the following observations concerning $\FF$.
\begin{enumerate}
\item For each real number $M\geq0$, 
\begin{align} 
\begin{aligned} 
\FF(\mathfrak{m}^M,\underline{\mathfrak{m}}^M,\mathfrak{m}^M,0,(0,0)) = \lrpar{0,0,0},
\end{aligned} \label{EQmappingPropertiesFF1}
\end{align}
where $\mathfrak{m}^M$ denotes the Schwarzschild sphere data.

\item The mapping $\FF$ is well-defined and smooth as mapping between the spaces
\begin{align*}
\begin{aligned}
\FF: \mathcal{X}(S_{0,1}) \times \mathcal{X}^+(\tilde{\HHb}_{[-\de,\de],2}) \times \XX(\HH_{0,[1,2]}) \times \mathcal{Y}_f \times \mathcal{Y}_{(j^1,j^2)} \to \mathcal{X}(S_{0,1}) \times \ZZ_\MMf(S_{0,2}) \times \mathcal{Z}_\CC
\end{aligned}
\end{align*}
in an open neighbourhood of 
\begin{align*} 
\begin{aligned} 
(x_{0,1},\tilde{\underline{x}},x,f,(j^1,j^2))=(\mathfrak{m},\underline{\mathfrak{m}},\mathfrak{m},0,(0,0)).
\end{aligned} 
\end{align*}

\item For real numbers $M\geq0$ sufficiently small, the linearization $\dot{\FF}^M$ of $\FF$ in $(x,f,(j^1,j^2))$ at 
\begin{align*} 
\begin{aligned} 
(x_{0,1},\tilde{\underline{x}},x,f,(j^1,j^2))=(\mathfrak{m}^M,\underline{\mathfrak{m}}^M,\mathfrak{m}^M,0,(0,0)),
\end{aligned} 
\end{align*}
is a well-defined, bounded linear operator between the spaces
\begin{align*} 
\begin{aligned} 
\dot{\FF}^M: \XX(\HH_{0,[1,2]}) \times \mathcal{Y}_f \times \mathcal{Y}_{(j^1,j^2)} \to  \mathcal{X}(S_{0,1}) \times \mathcal{Z}_\MMf(S_{0,2}) \times \ZZ_\CC,
\end{aligned} 
\end{align*}
and explicitly given by
\begin{align} 
\begin{aligned} 
\dot{\FF}^M(\dot{x},\dot{f},(\dot{j}^1,\dot{j}^2)) = \lrpar{ \dot{x}\vert_{S_{0,1}},\MMf\lrpar{\dot{x} \vert_{S_{0,2}}-\dot{\PP}^M_{(j^1,j^2)}(\dot{j}^1,\dot{j}^2) -\dot{\PP}^M_f\lrpar{\dot{f}}},\lrpar{\dot{\CC}^M_i}_{1\leq i \leq 10}},
\end{aligned} \label{EQlinearizedDerivationProblem9909}
\end{align}
where the linearized constraint functions $\dot{\CC}^M_i$, $1\leq i \leq 10$, are given in Section \ref{SECconformalMethod1112} and the linearized perturbations $\dot{\PP}^M_f$ and $\dot{\PP}^M_{(j^1,j^2)}$ are given in Section \ref{SECdefEquivalenceFirstOrderSphereData}.
\end{enumerate}

\ni Importantly, the linearization $\dot{\FF}^0$ at Minkowski, given in \eqref{EQlinearizedDerivationProblem9909}, is in accordance with the setup of the linearized characteristic gluing problem at Minkowski in Section \ref{SEClinearizedProblem}, so that \textbf{Theorem \ref{PROPlingluing22} implies that the linearization $\dot{\FF}^0$ is surjective. This constitutes the central ingredient for the proof of Theorem \ref{PROPNLgluingOrthA121}.}

As $\dot{\FF}^0$ is a bounded linear mapping between Hilbert spaces, its kernel $\mathrm{ker}(\dot{\FF}^0)$ is a closed subspace and the following splitting holds,
\begin{align*} 
\begin{aligned} 
\XX(\HH_{0,[1,2]}) \times \mathcal{Y}_f \times \mathcal{Y}_{(j^1,j^2)} = \mathrm{ker}(\dot{\FF}^0) \oplus \lrpar{\mathrm{ker}(\dot{\FF}^0)}^\perp.
\end{aligned} 
\end{align*}
In the following we consider only the restriction $\overline{\FF}$ of $\FF$ to
\begin{align*} 
\begin{aligned} 
\overline{\FF}: \mathcal{X}(S_{0,1}) \times \mathcal{X}^+(\tilde{\HHb}_{[-\de,\de],2}) \times \lrpar{\mathrm{ker}(\dot{\FF}^0)}^\perp \to \mathcal{X}(S_{0,1}) \times \ZZ_\MMf(S_{0,2}) \times \mathcal{Z}_\CC.
\end{aligned} 
\end{align*}
In this setting, the linearization $\dot{\overline{\FF}}^0$ is a \emph{bijection} between Hilbert spaces. 

By the continuity in $M\geq0$ of the family of linearizations
\begin{align*} 
\begin{aligned} 
\dot{\overline{\FF}}^M: \lrpar{\mathrm{ker}(\dot{\FF}^0)}^\perp \to  \mathcal{X}(S_{0,1}) \times \mathcal{Z}_\MMf(S_{0,2}) \times \ZZ_\CC,
\end{aligned} 
\end{align*}
and the classical functional analysis result that \emph{bijectivity is an open property of bounded linear operators}, we have the following corollary.

\begin{corollary}[Bijectivity of $\dot{\overline{\FF}}^M$] \label{CORFFMsurj} For real numbers $M\geq0$ sufficiently small, the linearization $\dot{\overline{\FF}}^M$ is a bijection, and the solution $(\xd, \dot{f}, (\dot{j}^1,\dot{j}^2))$ of 
\begin{align*} 
\begin{aligned} 
\dot{\overline{\FF}}^M(\xd, \dot{f}, (\dot{j}^1,\dot{j}^2)) = \lrpar{\dot{\mathfrak{X}}_{0,1}, \dot{\MMf}_{0,2}, (\mfq_i)_{1\leq i \leq 10}}
\end{aligned} 
\end{align*}
satisfies the following estimate,
\begin{align*} 
\begin{aligned} 
&\Vert \xd \Vert_{\mathcal{X}(\HH_{0,[1,2]})} + \Vert \dot f \Vert_{\YY_f} + \Vert (\dot{j}^1,\dot{j}^2) \Vert_{\YY_{({j}^1,{j}^2)}}+\left\Vert \dot{\PP}^M_{(j^1,j^2)}(\dot{j}_1,\dot{j}_2) \right\Vert_{\mathcal{X}(S_{0,2})}+\left\Vert \dot{\PP}^M_f\lrpar{\dot{f}} \right\Vert_{\mathcal{X}(S_{0,2})}\\
 \les& \Vert \dot{\mathfrak{X}}_{0,1} \Vert_{\mathcal{X}(S_{0,1})}+\Vert \dot{\MMf}_{0,2} \Vert_{\ZZ_\MMf(S_{0,2})}+ \Vert (\mfq_i)_{1\leq i \leq 10} \Vert_{\mathcal{Z}_\CC}.
\end{aligned} 
\end{align*}

\end{corollary}

\subsection{Construction of solution to the null structure equations}  \label{SECconclusion} 
In this section we apply the implicit function theorem to $\overline{\FF}$ to construct solutions $x$ to the null structure equations satisfying matching conditions.

With view on applying the implicit function theorem (see Theorem \ref{thm:InverseMars14}) to $\overline{\FF}$ at Schwarzschild of small mass $M\geq0$, we recall the following properties from Section \ref{SECimplicitA12setup}.

\begin{enumerate}

\item For $M\geq0$ sufficiently small, the mapping
$$\overline{\FF}: \mathcal{X}(S_{0,1}) \times \mathcal{X}^+(\tilde{\HHb}_{[-\de,\de],2}) \times \lrpar{\mathrm{ker}(\dot{\FF}^0)}^\perp \to  \mathcal{X}(S_{0,1}) \times \ZZ_\MMf(S_{0,2}) \times \mathcal{Z}_\CC$$
is a well-defined and smooth mapping between Hilbert spaces in an open neighbourhood of 
\begin{align} 
\begin{aligned} 
(x_{0,1},\tilde{\underline{x}},x,f,(j^1,j^2))=(\mathfrak{m}^M,\underline{\mathfrak{m}}^M,\mathfrak{m}^M,0,(0,0)),
\end{aligned} \label{EQevalSSMappIFT1}
\end{align}
where the size of the neighbourhood is independent of $M$.

\item By \eqref{EQmappingPropertiesFF1} it holds that
$$\overline{\FF}(\mathfrak{m}^M,\underline{\mathfrak{m}}^M,\mathfrak{m}^M,0,(0,0)) = \lrpar{0,0,0}.$$

\item For $M\geq0$ sufficiently small, the linearization $\dot{\overline{\FF}}^M$ of $\overline{\FF}$ in $(x,f,(j^1,j^2))$ evaluated at \eqref{EQevalSSMappIFT1} is a bijection. 

\end{enumerate}

\ni By the above, for $M\geq0$ sufficiently small, we can apply the implicit function theorem to $\overline{\FF}$ at \eqref{EQevalSSMappIFT1}. We conclude that there is a universal radius $r_0>0$ and a smooth mapping
\begin{align*} 
\begin{aligned} 
\GG^M: B\lrpar{(\mathfrak{m}^M,\underline{\mathfrak{m}}^M),r_0} \to \lrpar{\mathrm{ker}(\dot{\FF}^0)}^\perp \subset \XX(\HH_{0,[1,2]}) \times \mathcal{Y}_f \times \mathcal{Y}_{(j^1,j^2)},
\end{aligned} 
\end{align*}
where $B\lrpar{(\mathfrak{m}^M,\underline{\mathfrak{m}}^M),r_0}$ denotes the open ball of radius $r_0>0$ centered at $(\mathfrak{m}^M,\underline{\mathfrak{m}}^M)$,
\begin{align*} 
\begin{aligned} 
B\lrpar{(\mathfrak{m}^M,\underline{\mathfrak{m}}^M),r_0}\subset \mathcal{X}(S_{0,1}) \times \mathcal{X}^+(\HHb_{[-\de,\de],2}),
\end{aligned} 
\end{align*}
such that for all $(x_{0,1},\tilde{\underline{x}}) \in B\lrpar{(\mathfrak{m}^M,\underline{\mathfrak{m}}^M),r_0}$, 
\begin{align} 
\begin{aligned} 
\overline{\FF}\lrpar{x_{0,1},\tilde{\underline{x}}, \GG^M(x_{0,1},\tilde{\underline{x}})} = (0,0,0).
\end{aligned} \label{EQinverseLocal12}
\end{align}
Defining for given $(x_{0,1},\tilde{\underline{x}}) \in B\lrpar{(\mathfrak{m}^M,\underline{\mathfrak{m}}^M),r_0}$,
\begin{align*} 
\begin{aligned} 
(x,f,(j_1,j_2)):= \GG^M(x_{0,1},\tilde{\underline{x}}),
\end{aligned} 
\end{align*} 
we have by \eqref{EQinverseLocal12} and the definition of $\overline{\FF}$ as restriction of the mapping $\FF$ introduced in \eqref{EQdefinitiontildeCC} that 
\begin{align} 
\begin{aligned} 
&\CC_i(x) = 0 \,\, \text{ on } \HH_{0,[1,2]}  \text{ for } 1\leq i \leq 10,\\
&x \vert_{S_{0,1}} = {x_{0,1}}, \,\, \MMf(x \vert_{S_{0,2}}) = \MMf\lrpar{\PP_{(j^1,j^2)} \PP_f(\tilde{\underline{x}})}.
\end{aligned} \label{EQmatchingProofMainTheoremFinalStatement}
\end{align}

\ni This proves \eqref{EQmaintheoremMATCHINGcond4}. The matching \eqref{EQspheredatamatchingMAINTHEOREMS02} under the charge matching condition \eqref{EQmainTheoremmatchingcondition122} follows directly from Lemma \ref{LEMconditionalMATCHING}.

We turn to the proof of \eqref{EQboundsFullRESULTMAINTHM1}. Applying Lemma \ref{LEMoperatorEstimates} to the smooth map $\GG^M$, we get that
\begin{align} 
\begin{aligned} 
&\Vert x - \mathfrak{m}^M \Vert_{\XX(\HH_{0,[1,2]})} + \Vert f \Vert_{\YY_f} + \Vert (j^1,j^2) \Vert_{\mathcal{Y}_{(j^1,j^2)}} \\
\les& \Vert x_{0,1}-\mathfrak{m}^M_{0,1} \Vert_{\XX(S_{0,1})} + \Vert \tilde{\underline{x}}_{[-\de,\de],2}-\underline{\mathfrak{m}}^M \Vert_{\tilde{\mathcal{X}}^+(\HHb_{[-\de,\de],2}) }.
\end{aligned} \label{EQfinalTHMfirstFinalEstimate144}
\end{align} 
By \eqref{EQfinalTHMfirstFinalEstimate144} and Proposition \ref{PropositionSmoothnessF} and \ref{PropositionSmoothnessF2}, it further follows that for $\varep>0$ sufficiently small,
\begin{align*} 
\begin{aligned} 
\Vert \PP_{(j^1,j^2)} \PP_f(\tilde{\underline{x}}) -\tilde{\underline{x}}_{0,2} \Vert_{\XX(S_{0,2})} \les& \Vert \PP_{(j^1,j^2)} \PP_f(\tilde{\underline{x}}) -  \PP_f(\tilde{\underline{x}}) \Vert_{\XX(S_{0,2})} + \Vert  \PP_f(\tilde{\underline{x}}) -\tilde{\underline{x}}_{0,2} \Vert_{\XX(S_{0,2})}\\
\les& \Vert (j^1,j^2) \Vert_{\YY_{(j^1,j^2)}} + \Vert  \PP_f(\tilde{\underline{x}}) -\underline{\mathfrak{m}}^M \Vert_{\XX(S_{0,2})} \\
&+ \Vert f \Vert_{\YY_f} + \Vert \tilde{\underline{x}}_{[-\de,\de],2}- \underline{\mathfrak{m}}^M \Vert_{\mathcal{X}^+(\tilde{\HHb}_{[-\de,\de],2}) } \\
\les& \Vert (j^1,j^2) \Vert_{\YY_{(j^1,j^2)}} + \Vert f \Vert_{\YY_f} + \Vert \tilde{\underline{x}}_{[-\de,\de],2}- \underline{\mathfrak{m}}^M \Vert_{\mathcal{X}^+(\tilde{\HHb}_{[-\de,\de],2}) }\\
\les & \Vert x_{0,1}-\mathfrak{m}^M \Vert_{\XX(S_{0,1})} + \Vert \tilde{\underline{x}}_{[-\de,\de],2}-\underline{\mathfrak{m}}^M \Vert_{\mathcal{X}^+(\tilde{\HHb}_{[-\de,\de],2}) }.
\end{aligned} 
\end{align*}
We underline that the radius $r_0>0$ and the constants in the above estimates are universal for small $M\geq0$. This follows from the smoothness of $\overline{\FF}$ and the continuity in $M\geq0$ of the Schwarzschild data \eqref{EQevalSSMappIFT1}; see also, for example, Proposition 2.5.6 in \cite{MarsdenImplicit}. This finishes the proof of the estimates \eqref{EQboundsFullRESULTMAINTHM1}. 

\subsection{Proof of the charge perturbation estimate \eqref{EQChargeEstimatesMainTheorem0}}  \label{SECconclusion2} 

In this section we prove \eqref{EQChargeEstimatesMainTheorem0}, that is, for $M\geq0$ and $\varep>0$ sufficiently small, the following estimate holds, 
\begin{align} 
\begin{aligned} 
\left\vert \lrpar{\mathbf{E},\mathbf{P},\mathbf{L},\mathbf{G}}\lrpar{x_{0,2}} - \lrpar{\mathbf{E},\mathbf{P},\mathbf{L},\mathbf{G}}\lrpar{\tilde{\underline{x}}_{0,2}} \right\vert \les \varep M+ \varep^2,
\end{aligned} \label{EQ3proofMAINTHM11}
\end{align}
where $x_{0,2}:= \PP_{(j^1,j^2)} \PP_f(\tilde{\underline{x}})$.

Consider first \eqref{EQ3proofMAINTHM11} for the charge $\mathbf{E}$. For this section, we introduce the map
\begin{align} 
\begin{aligned} 
\mathbf{E}: \XX^+(\tilde{\HHb}_{[-\de,\de],2}) \times \YY_f \times \YY_{(j^1,j^2)} \to& H^4(S_{0,2}),\\
(\tilde{\underline{x}}, f, (j^1,j^2)) \mapsto& \mathbf{E}\lrpar{\tilde{\underline{x}}, f, (j^1,j^2)}:= \mathbf{E}\lrpar{\PP_{(j^1,j^2)} \PP_f (\tilde{\underline{x}})}.
\end{aligned} \label{EQmappingEperturbationestimate11}
\end{align}
From the smoothness of the perturbations $\PP_{(j^1,j^2)}$ and $\PP_f$, and definition of $\mathbf{E}$ in Definition \ref{DEFnonlinearcharges6}, it follows that $\mathbf{E}$ in \eqref{EQmappingEperturbationestimate11} is a smooth map in an open neighbourhood of
\begin{align*} 
\begin{aligned} 
\lrpar{\tilde{\underline{x}}, f, (j^1,j^2)} = \lrpar{\underline{\mathfrak{m}}, 0, (0,0)}.
\end{aligned} 
\end{align*}
By the fundamental theorem of calculus and \eqref{EQmappingEperturbationestimate11},
\begin{align} 
\begin{aligned} 
\mathbf{E}\lrpar{x_{0,2}} - \mathbf{E}\lrpar{\tilde{\underline{x}}_{0,2}} =&\mathbf{E}\lrpar{\tilde{\underline{x}}, f, (j^1,j^2)} - \mathbf{E}\lrpar{\tilde{\underline{x}}, 0, (0,0)} \\
=& \int\limits_{0}^1 \dot{\mathbf{E}}  \vert_{\lrpar{\tilde{\underline{x}}, f \cdot s, (j^1,j^2) \cdot s} } \lrpar{f, (j^1,j^2)}  ds.
\end{aligned} \label{EQfdtlEestimatePERT}
\end{align}
where $\dot{\mathbf{E}}$ denotes the linearization of \eqref{EQmappingEperturbationestimate11} in $(f,(j^1,j^2))$.

We estimate the integrand on the right-hand side of \eqref{EQfdtlEestimatePERT} as follows. For $0\leq s \leq1$, we have that
\begin{align} 
\begin{aligned} 
\dot{\mathbf{E}}  \vert_{\lrpar{\tilde{\underline{x}}, f \cdot s, (j^1,j^2) \cdot s} } = \lrpar{\dot{\mathbf{E}}  \vert_{\lrpar{\tilde{\underline{x}}, f \cdot s, (j^1,j^2) \cdot s} }- \dot{\mathbf{E}} \vert_{\lrpar{\underline{\mathfrak{m}}^M, 0, (0,0)}}} + \dot{\mathbf{E}} \vert_{\lrpar{\underline{\mathfrak{m}}^M, 0, (0,0)}}.
\end{aligned} \label{EQdotestimateE7788}
\end{align}

\ni On the one hand, by the smoothness of the mapping $\mathbf{E}$ defined in \eqref{EQmappingEperturbationestimate11}, it holds for all $(\dot{f}, (\dot{j}^1,\dot{j}^2))$ that for $\varep>0$ sufficiently small,
\begin{align} 
\begin{aligned} 
\left\vert \lrpar{\dot{\mathbf{E}}  \vert_{\lrpar{\tilde{\underline{x}}, f \cdot s, (j^1,j^2) \cdot s} }- \dot{\mathbf{E}} \vert_{\lrpar{\underline{\mathfrak{m}}^M, 0, (0,0)}}}(\dot{f}, (\dot{j}^1,\dot{j}^2))\right\vert
\les \varep \cdot \lrpar{ \Vert \dot{f} \Vert_{\YY_f} + \Vert (\dot{j}^1,\dot{j}^2) \Vert_{\YY_{({j}^1,{j}^2)}}}.
\end{aligned} \label{EQEperturbationEstimate1}
\end{align}

\ni On the other hand, using definition of $\mathbf{E}$ in Definition \ref{DEFnonlinearcharges6}, the properties that 
\begin{align*} 
\begin{aligned} 
\rh(\mathfrak{m}^M) = - \frac{2M}{r_M^3}, \,\, \be(\mathfrak{m}^M) =0,
\end{aligned} 
\end{align*}
and that, linearizing at Schwarzschild of mass $M\geq0$, see Lemmas \ref{LEMlinearizedTransversalSCHWARZSCHILD} and \ref{LEMspherediffLINSCHWARZSCHILD},
\begin{align} 
\begin{aligned} 
\dot{\rho} \vert_{(\underline{\mathfrak{m}}^M,0,(0,0))} \lrpar{\dot f,(\dot{j}^1, \dot{j}^2)}=-\frac{6M\Om_M^2}{r_M^4} \dot{f}, \,\, \dot{\beta} \vert_{(\underline{\mathfrak{m}}^M,0,(0,0))}\lrpar{\dot f,(\dot{j}^1, \dot{j}^2)}= -\frac{6M\Om_M}{r_M^3} \di \dot{f},
\end{aligned} \label{EQlinearizationAPPLSCHWARZSCHILD}
\end{align}
it is straight-forward to show that for all $\dot f$ and $(\dot{j}^1, \dot{j}^2)$,
\begin{align} 
\begin{aligned} 
\left\vert \dot{\mathbf{E}} \vert_{\lrpar{\underline{\mathfrak{m}}^M, 0, (0,0)}}\lrpar{\dot f,(\dot{j}^1, \dot{j}^2)} \right\vert \les M \cdot \lrpar{ \Vert \dot{f} \Vert_{\YY_f} + \Vert (\dot{j}^1,\dot{j}^2) \Vert_{\YY_{({j}^1,{j}^2)}}}.
\end{aligned} \label{EQEperturbationEstimate2}
\end{align}

\ni Plugging \eqref{EQEperturbationEstimate1} and \eqref{EQEperturbationEstimate2} into \eqref{EQdotestimateE7788}, we get that for all $\dot f$ and $(\dot{j}^1, \dot{j}^2)$, and $0\leq s \leq 1$,
\begin{align*} 
\begin{aligned} 
\left\vert \dot{\mathbf{E}}  \vert_{\lrpar{\tilde{\underline{x}}, f \cdot s, (j^1,j^2) \cdot s} }\lrpar{\dot f,(\dot{j}^1, \dot{j}^2)} \right\vert \les \lrpar{M +\varep} \cdot \lrpar{ \Vert \dot{f} \Vert_{\YY_f} + \Vert (\dot{j}^1,\dot{j}^2) \Vert_{\YY_{({j}^1,{j}^2)}}},
\end{aligned} 
\end{align*}
and subsequently, by \eqref{EQfdtlEestimatePERT},
\begin{align*} 
\begin{aligned} 
\left\vert \mathbf{E}\lrpar{x_{0,2}} - \mathbf{E}\lrpar{\tilde{\underline{x}}_{0,2}} \right\vert \les (M+ \varep)\cdot \varep = M\varep+ \varep^2.
\end{aligned} 
\end{align*}
This finishes the proof of \eqref{EQ3proofMAINTHM11} for $\mathbf{E}$. The proofs for $\mathbf{P}, \mathbf{L}$ and $\mathbf{G}$ are similar. Indeed, the crucial estimate \eqref{EQEperturbationEstimate2} similarly holds for $\mathbf{P}, \mathbf{L}$ and $\mathbf{G}$, so that the same argument as above applies. This finishes the proof of \eqref{EQ3proofMAINTHM11}.

\subsection{Proof of the charge transport estimate \eqref{EQChargeEstimatesMainTheorem}}  \label{SECconclusion27778} 

In this section we prove \eqref{EQChargeEstimatesMainTheorem}, that is, for $M\geq0$ and $\varep>0$ sufficiently small, the following estimate holds, 
\begin{align} 
\left\vert \lrpar{\mathbf{E},\mathbf{P},\mathbf{L},\mathbf{G}}\lrpar{x \vert_{S_{0,2}}} - \lrpar{\mathbf{E},\mathbf{P},\mathbf{L},\mathbf{G}}\lrpar{x \vert_{S_{0,1}}} \right\vert \les \varep M+ \varep^2. 
\label{EQ3proofMAINTHM12}
\end{align}

\ni First we prove the component $\mathbf{E}$ of \eqref{EQ3proofMAINTHM12},
\begin{align} 
\begin{aligned} 
\left\vert \mathbf{E}\lrpar{x \vert_{S_{0,2}}} - \mathbf{E}\lrpar{x \vert_{S_{0,1}}} \right\vert \les \varep M+\varep^2.
\end{aligned} \label{EQspecializedEestimtransport1}
\end{align}
Indeed, let $x$ be the constructed solution to the null structure equations on $\HH_{0,[1,2]}$. By the fundamental theorem of calculus,
\begin{align} 
\begin{aligned} 
\mathbf{E}\lrpar{x }\vert_{S_{0,2}} -\mathbf{E}\lrpar{x }\vert_{S_{0,1}} = \int\limits_{1}^2 D\mathbf{E}(x) \vert_{S_{0,v}} dv.
\end{aligned} \label{EQintegralEestim}
\end{align}

\ni In the following we analyze $D\mathbf{E}(x)\vert_{S_{0,v}}$ for $1\leq v \leq 2$. Using that for Schwarzschild reference data $\mathfrak{m}^M$ it holds that
\begin{align*} 
\begin{aligned} 
\mathbf{E}(\mathfrak{m}^M)=M, \,\, D\mathbf{E}(\mathfrak{m}^M)=0 \,\, \text{ on } \HH_{0,[1,2]},
\end{aligned} 
\end{align*}
we can express the integrand $D\mathbf{E}(x)$ on the right-hand side of \eqref{EQintegralEestim} by 
\begin{align} 
\begin{aligned} 
D\mathbf{E}(x) =& D\mathbf{E}(x) - D\mathbf{E}(\mathfrak{m}^M).
\end{aligned} \label{EQdifference27778}
\end{align}

\ni We make two observations. First, the map
\begin{align*} 
\begin{aligned} 
D\mathbf{E}: \XX(\HH_{0,[1,2]}) \to& H^4_1(\HH_{0,[1,2]}),\\
x \mapsto& D\mathbf{E}\lrpar{x},
\end{aligned}
\end{align*}
is smooth in an open neighbourhood of $x= \mathfrak{m}$. Second, at the end of this section (see Lemma \ref{CLAIMsmoothfamily}) we show that by the implicit function theorem construction of our solution $x$, there is a smooth family
\begin{align} 
\begin{aligned} 
(x_s)_{0\leq s \leq1} \subset \XX(\HH_{0,[1,2]})
\end{aligned} \label{EQdefFAMILYsolutions1117778}
\end{align}
of solutions to the null structure equations on $\HH_{0,[1,2]}$ such that
\begin{align} 
\begin{aligned} 
x_{s=0} = \mathfrak{m}^M \text{ and } x_{s=1} =x \text{ on } \HH_{0,[1,2]},
\end{aligned} \label{EQmatchingSMOOTHfamilySolutions}
\end{align}
and satisfying the estimate
\begin{align} 
\begin{aligned} 
\sup\limits_{0\leq s \leq 1} \Vert x_s- \mathfrak{m}^M \Vert_{\XX(\HH)} \les \varep, \,\, \sup\limits_{0\leq s \leq 1} \left\Vert \xd_{s} \right\Vert_{\XX(\HH_{0,[1,2]})} \les \varep,
\end{aligned} \label{EQestimatefamilyofsolutions}
\end{align}
where $\xd_s$ denotes the variation through the family \eqref{EQdefFAMILYsolutions1117778}. Hence we can rewrite \eqref{EQdifference27778} by the fundamental theorem of calculus as
\begin{align} 
\begin{aligned} 
D\mathbf{E}(x) =& D\mathbf{E}(x) - D\mathbf{E}(\mathfrak{m}^M) = \int\limits_{0}^1 \frac{d}{ds}\lrpar{D\mathbf{E}(x_s)} ds =  \int\limits_{0}^1D\dot{\mathbf{E}} \vert_{x_s}(\xd_s) ds,
\end{aligned} \label{EQDEfdtltheorem778}
\end{align}
where $D\dot{\mathbf{E}} \vert_{x_s}$ denotes the linearisation of $D\mathbf{E}(x)$ in $x$ evaluated at $x_s$. 

By construction, see \eqref{EQmatchingSMOOTHfamilySolutions}, $\xd_{s=0}$ is a solution to the \emph{homogeneous linearized null structure equations at Schwarzschild}, that is, for $\xd= \xd_{s=0}$,
\begin{align} 
\begin{aligned} 
\dot{\CC}^M(\xd) =0.
\end{aligned} \label{EQhomlinNSE77788}
\end{align}
In Appendix \ref{SEClinearizedCHARGEequationsSSAPP}, see \eqref{EQMsmalllinearizedEPLGSS99} it is shown that for all solutions $\xd$ to \eqref{EQhomlinNSE77788},
\begin{align} 
\begin{aligned} 
\Vert D\dot{\mathbf{E}} \vert_{\mathfrak{m}^M}(\xd) \Vert_{H^4_1(\HH_{0,[1,2]})} \les M \Vert \xd \Vert_{\XX(\HH_{0,[1,2]})}.
\end{aligned} \label{EQlinearizedESSestimate292}
\end{align}
Using the estimate \eqref{EQlinearizedESSestimate292}, we can estimate the integrand $D\dot{\mathbf{E}} \vert_{x_s}(\xd_s)$ on the right-hand side of \eqref{EQDEfdtltheorem778} as follows. We write
\begin{align*} 
\begin{aligned} 
D\dot{\mathbf{E}} \vert_{x_s}(\xd_s) = \underbrace{\lrpar{D\dot{\mathbf{E}} \vert_{x_s} - D\dot{\mathbf{E}} \vert_{\mathfrak{m}^M}}(\xd_s)}_{:=\II_1} +\underbrace{D\dot{\mathbf{E}} \vert_{\mathfrak{m}^M}(\xd_s)}_{:=\II_2}.
\end{aligned} 
\end{align*}
The term $\II_1$ is estimated by the smoothness of $D\mathbf{E}$ and \eqref{EQestimatefamilyofsolutions} as follows,
\begin{align*} 
\begin{aligned} 
\vert \II_1 \vert \les& \Vert x_s - \mathfrak{m}^M \Vert_{\XX(\HH_{0,[1,2]})} \cdot \Vert \xd_s \Vert_{\XX(\HH_{0,[1,2]})} \\
\les& \varep^2.
\end{aligned} 
\end{align*}
The term $\II_2$ can be analyzed by decomposing 
\begin{align} 
\begin{aligned} 
\xd_s = (\xd_s)^{\perp} + \lrpar{\xd_s}^\top,
\end{aligned} \label{EQdefinitionDecompositionxds}
\end{align}
where $(\xd_s)^{\perp} \in \lrpar{\mathrm{ker} \, \dot{\CC}^M }^\perp$ is defined as solution to
\begin{align} 
\begin{aligned} 
\dot{\CC}^M((\xd_s)^{\perp}) = \dot{\CC}^M(\xd_s), \,\, \Omd((\xd_s)^{\perp}) =0, \,\, \chihd((\xd_s)^{\perp})&=0 \text{ on } \HH_{0,[1,2]}, \\
(\xd_s)^{\perp} \vert_{S_{0,1}} &= 0.
\end{aligned} \label{EQsystemforperpxds}
\end{align}
and $\lrpar{\xd_s}^\top \in \mathrm{ker} \, \dot{\CC}^M$ is defined as
\begin{align*} 
\begin{aligned} 
\lrpar{\xd_s}^\top := \xd_s - (\xd_s)^{\perp}.
\end{aligned} 
\end{align*}
By bounds for the system \eqref{EQsystemforperpxds} analogous to Theorem \ref{PROPlingluing22} and Corollary \ref{CORFFMsurj}, together with the estimate
\begin{align*} 
\begin{aligned} 
\Vert \dot{\CC}^M(\xd_s) \Vert_{\ZZ_\CC} =& \Vert (\dot{\CC}^M-\dot{\CC}\vert_{x_s})(\xd_s) \Vert_{\ZZ_\CC} \\
\les& \Vert x_s -\mathfrak{m}^M \Vert_{\XX(\HH_{0,[1,2]})} \cdot \Vert \xd_s \Vert_{\XX(\HH_{0,[1,2]})} \\
\les& \varep^2,
\end{aligned} 
\end{align*}
where we used \eqref{EQestimatefamilyofsolutions} and that $\dot{\CC}\vert_{x_s}(\xd_s)=0$ by definition of $x_s$ in \eqref{EQdefFAMILYsolutions1117778}, it follows that
\begin{align} 
\begin{aligned} 
\Vert (\xd_s)^{\perp} \Vert_{\XX(\HH_{0,[1,2]})} \les \varep^2.
\end{aligned}\label{EQestimforperp}
\end{align}
This further implies that
\begin{align} 
\begin{aligned} 
\Vert (\xd_s)^{\top} \Vert_{\XX(\HH_{0,[1,2]})} =& \Vert \xd_s - (\xd_s)^{\perp} \Vert_{\XX(\HH_{0,[1,2]})} \\
\les& \varep + \varep^2\\
\les& \varep.
\end{aligned} \label{EQestimfortop}
\end{align}
By \eqref{EQlinearizedESSestimate292}, \eqref{EQdefinitionDecompositionxds}, \eqref{EQestimforperp} and \eqref{EQestimfortop}, and using that $D\mathbf{E}\vert_{\mathfrak{m}^M}$ is a bounded operator, we have
\begin{align*} 
\begin{aligned} 
\II_2 =& D\dot{\mathbf{E}} \vert_{\mathfrak{m}^M}((\xd_s)^\top) + D\dot{\mathbf{E}} \vert_{\mathfrak{m}^M}((\xd_s)^\perp) \\
\les& M  \Vert (\xd_s)^\top \Vert_{\XX(\HH_{0,[1,2]})} + \Vert (\xd_s)^\perp \Vert_{\XX(\HH_{0,[1,2]})} \\
\les& M\varep + \varep^2.
\end{aligned} 
\end{align*}

\ni To summarize the above, we conclude that for $0\leq s \leq 1$,
\begin{align*} 
\begin{aligned} 
\Vert D\dot{\mathbf{E}}\vert_{x_s}\lrpar{\xd_s} \Vert_{H^4_1(\HH_{0,[1,2]})} \les M \varep + \varep^2,
\end{aligned} 
\end{align*}
which, plugged into \eqref{EQDEfdtltheorem778}, yields that
\begin{align*} 
\begin{aligned} 
\Vert D\mathbf{E}(x) \Vert_{H^4_1(\HH_{0,[1,2]})} \les M\varep + \varep^2,
\end{aligned} 
\end{align*}
and, consequently, by plugging into \eqref{EQintegralEestim}, proves the charge estimate \eqref{EQspecializedEestimtransport1} for $\mathbf{E}$.

We claim that the charge estimates for $\mathbf{P}, \mathbf{L}$ and $\mathbf{G}$ are proved similarly. Indeed, in Appendix \ref{SEClinearizedCHARGEequationsSSAPP}, see \eqref{EQMsmalllinearizedEPLGSS99}, in addition to \eqref{EQlinearizedESSestimate292} it is shown that for solutions $\xd$ to 
\begin{align*} 
\begin{aligned} 
\dot{\CC}^M(\xd)=0,
\end{aligned} 
\end{align*} 
it holds that
\begin{align} 
\begin{aligned} 
&\Vert D\dot{\mathbf{P}} \vert_{\mathfrak{m}^M}(\xd) \Vert_{H^4_1(\HH_{0,[1,2]})} + \Vert D\dot{\mathbf{L}} \vert_{\mathfrak{m}^M}(\xd) \Vert_{H^5_1(\HH_{0,[1,2]})} + \Vert D\dot{\mathbf{G}} \vert_{\mathfrak{m}^M}(\xd) \Vert_{H^5_1(\HH_{0,[1,2]})} \\
\les& M \Vert \xd \Vert_{\XX(\HH_{0,[1,2]})}.
\end{aligned} \label{EQAppAPPLSCHWARZSCHILD}
\end{align}
Thus, the remaining charge estimates in \eqref{EQ3proofMAINTHM12} are proved by following the same argument as above.

It remains to prove \eqref{EQdefFAMILYsolutions1117778}, that is, the existence of the smooth family
\begin{align*} 
\begin{aligned} 
(x_s)_{0\leq s \leq1} \subset \XX(\HH_{0,[1,2]})
\end{aligned} 
\end{align*}
of solutions to the null structure equations satisfying \eqref{EQmatchingSMOOTHfamilySolutions} and \eqref{EQestimatefamilyofsolutions}, that is,
\begin{align*} 
\begin{aligned} 
x_{s=0} = \mathfrak{m}^M \text{ and } x_{s=1} =x \text{ on } \HH_{0,[1,2]},
\end{aligned} 
\end{align*}
and
\begin{align*} 
\begin{aligned} 
\sup\limits_{0\leq s \leq 1} \Vert x_s- \mathfrak{m}^M \Vert_{\XX(\HH)} \les \varep, \,\, \sup\limits_{0\leq s \leq 1} \left\Vert \xd_{s} \right\Vert_{\XX(\HH_{0,[1,2]})} \les \varep.
\end{aligned} 
\end{align*}
Indeed, using the smooth map $\GG^M$ constructed in Section \ref{SECconclusion}, this follows directly from the following lemma. We remark here that the linearization $\dot{\GG}^M$ is by construction bijective and uniformly bounded for $M\geq0$ sufficiently small, see Section \ref{SECconclusion}.

\begin{lemma}[Existence of smooth family of data] \label{CLAIMsmoothfamily} Let $M\geq0$ and $\varep>0$ be sufficiently small. There are smooth families of
\begin{itemize}
\item sphere data $(x_{0,1})_s \in \XX(S_{0,1})$ for $0\leq s \leq 1$,
\item ingoing null data $\tilde{\underline{x}}_s \in \XX^+(\tilde{\HHb}_{[-\de,\de],2})$ for $0\leq s \leq 1$, solving the null structure equations on $\tilde{\HHb}_{[-\de,\de],2}$, 
\end{itemize}
such that
\begin{align*} 
\begin{aligned} 
\lrpar{(x_{0,1})_s , \underline{x}_{s}} \vert_{s=0} = \lrpar{\mathfrak{m}^M, \underline{\mathfrak{m}}^M}, \,\, \lrpar{(x_{0,1})_{s}, \tilde{\underline{x}}_{s}} \vert_{s=1} = \lrpar{x_{0,1},\tilde{\underline{x}}},
\end{aligned} 
\end{align*}
and satisfying
\begin{align} 
\begin{aligned} 
\Vert (x_{0,1})_s -\mathfrak{m}^M \Vert_{\XX(S_{0,1})} + \Vert \dot{(x_{0,1})}_s \Vert_{\XX(S_{0,1})} \les& \varep, \\
\Vert \underline{x}_s -\mathfrak{m}^M \Vert_{\XX^+(\tilde{\HHb}_{[-\de,\de],2})} + \Vert \dot{\underline{x}}_s \Vert_{\XX^+(\tilde{\HHb}_{[-\de,\de],2})} \les& \varep,
\end{aligned} \label{EQsmoothboundarydataestimates}
\end{align}
where $\dot{(x_{0,1})}_s$ and $\dot{\underline{x}}_s$ denote the variations of $(x_{0,1})_s$ and $\underline{x}_s$, respectively.
\end{lemma}
First, the smooth family of sphere data $(x_{0,1})_s$ can be defined by 
\begin{align} 
\begin{aligned} 
(x_{0,1})_s := (x_{0,1}-\mathfrak{m}^M) \cdot s + \mathfrak{m}^M.
\end{aligned} \label{EQsphereDatasmoothfamily}
\end{align}
\begin{remark} \label{REMspheredatafamily} In \eqref{EQsphereDatasmoothfamily} we abuse notation, as by Definition \ref{DEFspheredata2} the tensors $\chih$, $\chibh$, $\a$ and $\ab$ are required to be symmetric $\gd$-tracefree $2$-tensors which is a constraint and not compatible with the linear operation depicted in \eqref{EQsphereDatasmoothfamily}. In \eqref{EQsphereDatasmoothfamily} we interpret the prescription of $\chih$, $\chibh$, $\a$ and $\ab$ in the sense that two tensor components are freely prescribable (these are added on the right-hand side of \eqref{EQsphereDatasmoothfamily}), and the other two tensor components are fully determined by the condition to be symmetric and tracefree with respect to $\gd$. Same goes for the prescription of the symmetric tensor $\gd$. In this sense, the prescription of sphere data is without constraint, and \eqref{EQsphereDatasmoothfamily} is well-defined. \end{remark}

\ni Second, the family of ingoing null data $\tilde{\underline{x}}_s$ is given by constructing solutions to the null structure equations (as proved in Section \ref{SECconclusion} but at the higher level of regularity $\XX^+$) on $\tilde{\HHb}_{[-\de,\de],2}$ from sphere data $(\tilde{x}_{0,2})_s$ on $\tilde{S}_{0,2}$ given for $0\leq s \leq 1$ by (see Remark \ref{REMspheredatafamily})
\begin{align} 
\begin{aligned} 
(\tilde{x}_{0,2})_s := (\tilde{x}\vert_{S_{0,2}}- \mathfrak{m}^M) \cdot s + \mathfrak{m}^M,
\end{aligned} \label{EQfamilyprescription1}
\end{align}
and free data $\Om_s$ and $\mathrm{conf}(\gd)_s$ on $\tilde{\HHb}_{[-\de,\de],2}$ given for $0\leq s \leq 1$ by
\begin{align} 
\begin{aligned} 
\Om_s := (\Om - \Om_M)\cdot s + \Om_M, \,\, \mathrm{conf}(\gd)_s:= (\mathrm{conf}(\gd)-\mathrm{conf}(\mathfrak{m}^M)) \cdot s + \mathrm{conf}(\mathfrak{m}^M).
\end{aligned} \label{EQfamilyprescription2}
\end{align}
The estimate \eqref{EQsmoothboundarydataestimates} follows by the general estimates proved for the construction of solutions to the null structure equations (see Section \ref{SECconclusion}) and the explicit prescriptions \eqref{EQfamilyprescription1} and \eqref{EQfamilyprescription2}. This finishes the proof of Lemma \ref{CLAIMsmoothfamily} and hence of the charge estimate \eqref{EQChargeEstimatesMainTheorem}.

\subsection{Outline of the proof of Theorem \ref{THMHIGHERorderLgluingMAINTHEOREM}}  \label{SECconclusion4} In this section we indicate how the proof of Theorem \ref{THMHIGHERorderLgluingMAINTHEOREM}, that is, the characteristing gluing with higher-order $L$-derivatives, is based on the proof of Theorem \ref{PROPNLgluingOrthA121}.
\begin{enumerate}
\item The proof of Theorem \ref{THMHIGHERorderLgluingMAINTHEOREM} is, similarly as in Sections \ref{SECimplicitA12setup} and \ref{SECconclusion}, reduced to studying the \emph{linearized} characteristic gluing problem with higher-order $L$-derivatives by the implicit function theorem.

\item At the linear level at Minkowski, the gauge perturbations $\PP_{{f}}$ and $\PP_{({j}^1,{j}^2)}$ leave $\DD^{L,m}$ invariant; see Remark \ref{REMhigherOrderSpherePerturbations}. Hence the linearized gluing of $\DD^{L,m}$ solely depends on the prescription of the free data $\Omd$ and $\chihd$ on $\HH_{0,[1,2]}$.

\item By the linearized null structure equations in Section \ref{SEClinearizedProblem}, $\dot{\DD}^{L,m}$ can directly be calculated from $\Omd$ and $\chihd$ on $\HH_{0,[1,2]}$. In particular, there is no obstruction to matching $\dot{\DD}^{L,m}$ on $S_{0,1}$ and $S_{0,2}$ by adjusting $\Omd$ and $\chihd$ on $\HH_{0,[1,2]}$, see Remark \ref{REMhigherLgluing2}. This shows that the linearized characteristic gluing problem for higher-order $L$-derivatives is solvable.

\end{enumerate}
This finishes our discussion of Theorem \ref{THMHIGHERorderLgluingMAINTHEOREM}.

\section{Bifurcate characteristic gluing} \label{SECHIGHERfull} \ni In this section we prove Theorem \ref{THMtransversalHIGHERv2}, that is, the codimension-$10$ characteristic gluing of higher-order sphere data along two null hypersurfaces bifurcating from an auxiliary sphere. Analogously to the perturbative characteristic gluing of Theorem \ref{PROPNLgluingOrthA121}, the two main ingredients are 
\begin{enumerate}
\item solving the \emph{linearized} characteristic gluing problem, 
\item applying the implicit function theorem.
\end{enumerate}

\ni As the application of the implicit function theorem is similar as in Sections \ref{SECimplicitA12setup} and \ref{SECconclusion}, we focus in this section on the new ideas necessary for (1). We proceed as follows.

\begin{itemize}
\item In Section \ref{SECingoingNULLstructureEQS} we discuss the linearized null structure equations and the conserved charges along $\HHb_{[-1,0],1}$ and $\HH_{-1,[1,2]}$. 
\item In Section \ref{SECprelimTRANSVERSE} we derive relations between charges on $\HHb_{[-1,0],1}$ and $\HH_{-1,[1,2]}$.
\item In Section \ref{SECstatementFULLhigher} we state and prove the linearized characteristic gluing along two transversely-intersecting null hypersurfaces.
\end{itemize}
In Section \ref{SECproofWgluing}, we prove Proposition \ref{PROPchargluingW}.
\subsection{Linearized null structure equations} \label{SECingoingNULLstructureEQS} In this section we analyze the linearized null structure equations at Minkowski along $\HHb_{[-1,0],1}$ and $\HH_{-1,[1,2]}$, and discuss the corresponding conserved charges. 

\begin{remark} In the following we study the \emph{homogeneous} linearized null structure equations, that is, we follow the formalism of \cite{DHR} and linearize through a family of solutions to the Einstein equations. The analysis of the \emph{inhomogeneous} linearized null structure equations along $\HHb_{[-1,0],1}$ with source terms $(\underline{\mathfrak{c}}_i)_{1\leq i \leq 10}$ (which is necessary for the application of the implicit function theorem) is then analogous to Section \ref{SECderivationConstraintFunctions2}.\end{remark}

\ni \textbf{Linearized null structure equations along $\HH_{-1,[1,2]}$.} In Section \ref{SECderivationConstraintFunctions2} we linearized the null structure equations along $\HH_{0,[1,2]}$, see Lemma \ref{LEMlinearizedConstraintsSCHWARZSCHILD}. Setting $M=0$ in Lemma \ref{LEMlinearizedConstraintsSCHWARZSCHILD} yields the linearized equations at Minkowski. This linearization and the resulting linearized equations clearly apply analogously to the null hypersurface $\HH_{-1,[1,2]}$ considered in this section.

We recall that in Section \ref{SECprelimAnalysis} we identified \emph{charges}  
\begin{align} 
\begin{aligned} 
\QQ_i \text{ for } 0\leq i \leq 7,
\end{aligned} \label{EQchargesHHyoufound}
\end{align}
which satisfy conservation laws along $\HH_{-1,[1,2]}$,
\begin{align*} 
\begin{aligned} 
D \QQ_i =0 \text{ for } 0\leq i \leq 7.
\end{aligned} 
\end{align*}
We refer to \eqref{EQdefChargesMinkowski8891} for the precise definitions of these charges. \\

\ni \textbf{Linearized null structure equations in $\Lb$-direction.} Similar to Lemma \ref{LEMlinearizedConstraintsSCHWARZSCHILD}, by linearization at Minkowski of the null structure equations along $\HHb_{[-1,0],1}$ we get the following, see also \cite{DHR}. The linearized first variation equation,
\begin{align} 
\begin{aligned} 
\Du \lrpar{\frac{\phid}{r}} = \frac{\omtrchibd}{2}, \,\, \Du \gdcd = \frac{2}{r^2} \chibhd,
\end{aligned} \label{EQingoingEQ1}
\end{align}
the linearized Raychauduri equation, 
\begin{align} 
\begin{aligned} 
\Du \lrpar{r^2 \lrpar{\omtrchibd + \frac{4}{r} \Omd}} = -4 \Omd, 
\end{aligned} \label{EQingoingEQ2}
\end{align}
as well as
\begin{align} 
\begin{aligned} 
\Du \lrpar{r^2 \etabd} =& \Divdo \chibhd - \frac{r^2}{2} \di \lrpar{\omtrchibd + \frac{4}{r} \Omd} - 2r \di \Omd, \\
\Du \lrpar{r^2 \omtrchid} =& 2 \Divdo \lrpar{- \etabd + 2 \di \Omd} - \Divdo \Divdo \gdcd + \frac{2}{r} \lrpar{\Ldo+2} \phid - 2r \omtrchibd - 4 \Omd.
\end{aligned} \label{EQingoingEQ3}
\end{align}

\ni As in Section \ref{SECprelimAnalysis}, we can derive null transport equations along $\HHb_{[-1,0],1}$ from the linearized null structure equations \eqref{EQingoingEQ1}, \eqref{EQingoingEQ2} and \eqref{EQingoingEQ3}; the explicit proof of the following lemma is omitted.

\begin{lemma}[Null transport equations in $\Lb$-direction] \label{LEMnulltransportTRANS} The linearized null structure equations imply the following null transport equations,\begin{align*} 
\begin{aligned} 
\Du \lrpar{\frac{r}{2}\lrpar{\omtrchibd+\frac{4}{r} \Omd} -\frac{\phid}{r}} =&0, \\
\Du \lrpar{r^2 \etabd - \frac{r^3}{2} \di \lrpar{\omtrchid - \frac{4}{r}\Omd}} =& - \Divdo \chibhd,
\end{aligned} 
\end{align*}
and
\begin{align*} 
\begin{aligned} 
\Du \lrpar{r^2 \omtrchid +\frac{2}{r} \Divdo \lrpar{r^2 \etabd -\frac{r^3}{2}\di \lrpar{\omtrchibd + \frac{4}{r}\Omd }} -2r^3\Kd - r^2 \lrpar{\omtrchibd + \frac{4}{r}\Omd}}=0.
\end{aligned} 
\end{align*}
\end{lemma}

\ni \emph{Remarks on Lemma \ref{LEMnulltransportTRANS}.}
\begin{enumerate}
\item By Lemma \ref{LEMnulltransportTRANS}, the charges $\underline{\QQ}_0, \underline{\QQ}_1$ and $\underline{\QQ}_2$ defined on the sphere $S_{u,v}$ with $r=v-u>0$ by
\begin{align} 
\begin{aligned} 
\underline{\QQ}_0 :=& \lrpar{ r^2 \etabd - \frac{r^3}{2} \di \lrpar{\omtrchid - \frac{4}{r}\Omd}}^{[1]}, \\
\underline{\QQ}_1 :=& \frac{r}{2} \lrpar{\omtrchibd + \frac{4}{r} \Omd} - \frac{\phid}{r}, \\
\underline{\QQ}_2 :=& r^2 \omtrchid + \frac{2}{r} \Divdo \lrpar{r^2 \etabd - \frac{r^3}{2}\di \lrpar{\omtrchibd + \frac{4}{r} \Omd }} \\
&- r^2 \lrpar{\omtrchibd + \frac{4}{r}\Omd} - 2 r^3 \Kd,
\end{aligned} \label{EQunderlinecharges12}
\end{align}
satisfy the following conservation laws along $\HHb_{[-1,0],1}$,
\begin{align*} 
\begin{aligned} 
\Du \underline{\QQ}_0 = 0, \,\,\Du \underline{\QQ}_1 = 0, \,\, \Du \underline{\QQ}_2 = 0.
\end{aligned} 
\end{align*}

\item In analogy to the analysis in Section \ref{SEClinPRELIM20}, the null structure equations \eqref{EQingoingEQ1}, \eqref{EQingoingEQ2}, \eqref{EQingoingEQ3} and the null transport equations of Lemma \ref{LEMnulltransportTRANS} imply that the quantities 
\begin{align} 
\begin{aligned} 
\phid, \gdcd, \lrpar{r^2 \etabd - \frac{r^3}{2} \di \lrpar{\omtrchibd + \frac{4}{r}\Omd}}^{[\geq2]}
\end{aligned} \label{EQtransgluing1}
\end{align}
can be glued without obstacles along $\HHb_{[0,-1],1}$ by using the degrees of freedom $\Omd$ and $\chibhd$ on $\HHb_{[-1,0],1}$.

\item As for the linearized null structure equations on $\HH_{-1,[1,2]}$, there are further charges $\uQQ_i$, $4\leq i \leq 7$, and higher-order charges along $\HHb_{[-1,0],1}$. However, for the purposes of this paper, the explicit expressions for $\uQQ_1, \uQQ_2$ and $\uQQ_3$ are sufficient. 

\end{enumerate}

\subsection{Preliminary analysis of charges} \label{SECprelimTRANSVERSE} In this section we derive relations between the charges
$$\QQ_0, \QQ_1, \QQ_2 \text{ and } \underline{\QQ}_0, \uQQ_1, \underline{\QQ}_2.$$ 

\ni The following lemma is the main result of this section.
\begin{lemma}[Charge identities] \label{LEMchargeRelationstransv} Consider linearized sphere data $\xd_{u,v}$ on a sphere $S_{u,v}$. Let $\QQ_i$ and $\uQQ_i$, $0\leq i \leq 2$, denote the associated charges on $S_{u,v}$. Then it holds that, with $r=v-u$,
\begin{align*} 
\begin{aligned} 
\QQ_1 -\frac{1}{2r} \underline{\QQ}_2 - \underline{\QQ}_1 =&  -2 \Omd + \half \Divdo \Divdo \gdcd- \frac{1}{r^2} \Divdo \lrpar{r^2 \etabd - \frac{r^3}{2} \di \lrpar{\omtrchibd + \frac{4}{r}\Omd }}    - \frac{1}{r} \Ldo \phid, 
\end{aligned} 
\end{align*}
and
\begin{align*} 
\begin{aligned} 
-\frac{1}{4r} \lrpar{\QQ_2 + \uQQ_2} - \half \Ldo\lrpar{\QQ_1 - \uQQ_1} =& \Ldo\Omd -\frac{1}{r^2} \Divdo \lrpar{r^2 \etabd - \frac{r^3}{2} \di \lrpar{\omtrchibd +\frac{4}{r} \Omd}} -\frac{1}{r} \Ldo \phid.
\end{aligned} 
\end{align*}
and moreover,
\begin{align} 
\begin{aligned} 
\lrpar{\QQ_0}_H = -\lrpar{\uQQ_0}_H, \,\, \lrpar{\QQ_0}_E = \lrpar{\uQQ_0}_E+\frac{r}{2} \di \underline{\QQ}_2^{[1]}, \,\, \QQ_2^{[0]} = - \underline{\QQ}_2^{[0]}, \,\, \QQ_2^{[1]} = \underline{\QQ}_2^{[1]}.
\end{aligned} \label{EQgaugerelationtransv2}
\end{align}

\end{lemma}

\ni \emph{Remarks on Lemma \ref{LEMchargeRelationstransv}.}
\begin{enumerate}

\item The significance of the first two identities of Lemma \ref{LEMchargeRelationstransv} is that by \eqref{EQtransgluing1} the terms on the right-hand side are freely glueable along $\HHb_{[-1,0],1}$. 

\item The relations \eqref{EQgaugerelationtransv2} show that $\uQQ_0$ and $\uQQ_2^{[\leq1]}$ fully determine $\QQ_0$ and $\QQ_2^{[\leq1]}$ on $S_{u,v}$, and vice versa.

\end{enumerate}

\ni The rest of this section contains the proof of Lemma \ref{LEMchargeRelationstransv}. In the following we use the definition of $\uQQ_1,\uQQ_2$ and $\uQQ_3$, see \eqref{EQunderlinecharges12}, as well as \eqref{EQssLinearizationRelations7778}, \eqref{EQgaussDHR} and \eqref{EQdefChargesMinkowski8891}, that is,
\begin{align*} 
\begin{aligned} 
\etabd = - \etad + 2 \di \Omd, \,\, \Kd = \frac{1}{2r^2} \Divdo \Divdo \gdcd - \frac{1}{r^3} (\Ldo +2) \phid,
\end{aligned} 
\end{align*}
and
\begin{align*} 
\begin{aligned} 
\QQ_0 :=& \lrpar{r^2\etad +\frac{r^3}{2} \di\lrpar{\omtrchid-\frac{4}{r}\Omd}}^{[1]}, \\
\QQ_1 :=& \frac{r}{2} \lrpar{\omtrchid - \frac{4}{r} \Omd} + \frac{\phid}{r}, \\
\QQ_2 :=& r^2 \omtrchibd -\frac{2}{r} \Divdo \lrpar{r^2 \etad + \frac{r^3}{2} \di \lrpar{\omtrchid-\frac{4}{r} \Omd}} - r^2 \lrpar{\omtrchid - \frac{4}{r} \Omd} + 2 r^3 \Kd
\end{aligned} 
\end{align*}

\ni \textbf{Analysis of $\QQ_2$.} We have that
\begin{align*} 
\begin{aligned} 
\QQ_2 + \uQQ_2 =& - \frac{2}{r} \Divdo \lrpar{r^2 \etad + \frac{r^3}{2} \di \lrpar{\omtrchid - \frac{4}{r} \Omd}} + \frac{2}{r} \Divdo \lrpar{r^2 \etabd - \frac{r^3}{2} \di \lrpar{\omtrchibd +\frac{4}{r} \Omd}} \\
=& \frac{4}{r} \Divdo \lrpar{r^2 \etabd - \frac{r^3}{2} \di \lrpar{\omtrchibd +\frac{4}{r} \Omd}} \\
&-\frac{2}{r} \Divdo \lrpar{\frac{r^3}{2} \di \omtrchid - \frac{r^3}{2} \di \omtrchibd - 2 r^2 \Omd} \\
=& \frac{4}{r} \Divdo \lrpar{r^2 \etabd - \frac{r^3}{2} \di \lrpar{\omtrchibd +\frac{4}{r} \Omd}} -\Ldo \lrpar{r^2 \omtrchid - r^2 \omtrchibd - 4r \Omd} \\
=& \frac{4}{r} \Divdo \lrpar{r^2 \etabd - \frac{r^3}{2} \di \lrpar{\omtrchibd +\frac{4}{r} \Omd}} \\
&-\Ldo \lrpar{\uQQ_2 - \frac{2}{r} \Divdo\lrpar{r^2 \etabd - \frac{r^3}{2} \di \lrpar{\omtrchibd +\frac{4}{r} \Omd}} + 2r^3 \Kd } \\
=& \frac{2}{r} \lrpar{\Ldo +2 } \Divdo \lrpar{r^2 \etabd - \frac{r^3}{2} \di \lrpar{\omtrchibd +\frac{4}{r} \Omd}} \\
&-\Ldo \uQQ_2 - r \Ldo \Divdo \Divdo \gdcd + 2 \Ldo \lrpar{\Ldo +2} \phid, 
\end{aligned} 
\end{align*}
which can be rewritten as
\begin{align} 
\begin{aligned} 
&-\frac{1}{2r}\lrpar{\QQ_2 +(\Ldo+1) \uQQ_2} \\
=& \Ldo \lrpar{-\frac{1}{r^2} \Divdo \lrpar{r^2 \etabd - \frac{r^3}{2} \di \lrpar{\omtrchibd +\frac{4}{r} \Omd}} +\half \Divdo \Divdo \gdcd-\frac{1}{r} \Ldo \phid} \\
& -\frac{2}{r^2} \Divdo \lrpar{r^2 \etabd - \frac{r^3}{2} \di \lrpar{\omtrchibd +\frac{4}{r} \Omd}} -\frac{2}{r} \Ldo \phid.
\end{aligned} \label{EQchargeEXPR2}
\end{align}

\ni Projecting \eqref{EQchargeEXPR2} onto the modes $l=0$ and $l=1$ proves the last two of \eqref{EQgaugerelationtransv2},
\begin{align*} 
\begin{aligned} 
\QQ_2^{[0]} = - \underline{\QQ}_2^{[0]}, \,\, \QQ_2^{[1]} = \underline{\QQ}_2^{[1]}.
\end{aligned} 
\end{align*}

\ni \textbf{Analysis of $\QQ_1$.} We have that
\begin{align*} 
\begin{aligned} 
\QQ_1
=& \frac{1}{2r} \lrpar{r^2 \omtrchid} - 2\Omd + \frac{\phid}{r} \\
=& \frac{1}{2r} \lrpar{\underline{\QQ}_2 - \frac{2}{r} \Divdo \lrpar{r^2 \etabd - \frac{r^3}{2} \di \lrpar{\omtrchibd + \frac{4}{r}\Omd }} }\\
&+\frac{1}{2r} \lrpar{r^2 \lrpar{\omtrchibd+\frac{4}{r}\Omd}+ 2r^3\Kd} -2\Omd + \frac{\phid}{r}\\
=& \frac{1}{2r} \underline{\QQ}_2 - \frac{1}{r^2} \Divdo \lrpar{r^2 \etabd - \frac{r^3}{2} \di \lrpar{\omtrchibd + \frac{4}{r}\Omd }} \\
&+ \lrpar{\underline{\QQ}_1 + \frac{\phid}{r}} + r^2 \Kd -2\Omd + \frac{\phid}{r} \\
=& \frac{1}{2r} \underline{\QQ}_2 - \frac{1}{r^2} \Divdo \lrpar{r^2 \etabd - \frac{r^3}{2} \di \lrpar{\omtrchibd + \frac{4}{r}\Omd }} \\
&+ \underline{\QQ}_1  -2 \Omd + \half \Divdo \Divdo \gdcd - \frac{1}{r} \Ldo \phid,
\end{aligned} 
\end{align*}
which proves the first equation of Lemma \ref{LEMchargeRelationstransv}, that is,
\begin{align} 
\begin{aligned} 
\QQ_1 - \uQQ_1 - \frac{1}{2r} \uQQ_2 =&  -2 \Omd + \half \Divdo \Divdo \gdcd  \\
&- \frac{1}{r^2} \Divdo \lrpar{r^2 \etabd - \frac{r^3}{2} \di \lrpar{\omtrchibd + \frac{4}{r}\Omd }} - \frac{1}{r} \Ldo \phid.
\end{aligned} \label{EQchargeEXPR1}
\end{align}

\ni Moreover, plugging \eqref{EQchargeEXPR1} into \eqref{EQchargeEXPR2}, we get
\begin{align*} 
\begin{aligned} 
&-\frac{1}{2r}\lrpar{\QQ_2 +(\Ldo+1) \uQQ_2} \\
=& \Ldo \lrpar{\QQ_1 - \uQQ_1 - \frac{1}{2r} \uQQ_2 + 2 \Omd} \\
& -\frac{2}{r^2} \Divdo \lrpar{r^2 \etabd - \frac{r^3}{2} \di \lrpar{\omtrchibd +\frac{4}{r} \Omd}} -\frac{2}{r} \Ldo \phid,
\end{aligned} 
\end{align*}
which proves the second equation of Lemma \ref{LEMchargeRelationstransv}, that is,
\begin{align*} 
\begin{aligned} 
&-\frac{1}{4r} \lrpar{\QQ_2 + \uQQ_2} - \half \Ldo\lrpar{\QQ_1 - \uQQ_1} \\
=& \Ldo\Omd -\frac{1}{r^2} \Divdo \lrpar{r^2 \etabd - \frac{r^3}{2} \di \lrpar{\omtrchibd +\frac{4}{r} \Omd}} -\frac{1}{r} \Ldo \phid.
\end{aligned} 
\end{align*}

\ni \textbf{Analysis of $\QQ_0$.} We have that
\begin{align*} 
\begin{aligned} 
\QQ_0 + \uQQ_0 =& \lrpar{r^2\etad +\frac{r^3}{2} \di\lrpar{\omtrchid-\frac{4}{r}\Omd}}^{[1]} + \lrpar{-r^2\etad -\frac{r^3}{2} \di \omtrchibd }^{[1]} \\
=& r^2 \di \lrpar{\frac{r}{2}\omtrchid - \frac{r}{2} \omtrchibd - 2 \Omd}^{[1]}\\
=& r^2 \di \lrpar{\frac{r}{2}\lrpar{\omtrchid-\frac{4}{r}\Omd}+\frac{\phid}{r} - \lrpar{\frac{r}{2}\lrpar{\omtrchibd+\frac{4}{r}\Omd}- \frac{\phid}{r} } + 2\Omd - \frac{2\phid}{r}}^{[1]}\\
=& r^2 \di \lrpar{\QQ_1 - \uQQ_1+ 2\Omd - \frac{2\phid}{r}}^{[1]}.
\end{aligned} 
\end{align*}
On the one hand, this shows the first of \eqref{EQgaugerelationtransv2}, that is, 
\begin{align*} 
\begin{aligned} 
\lrpar{\QQ_0}_H = -\lrpar{\uQQ_0}_H.
\end{aligned}
\end{align*}
On the other hand, together with \eqref{EQchargeEXPR1} we get that
\begin{align*} 
\begin{aligned} 
\lrpar{\QQ_0}_E =& -\lrpar{\uQQ_0}_E+ r^2 \di \lrpar{\QQ_1 - \uQQ_1+ 2\Omd - \frac{2\phid}{r}}^{[1]}_E \\
=& -\lrpar{\uQQ_0}_E+ r^2 \di \lrpar{\frac{1}{2r} \underline{\QQ}_2 -\frac{1}{r^2} \Divdo \uQQ_0}^{[1]}_E \\
=& \lrpar{\uQQ_0}_E+\frac{r}{2} \di \underline{\QQ}_2^{[1]},
\end{aligned} 
\end{align*}
where we used that $\lrpar{\di\Divdo \uQQ_0}_E = -2 (\uQQ_0)_E$ due to $\uQQ_0= \uQQ_0^{[1]}$ and \eqref{EQFourierMultiplier7778}. This proves the second of \eqref{EQgaugerelationtransv2}. This finishes the proof of Lemma \ref{LEMchargeRelationstransv}. 

\subsection{Linearized bifurcate characteristic gluing} \label{SECstatementFULLhigher} 

In this section we state and solve the linearized codimension-$10$ bifurcate characteristic gluing problem.\\

\ni \textbf{Notation.} To ease presentation, we do not explicitly state the corresponding higher regularity norms and assume to work in a smooth setting. In the following, let $m\geq0$ be an integer.

\begin{theorem}[Linearized bifurcate characteristic gluing] \label{THMlinearizedGluingTWO} Consider on spheres $S_{0,1}$ and $S_{-1,2}$, respectively, the following smooth linearized higher-order sphere data,
\begin{align*} 
\begin{aligned} 
(\xd_{0,1}, \dot{\DD}^{L,m}_{0,1}, \dot{\DD}^{\Lb,m}_{0,1}) \text{ and } (\xd_{-1,2}, \dot{\DD}^{L,m}_{-1,2},\dot{\DD}^{\Lb,m}_{-1,2}).
\end{aligned} 
\end{align*}
There exist
\begin{itemize}
\item a smooth solution $(\xd,\dot{\DD}^{L,m},\dot{\DD}^{\Lb,m})$ to the higher-order null structure equations on $\HH_{-1,[1,2]}$
\item a smooth solution $(\dot{\underline{x}},\dot{\underline{\DD}}^{L,m},\dot{\underline{\DD}}^{\Lb,m})$ to the higher-order null structure equations on $\HHb_{[-1,0],1}$,
\item smooth higher-order sphere data $(\xd_{-1,1},\dot{\DD}^{L,m}_{-1,1},\dot{\DD}^{\Lb,m}_{-1,1})$ on $S_{-1,1}$,
\end{itemize}
fully matching on $S_{-1,1}$,
\begin{align*} 
\begin{aligned} 
(\xd,\dot{\DD}^{L,m},\dot{\DD}^{\Lb,m}) \vert_{S_{-1,1}} =(\dot{\underline{x}},\dot{\underline{\DD}}^{L,m},\dot{\underline{\DD}}^{\Lb,m})\vert_{S_{-1,1}}= (\xd_{-1,1},\dot{\DD}^{L,m}_{-1,1},\dot{\DD}^{\Lb,m}_{-1,1}),
\end{aligned} 
\end{align*}
such that we have higher-order matching on $S_{0,1}$,
\begin{align*} 
\begin{aligned} 
(\xd,\dot{\DD}^{L,m},\dot{\DD}^{\Lb,m}) \vert_{S_{0,1}} = (\xd_{0,1}, \dot{\DD}^{L,m}_{0,1}, \dot{\DD}^{\Lb,m}_{0,1}),
\end{aligned} 
\end{align*}
and higher-order matching up to the charges $\QQ_0$ and $\QQ_2^{[\leq1]}$ on $S_{-1,2}$, that is,
\begin{align} 
\begin{aligned} 
\MMf\lrpar{ \xd \vert_{S_{-1,2}}} = \MMf\lrpar{\xd_{-1,2}}, \,\, \dot{\DD}^{L,m} \vert_{S_{-1,2}} = \dot{\DD}^{L,m}_{-1,2}, 
\end{aligned} \label{EQdefMatchingHigherConditionS12}
\end{align}
where $\MMf$ denotes the matching map defined in Definition \ref{DEFmatchingMAP} applied to linearized sphere data, and 
\begin{align} 
\begin{aligned} 
&\lrpar{{\QQ_1}, {\QQ_2^{[\geq2]}}, {\QQ_3}, \dots, \QQ_{i(m)} }\lrpar{({\xd},\dot{\DD}^{L,m},\dot{\DD}^{\Lb,m}) \vert_{S_{-1,2}}} \\
=& \lrpar{{\QQ_1}, {\QQ_2^{[\geq2]}}, {\QQ_3}, \dots, \QQ_{i(m)} }\lrpar{\xd_{-1,2},\dot{\DD}_{-1,2}^{L,m},\dot{\DD}_{-1,2}^{\Lb,m}}.
\end{aligned} \label{EQdefMatchingHigherConditionS123337778}
\end{align}
\end{theorem}

\ni \emph{Remarks on Theorem \ref{THMlinearizedGluingTWO}.}
\begin{enumerate}
\item The matching \eqref{EQdefMatchingHigherConditionS12} and \eqref{EQdefMatchingHigherConditionS123337778} equals \emph{higher-order matching up to charges $\QQ_0$ and $\QQ_2^{[\leq1]}$} in the following sense. Analogous to Lemma \ref{LEMconditionalMATCHING}, if it holds on $S_{-1,2}$ in addition to \eqref{EQdefMatchingHigherConditionS12} that
\begin{align*} 
\begin{aligned} 
\lrpar{{\QQ_0}, {\QQ_2^{[\leq1]}}}\lrpar{{\xd} \vert_{S_{-1,2}}} = \lrpar{{\QQ_0}, {\QQ_2^{[\leq1]}}}\lrpar{\xd_{-1,2}},
\end{aligned} 
\end{align*}
then we have that 
\begin{align*} 
\begin{aligned} 
\xd \vert_{S_{-1,2}} = \xd_{-1,2},
\end{aligned} 
\end{align*}
by which we can subsequently deduce from \eqref{EQdefMatchingHigherConditionS12} and \eqref{EQdefMatchingHigherConditionS123337778} the full higher-order matching
\begin{align*} 
\begin{aligned} 
({\xd},\dot{\DD}^{L,m},\dot{\DD}^{\Lb,m}) \vert_{S_{-1,2}} = (\xd_{-1,2},\dot{\DD}_{-1,2}^{L,m},\dot{\DD}_{-1,2}^{\Lb,m}).
\end{aligned} 
\end{align*}

\end{enumerate}

\ni In the following we outline the proof of Theorem \ref{THMlinearizedGluingTWO}. We first make some remarks.
\begin{itemize}

\item In accordance with the definition of free data, see Section \ref{SECderivationConstraintFunctions1} and Remark \ref{REMARKfreedatalinear}, our degrees of freedom in the linearized gluing problem are the prescriptions of
\begin{align} 
\begin{aligned} 
\chibhd \text{ and } \Omd \text{ on } \HHb_{[-1,0],1}, \,\, \chihd \text{ and } \Omd \text{ on } \HH_{-1,[1,2]}.
\end{aligned} \label{EQfreedomtransversal}
\end{align}

\item The gluing of 
\begin{align} 
\begin{aligned} 
\Omd, \omd, \chihd, \ad, \DD^{L,m} \text{ along } \HH_{-1,[1,2]}
\end{aligned} \label{EQgluingLderivativesHIGHER}
\end{align}
and of 
\begin{align} 
\begin{aligned} 
\Omd, \ombd, \chibhd, \abd, \DD^{\Lb,m} \text{ along } \HHb_{[-1,0],1}
\end{aligned} \label{EQgluingLbderivativesHIGHER}
\end{align}
follows without obstructions from the degrees of freedom \eqref{EQfreedomtransversal}. Indeed, see the discussions concerning the linearized higher-order $L$-gluing along $\HH_{-1,[1,2]}$ in Remark \ref{REMhigherLgluing1}, which also generalizes to higher-order $\Lb$-gluing along $\HHb_{[-1,0],1}$.

\item Given linearized higher-order sphere data $(\xd_{0,1}, \dot{\DD}^{L,m}_{0,1}, \dot{\DD}^{\Lb,m}_{0,1})$ on $S_{0,1}$, the conservation laws along $\HHb_{[-1,0],1}$ determine the charges $\uQQ_i$, $0\leq i \leq i(m)$ on $S_{-1,1}$. Similarly, given linearized higher-order sphere data $(\xd_{-1,2}, \dot{\DD}^{L,m}_{-1,2}, \dot{\DD}^{\Lb,m}_{-1,2})$ on $S_{-1,2}$, the conservation laws along $\HH_{-1,[1,2]}$ determine the charges $\QQ_i$, $0\leq i \leq i(m)$, on $S_{-1,1}$. The latter are denoted in the following by
\begin{align} 
\begin{aligned} 
(\QQ_i)_0 \text{ for } 0\leq i \leq i(m).
\end{aligned} \label{EQdefChargescomingfromS12}
\end{align}
\end{itemize}

\ni In the following, we prove Theorem \ref{THMlinearizedGluingTWO} in two steps. 
\begin{enumerate}
\item We construct a solution $(\underline{x},\underline{\DD}^{L,m},\underline{\DD}^{\Lb,m})$ to the linearized higher-order null structure equations on $\HHb_{[-1,0],1}$ such that on $S_{0,1}$ it fully matches the given higher-order sphere data, and on $S_{-1,1}$ we have the charge matching 
\begin{align} 
\begin{aligned} 
&\lrpar{\QQ_1, \QQ_2^{[\geq2]}, \QQ_3, \dots, \QQ_{i(m)}}(\underline{x},\underline{\DD}^{L,m},\underline{\DD}^{\Lb,m}) \vert_{S_{-1,1}}\\
=&
\lrpar{\lrpar{\QQ_1}_0, \lrpar{\QQ_2^{[\geq2]}}_0, \lrpar{\QQ_3}_0, \dots, \lrpar{\QQ_{i(m)}}_0 },
\end{aligned} \label{EQchargeMatchingS117778}
\end{align}
where the right-hand side charges are defined in \eqref{EQdefChargescomingfromS12}. This is the content of Proposition \ref{PROPlinearindep} below.

\item We construct a solution $(\dot{x},\dot{\DD}^{L,m},\dot{\DD}^{\Lb,m})$ to the linearized higher-order null structure equations on $\HH_{-1,[1,2]}$ such that on $S_{-1,1}$ it fully matches the given higher-order sphere data, and on $S_{-1,2}$ we have the matching
\begin{align*} 
\begin{aligned} 
\MMf(\xd\vert_{S_{-1,2}}) = \MMf(\xd_{-1,2}), \,\, \dot{\DD}^{L,m} \vert_{S_{-1,2}} = \dot{\DD}^{L,m}_{-1,2}.
\end{aligned} 
\end{align*}
By the charge matching \eqref{EQchargeMatchingS117778} and conservation laws along $\HH_{-1,[1,2]}$, it follows that the condition \eqref{EQdefMatchingHigherConditionS123337778} is satisfied on $S_{-1,2}$.
\end{enumerate}

\ni The above two steps can be seen as generalization of Sections \ref{SECmatchingCharges11} and \ref{SECsecondchoice12}, respectively. 
\begin{proposition}[Characteristic gluing along $\HHb_{[-1,0],1}$ with charge matching on $S_{-1,1}$] \label{PROPlinearindep} Let $(\xd_{0,1},\dot{\DD}^{L,m}_{0,1}, \dot{\DD}^{\Lb,m}_{0,1})$ be given linearized higher-order sphere data on $S_{0,1}$ and let 
\begin{align} 
\begin{aligned} 
\lrpar{\lrpar{\QQ_1}_0, \lrpar{\QQ_2^{[\geq2]}}_0, \lrpar{\QQ_3}_0, \dots, \lrpar{\QQ_{i(m)}}_0 }
\end{aligned} \label{EQprescribedCHARGEvalueshigher}
\end{align}
be a given tuple of charge values on $S_{-1,1}$. There exists higher-order ingoing null data $(\underline{\xd},\underline{\DD}^{L,m},\underline{\DD}^{\Lb,m})$ on $\HHb_{[-1,0],1}$ solving the linearized higher-order null structure equations such that
\begin{align*} 
\begin{aligned} 
(\underline{\xd},\underline{\DD}^{L,m},\underline{\DD}^{\Lb,m}) \vert_{S_{0,1}}= (\xd_{0,1},\dot{\DD}^{L,m}_{0,1}, \dot{\DD}^{\Lb,m}_{0,1}),
\end{aligned} 
\end{align*}
and
\begin{align*} 
\begin{aligned} 
&\lrpar{{\QQ_1}, {\QQ_2^{[\geq2]}}, {\QQ_3}, \dots, \QQ_{i(m)} }\lrpar{(\underline{\xd},\underline{\DD}^{L,m},\underline{\DD}^{\Lb,m}) \vert_{S_{-1,1}}} \\
=& \lrpar{\lrpar{\QQ_1}_0, \lrpar{\QQ_2^{[\geq2]}}_0, \lrpar{\QQ_3}_0, \dots, \lrpar{\QQ_{i(m)}}_0 }.
\end{aligned} 
\end{align*}

\end{proposition}
\emph{Remarks on Proposition \ref{PROPlinearindep}.}
\begin{enumerate}
\item The charges $\QQ_0$ and $\QQ_2^{[\geq1]}$ can in general not be matched on $S_{-1,1}$ because by Lemma \ref{LEMchargeRelationstransv}, they are determined from $\uQQ_0$ and $\uQQ_2^{[\geq1]}$ which are in turn determined from $\xd_{0,1}$ by the conservation laws on $\HHb_{[-1,0],1}$.
\item The sphere data $\underline{\xd}\vert_{S_{-1,1}}$ is a priori not fully determined by the matching conditions of Proposition \ref{PROPlinearindep} and thus our construction admits some freedom of choice.
\end{enumerate}

\begin{proof}[Proof of Proposition \ref{PROPlinearindep}] First consider the matching 
\begin{align} 
\begin{aligned} 
\lrpar{{\QQ_1}, {\QQ_2^{[\geq2]}} }\lrpar{\underline{\xd} \vert_{S_{-1,1}}} = \lrpar{\lrpar{\QQ_1}_0, \lrpar{\QQ_2^{[\geq2]}}_0 }.
\end{aligned} \label{EQfirstmatching127778}
\end{align}
By Lemma \ref{LEMchargeRelationstransv}, we have to solve the following system on $S_{-1,1}$,
\begin{align} 
\begin{aligned} 
\lrpar{\QQ_1}_0 -\frac{1}{2r} \underline{\QQ}_2 - \underline{\QQ}_1 =&  -2 \Omd + \half \Divdo \Divdo \gdcd \\
&- \frac{1}{r^2} \Divdo \lrpar{r^2 \etabd - \frac{r^3}{2} \di \lrpar{\omtrchibd + \frac{4}{r}\Omd }}    - \frac{1}{r} \Ldo \phid, 
\end{aligned} \label{EQtranssystem1}
\end{align}
and
\begin{align} 
\begin{aligned} 
&-\frac{1}{4r} \lrpar{\lrpar{\QQ^{[\geq2]}_2}_0 + \uQQ_2^{[\geq2]}} - \half \Ldo\lrpar{\lrpar{\QQ_1}_0 - \uQQ_1}^{[\geq2]} \\
=& \Ldo\Omd^{[\geq2]} -\frac{1}{r^2} \Divdo \lrpar{r^2 \etabd - \frac{r^3}{2} \di \lrpar{\omtrchibd +\frac{4}{r} \Omd}}^{[\geq2]} -\frac{1}{r} \Ldo \phid^{[\geq2]},
\end{aligned} \label{EQtranssystem2}
\end{align}
where the charges $\uQQ_1$ and $\uQQ_2$ are determined on $S_{-1,1}$ from the sphere data $\xd_{0,1}$ on $S_{0,1}$. 

We recall from \eqref{EQtransgluing1} that the quantities
\begin{align*} 
\begin{aligned} 
\Omd,\,\, \phid, \,\, \gdcd, \,\, \lrpar{r^2 \etabd - \frac{r^3}{2} \di \lrpar{\omtrchibd +\frac{4}{r} \Omd}}^{[\geq2]}
\end{aligned} 
\end{align*}
on the right-hand side of \eqref{EQtranssystem1} and \eqref{EQtranssystem2} are glueable along $\HHb_{[-1,0],1}$ and can thus be freely prescribed on $S_{-1,1}$. In the following, we show in detail how to prescribe them such that \eqref{EQtranssystem1} and \eqref{EQtranssystem2} are satisfied. \\

\ni \textbf{Matching of modes $l=0$.} By projecting \eqref{EQtranssystem1} onto the modes $l=0$, we get
\begin{align*} 
\begin{aligned} 
\lrpar{\QQ_1}_0^{[0]} = \frac{1}{2r} \underline{\QQ}_2^{[0]} + \underline{\QQ}_1^{[0]}  -2 \Omd^{[0]}.
\end{aligned} 
\end{align*}
Hence we prescribe on $S_{-1,1}$,
\begin{align*} 
\begin{aligned} 
\Omd^{[0]} = \half \lrpar{- \QQ_1^{[0]} +\frac{1}{2r} \underline{\QQ}_2^{[0]} + \underline{\QQ}_1^{[0]}}, \,\, \phid^{[0]} =0.
\end{aligned} 
\end{align*}

\ni \textbf{Matching of modes $l=1$.} By projecting \eqref{EQtranssystem1} onto the modes $l=1$, we get
\begin{align*} 
\begin{aligned} 
\lrpar{\QQ_1}_0^{[1]}=& \frac{1}{2r} \underline{\QQ}_2^{[1]} - \frac{1}{r^2} \Divdo \lrpar{r^2 \etabd - \frac{r^3}{2} \di \lrpar{\omtrchibd + \frac{4}{r}\Omd }}^{[1]} + \underline{\QQ}_1^{[1]}  -2 \Omd^{[1]} + \frac{2}{r} \phid^{[1]} \\
=&\frac{1}{2r} \underline{\QQ}_2^{[1]} - \frac{1}{r^2} \Divdo \underline{\QQ}_0 + \underline{\QQ}_1^{[1]}  -2 \Omd^{[1]} + \frac{2}{r} \phid^{[1]}.
\end{aligned} 
\end{align*}
Hence we prescribe on $S_{-1,1}$,
\begin{align*} 
\begin{aligned} 
\Omd^{[1]}=0, \,\, \phid^{[1]}= \frac{r}{2} \lrpar{\lrpar{\QQ_1}_0^{[1]} -\frac{1}{2r} \underline{\QQ}_2^{[1]}+\frac{1}{r^2} \Divdo \underline{\QQ}_0-\underline{\QQ}_1^{[1]}}.
\end{aligned} 
\end{align*}

\ni \textbf{Matching of modes $l\geq2$.} By projecting \eqref{EQtranssystem1} and \eqref{EQtranssystem2} onto the modes $l\geq2$, we get

\begin{align*} 
\begin{aligned} 
\lrpar{\QQ_1}_0^{[\geq2]} -\frac{1}{2r} \underline{\QQ}_2^{[\geq2]} - \underline{\QQ}_1^{[\geq2]} =&  -2 \Omd^{[\geq2]} + \half \Divdo \Divdo \gdcd^{[\geq2]} \\
&- \frac{1}{r^2} \Divdo \lrpar{r^2 \etabd - \frac{r^3}{2} \di \lrpar{\omtrchibd + \frac{4}{r}\Omd }}^{[\geq2]}    - \frac{1}{r} \Ldo \phid^{[\geq2]}, 
\end{aligned} 
\end{align*}
and
\begin{align*} 
\begin{aligned} 
&-\frac{1}{4r} \lrpar{\lrpar{\QQ^{[\geq2]}_2}_0 + \uQQ_2^{[\geq2]}} - \half \Ldo\lrpar{\lrpar{\QQ_1}_0 - \uQQ_1}^{[\geq2]} \\
=& \Ldo\Omd^{[\geq2]} -\frac{1}{r^2} \Divdo \lrpar{r^2 \etabd - \frac{r^3}{2} \di \lrpar{\omtrchibd +\frac{4}{r} \Omd}}^{[\geq2]} -\frac{1}{r} \Ldo \phid^{[\geq2]},
\end{aligned} 
\end{align*}

\ni It is straight-forward to check that the following prescription on $S_{-1,1}$ is a solution of the above system,
\begin{align*} 
\begin{aligned} 
\Omd^{[\geq2]} =0, \,\, \lrpar{r^2 \etabd - \frac{r^3}{2} \di \lrpar{\omtrchibd + \frac{4}{r}\Omd }}^{[\geq2]} =0, 
\end{aligned} 
\end{align*}
with $\phid^{[\geq2]}$ and $\gdcd^{[\geq2]}$ defined subsequently on $S_{-1,1}$ as solutions to
\begin{align*} 
\begin{aligned} 
 -\frac{1}{r} \Ldo \phid^{[\geq2]} =& -\frac{1}{4r} \lrpar{\lrpar{\QQ^{[\geq2]}_2}_0 + \uQQ_2^{[\geq2]}} - \half \Ldo\lrpar{\lrpar{\QQ_1}_0 - \uQQ_1}^{[\geq2]},\\
\half \Divdo \Divdo \gdcd^{[\geq2]} =& \lrpar{\QQ_1}_0^{[\geq2]} -\frac{1}{2r} \underline{\QQ}_2^{[\geq2]} - \underline{\QQ}_1^{[\geq2]} + \frac{1}{r} \Ldo \phid^{[\geq2]}, \\
(\Divdo \gdcd)_H=&0.
\end{aligned} 
\end{align*}
Using that the Laplacian and div-curl are elliptic Hodge systems, see Appendix \ref{SECellEstimatesSpheres}, it is straight-forward to prove regularity estimates for $\phid$ and $\gdcd$ which show that the above construction is consistent with the regularity hierarchy of the linearized null structure equations. This proves the matching \eqref{EQfirstmatching127778} of $\QQ_1$ and $\QQ_2^{[\geq2]}$.

It remains to realize the matching
\begin{align} 
\begin{aligned} 
\lrpar{{\QQ_3}, \dots, \QQ_{i(m)} }\lrpar{\underline{\xd} \vert_{S_{-1,1}}} = \lrpar{\lrpar{\QQ_3}_0, \dots, \lrpar{\QQ_{i(m)}}_0 }.
\end{aligned} \label{EQhigherregcharges3up}
\end{align}
First, from \eqref{EQdefChargesMinkowski8891} it follows that the matching condition
\begin{align} 
\begin{aligned} 
\lrpar{{\QQ_3}, \dots, \QQ_7 }\lrpar{\underline{\xd} \vert_{S_{-1,1}}} = \lrpar{\lrpar{\QQ_3}_0, \dots, \lrpar{\QQ_7}_0},
\end{aligned} \label{EQmatching3to7HIGHER}
\end{align}
can be realized by an appropriate choice of
\begin{align} 
\begin{aligned} 
\chibhd, \abd, \ombd, \Du\ombd \text{ on } S_{-1,1}.
\end{aligned} \label{EQchoiceOfquants1s11}
\end{align}
By \eqref{EQgluingLbderivativesHIGHER}, the quantities in \eqref{EQchoiceOfquants1s11} can be glued without obstructions along $\HHb_{[-1,0],1}$, and can thus be freely prescribed and realized on $S_{-1,1}$. This proves the matching \eqref{EQmatching3to7HIGHER}.

Second, the matching of higher-order charges
\begin{align*} 
\begin{aligned} 
\lrpar{{\QQ_8}, \dots, \QQ_{i(m)} }\lrpar{\underline{\xd} \vert_{S_{-1,1}}} = \lrpar{\lrpar{\QQ_8}_0, \dots, \lrpar{\QQ_{i(m)}}_0 }
\end{aligned} 
\end{align*}
follows similarly. Indeed, by the definition of higher-order charges as conserved quantities of null transport equations for higher-order transversal derivatives along $\HH_{-1,[1,2]}$, it follows that they can be realized by appropriate prescription of $\DD^{\Lb,m}$ on $S_{-1,1}$. The ability to glue along $\HHb_{[-1,0],1}$ to prescribed $\DD^{\Lb,m}$ on $S_{-1,1}$ follows from \eqref{EQgluingLbderivativesHIGHER}. This finishes the proof of Proposition \ref{PROPlinearindep}. \end{proof}

\ni We are now in position to conclude the proof of Theorem \ref{THMlinearizedGluingTWO} by step (2) outlined before. We recall that from Proposition \ref{PROPlinearindep} we have a smooth solution
\begin{align*} 
\begin{aligned} 
(\underline{\xd},\dot{\underline{\DD}}^{L,m}, \dot{\underline{\DD}}^{\Lb,m}) \text{ on } \HHb_{[-1,0],1}
\end{aligned} 
\end{align*}
to the linearized higher-order null structure equations such that we have full higher-order matching on $S_{0,1}$ and the following charge matching on $S_{-1,1}$,
\begin{align*} 
\begin{aligned} 
&\lrpar{{\QQ_1}, {\QQ_2^{[\geq2]}}, {\QQ_3}, \dots, \QQ_{i(m)} }\lrpar{(\underline{\xd},\dot{\underline{\DD}}^{L,m},\dot{\underline{\DD}}^{\Lb,m}) \vert_{S_{-1,1}}} \\
=& \lrpar{\lrpar{\QQ_1}_0, \lrpar{\QQ_2^{[\geq2]}}_0, \lrpar{\QQ_3}_0, \dots, \lrpar{\QQ_{i(m)}}_0 },
\end{aligned} 
\end{align*}
where the right-hand side are the conserved charges determined by conservation laws along $\HH_{-1,[1,2]}$ from the higher-order sphere data on $S_{-1,2}$,
\begin{align*} 
\begin{aligned} 
(\xd_{-1,2}, \dot{\DD}^{L,m}_{-1,2}, \dot{\DD}^{\Lb,m}_{-1,2}).
\end{aligned} 
\end{align*}

\ni By applying the characteristic gluing of Section \ref{SEClinearizedProblem} along $\HH_{-1,[1,2]}$, we construct a solution
\begin{align*} 
\begin{aligned} 
(\xd, \dot{\DD}^{L,m}, \dot{\DD}^{\Lb,m}) \text{ on } \HH_{-1,[1,2]}
\end{aligned} 
\end{align*}
to the linearized higher-order null structure such that we have full higher-order matching on $S_{-1,1}$ and we have the following higher-order matching on $S_{-1,2}$,
\begin{align*} 
\begin{aligned} 
\MMf(\xd \vert_{S_{-1,2}}) =& \MMf(\xd_{-1,2}), \,\, 
\dot{\DD}^{L,m}\vert_{S_{-1,2}} =& \dot{\DD}^{L,m}_{-1,2}, 
\end{aligned}
\end{align*}
and
\begin{align*} 
\begin{aligned} 
&\lrpar{{\QQ_1}, {\QQ_2^{[\geq2]}}, {\QQ_3}, \dots, \QQ_{i(m)} }\lrpar{({\xd},\dot{\DD}^{L,m},\dot{\DD}^{\Lb,m}) \vert_{S_{-1,2}}} \\
=& \lrpar{{\QQ_1}, {\QQ_2^{[\geq2]}}, {\QQ_3}, \dots, \QQ_{i(m)} }\lrpar{\xd_{-1,2},\dot{\DD}_{-1,2}^{L,m},\dot{\DD}_{-1,2}^{\Lb,m}}.
\end{aligned}
\end{align*}

\ni The following remark on the linearized charges finishes the proof of Theorem \ref{THMlinearizedGluingTWO}.

\begin{remark}[Relations for $\QQ_0$ and $\QQ_2^{[\leq1]}$] By Lemma \ref{LEMchargeRelationstransv}, that is, the relations
\begin{align} 
\begin{aligned} 
\lrpar{\QQ_0}_H = -\lrpar{\uQQ_0}_H, \,\, \lrpar{\QQ_0}_E = \lrpar{\uQQ_0}_E+\frac{r}{2} \di \underline{\QQ}_2^{[1]}, \,\, \QQ_2^{[0]} = - \underline{\QQ}_2^{[0]}, \,\, \QQ_2^{[1]} = \underline{\QQ}_2^{[1]}.
\end{aligned} \label{EQrelationschargeschangepa7778}
\end{align}
and the conservation laws along $\HHb_{[-1,0],1}$ and $\HH_{-1,[1,2]}$, we have that
\begin{align*} 
\begin{aligned} 
(\lrpar{\QQ_0}_H, \QQ_2^{[0]}, \QQ_2^{[1]})(\xd \vert_{S_{-1,2}}) = (\lrpar{\QQ_0}_H, \QQ_2^{[0]}, \QQ_2^{[1]})(\xd_{0,1}).
\end{aligned} 
\end{align*}
Moreover, from \eqref{EQrelationschargeschangepa7778} and the conservation laws we deduce that
\begin{align*} 
\begin{aligned} 
\lrpar{\QQ_0}_E(\xd \vert_{S_{-1,2}}) =& \lrpar{\QQ_0}_E(\xd \vert_{S_{-1,1}}) \\
=& \lrpar{\uQQ_0}_E(\xd \vert_{S_{-1,1}})+\frac{2}{2} \di \underline{\QQ}_2^{[1]}(\xd \vert_{S_{-1,1}}) \\
=&\lrpar{\uQQ_0}_E(\underline{\xd} \vert_{S_{-1,1}})+\frac{2}{2} \di \underline{\QQ}_2^{[1]}(\underline{\xd} \vert_{S_{-1,1}})\\
=& \lrpar{\uQQ_0}_E(\underline{\xd} \vert_{S_{0,1}})+\frac{2}{2} \di \underline{\QQ}_2^{[1]}(\underline{\xd} \vert_{S_{0,1}}) \\
=& \lrpar{\lrpar{\uQQ_0}_E(\underline{\xd} \vert_{S_{0,1}})+\frac{1}{2} \di \underline{\QQ}_2^{[1]}(\underline{\xd} \vert_{S_{0,1}})}  +\frac{1}{2} \di \underline{\QQ}_2^{[1]}(\underline{\xd} \vert_{S_{0,1}}) \\
=& \lrpar{\QQ_0}_E(\xd_{0,1}) + \frac{1}{2} \di {\QQ}_2^{[1]}({\xd}_{0,1}),
\end{aligned} 
\end{align*}
where we used that by construction, $\underline{\xd} \vert_{S_{0,1}} = \xd_{0,1}$. Using \eqref{EQFourierMultiplier7778} we can rewrite the above as, for $m=-1,0,1$,
\begin{align*} 
\begin{aligned} 
\lrpar{\QQ_0}^{(1m)}_E(\xd \vert_{S_{-1,2}}) = \lrpar{\QQ_0}^{(1m)}_E(\xd_{0,1}) - \frac{\sqrt{2}}{2} {\QQ}_2^{(1m)}({\xd}_{0,1}).
\end{aligned} 
\end{align*}
In terms of $\dot{\mathbf{P}}$ and $\dot{\mathbf{G}}$ (see Remark \eqref{EQREMchargeslinatMinkowski}), this can be written as
\begin{align*} 
\begin{aligned} 
\dot{\mathbf{G}}^m(\xd \vert_{S_{-1,2}}) = \dot{\mathbf{G}}^m(\xd \vert_{S_{0,1}}) - 2 \dot{\mathbf{P}}^m(\xd \vert_{S_{0,1}}).
\end{aligned} 
\end{align*}
\end{remark}

\subsection{Proof of Proposition \ref{PROPchargluingW}} \label{SECproofWgluing} \ni In this section we prove Proposition \ref{PROPchargluingW}, that is, the bifurcate characteristic gluing with localized sphere data perturbation $W$. 

The proof is a slight generalization of the proof of Theorem \ref{THMtransversalHIGHERv2}. As before, by the implicit function theorem, the proof can be reduced to solving the linearized problem. In this case, the linearized characteristic gluing problem admits the additional freedom of adding a linearized localized sphere data perturbation $\dot{W}$ to the linearized sphere data on $S_{-1,2}$. 

Given the explicit formulas in Section \ref{SECingoingNULLstructureEQS}, it is straight-forward to construct a smooth sphere data perturbation $\dot{W}$, compactly supported in the angular region $K$, with prescribed values for $\QQ_0(\dot{W})$ and $\QQ_2^{[\leq1]}(\dot{W})$, and derive appropriate bounds; see, for example, the explicit choice in \cite{ACR3}. This allows to match $\QQ_0$ and $\QQ_2^{[\leq1]}$ on $S_{-1,2}$.

Adding $\dot{W}$ to the linearized sphere data on $S_{-1,2}$ changes the gauge-dependent charges on $S_{-1,2}$ (by a well-controlled amount). Using the results of Section \ref{SECstatementFULLhigher} for the linearized characteristic gluing along two transversely-intersecting null hypersurfaces, we can match the solution to the null constraint equations to these new values of gauge-dependent charges on $S_{-1,2}$. This solves the linearized characteristic gluing problem for Proposition \ref{PROPchargluingW}.

We remark that for controlling the support of $W$ in the application of the implicit function theorem, $W$ is bounded in an $L^2$-based Sobolev space with weights (in particular, this space is Hilbert) which ensure its smooth vanishing towards the boundary of the angular region $K$, see, for example, \cite{Corvino,CorvinoSchoen}.

\appendix
\section{Perturbations of sphere data} \label{SECproofTEClemmasmoothness} In this section we define the perturbations of sphere data of Section \ref{SECdefEquivalenceFirstOrderSphereData}, and prove Propositions \ref{PropositionSmoothnessF} and \ref{PropositionSmoothnessF2}. We consider transversal perturbations in Section \ref{SECapp1231} and angular perturbations in Section \ref{SECapp1232}. The proof of Propositions \ref{PropositionSmoothnessF} and \ref{PropositionSmoothnessF2} is given in Section \ref{APPproofSmoothnessFINAL}.

\subsection{Transversal perturbations} \label{SECapp1231} In this section we define \emph{transversal perturbations} of sphere data.

\subsubsection{Null geometry} First we recall the null geometry setup. Let $\tilde S$ be a spacelike $2$-sphere in a spacetime $(\MM,\g)$. Let $(\tilde{u},\tilde{v},\tilde{\th}^1,\tilde{\th}^2)$ be a local double null coordinate system around $\tilde S$, that is
\begin{align} 
\begin{aligned}
\g =& - 4 \tilde{\Om}^2 d\tilde{u} d\tilde{v} + \tilde{\gd}_{CD} (d\tilde{\th}^C - \tilde{b}^C d\tilde{v})(d\tilde{\th}^D - \tilde{b}^D d\tilde{v}),
\end{aligned} \label{EQdoublenullform12342} 
\end{align}
such that $\tilde S = \tilde{S}_{0,2}:= \{\tilde u=0, \tilde v=2\}.$ We recall the following standard notation, see, for example, Section 1 of \cite{ChrFormationBlackHoles}. 
\begin{itemize}
\item The \emph{geodesic null vectorfields} are defined by
\begin{align} 
\begin{aligned} 
\tilde{\Lb}' := -2 \D\tilde{v}, \,\, \tilde{L}' := -2 \D\tilde{u},
\end{aligned} \label{EQdefinitionGeodesicFields}
\end{align}
where $\D$ denotes the covariant derivative on $(\MM,\g)$.

\item The \emph{normalised null vectorfields} are defined by
\begin{align*} 
\begin{aligned} 
\widehat{\tilde{L}} := \Om  \tilde{L}', \,\, \widehat{\tilde{\Lb}} := \Om  \tilde{\Lb}'.
\end{aligned}
\end{align*}
\item The \emph{equivariant null vectorfields} are defined by
\begin{align} 
\begin{aligned} 
\tilde L := \tilde{\Om}^2 L' , \,\, \tilde\Lb := \tilde{\Om}^2 \Lb'.
\end{aligned} \label{EQLLbdef}
\end{align}
\item The Ricci coefficients are defined with respect to the above vectorfields as follows,
\begin{align} 
\begin{aligned} 
\tilde{\chi}_{AB} :=& \g(\D_{\tilde A} \widehat{\tilde L},\pr_{\tilde{B}}), & \tilde{\chib}_{AB} :=& \g(\D_{\tilde{A}} \widehat{\tilde \Lb},\pr_{\tilde{B}}), & \tilde{\zeta}_A :=& \half \g(\D_{\tilde{A}} \widehat{\tilde L},\widehat{\tilde \Lb}), \\
 \tilde\eta :=& \tilde \zeta + \tilde\di \log \tilde\Om, & \tilde\om :=& \tilde L \log \tilde \Om, & \tilde{\omb} :=& \tilde{\Lb} \log \tilde\Om,
\end{aligned} \label{EQdefRicciCOEFF555}
\end{align}
where $\tilde{\di}$ denotes the exterior derivative on spheres $\tilde{S}_{u,v}$.
\end{itemize}

\ni We have the following practical lemma, see, for example, \cite{ChrFormationBlackHoles}.
\begin{lemma}[Properties of double null coordinates] \label{LEMcalculusDoubleNull} The following holds.
\begin{enumerate}
\item The inverse $\g^{-1}$ of \eqref{EQdoublenullform12342} is given by
\begin{align} 
\begin{aligned} 
\g^{-1} =& -\frac{1}{2\tilde{\Om}^2} \lrpar{\pr_{\tilde u} \otimes \pr_{\tilde v} + \pr_{\tilde v} \otimes \pr_{\tilde u}}-\frac{\tilde{b}^C}{2\tilde{\Om}^2 } \lrpar{\pr_{\tilde u} \otimes \pr_{\tilde C} + \pr_{\tilde C} \otimes \pr_{\tilde u}} + \tilde{\gd}^{AB} \pr_{\tilde A} \otimes \pr_{\tilde B}.
\end{aligned} \label{EQinverseExpressionAPP12234}
\end{align}
Specifically,
\begin{align} 
\begin{aligned} 
\g^{\tilde v \tilde v}=\g^{\tilde v \tilde A} =0.
\end{aligned} \label{EQspecificrelationsAPP}
\end{align}
\item It holds that $\g\lrpar{L',\Lb'}=-2\tilde{\Om}^{-2}$, and
\begin{align} 
\begin{aligned} 
\tilde L = \pr_{\tilde v} + \tilde{b}^A \pr_{\tth^A}, \,\, \tilde\Lb = \pr_{\tilde u}.
\end{aligned} \label{EQvectorRelations3355}
\end{align}
\item It holds that for $A=1,2$,
\begin{align} 
\begin{aligned} 
\pr_{\tilde u} \tilde{b}^A = 4\tilde{\Om}^2 \tilde{\zeta}^A.
\end{aligned} \label{EQbDerivativedoublenull}
\end{align}
\item It holds that
\begin{align} 
\begin{aligned} 
\Ga^{\tilde v}_{\tilde u \tilde u}=& \Ga^{\tilde v}_{\tilde u \tilde v} = \Ga^{\tilde v}_{\tilde u \tilde A} =0, &
\Ga^{\tilde v}_{\tilde A \tilde v} =& \pr_{\tilde A} \log \tilde \Om - \tilde{\zeta}_A - \frac{1}{2\tilde\Om} \tilde{\chib}_{AB} \tilde{b}^B, &
\Ga^{\tilde v}_{\tilde A \tilde B} =& \frac{1}{2\tilde\Om} \tilde{\chib}_{AB},
\end{aligned} \label{EQchristoffelSymbols12}
\end{align}
where the Christoffel symbols are defined by $\Ga^\ga_{\mu\nu}:= \half \g^{\ga \a} \lrpar{\pr_\mu \g_{\a\nu}+\pr_\nu \g_{\a\mu}-\pr_\a \g_{\mu\nu}}$.
\end{enumerate}
\end{lemma}

\subsubsection{Definition of $u$ on $\protect\tilde{\protect\underline{\protect{\mathcal{H}}}}_2$ and analysis of foliation geometry} In the following we change $\tilde u$ to $u$ on $\tilde{\HHb}_2:=\{\tilde v=2 \}$ and analyse how the foliation geometry of the resulting local double null coordinates $(u, v, \th^1,\th^2)$ (with $v=\tilde v$ on $\MM$) relates to the foliation geometry of the local double null coordinates $(\tilde u, \tilde v, \tth^1,\tth^2)$.

For a given scalar function $f=f(u,\th^1,\th^2)$, define $(u,\th^1,\th^2)$ on $\tilde{\HHb}_2$ by
\begin{align} 
\begin{aligned} 
\tilde{u}=u+f(u,\th^1,\th^2), \,\, \tth^1=\th^1, \,\, \tth^2=\th^2.
\end{aligned} \label{EQdefU}
\end{align}
For $f$ sufficiently small, $(u,\th^1, \th^2)$ are a coordinate system on $\tilde{\HHb}_2$ and we have that
\begin{align} 
\begin{aligned} 
\pr_u = \lrpar{1+\pr_u f} \pr_{\tilde u}, \,\, \pr_{\th^A} = \pr_{\tth^A} + (\pr_{\th^A} f) \pr_{\tilde u}, \,\, \pr_{\th^A} f =\lrpar{1+\pr_u f} \pr_{\tth^A} f .
\end{aligned} \label{EQpartialsRelation}
\end{align}

\ni In accordance with \eqref{EQdefinitionGeodesicFields} and \eqref{EQvectorRelations3355}, define on $\tilde{\HHb}_2$
\begin{align} 
\begin{aligned} 
\Lb := \pr_u, \,\, \Lb' := -2 \D \tilde{v} = \tilde{\Lb}',
\end{aligned} \label{EQgaugeChoice444}
\end{align}
and define in accordance with \eqref{EQLLbdef} the null lapse $\Om$ on $\tilde{\HHb}_2$ through the relation
\begin{align} 
\begin{aligned} 
\Lb= {\Om}^2 \Lb'.
\end{aligned} \label{EQdefNEWOMEGA}
\end{align}

\ni We can relate the foliation geometry of $(u,\th^1,\th^2)$ to the geometry of $(\tilde u, \tth^1, \tth^2)$ as follows. 
\begin{enumerate}
\item We explicitly calculate $\Om$ on $\tilde{\HHb}_2$ as follows. Using \eqref{EQLLbdef}, \eqref{EQdefU}, \eqref{EQpartialsRelation} and \eqref{EQgaugeChoice444}, it holds that on $\{\tilde v =2\}$,
\begin{align} 
\begin{aligned} 
 \Lb = \lrpar{1+\pr_u f}\tilde \Lb =  \lrpar{1+\pr_u f} \tilde{\Om}^{2}\tilde{\Lb}' =  \lrpar{1+\pr_u f}\tilde{\Om}^{2} {\Lb}',
\end{aligned} \label{EQrelationstildeLL}
\end{align}
from which we conclude by \eqref{EQdefNEWOMEGA} that on $\tilde{\HHb}_2$,
\begin{align} 
\begin{aligned} 
\Om^2 = \tilde{\Om}^2 \lrpar{1+\pr_u f}.
\end{aligned} \label{EQexpressionOMEGASQUARED123}
\end{align}

\item By \eqref{EQpartialsRelation} it follows that the induced metric $\gd$ on level sets of $u$ on $\tilde{\HHb}_2$ is given for $A,B=1,2$ by
\begin{align} 
\begin{aligned} 
\gd_{AB} := \g\lrpar{\pr_A, \pr_B}= \g\lrpar{\pr_{\tilde A}, \pr_{\tilde B}} = \tilde{\gd}_{AB}.
\end{aligned} \label{EQinducedmetricFormula15}
\end{align}
This implies further that
\begin{align} 
\begin{aligned} 
\gd^{AB} = \tilde{\gd}^{AB}.
\end{aligned} \label{EQinducedmetricFormula152}
\end{align}
We remark that in explicit notation, \eqref{EQexpressionOMEGASQUARED123} and \eqref{EQinducedmetricFormula15} are
\begin{align*} 
\begin{aligned} 
\Om^2(u,\th^1,\th^2) =& \lrpar{1+(\pr_u f)(u,\th^1,\th^2)} \tilde{\Om}^2(u+f(u,\th^1,\th^2),\th^1,\th^2),\\
\gd_{AB}(u,\th^1,\th^2)=& \tilde{\gd}_{AB}\lrpar{u+f(u,\th^1,\th^2),\th^1,\th^2}.
\end{aligned} 
\end{align*}

\item The vectorfield $\Lb$ and the scalar function $\Om$ uniquely determine the null vectorfield $L$ on $\{\tilde v=2\}$ defined by
\begin{align} 
\begin{aligned} 
\g\lrpar{L,\Lb} = -2 \Om^2, \,\, \g\lrpar{L,\pr_{\th^1}}=\g\lrpar{L,\pr_{\th^2}}=0.
\end{aligned} \label{EQrefLLBomfacts}
\end{align}
An explicit calculation shows that $L$ is given by
\begin{align} 
\begin{aligned} 
L =& \lrpar{ \tilde{\Om}^2\vert \Nd f \vert_\gd^2 } \tilde{\Lb} + \tilde{L} + \lrpar{2\tilde{\Om}^2 \tilde{\gd}^{AC}\pr_C f } \pr_{\tth^A}.
\end{aligned} \label{EQLvectorfieldFormula15}
\end{align}
where $\vert \Nd f \vert_\gd^2 := \tilde{\gd}^{AB}\pr_Af \pr_B f$. We define further $\widehat{L} := \Om^{-1} L$.

\end{enumerate}

\subsubsection{Analysis of Ricci coefficients on $\protect\tilde{\protect\underline{\protect\mathcal{H}}}_2$} The Ricci coefficients with respect to $\lrpar{\widehat L, \widehat \Lb}$ are defined as follows, 
\begin{align*} 
\begin{aligned} 
\chi_{AB} :=& \g(\D_{A} \widehat{L},\pr_B), & \chib_{AB} :=& \g(\D_A \widehat{\Lb},\pr_B), & \zeta_A :=& \half \g(\D_A \widehat{L},\widehat{\Lb}), \\
 \eta :=& \zeta + \di \log \Om, & \om :=& L \log \Om, & \omb :=& \Lb \log \Om.
\end{aligned} 
\end{align*}
We analyze the Ricci coefficients in the order $\lrpar{\omb, \chib, \om, \zeta, \eta, \chi, \Du\omb, D\om}.$\\

\ni \textbf{Analysis of $\omb$.} On the one hand we have by \eqref{EQgaugeChoice444} that
\begin{align*} 
\begin{aligned} 
\omb := \Lb \log \Om = \Om^{-1} \pr_u \Om = \frac{1}{2\Om^2} \pr_u \lrpar{\Om^2}.
\end{aligned} 
\end{align*}
On the other hand we have by \eqref{EQpartialsRelation} and \eqref{EQexpressionOMEGASQUARED123} that
\begin{align*} 
\begin{aligned} 
\pr_u \lrpar{\Om^2} =& \pr_u \lrpar{\tilde{\Om}^2\lrpar{1+\pr_u f}} = 2\tilde{\Om} \pr_{\tilde u} \tilde{\Om} \lrpar{1+\pr_u f}^2 + \tilde{\Om}^2 \pr_u^2 f.
\end{aligned} 
\end{align*}
Combining the above two and using \eqref{EQdefRicciCOEFF555}, it follows that
\begin{align} 
\begin{aligned} 
\omb =& \frac{1}{2\Om^2} \lrpar{2\tilde{\Om} \pr_{\tilde u} \tilde{\Om} \lrpar{1+\pr_u f}^2 + \tilde{\Om}^2 \pr_u^2 f} = \tilde{\omb} \lrpar{1+\pr_u f} + \half \frac{\tilde{\Om}^2}{\Om^2} \pr_u^2 f.
\end{aligned} \label{EQombExpresssion4}
\end{align}

\ni \textbf{Analysis of $\chib$.} By explicit computation, we have that
\begin{align} 
\begin{aligned} 
\chib_{AB} := \g(\D_A \widehat{\Lb},\pr_B) =\Om^{-1} \lrpar{1+\pr_u f} \tilde{\Om} \tilde{\chib}_{AB},
\end{aligned} \label{EQchib5556}
\end{align}
where we used that
\begin{align*} 
\begin{aligned} 
\g\lrpar{\D_{\pr_{\tilde u}} \pr_{\tilde u}, \pr_{\tilde B}} = \tilde{\Om}^4 \g\lrpar{\D_{\tilde{\Lb}'} \tilde{\Lb}', \pr_{\tilde B}} = 0.
\end{aligned} 
\end{align*}
We can separate \eqref{EQchib5556} into
\begin{align} 
\begin{aligned} 
\Om \trchib = \lrpar{1+\pr_u f} \tilde{\Om}{\tr \tilde{\chib}}, \,\, \chibh_{AB} = \frac{\tilde{\Om}}{\Om} \lrpar{1+\pr_u f} \tilde{\chibh}_{AB},
\end{aligned} \label{EQformulatrchichih555}
\end{align}
that is, in explicit notation,
\begin{align*} 
\begin{aligned} 
\lrpar{\Om \trchib}(u,\th^1,\th^2) =& \lrpar{1+\pr_u f(u,\th^1,\th^2)} \tilde{\Om} \tr \tilde{\chib}(u+f(u,\th^1,\th^2),\th^1,\th^2),\\
\chibh_{AB}(u,\th^1,\th^2) =&\frac{\tilde{\Om}}{\Om}(u+f(u,\th^1,\th^2),\th^1,\th^2) \lrpar{1+\pr_u f(u,\th^1,\th^2)} \tilde{\chibh}_{AB}(u+f(u,\th^1,\th^2),\th^1,\th^2).
\end{aligned} 
\end{align*}

\ni \textbf{Analysis of $\om$.} We calculate 
\begin{align*} 
\begin{aligned} 
\om := L \log \Om
\end{aligned} 
\end{align*}
as follows. First, by construction of the double coordinates $(u, \tilde v, \th^1, \th^2)$, see \eqref{EQrefLLBomfacts} and \eqref{EQgaugeChoice444}, $\Om$ is defined on $\MM$ through
\begin{align} 
\begin{aligned} 
L'(\tilde v) = \Om^{-2}. 
\end{aligned} \label{EQdoubleNULLrelation334}
\end{align}
In particular, this implies with the geodesic equation satisfied by $L'$ that 
\begin{align} 
\begin{aligned} 
L'\lrpar{\Om^{-2}} = L'\lrpar{L'(\tilde v) } =\D_{L'}\lrpar{\D_{L'} \tilde v}=  \D_{\underbrace{\D_{L'}L'}_{=0}} \tilde v + \D_{L'} \D_{L'} \tilde v= \Om^{-4} \D_L \D_L \tilde v.
\end{aligned} \label{EQomegaLprime2}
\end{align}
Second, we have the algebraic relation
\begin{align} 
\begin{aligned} 
L'\lrpar{\Om^{-2}} = -\frac{2}{\Om^3} L' \lrpar{\Om}.
\end{aligned} \label{EQomegaLprime1}
\end{align}
By \eqref{EQomegaLprime2} and \eqref{EQomegaLprime1}, we get that
\begin{align} 
\begin{aligned} 
\om = L \lrpar{\log\Om} = \Om L' \lrpar{\Om} = -\frac{\Om^4}{2} L' \lrpar{\Om^{-2}} = -\frac{1}{2} \D_{L}\D_L \tilde v.
\end{aligned} \label{EQexpressionome4}
\end{align}
Plugging \eqref{EQLvectorfieldFormula15} into \eqref{EQexpressionome4} and using \eqref{EQLLbdef}, we get that
\begin{align} 
\begin{aligned} 
\om =& -\frac{1}{2} \lrpar{ \tilde{\Om}^2 \vert \Nd f \vert^2_\gd}^2 \D_{\tilde u }\D_{\tilde u} \tilde v \underbrace{-\frac{1}{2} \D_{\tilde L} \D_{\tilde L} \tilde{v}}_{= \tilde\om} -2\tilde{\Om}^4 \tilde{\gd}^{AB} \tilde{\gd}^{CD}(\pr_B f)(\pr_C f) \D_{\tilde A} \D_{\tilde D} \tilde v \\
& - \lrpar{ \tilde{\Om}^2 \vert \Nd f \vert^2_\gd} \D_{\tilde u} \D_{\pr_{\tilde v} +\tilde{b}^C \pr_{\tilde C} } \tilde v - \lrpar{ \tilde{\Om}^2 \vert \Nd f \vert^2_\gd} \lrpar{2\tilde{\Om}^2 \tilde{\gd}^{AB}\pr_B f } \D_{\tilde u} \D_{\tilde A} \tilde v \\
&- \lrpar{2\tilde{\Om}^2 \tilde{\gd}^{AB}\pr_B f }  \D_{\pr_{\tilde v} +\tilde{b}^C \pr_{\tilde C}}\D_{\tilde A} \tilde v.
\end{aligned} \label{EQexpressionoMEGA5}
\end{align}
Here the Hessian $\D\D\tilde v$ is given in coordinates $\mu, \nu \in \{\tilde u, \tilde v, \tth^1, \tth^2\} $ by
\begin{align} 
\begin{aligned} 
\D_\mu \D_\nu \tilde v =& \pr_\mu \pr_\nu \tilde v - \Ga^{\la}_{\mu\nu} \pr_{\la} \tilde v = - \Ga^{\tilde v}_{\mu\nu}.
\end{aligned} \label{EQHessian}
\end{align}
From \eqref{EQchristoffelSymbols12} and \eqref{EQHessian}, we conclude that $\D_{\tilde u}\D_{\tilde u}\tilde v= \D_{\tilde u} \D_{\tilde v}\tilde v = \D_{\tilde u} \D_{\tilde A}\tilde v =0$ and
\begin{align} 
\begin{aligned} 
\D_{\tilde A} \D_{\tilde v} \tilde v= - \pr_{\tilde A} \log \tilde \Om + \tilde{\zeta}_A + \frac{1}{2\tilde\Om} \tilde{\chib}_{AB} \tilde{b}^B, \,\,
\D_{\tilde A} \D_{\tilde B} \tilde v= - \frac{1}{2\tilde\Om} \tilde{\chib}_{AB}.
\end{aligned} \label{EQHessianAnalysis}
\end{align}
Plugging \eqref{EQHessianAnalysis} into \eqref{EQexpressionoMEGA5}, we get that
\begin{align} 
\begin{aligned} 
\om =&  \tilde\om +\tilde{\Om}^3   \tilde{\chib}^{AB}(\pr_A f)(\pr_B f) +\lrpar{2\tilde{\Om}^2 \tilde{\gd}^{AB}\pr_B f } \lrpar{ \pr_{\tilde A} \log \tilde \Om - \tilde{\zeta}_A }.
\end{aligned} \label{EQformulaLittleOM555}
\end{align}

\ni \textbf{Analysis of $\zeta$ and $\eta$.} Using \eqref{EQexpressionOMEGASQUARED123}, we have by explicit computation that
\begin{align} 
\begin{aligned} 
\zeta_A := \half \g(\D_A \widehat{L}, \widehat{\Lb}) = \frac{\tilde{\Om}^2}{\Om^2} \pr_A\pr_u f - \frac{1}{2\Om^2} \lrpar{1+\pr_u f} \g\lrpar{L, \D_A \pr_{\tilde u}} - \pr_A \log \Om. 
\end{aligned} \label{EQzetafirstexpr55}
\end{align}

\ni By \eqref{EQpartialsRelation}, \eqref{EQrelationstildeLL}, \eqref{EQLvectorfieldFormula15} and the geodesic equation for $\tilde{\Lb}'$, we have that
\begin{align} 
\begin{aligned} 
\g\lrpar{L, \D_A \pr_{\tilde u}} =& -2 \tilde{\Om}^2 \lrpar{\tilde{\zeta}_A + \pr_{\tilde A} \log \tilde{\Om} + 2 (\pr_A f) \tilde\omb - \tilde{\Om}\tilde{\gd}^{BC}(\pr_Bf) \tilde{\chib}_{AC}}.
\end{aligned} \label{EQcalculationzeta555}
\end{align}
Plugging \eqref{EQcalculationzeta555} into \eqref{EQzetafirstexpr55}, we get that
\begin{align} 
\begin{aligned} 
\zeta_A =& - \pr_A \log \Om + \frac{\tilde{\Om}^2}{\Om^2} \pr_A\pr_u f \\
&+ \frac{\tilde{\Om}^2}{\Om^2} \lrpar{1+\pr_u f} \lrpar{\tilde{\zeta}_A + \pr_{\tilde A} \log \tilde{\Om} + 2 (\pr_A f) \tilde\omb - \tilde{\Om}\tilde{\gd}^{BC}(\pr_Bf) \tilde{\chib}_{AC}}  
\end{aligned} \label{EQzetavariationf23332}
\end{align}

\ni We conclude from the above that
\begin{align} 
\begin{aligned} 
\eta_A :=& \zeta_A +\pr_A \log \Om \\
=& \frac{\tilde{\Om}^2}{\Om^2} \pr_A\pr_u f + \frac{\tilde{\Om}^2}{\Om^2} \lrpar{1+\pr_u f} \lrpar{\tilde{\zeta}_A + \pr_{\tilde A} \log \tilde{\Om} + 2 (\pr_A f) \tilde\omb - \tilde{\Om}\tilde{\gd}^{BC}(\pr_Bf) \tilde{\chibh}_{AC}} \\
&- \frac{\tilde{\Om}^3}{2\Om^2} \lrpar{1+\pr_u f} (\pr_Af) \tr\tilde{\chib}.
\end{aligned} \label{EQetaFinal555}
\end{align}

\ni \textbf{Analysis of $\chi$.} By \eqref{EQpartialsRelation}, \eqref{EQLvectorfieldFormula15} and \eqref{EQchib5556} we have by explicit computation that
\begin{align} 
\begin{aligned} 
\chi_{AB} := \g(\D_{A} \widehat{L},\pr_B) =& \frac{\tilde{\Om}^3}{\Om}  \vert \Nd f \vert_\gd^2 \tilde{\chib}_{AB} +\frac{\tilde \Om}{\Om}\tilde{\chi}_{AB} + \Om^{-1} (\pr_A f) \g\lrpar{\D_{\tilde \Lb} \tilde L, \pr_{\tilde B}} \\
&+ \Om^{-1} (\pr_A f)(\pr_Bf) \g\lrpar{\D_{\tilde \Lb}\tilde L, \tilde \Lb} +\Om^{-1}(\pr_Bf)\g\lrpar{\D_{\tilde A} \tilde L, \tilde \Lb}\\
&+ \Om^{-1}\g\lrpar{\D_A \lrpar{\lrpar{2\tilde{\Om}^2 \tilde{\gd}^{CD}\pr_C f } \pr_{\tilde D}},\pr_B}.
\end{aligned} \label{EQchiExpr5555}
\end{align}
By \eqref{EQvectorRelations3355}, \eqref{EQpartialsRelation} and \eqref{EQchib5556} we have
\begin{align*} 
\begin{aligned} 
\g\lrpar{\D_{\tilde \Lb} \tilde L, \pr_{\tilde B}} =&2 \tilde{\Om}^2 \tilde{\eta}_B, \,\, \g\lrpar{\D_{\tilde A} \tilde L, \tilde \Lb} = 2\tilde{\Om}^2 \lrpar{\tilde{\eta}_A - 2 \pr_{\tilde A} \log \tilde{\Om}},
\end{aligned} 
\end{align*}

\ni as well as
\begin{align*} 
\begin{aligned} 
\g\lrpar{\D_{\tilde \Lb}\tilde L, \tilde \Lb}=& \tilde{\Om}^2 \g \lrpar{\D_{\tilde \Lb} \tilde L, \tilde{\Lb}'}
= -\tilde{\Om}^2 \g\lrpar{\tilde L, \D_{\tilde \Lb} \tilde{\Lb}'} 
= -\tilde{\Om}^4 \g(\tilde L, \underbrace{\D_{\tilde{\Lb}'} \tilde{\Lb}'}_{=0})
=0,
\end{aligned} 
\end{align*}
and
\begin{align*} 
\begin{aligned} 
\g\lrpar{\D_A \lrpar{2\tilde{\Om} \tilde{\gd}^{CD} (\pr_C f) \pr_{\tilde D} }, \pr_B} 
= \pr_A \lrpar{2\tilde{\Om}^2} \pr_B f + \lrpar{2\tilde{\Om}^2} \lrpar{(\pr_A\pr_B f) + (\pr_Cf) \g\lrpar{\D_A \lrpar{\tilde{\gd}^{CD}\pr_{\tilde D} }, \pr_B}},
\end{aligned} 
\end{align*}
where on the right-hand side we can rewrite with \eqref{EQpartialsRelation} 
\begin{align*} 
\begin{aligned} 
\g\lrpar{\D_A \lrpar{\tilde{\gd}^{CD}\pr_{\tilde D} }, \pr_B} = -\tilde{\Ga}^C_{AB} - (\pr_Af) \tilde{\Om} \tilde{\chib}_{BD} \tilde{\gd}^{CD}- (\pr_Bf) \tilde{\gd}^{CD} \tilde{\Om} \tilde{\chib}_{AD},
\end{aligned} 
\end{align*}
yielding that
\begin{align*} 
\begin{aligned} 
&\g\lrpar{\D_A \lrpar{2\tilde{\Om} \tilde{\gd}^{CD} (\pr_C f) \pr_{\tilde D} }, \pr_B} \\
=& \pr_A \lrpar{2\tilde{\Om}^2} \pr_B f + \lrpar{2\tilde{\Om}^2} \lrpar{(\pr_A\pr_B f) - (\pr_Cf) \tilde{\Ga}^C_{AB}}\\
&- \lrpar{2\tilde{\Om}^2} \lrpar{(\pr_Af)\tilde{\Om}\tilde{\chib}_{BD}\tilde{\gd}^{CD}(\pr_Cf) + (\pr_Bf)\tilde{\Om}\tilde{\chib}_{AD}\tilde{\gd}^{CD}(\pr_Cf)}.
\end{aligned} 
\end{align*}

\ni Plugging the above into \eqref{EQchiExpr5555}, we have that
\begin{align} 
\begin{aligned} 
\chi_{AB} =& \frac{\tilde{\Om}^3}{\Om} \vert \Nd f \vert_\gd^2 \tilde{\chib}_{AB} +\frac{\tilde \Om}{\Om}\tilde{\chi}_{AB} \\
&+\frac{2\tilde{\Om}^2}{\Om} \lrpar{ (\pr_A f)\tilde{\eta}_B + (\pr_B f)\tilde{\eta}_A} +\frac{2\tilde{\Om}^2}{\Om} \lrpar{\pr_A\pr_B f - \tilde{\Ga}^C_{AB} \pr_C f} \\
&- \frac{2\tilde{\Om}^2}{\Om} \lrpar{(\pr_Af)\tilde{\Om}\tilde{\chib}_{BD}\tilde{\gd}^{CD}(\pr_Cf) + (\pr_Bf)\tilde{\Om}\tilde{\chib}_{AD}\tilde{\gd}^{CD}(\pr_Cf)}.
\end{aligned} \label{EQchiFINAL555expr}
\end{align}

\ni \textbf{Analysis of $\Du\omb$.} We have by explicit computation that
\begin{align} 
\begin{aligned} 
\Du \omb := \pr_u \lrpar{\pr_u \log \Om} = \frac{\pr_u^3 f}{2\lrpar{1+\pr_u f}} - \frac{\lrpar{\pr_u^2 f}^2}{2\lrpar{1+\pr_u f}^2}+ \tilde{\Du} \tilde \omb \lrpar{1+\pr_uf}^2 + \tilde\omb \pr_u^2 f.
\end{aligned} \label{EQformulaDUOMB55}
\end{align}

\ni \textbf{Analysis of $D \om$.} Using that (see also \eqref{EQdoubleNULLrelation334})
\begin{align*} 
\begin{aligned} 
L'\lrpar{\Om} = -\frac{\Om^3}{2} L'\lrpar{\Om^{-2}} = -\frac{\Om^3}{2} L'\lrpar{L'(\tilde v)},
\end{aligned} 
\end{align*}
we have that
\begin{align*} 
\begin{aligned} 
D\om = L \lrpar{L \log \Om } = \Om^2 L' \lrpar{\Om L' \lrpar{ \Om}} = \Om^2 L' \lrpar{-\frac{\Om^4}{2}L'\lrpar{L'(\tilde v)}} = 4 \om^2 - \frac{\Om^6}{2} L'\lrpar{L' \lrpar{L'(\tilde v)}}.
\end{aligned} 
\end{align*}
Using that $\D_{L'}L' =0$, it follows further that
\begin{align*} 
\begin{aligned} 
D\om = & 4 \om^2 - \frac{\Om^6}{2} \D_{L'} \D_{L'} \D_{L'} \tilde v= 4\om^2 - \frac{1}{2} \D_L \D_L \D_L \tilde v.
\end{aligned} 
\end{align*}
By \eqref{EQLvectorfieldFormula15}, we can furthermore expand $\D_L \D_L \D_L \tilde v$ (explicit calculation omitted here), to conclude that $D\om$ can be written as a sum of products of first angular derivatives of $f$ and (the following all with tilde) null curvature components, first derivatives of Ricci coefficients and second derivatives of metric coefficients.

\begin{remark} The only linear terms in $f$ in the expression for $D\om$ are
\begin{align*} 
\begin{aligned} 
2 \lrpar{2\tilde{\Om}^2 \tilde{\gd}^{AB}(\pr_Af)} \D_{\tilde L} \D_{\tilde L} \D_{\tilde B} \tilde v \text{ and } 2\tilde{\Om}^2 \tilde{\gd}^{EF}(\pr_E f) \D_{\tilde F} \D_{\tilde L} \D_{\tilde L} \tilde v,
\end{aligned} 
\end{align*}
and we note that at Minkowski,
\begin{align*} 
\begin{aligned} 
\D_{\tilde L}\D_{\tilde L}\D_{\tilde L}\tilde v= \D_{\tilde L} \D_{\tilde L} \D_{\tilde B} \tilde v = \D_{\tilde F} \D_{\tilde L} \D_{\tilde L} \tilde v =0.
\end{aligned} 
\end{align*}
Hence the linearization of $D\om$ vanishes at Minkowski.
\end{remark}

\subsubsection{Calculation of null curvature components on $\protect\tilde{\protect\underline{\protect\mathcal{H}}}_2$} \label{SECnullCurvatureComponentsAPP} We recall from \eqref{EQnullcurvatureCOMPDEF} the definition of the null curvature components,
\begin{align} \begin{aligned}
\alpha_{AB} :=& \Rbf(\pr_A,\widehat{L}, \pr_B, \widehat{L}), & \beta_A :=& \half \Rbf(\pr_A, \widehat{L},\widehat{\Lb},\widehat{L}), & \rh :=& \frac{1}{4} \Rbf(\widehat{\Lb}, \widehat{L}, \widehat{\Lb}, \widehat{L}), \\
 \sigma \iin_{AB} :=& \half \Rbf(\pr_A,\pr_B,\widehat{\Lb}, \widehat{L}), & 
\beb_A :=& \half \Rbf(\pr_A, \widehat{\Lb},\widehat{\Lb},\widehat{L}), & \ab_{AB} :=& \Rbf(\pr_A,\widehat{\Lb}, \pr_B, \widehat{\Lb}).
\end{aligned} \label{EQredefinitionNullCurvatureComp} \end{align}

\ni Plugging \eqref{EQpartialsRelation}, \eqref{EQrelationstildeLL}, \eqref{EQexpressionOMEGASQUARED123} and \eqref{EQLvectorfieldFormula15}, that is
\begin{align*} 
\begin{aligned} 
\pr_{\th^A} =& \pr_{\tth^A} + (\pr_{\th^A} f) \pr_{\tilde u}, &
\Om^2 =& \tilde{\Om}^2 \lrpar{1+\pr_u f},  \\
\Lb =& \lrpar{1+\pr_u f}\tilde \Lb, &
L =& \lrpar{ \tilde{\Om}^2 \vert \Nd f\vert^2_\gd} \tilde{\Lb} + \tilde{L} + \lrpar{2\tilde{\Om}^2 \tilde{\gd}^{AC}\pr_C f } \pr_{\tth^A},
\end{aligned} 
\end{align*}
into \eqref{EQredefinitionNullCurvatureComp}, it follows that the null curvature components $\lrpar{\a, \be, \rh, \si, \beb, \ab}$ can be expressed as sum of products of $\lrpar{\tilde \a, \tilde \be, \tilde \rh, \tilde \si, \tilde \beb, \tilde \ab}$ and $f, \pr_A f$, $A=1,2$, and $\pr_u f$.

\subsection{Angular perturbations} \label{SECapp1232} In this section we discuss in detail \emph{angular perturbations} of sphere data. Consider on a sphere $S_{0,2}$ with coordinates $(\tth^1,\tth^2)$ the sphere data
\begin{align*} 
\begin{aligned} 
x_{0,2} = \lrpar{\Om, \gd, \Om\trchi, \chih, \Om\trchib, \chibh, \eta, \om, D\om, \omb, \Du\omb, \a, \ab}.
\end{aligned} 
\end{align*}
In the following, we change coordinates from $(\tth^1,\tth^2)$ to new coordinates $(\th^1,\th^2)$ and analyse the change of the sphere data on $S_{0,2}$.

For scalar functions $j^1(\th^1,\th^2)$ and $j^2(\th^1,\th^2)$, define $(\th^1,\th^2)$ on $S_{0,2}$ by
\begin{align} 
\begin{aligned} 
\tth^1 = \th^1 + j^1(\th^1,\th^2), \,\, \tth^2 = \th^2 + j^2(\th^1,\th^2).
\end{aligned} \label{EQdefCoordinateChangeSpheres}
\end{align}
For $(j_1, j_2)$ sufficiently small, $(\th^1,\th^2)$ are local coordinates on $S_{0,2}$. The Jacobian of this coordinate change is given by
\begin{align} 
\begin{aligned} 
J := \begin{pmatrix}
\frac{\pr \tth^1}{\pr \th^1} & \frac{\pr\tth^1}{\pr \th^2}  \\
\frac{\pr\tth^2}{\pr\th^1} & \frac{\pr\tth^2}{\pr \th^2} 
\end{pmatrix} 
=\begin{pmatrix}
1+ \pr_1 j^1 & \pr_2 j^1  \\
\pr_1 j^2 & 1+ \pr_2 j^2 
\end{pmatrix}.
\end{aligned} \label{EQJacobianformula}
\end{align}

\ni First, the scalar functions
\begin{align*} 
\begin{aligned} 
\lrpar{\Om, \Om\trchi, \Om\trchib, \om, D\om, \omb, \Du\omb}
\end{aligned} 
\end{align*}
do not transform under the change \eqref{EQdefCoordinateChangeSpheres}.

Second, the $1$-forms $\be$ and $\beb$ transform under the tensor transformation law for $1$-forms: Let $X$ be a $1$-form on $S_{0,2}$. Denote by $X_{A}$ and $\tilde{X}_{A}$ its coordinate components with respect to $(\th^1,\th^2)$ and $(\tth^1,\tth^2)$, respectively. Then, using \eqref{EQJacobianformula}, it holds that
\begin{align*} 
\begin{aligned} 
\begin{pmatrix} 
X_1 \\ X_2
\end{pmatrix} = J^T \begin{pmatrix} 
\tilde{X}_1 \\ \tilde{X}_2
\end{pmatrix}= \begin{pmatrix}
1+ \pr_1 j^1 & \pr_1 j^2  \\
\pr_2 j^1 & 1+ \pr_2 j^2 
\end{pmatrix} \begin{pmatrix} 
\tilde{X}_1 \\ \tilde{X}_2
\end{pmatrix}.
\end{aligned} 
\end{align*}

\ni Third, the $2$-tensors
\begin{align*} 
\begin{aligned} 
\lrpar{\chih, \chibh, \a, \ab}
\end{aligned} 
\end{align*}
change under \eqref{EQdefCoordinateChangeSpheres} according to the tensor transformation law: Let $T$ be a symmetric $2$-tensor on $S_{0,2}$. Denote by $T_{AB}$ and $\tilde{T}_{AB}$ its coordinate components with respect to $(\th^1,\th^2)$ and $(\tth^1,\tth^2)$, respectively. Then, using \eqref{EQJacobianformula}, it holds that
\begin{align*} 
\begin{aligned} 
\begin{pmatrix}
T_{11} & T_{12}  \\
T_{12} & T_{22} 
\end{pmatrix}= J^T \begin{pmatrix}
\tilde{T}_{11} & \tilde{T}_{12}  \\
\tilde{T}_{12} & \tilde{T}_{22} 
\end{pmatrix} J =
\begin{pmatrix}
1+ \pr_1 j^1 & \pr_1 j^2  \\
\pr_2 j^1 & 1+ \pr_2 j^2 
\end{pmatrix}\begin{pmatrix}
\tilde{T}_{11} & \tilde{T}_{12}  \\
\tilde{T}_{12} & \tilde{T}_{22} 
\end{pmatrix}\begin{pmatrix}
1+ \pr_1 j^1 & \pr_2 j^1  \\
\pr_1 j^2 & 1+ \pr_2 j^2 
\end{pmatrix}.
\end{aligned} 
\end{align*}

\ni Specifically, we note that the components of the induced metric $\gd$ change according to
\begin{align} 
\begin{aligned} 
\begin{pmatrix}
\gd_{11} & \gd_{12}  \\
\gd_{12} & \gd_{22} 
\end{pmatrix}= 
\begin{pmatrix}
1+ \pr_1 j^1 & \pr_1 j^2  \\
\pr_2 j^1 & 1+ \pr_2 j^2 
\end{pmatrix}\begin{pmatrix}
\tilde{\gd}_{11} & \tilde{\gd}_{12}  \\
\tilde{\gd}_{12} & \tilde{\gd}_{22} 
\end{pmatrix}\begin{pmatrix}
1+ \pr_1 j^1 & \pr_2 j^1  \\
\pr_1 j^2 & 1+ \pr_2 j^2 
\end{pmatrix},
\end{aligned} \label{EQchangeofInducedMetric}
\end{align}
which implies that
\begin{align*} 
\begin{aligned} 
\sqrt{\det \gd} =& \lrpar{(1+ \pr_1 j^1)(1+ \pr_2 j^2)- \pr_2 j^1\pr_1 j^2} \sqrt{\det \tilde{\gd}}\\
=&\lrpar{1+ 2(\pr_1j^1 +\pr_2 j^2)+ (\pr_1j^1)^2 + (\pr_2j^2)^2 -\pr_2 j^1\pr_1 j^2} \sqrt{\det \tilde{\gd}}.
\end{aligned} 
\end{align*}
\subsection{Proof of Propositions \ref{PropositionSmoothnessF} and \ref{PropositionSmoothnessF2}} \label{APPproofSmoothnessFINAL} In this section we prove Propositions \ref{PropositionSmoothnessF} and \ref{PropositionSmoothnessF2}. We show that for given real numbers $\de>0$, 

\begin{enumerate}
\item the mapping
\begin{align*} 
\begin{aligned} 
\PP_f: \, \XX^+\lrpar{\HHb_{ [-\de,\de],2}} \times \YY_f &\to \XX(S_2), \\
(\underline{x},f) &\mapsto \PP_f(\underline{x}),
\end{aligned} 
\end{align*}
is well-defined and smooth in an open neighbourhood of
$$(\underline{x},f)=(\underline{\mathfrak{m}},0) \in \XX^+\lrpar{\HHb_{ [-\de,\de],2}} \times \YY_f,$$ 
\item the mapping
\begin{align*} 
\begin{aligned} 
\PP_{(j_1,j_2)}: \, \XX(S_{0,2}) \times \YY_{(j_1,j_2)} &\to \XX(S_{0,2}), \\
(x_{0,2},(j_1,j_2)) &\mapsto \PP_{(j_1,j_2)}(x_{0,2}),
\end{aligned} 
\end{align*}
is well-defined and smooth in an open neighbourhood of
$$(x_{0,2}, {(j_1,j_2)})=(\mathfrak{m}_{0,2},(0,0)) \in \XX(S_{0,2}) \times \YY_{(j_1,j_2)}.$$ 
\end{enumerate}
The explicit estimates for $\PP_f$ and $\PP_{(j_1,j_2)}$, see \eqref{EQestimatePPfsmoothness1stder} and \eqref{EQestimatePPJJsmoothness1stder}, subsequently follow from the above by Lemma \ref{LEMoperatorEstimates}.

We are now in position to prove (1) and (2) above. First, it follows directly from the formulas of Sections \ref{SECapp1231} and \ref{SECapp1232} that the formal expression for $\PP_f$ and $\PP_{(j_1,j_2)}$ are well-defined locally.

Second, in the following, let $\varep>0$ be a real number, and assume that
\begin{align*} 
\begin{aligned} 
\Vert \underline{x}-\underline{\mathfrak{m}} \Vert_{\XX^+\lrpar{\HHb_{ [-\de,\de],2}}} + \Vert x_{0,2}-\mathfrak{m}_{0,2} \Vert_{\XX(S_{0,2})} + \Vert f \Vert_{\YY_f} + \Vert (j_1,j_2) \Vert_{\YY_{(j_1,j_2)}} \leq \varep.
\end{aligned} 
\end{align*}
In Sections \ref{SECEQnonlinearESTIM1} and \ref{SECEQnonlinearESTIM2} below we show that for $\varep>0$ sufficiently small,
\begin{align} 
\Vert \PP_f(\underline{x}) - \mathfrak{m}_{0,2} \Vert_{\XX(S_{0,2})} \les& \Vert \underline{x} -\underline{\mathfrak{m}} \Vert_{\XX^+\lrpar{\HHb_{ [-\de,\de],2}}} + \Vert f \Vert_{\YY_f}, \label{EQnonlinearESTIM1} \\
\Vert \PP_{(j_1,j_2)}(x_{0,2})-\mathfrak{m}_{0,2} \Vert_{\XX(S_{0,2})} \les& \Vert x_{0,2}-\mathfrak{m}_{0,2} \Vert_{\XX(S_{0,2})} + \Vert (j_1,j_2) \Vert_{\YY_{(j_1,j_2)}}, \label{EQnonlinearESTIM2}
\end{align}
The estimates \eqref{EQnonlinearESTIM1} and \eqref{EQnonlinearESTIM2} establish that $\PP_f$ and $\PP_{(j^1,j^2)}$ are well-defined mappings between the indicated Hilbert spaces. 

Third, the smoothness of the mappings $\PP_f$ and $\PP_{(j^1,j^2)}$ follows from the explicit formulas of Sections \ref{SECapp1231} and \ref{SECapp1232}.

\subsubsection{Proof of Proposition \ref{PropositionSmoothnessF}} \label{SECEQnonlinearESTIM1} In the following give the proof of \eqref{EQnonlinearESTIM1} by using the formulas of Section \ref{SECapp1231} and standard product estimates. By Definition \ref{DEFnormFirstOrderDATA} of $\XX(S_{0,2})$, we have to show that for $\varep>0$ sufficiently small,
\begin{align} 
\begin{aligned} 
& \Vert \Om-1 \Vert_{H^{6}(S_{0,2})}+\Vert \gd-r^2\gac \Vert_{H^{6}(S_{0,2})} \\
&+ \left\Vert \trchi-\frac{2}{r} \right\Vert_{H^{6}(S_{0,2})} + \Vert \chih \Vert_{H^{6}(S_{0,2})} + \left\Vert \trchib+\frac{2}{r} \right\Vert_{H^{4}(S_2)} + \Vert \chibh \Vert_{H^{4}(S_{0,2})} \\
&+ \Vert \eta \Vert_{H^{5}(S_{0,2})} + \Vert \om \Vert_{H^{6}(S_{0,2})}+ \Vert D\om \Vert_{H^{6}(S_{0,2})}+\Vert \omb \Vert_{H^{4}(S_{0,2})} + \Vert \Du\omb \Vert_{H^{2}(S_{0,2})} \\
&+ \Vert \a \Vert_{H^{6}(S_{0,2})} +\Vert \ab \Vert_{H^{2}(S_{0,2})}\\
\les& \Vert \underline{x} \Vert_{\XX^+\lrpar{\HHb_{ [-\de,\de],2}}} + \Vert f \Vert_{\YY_f}.
\end{aligned} \label{EQestimatesPPfbig1}
\end{align}

\ni We derive the control of $\Om$ as an instructive example. By \eqref{EQdefU} and \eqref{EQexpressionOMEGASQUARED123}, it holds that on 
\begin{align*} 
\begin{aligned} 
S_{0,2} := \{ u=0, v=2\},
\end{aligned} 
\end{align*}
the null lapse $\Om$ can be expressed as 
\begin{align*} 
\begin{aligned} 
\Om^2(0, \th^1,\th^2) =& \lrpar{1+\pr_u f(0,\th^1,\th^2)} \tilde{\Om}^2(f(0,\th^1,\th^2),\th^1,\th^2)).
\end{aligned} 
\end{align*}
By Definition \ref{DEFspacetimeNORM} and standard product estimates, it holds for $\varep>0$ sufficiently small that
\begin{align*} 
\begin{aligned} 
\Vert \Om-1 \Vert_{H^6(S_2)} \les& \Vert f(0) \Vert_{H^8(S_2)} + \Vert \pr_u f(0) \Vert_{H^6(S_2)} + \Vert \tilde \Om -1 \Vert_{\XX^+\lrpar{\HHb_{ [-\de,\de],2}}} \\
\les& \Vert f \Vert_{\YY_f} + \Vert \underline{x} \Vert_{\XX^+\lrpar{\HHb_{ [-\de,\de],2}}}.
\end{aligned} 
\end{align*}

\ni Similarly, from the expressions for 
\begin{align*} 
\begin{aligned} 
\lrpar{\gd, \Om\trchi, \chih, \chibh, \Om\trchib, \eta, \om, D\om, \omb, \Du\omb, \a, \be, \rh, \si, \beb, \ab}
\end{aligned} 
\end{align*}
provided in \eqref{EQinducedmetricFormula15}, \eqref{EQombExpresssion4}, \eqref{EQformulatrchichih555}, \eqref{EQformulaLittleOM555}, \eqref{EQetaFinal555}, \eqref{EQchiFINAL555expr}, \eqref{EQformulaDUOMB55} and Section \ref{SECnullCurvatureComponentsAPP}, the rest of \eqref{EQestimatesPPfbig1} follows.

\subsubsection{Proof of Proposition \ref{PropositionSmoothnessF2}} \label{SECEQnonlinearESTIM2}
In the following we prove \eqref{EQnonlinearESTIM2} by using the formulas of Section \ref{SECapp1232} and standard product estimates. By Definition \ref{DEFnormFirstOrderDATA} of $\XX(S_{0,2})$, we have to show that for $\varep>0$ sufficiently small,
\begin{align} 
\begin{aligned} 
& \Vert \Om-1 \Vert_{H^{6}(S_{0,2})}+\Vert \gd -r^2 \gac \Vert_{H^{6}(S_{0,2})} \\
&+ \left\Vert \trchi-\frac{2}{r} \right\Vert_{H^{6}(S_{0,2})} + \Vert \chih \Vert_{H^{6}(S_2)} + \left\Vert \trchib+\frac{2}{r} \right\Vert_{H^{4}(S_{0,2})} + \Vert \chibh \Vert_{H^{4}(S_2)} \\
&+ \Vert \eta \Vert_{H^{5}(S_{0,2})} + \Vert \om \Vert_{H^{6}(S_{0,2})}+ \Vert D\om \Vert_{H^{6}(S_{0,2})}+\Vert \omb \Vert_{H^{4}(S_{0,2})} + \Vert \Du\omb \Vert_{H^{2}(S_{0,2})} \\
&+ \Vert \a \Vert_{H^{6}(S_{0,2})} +\Vert \ab \Vert_{H^{2}(S_{0,2})}\\
\les& \Vert x_{0,2} \Vert_{\XX(S_{0,2})} + \Vert (j^1,j^2) \Vert_{\YY_{(j^1,j^2)}}.
\end{aligned} \label{EQestimatesPPfbig12222}
\end{align}
We derive the control of $\gd$ as an instructive example. By \eqref{EQchangeofInducedMetric}, it holds that
\begin{align*} 
\begin{aligned} 
\begin{pmatrix}
\gd_{11} & \gd_{12}  \\
\gd_{12} & \gd_{22} 
\end{pmatrix}= \begin{pmatrix}
1+ \pr_1 j^1 & \pr_1 j^2  \\
\pr_2 j^1 & 1+ \pr_2 j^2 
\end{pmatrix}\begin{pmatrix}
\tilde{\gd}_{11} & \tilde{\gd}_{12}  \\
\tilde{\gd}_{12} & \tilde{\gd}_{22} 
\end{pmatrix}\begin{pmatrix}
1+ \pr_1 j^1 & \pr_2 j^1  \\
\pr_1 j^2 & 1+ \pr_2 j^2 
\end{pmatrix}.
\end{aligned} 
\end{align*}
Hence by standard product estimates, it holds that for $\varep_1>0$ sufficiently small,
\begin{align*} 
\begin{aligned} 
\Vert \gd - r^2 \gac \Vert_{H^6(S_{0,2})} \les& \Vert j^1 \Vert_{H^7(S_{0,2})}+\Vert j^1 \Vert_{H^7(S_{0,2})} + \Vert \tilde{\gd}_{AB} - r^2 \gac \Vert_{H^6(S_{0,2})} \\
\les& \Vert (j^1,j^2) \Vert_{\YY_{(j^1,j^2)}}+ \Vert x_{0,2} \Vert_{\XX(S_{0,2})}.
\end{aligned} 
\end{align*}
By the arguments of Section \ref{SECapp1232}, the rest of \eqref{EQestimatesPPfbig12222} is derived similarly. This finishes the proof of Proposition \ref{PropositionSmoothnessF2}.

\section{Derivation of null transport equations} \label{SECnulltransportAPPENDIX} \ni In this section we prove null transport equations used in this paper. In Section \ref{SECappendixDerivationTransveralEquations} we prove the non-linear null transport equation \eqref{EQDUOMU1} for $\Du\omb$ along $\HH$. In Section \ref{SECDerivation1} we derive the linearized null transport equations of Lemma \ref{LEMtransportEQUStrans} for $\ombd, \Du\ombd$ and $\abd$.

\subsection{Derivation of null transport equation for $\protect\Du \protect\omb$} \label{SECappendixDerivationTransveralEquations}

\ni In this section, we prove the transport equation \eqref{EQDUOMU1} for $\Du \omb$ along $\HH$. We remark that in case of a geodesic foliation on $\HH=\HH_{0,[1,2]}$, that is, $\Om \equiv 1$ on $\HH$, this transport equation is readily available in  \cite{ChrFormationBlackHoles}. We first have the following commutator identities, see Chapter 1 in \cite{ChrFormationBlackHoles}.

\begin{lemma}[Commutator identity] \label{LemmaCommutatorIdentity} Let $W$ be an $S_v$-tangent tensorfield. Then,
\begin{align*}
\Du D W - D \Du W = \Lied_{4 \Om^2 \zeta} W.
\end{align*}
\end{lemma}

\ni We are now in position to derive the null transport equation for $\Du \omb$. From the null structure equations \eqref{EQtransportEQLnullstructurenonlinear}, we have that
\begin{align}
D \omb = \Om^2\lrpar{2 (\eta,\etab) - \vert \eta \vert^2 -\rh}. \label{EQtransportomb5551}
\end{align}
Applying the $\Du$-derivative to \eqref{EQtransportomb5551} and using \eqref{EQfirstvariation1}, \eqref{EQtransportEQLnullstructurenonlinear}, \eqref{EQnullBianchiEquations} and Lemma \ref{LemmaCommutatorIdentity} with $W=\omb$, we have that
\begin{align*} 
\begin{aligned} 
D\Du\omb =& -4\Om^2 \zeta(\omb) + \Du D\omb\\
=& -4\Om^2 \zeta(\omb)+ \Du\lrpar{\Om^2 \lrpar{2(\eta,\etab)-\vert \eta\vert^2 - \rh}}\\ 
=&-4\Om^2 \zeta(\omb) +2 \Om^2 \omb \lrpar{2(\eta,\etab)- \vert \eta\vert^2 - \rh} \\
&+\Om^2\lrpar{-4\Om \chib(\eta,\etab) + 2 \lrpar{-\Om\lrpar{\chib \cdot \eta + \beb}+ 2 \di \omb, \etab} +2 \lrpar{\eta, \Om \lrpar{\chib \cdot \etab + \beb}} }\\
&+ \Om^2 \lrpar{2\Om \chib(\eta, \eta) - 2\lrpar{\Om \lrpar{\chib \cdot \eta + \beb}+ 2\di \omb, \eta} - \Du \rh} \\
=& -4\Om^2 \zeta(\omb) +2 \Om^2 \omb \lrpar{2(\eta,\etab)- \vert \eta\vert^2 - \rh} \\
&-6 \Om^3 \chib(\eta, \etab)-2\Om^3 (\beb,\etab) +4 \Om^2 \etab(\omb)\\
&+ 6 \Om^3 \chib(\eta, \eta) + 4 \Om^3 (\beb, \eta) - 4\Om^2 \eta(\omb) \\
&+ \frac{3}{2} \Om^3\trchib\rh+\Om^3 \lrpar{\Divd \beb +(2\eta-\zeta,\beb)+ \half (\chih, \ab)},
\end{aligned} 
\end{align*}
which, using \eqref{DEFricciCoefficients} and \eqref{EQriccirelationetabeta}, can be rewritten as
\begin{align*} 
\begin{aligned} 
D\Du\omb =& 12 \Om^3 (\di\log\Om-\eta)\omb +2 \Om^2 \omb \lrpar{2(\eta,\etab)- \vert \eta\vert^2} + \rh \lrpar{\frac{3}{2}\Om^3\trchib -2\Om^2 \omb} \\
&+12 \Om^3 \chib(\eta, \eta-\di\log\Om) + \Om^3 (\beb, 7\eta-3\di\log\Om)+ \Om^3 \Divd \beb + \half \Om^3 (\chih, \ab).
\end{aligned} 
\end{align*}
This finishes the proof of \eqref{EQDUOMU1}.

\subsection{Derivation of transport equations for $\protect\ombd$, $\protect\abd$ and $\protect\Du\protect\ombd$} \label{SECDerivation1} In this section we prove Lemma \ref{LEMtransportEQUStrans}. To simplify notation, we use that in Minkowski on $\HH=\HH_{0,[1,2]}$ it holds that $r=v$. First we recall from \eqref{EQdefChargesMinkowski8891} that
\begin{align*} \begin{aligned} 
\QQ_1 :=&\frac{v}{2} \lrpar{\omtrchid-\frac{4}{v}\Omd} + \frac{\phid}{v}, \\
\QQ_2 :=& v^2 \omtrchibd -\frac{2}{v}\Divdo \lrpar{v^2\etad+\frac{v^3}{2}\di\lrpar{\omtrchid-\frac{4}{v}\Omd}} -v^2 \lrpar{\omtrchid-\frac{4}{v}\Omd} +2v^3 \Kd, \\
\QQ_3:=& \frac{\chibhd}{v} -\half \lrpar{ \DDd_2^\ast \Divdo +1} \gdcd + \DDd_2^\ast \lrpar{ \etad + \frac{v}{2}\di \lrpar{\omtrchid-\frac{4}{v}\Omd}} - v \DDd_2^\ast \di \lrpar{\omtrchid-\frac{4}{v}\Omd}.
\end{aligned} \end{align*}
and that by Lemmas \ref{CORtransport1EQS} and \ref{LEMgaugeconservationlaws},
\begin{align*} 
\begin{aligned} 
D\QQ_1    =&\mfq_1-D\lrpar{\frac{1}{2v} \mfq_2}, \\
D\QQ_2 =& v^2 \mfq_5 -2v \Divdo \mfq_4 -2v(\Ldo+1)\mfq_1+ D\lrpar{(\Ldo+1)\mfq_2} -\frac{1}{v}(\Ldo+2) \mfq_2+ \frac{1}{v}\Divdo\Divdo \mfq_3, \\
D\QQ_3=& \frac{1}{v}\mfq_6 -\frac{1}{2v^2}\lrpar{\DDd_2^\ast \Divdo +1}{\mfq_3}+ \DDd_2^\ast \mfq_4-\DDd_2^\ast \di \mfq_1 +D\lrpar{\frac{1}{2v}\DDd_2^\ast \di \mfq_2}+ \frac{1}{2v^2}\DDd_2^\ast \di \mfq_2.
\end{aligned} 
\end{align*}

\ni \textbf{Transport equation for $\ombd$.} We have that
\begin{align*} 
\begin{aligned} 
&D\lrpar{ \ombd +\frac{1}{4v^2}\QQ_2 +\frac{1}{3v} \Divdo \lrpar{\etad+\frac{v}{2}\di\lrpar{\omtrchid-\frac{4}{v}\Omd}}} \\
=& \mfq_7 + \Kd + \frac{1}{2v} \omtrchibd -\frac{1}{2v} \omtrchid + \frac{2}{v^2} \Omd + \frac{1}{4v^2}D\QQ_2 \\
&-\frac{1}{2v^3} \lrpar{v^2\omtrchibd -\frac{2}{v}\Divdo \lrpar{v^2\etad+\frac{v^3}{2}\di \lrpar{\omtrchid-\frac{4}{v}\Omd}}-v^2 \lrpar{\omtrchid-\frac{4}{v}\Omd}+2v^3 \Kd}\\
&-\frac{1}{v^2} \Divdo \lrpar{\etad +\frac{v}{2}\di\lrpar{\omtrchid-\frac{4}{v}\Omd}}\\
&+\frac{1}{3v^3} \Divdo \lrpar{\Divdo \chihd +v^2\mfq_4+v^2 \di\mfq_1+\half \di \mfq_2-D\lrpar{\frac{v}{2}\di\mfq_2}} \\
=& \mfq_7+ \frac{1}{4v^2}D\QQ_2+\frac{1}{3v^3} \Divdo \lrpar{\Divdo \chihd +v^2\mfq_4+v^2 \di\mfq_1+\half \di \mfq_2-D\lrpar{\frac{v}{2}\di\mfq_2}}\\
=& \mfq_7+ \frac{1}{4v^2}D\QQ_2+\frac{1}{3v^3} \Divdo \lrpar{\Divdo \chihd +v^2\mfq_4+v^2 \di\mfq_1 - \di \mfq_2}-D\lrpar{\frac{1}{6v^2}\Ldo \mfq_2} \\
=&  \mfq_7+ \frac{1}{4}\mfq_5 + \frac{1}{4v^3} \Divdo \Divdo \mfq_3 - \frac{1}{6v} \Divdo \mfq_4-\frac{1}{6v}(\Ldo+3) \mfq_1 -\frac{1}{12v^3} \Ldo \mfq_2 + \frac{1}{3v^3} \Divdo \Divdo \chihd.
\end{aligned} 
\end{align*}

\ni \textbf{Transport equation for $\abd$.} Using the above definition of $\QQ_1$, $\QQ_2$ and $\QQ_3$, we have by the system \eqref{EQlinearizedOPsystem2},
\begin{align*} 
\begin{aligned} 
&D\left( \frac{\abd}{v} +\frac{2}{v}\DDd_2^\ast \Divdo \QQ_3 - \frac{1}{2v^2} \DDd_2^\ast \di \QQ_2 - \frac{2}{v} \DDd_2^\ast \di \lrpar{\Ldo+2}\QQ_1 \right)\\
=& \frac{1}{v}\mfq_8 +\frac{2}{v}\DDd_2^\ast \lrpar{\frac{1}{v^2}\Divdo \chibhd -\half \di \omtrchibd-\frac{1}{v}\etad}-\frac{2}{v^2}\DDd_2^\ast \Divdo \QQ_3 +\frac{2}{v}\DDd_2^\ast \Divdo \lrpar{D\QQ_3}\\
&+\frac{1}{v^3}\DDd_2^\ast \di \QQ_2 -\frac{1}{2v^2}\DDd_2^\ast \di \lrpar{D\QQ_2}+\frac{2}{v^2}\DDd_2^\ast \di \lrpar{\Ldo+2} \QQ_1 -\frac{2}{v}\DDd_2^\ast \di \lrpar{\Ldo+2} \lrpar{D\QQ_1}\\
=& \frac{1}{v}\mfq_8 +\frac{2}{v}\DDd_2^\ast \lrpar{\frac{1}{v^2}\Divdo \chibhd -\half \di \omtrchibd-\frac{1}{v}\etad} \\
&+\frac{2}{v}\DDd_2^\ast \Divdo \lrpar{D\QQ_3} -\frac{1}{2v^2}\DDd_2^\ast \di \lrpar{D\QQ_2} -\frac{2}{v}\DDd_2^\ast \di\lrpar{\Ldo+2} \lrpar{D\QQ_1}\\
&-\frac{2}{v^2}\DDd_2^\ast \Divdo \lrpar{ \frac{\chibhd}{v} -\half \lrpar{ \DDd_2^\ast \Divdo +1} \gdcd + \DDd_2^\ast \lrpar{ \etad + \frac{v}{2}\di \lrpar{\omtrchid-\frac{4}{v}\Omd}} - v \DDd_2^\ast \di \lrpar{\omtrchid-\frac{4}{v}\Omd}}\\
&+\frac{1}{v^3}\DDd_2^\ast \di \lrpar{v^2 \omtrchibd -\frac{2}{v}\Divdo \lrpar{v^2\etad+\frac{v^3}{2}\di\lrpar{\omtrchid-\frac{4}{v}\Omd}} -v^2 \lrpar{\omtrchid-\frac{4}{v}\Omd} +2v^3 \Kd}\\
&+\frac{2}{v^2}\DDd_2^\ast \di \lrpar{\Ldo+2} \lrpar{\frac{v}{2} \lrpar{\omtrchid-\frac{4}{v}\Omd} + \frac{\phid}{v}},
\end{aligned} 
\end{align*}
and
\begin{align*} 
\begin{aligned} 
&-\frac{2}{3}\DDd_2^\ast \lrpar{\Divdo \DDd_2^\ast +1 + \di\Divdo}D\lrpar{\frac{1}{v^3}\lrpar{v^2\etad+\frac{v^3}{2}\di\lrpar{\omtrchid-\frac{4}{v}\Omd}}}\\
=&\frac{2}{v^2} \DDd_2^\ast \lrpar{\Divdo \DDd_2^\ast +1 + \di\Divdo} \lrpar{\etad+\frac{v}{2}\di\lrpar{\omtrchid-\frac{4}{v}\Omd}}\\
&-\frac{2}{3v^3} \DDd_2^\ast \lrpar{\Divdo \DDd_2^\ast +1 + \di\Divdo}\lrpar{\Divdo \chihd + v^2\mfq_4+v^2 \di\mfq_1+\half \di \mfq_2-D\lrpar{\frac{v}{2}\di\mfq_2}}.
\end{aligned} 
\end{align*}
and
\begin{align*} 
\begin{aligned} 
& \DDd_2^\ast \lrpar{\Divdo \DDd_2^\ast + 1 + \di\Divdo} \Divdo D\lrpar{\frac{1}{v}\gdcd}\\
=& -\frac{1}{v^2}\DDd_2^\ast \lrpar{\Divdo \DDd_2^\ast + 1 + \di\Divdo} \Divdo \gdcd + \frac{1}{v} \DDd_2^\ast \lrpar{\Divdo \DDd_2^\ast + 1 + \di\Divdo} \Divdo \lrpar{\frac{2}{v^2}\chihd+\frac{1}{v^2}{\mfq_3}}.
\end{aligned} 
\end{align*}
Summing up the above and using that, see Section \ref{SECspecAnalysis},
\begin{align*} 
\begin{aligned} 
2\Divdo \DDd_2^\ast \di + \di (\Ldo+2)= 2 \lrpar{\Divdo \DDd_2^\ast + \half \di \Divdo+ 1}\di =0
\end{aligned} 
\end{align*}
yields the transport equation for $\abd$ in Lemma \ref{LEMtransportEQUStrans}.\\

\ni \textbf{Transport equation for $\Du \ombd$.} We have that
\begin{align*} 
\begin{aligned} 
&D\lrpar{ \Du\ombd -\frac{1}{6v^3} \lrpar{\Ldo-3}\QQ_2+ \frac{1}{2v^2} \Divdo \Divdo \QQ_3 +\frac{1}{v^2} \Divdo \Divdo \DDd_2^\ast \di \QQ_1}\\
=&\mfq_9 +\frac{3}{v} \lrpar{\Kd+\frac{1}{2v}\omtrchibd-\frac{1}{2v}\omtrchid + \frac{2}{v^2}\Omd} + \frac{1}{v^2} \Divdo \lrpar{\frac{1}{v^2}\Divdo \chibhd -\frac{1}{v}\etad-\half \di \omtrchibd}\\
&-\frac{1}{6v^3}(\Ldo-3)D\QQ_2 +\frac{1}{2v^2}\Divdo \Divdo D\QQ_3 +\frac{1}{v^2}\Divdo \Divdo \DDd_2^\ast \di D\QQ_1\\
&+\frac{1}{2v^4}(\Ldo-3)\lrpar{v^2\omtrchibd -\frac{2}{v}\Divdo \lrpar{v^2\etad+\frac{v^3}{2}\di \lrpar{\omtrchid-\frac{4}{v}\Omd}}-v^2 \lrpar{\omtrchid-\frac{4}{v}\Omd}+2v^3 \Kd}\\
&-\frac{1}{v^3} \Divdo \Divdo \lrpar{\frac{\chibhd}{v} -\half \lrpar{ \DDd_2^\ast \Divdo +1} \gdcd + \DDd_2^\ast \lrpar{ \etad + \frac{v}{2}\di \lrpar{\omtrchid-\frac{4}{v}\Omd}} - v \DDd_2^\ast \di \lrpar{\omtrchid-\frac{4}{v}\Omd}}\\
&-\frac{2}{v^3} \Divdo \Divdo \DDd_2^\ast \di \lrpar{\frac{v}{2}\lrpar{\omtrchid-\frac{4}{v}\Omd}+\frac{\phid}{v}},
\end{aligned} 
\end{align*}
as well as
\begin{align*} 
\begin{aligned} 
&-D\lrpar{ \frac{1}{4v^2} \lrpar{\lrpar{\Ldo-2}\Divdo+ \Divdo \Divdo \DDd_2^\ast} \lrpar{\etad+ \frac{v}{2}\di\lrpar{\omtrchid-\frac{4}{v}\Omd}}} \\
=& \frac{1}{v^3} \lrpar{\lrpar{\Ldo-2}\Divdo+ \Divdo \Divdo \DDd_2^\ast}\lrpar{\etad+ \frac{v}{2}\di\lrpar{\omtrchid-\frac{4}{v}\Omd}} \\
&-\frac{1}{4v^4} \lrpar{\lrpar{\Ldo-2}\Divdo+ \Divdo \Divdo \DDd_2^\ast}\lrpar{\Divdo \chihd +v^2\mfq_4+v^2 \di\mfq_1+\half \di \mfq_2 -D\lrpar{\frac{v}{2}\di\mfq_2}},
\end{aligned} 
\end{align*}
and
\begin{align*} 
\begin{aligned} 
D\lrpar{\frac{1}{8v^2} \Divdo \di \Divdo \Divdo \gdcd} = -\frac{1}{4v^3}\Divdo \di \Divdo \Divdo \gdcd + \frac{1}{8v^2}\Divdo \di \Divdo \Divdo\lrpar{\frac{2}{v^2}\chihd+\frac{1}{v^2}\mfq_3}.
\end{aligned} 
\end{align*}
Summing up the above and using Lemma \ref{LEMgaugeconservationlaws} yields the transport equation for $\Du\ombd$ and thus finishes the proof of Lemma \ref{LEMtransportEQUStrans}. 

\section{Linearization at Schwarzschild} \label{SEClinearizedCHARGEequationsSSAPP}

\ni In this section we first derive the linearizations of $\PP_f$ and $\PP_{(j^1,j^2)}$ at Schwarzschild of mass $M\geq0$, see Section \ref{SECPerturbationslinearizesd222111SCHWARZSCHILD}. The linearizations are used in \eqref{EQlinearizationAPPLSCHWARZSCHILD} in Section \ref{SECconclusion2} for proving the perturbation estimate \eqref{EQChargeEstimatesMainTheorem0} for $(\mathbf{E},\mathbf{P},\mathbf{L},\mathbf{G})$. Then in Sections \ref{SECderivationConstraintFunctions2SCHWARZSCHILD} and \ref{SEClinearizedTRANSPORTEPLGSCHWARZSCHILD}, we calculate the linearizations of the constraint functions and the null transport equations for the linearizations of $(\mathbf{E},\mathbf{P},\mathbf{L},\mathbf{G})$ at Schwarzschild of mass $M\geq0$, respectively. The latter are used for \eqref{EQlinearizedESSestimate292} and \eqref{EQAppAPPLSCHWARZSCHILD} in Section \ref{SECconclusion27778} to prove the transport estimate \eqref{EQChargeEstimatesMainTheorem} for $(\mathbf{E},\mathbf{P},\mathbf{L},\mathbf{G})$.
\subsection{Linearization of $\PP_f$ and $\PP_{(j^1,j^2)}$ at Schwarzschild of mass $M\geq0$} \label{SECPerturbationslinearizesd222111SCHWARZSCHILD} 

\ni In this section we define the linearization of $\PP_f$ and $\PP_{(j^1,j^2)}$. As visible in the proofs of Lemmas \ref{LEMlinearizedTransversalSCHWARZSCHILD} and \ref{LEMspherediffLINSCHWARZSCHILD} below, their linearization is closely connected to the more general \emph{linearized pure gauge solutions} of \cite{DHR}.

First, we have the following lemma for $\PP_f$.
\begin{lemma}[Linearization of $\PP_f$] \label{LEMlinearizedTransversalSCHWARZSCHILD} 

Let $\dot{\PP}_{f}^M$ denote the linearization of $\PP_{f}$ in $f$ at $f=0$ and Schwarzschild of mass $M\geq0$. For a given linearized perturbation function $\dot f$,
\begin{align*} 
\begin{aligned} 
\dot f := \lrpar{\dot{f}(0,\th^1,\th^2), \pr_u \dot{f}(0,\th^1,\th^2), \,\, \pr_u^2 \dot{f}(0,\th^1,\th^2), \,\, \pr_u^3 \dot{f}(0,\th^1,\th^2) },
\end{aligned} 
\end{align*}
the non-trivial components of $\dot{\PP}^M_{f}\lrpar{\dot{f}}$ are given by
\begin{align*} \begin{aligned} 
\Omd =& \frac{1}{2\Om_M} \pr_{u} \lrpar{\dot{f} \Om_M^2}, & \phid=&-\Om_M^2 \dot{f}, & \etad=&  \frac{r_M}{\Om_M^2} \di \lrpar{\pr_u \lrpar{\frac{\Om_M^2}{r_M}f}},\\
\chihd =& - 2\Om_M \DDd_2^\ast \di \dot{f}, & \omtrchibd =& -2\pr_u \lrpar{\frac{f\Om_M^2}{r_M}}, & \omtrchid=& \frac{2\Om_M^2}{r_M^2} \lrpar{\Ldo \dot{f} - \dot{f}\lrpar{1-2\Om_M^2}},
\end{aligned} \end{align*}
and
\begin{align*} 
\begin{aligned} 
\ombd =\pr_u\lrpar{\frac{1}{2\Om^2_M} \pr_{u} \lrpar{\dot{f} \Om_M^2}}, \,\,
\Du\ombd = \pr_u^2 \lrpar{\frac{1}{2\Om^2_M} \pr_{u} \lrpar{\dot{f} \Om_M^2}},
\end{aligned} 
\end{align*}
and
\begin{align*} 
\begin{aligned} 
\dot{\beta}=-\frac{6M\Om_M}{r_M^3} \di \dot{f}, \,\, \rhd=& -\frac{6M\Om_M^2}{r_M^4} \dot{f}.
\end{aligned} 
\end{align*}
where we tacitly evaluated at $u=0$.
\end{lemma}

\begin{proof}[Proof of Lemma \ref{LEMlinearizedTransversalSCHWARZSCHILD}] The direct way to prove Lemma \ref{LEMlinearizedTransversalSCHWARZSCHILD} is to linearize $\PP_f$ by hand, using the explicit formulas of Appendix \ref{SECproofTEClemmasmoothness}. In the following we argue that $\dot{\PP}_f$ 
is readily calculated in \cite{DHR}.

Indeed, in \cite{DHR} the following mapping is studied. Let $(\tilde{u},\tilde{v},\tilde{\th}^1,\tilde{\th}^2)$ be Eddington-Finkelstein coordinates on the exterior region of a Schwarzschild spacetime of small mass $M>0$, see \eqref{EQSSformulaMETRIC}. Consider the sphere
\begin{align*} 
\begin{aligned} 
\tilde S = \tilde{S}_{0,2}:= \{ \tilde{u}=0, \tilde{v}=2 \}.
\end{aligned} 
\end{align*}
Given a smooth and sufficiently small scalar function $f=f(u,\th^1,\th^2)$, define new local coordinates $(u,v,\th^1,\th^2)$
\begin{align*} 
\begin{aligned} 
\tilde{u}= u+ f(u,\th^1,\th^2),\,\, \tilde{v}=v, \,\, \tilde{\th}^1=\th^1, \,\, \tilde{\th}^2=\th^2.
\end{aligned} 
\end{align*}
The coordinate system $(u,v,\th^1,\th^2)$ is not double null. However, it is shown in (173) in \cite{DHR} that $(u,v,\th^1,\th^2)$ is double null \emph{to first order in $f$}. 

Hence, to first order in $f$, the sphere data calculated with respect to $(u,v,\th^1,\th^2)$ on the sphere
\begin{align*} 
\begin{aligned} 
S= S_{0,2} := \{u=0, v=2\},
\end{aligned} 
\end{align*}
agrees with the sphere data $x_{0,2}:=\PP_f({\underline{\mathfrak{m}}}^M)$ constructed by our mapping $\PP_f$. Consequently, their linearizations in $f$ (evaluated at $f=0$ and Schwarzschild of mass $M>0$) agree. This linearization is calculated in Lemma 6.1.1 in \cite{DHR} and agrees with our expressions in Lemma \ref{LEMlinearizedTransversal}. We note that the expression for $\ombd$ follows from \eqref{EQssLinearizationRelations7778}.

We remark that in \cite{DHR} the linearization is calculated at Schwarzschild of mass $M>0$, but a straight-forward inspection shows that the calculation goes through for $M\geq0$. This finishes the proof of Lemma \ref{LEMlinearizedTransversalSCHWARZSCHILD}. \end{proof}

\ni Second, we have the following lemma for $\PP_{(j^1,j^2)}$. It is a corollary of Lemma 6.1.3 in \cite{DHR}, where we note that our notation connects to \cite{DHR} as follows,
\begin{align*} 
\begin{aligned} 
\widehat{\dot{\gd}} = r_M^2 \gdcd, \,\, \frac{\dot{\sqrt{\det\gd}}}{\sqrt{\det\gd}} = \frac{2\phid}{r_M}.
\end{aligned} 
\end{align*}
\begin{lemma}[Linearized angular perturbations] \label{LEMspherediffLINSCHWARZSCHILD} Let $\dot{\PP}_{(j_1,j_2)}^M$ denote the linearization of $\PP_{(j_1,j_2)}$ in $(j_1,j_2)$ at $(j_1,j_2)=(0,0)$ and Schwarzschild sphere data $\mathfrak{m}^M$. The nontrivial components of $\dot{\PP}^M_{(j_1,j_2)} (\dot{j}_1,\dot{j}_2)$ are given by
\begin{align*}
\phid= \frac{r_M}{2} \Ldo \dot{q}_1, \,\, \gdcd = 2 \DDd_2^\ast \DDd_1^\ast (\dot{q}_1,\dot{q}_2),
\end{align*}
where the scalar functions $q_1$ and $q_2$ are related to $\dot{j}^1$ and $\dot{j}^2$ by
\begin{align*} 
\begin{aligned} 
\dot{j}^1 d\th^1 + \dot{j}^2 d\th^2 = - r_M^2 \DDd_1^\ast(\dot{q}_1,\dot{q}_2).
\end{aligned} 
\end{align*}
\end{lemma}

\subsection{Linearized constraint functions at Schwarzschild of mass $M\geq0$}\label{SECderivationConstraintFunctions2SCHWARZSCHILD}
In this section we linearize the constraint functions $(\CC_i(x))_{1\leq i \leq 10}$ at Schwarzschild of mass $M\geq0$, that is, at $x=\mathfrak{m}^M$. The linearization procedure is adapted from \cite{DHR}: We expand the sphere data
\begin{align*} 
\begin{aligned} 
x =& \lrpar{\Om, \gd, \Om\trchi, \chih, \Om\trchib, \chibh, \eta, \om, D\om, \omb, \Du\omb, \a,  \ab} \\
=& \lrpar{\Om_M, r_M^2 \gac, \frac{2\Om_M}{r}, 0, -\frac{2\Om_M}{r}, 0, 0, \frac{M}{r_M^2}, -\frac{2M\Om_M^2}{r_M^3}, -\frac{M}{r_M^2}, -\frac{2M\Om_M^2}{r_M^3}, 0,0}\\
&+ \varep \cdot \lrpar{\dot\Om, \dot{\gd}, \omtrchid, \dot\chih, \omtrchibd, \dot\chibh, \dot\eta, \omd, D\omd, \ombd, \Du\ombd, \ad, \abd} + \mathcal{O}(\varep^2),
\end{aligned} 
\end{align*}
and differentiate in $\varep$ at $\varep=0$. Here, the Schwarzschild quantities $r_M=r_M(u,v)$ and $\Om_M=\Om_M(u,v)$ are defined in \eqref{EQdefRbyUV} and \eqref{EQSSnullLapseFormula}, respectively.

The proof of the next lemma follows by explicit calculation, see also Section 5 in \cite{DHR}.

\begin{lemma}[Linearization of constraint functions at Schwarzschild] \label{LEMlinearizedConstraintsSCHWARZSCHILD} Let $M\geq0$ be a real number and let $(\dot{\CC}_i^M)_{1\leq i \leq 10}$ denote the linearization of the constraint functions $({\CC}_i)_{1\leq i \leq 10}$ at Schwarzschild of mass $M$. Then it holds that
\begin{align*} 
\begin{aligned} 
\dot{\CC}^M_1=& D^2 \phid - 2 \Om_M^2 \omd - \frac{M}{r_M} \omtrchid - \frac{2M\Om_M^2}{r_M^3} \phid, \\
\dot{\CC}^M_2=& r_M^2 \lrpar{2D\lrpar{\frac{\phid}{r_M}}-\omtrchid} ,\\
\dot{\CC}^M_3=& r_M^2 D\gdcd - 2 \Om_M \chihd,\\
\dot{\CC}^M_4=& \frac{1}{r_M^2} D\lrpar{r_M^2 \etad} - \Om_M \lrpar{\frac{1}{r_M^2} \Divdo \chihd - \frac{1}{2\Om_M} \di \omtrchid + \frac{4}{r_M} \di \Omd}, \\
\dot{\CC}^M_5=& \frac{1}{r_M^2} D\lrpar{r_M^2 \omtrchibd} - \frac{2\Om_M^2}{r_M} \omtrchid + \frac{2\Om_M^2}{r_M^2} \Divdo \lrpar{\etad-\frac{2}{\Om_M}\di \Omd} + \frac{4\Om_0}{r_M^2} \Omd +2 \Om_M^2 \Kd,
\end{aligned} 
\end{align*}
and moreover
\begin{align*} 
\begin{aligned} 
\dot{\CC}^M_6=& \Om_M r_M D\lrpar{\frac{\chibhd}{r_M}} + \chibhd  \frac{\Om_M M}{r_M^2} - 2 \Om_M^2 \DDd_2^\ast \lrpar{\etad-\frac{2}{\Om_M}\di\Omd} -\frac{\Om_M^3}{r_M} \chihd ,\\
\dot{\CC}^M_7=& D\ombd - \frac{2\Om_M}{r_M^2} \Omd - \Om_M^2 \Kd + \frac{\Om_M^2}{2r_M} \lrpar{\omtrchid-\omtrchibd}, \\
\dot{\CC}^M_8=& D\abd - \frac{\Om_M^2}{r_M} \abd + \frac{2M}{r_M^2} \abd - 2 \DDd_2^\ast \lrpar{\frac{\Om_M}{r_M^2} \Divdo \chibhd - \half \di \omtrchibd - \frac{\Om_M^2}{r_M} \etad}- \frac{6M\Om_M}{r_M^3} \chibhd, \\
\dot{\CC}^M_9=& D \lrpar{\Du \ombd} + \frac{12 M \Om_M^3}{r_M^2} \lrpar{\frac{\di\Omd}{\Om_M}-\etad} + \frac{2M\Om_M^2}{r_M^3} \lrpar{\frac{3}{2} \omtrchibd - 2 \ombd} \\
&+ \lrpar{\frac{2\Om_M}{r_M^2} \Omd + \Om_M^2 \Kd - \frac{\Om_M^2}{2r_M} \lrpar{\omtrchid-\omtrchibd}} \lrpar{-\frac{3\Om_M^2}{r_M} + \frac{2M}{r_M^2}} \\
&- \frac{\Om_M^3}{r_M^2} \Divdo \lrpar{\frac{1}{r_M^2} \Divdo \chibhd - \frac{1}{2\Om_M} \di \omtrchibd - \frac{\Om_M}{r_M} \etad}.\\
\CCd^M_{10}=& \Om_M \ad + D\chihd - \frac{M}{r_M^2} \chihd.
\end{aligned} 
\end{align*}
\end{lemma}
\begin{remark} In addition to the above, we have by \eqref{EQdefdecomposition1} and \eqref{DEFricciCoefficients} that
\begin{align} 
\begin{aligned} 
\dot{\gd} = 2r_M \phid \gac + r_M^2 \gdcd, \,\, \omd = D\lrpar{\frac{\Omd}{\Om_M}}, \,\, \ombd = \Du\lrpar{\frac{\Omd}{\Om_M}}, \,\, \etabd =& - \etad + \frac{2}{\Om_M} \di \Omd.
\end{aligned} \label{EQssLinearizationRelations7778SCHWARZSCHILD}
\end{align}
Moreover, by (242) in \cite{DHR} the linearization of the Gauss curvature $\Kd$ is given by
\begin{align} \label{EQgaussDHRSCHWARZSCHILD}
\Kd =& \frac{1}{2r_M^2} \Divdo \Divdo \gdcd - \frac{1}{r_M^3} (\Ldo +2) \phid.
\end{align}
Moreover, using that for Schwarzschild sphere data, $\phi= r_M$ and $r=r_M$, we have that $\dot{r}^{[\geq1]}=0$ and
\begin{align} 
\begin{aligned} 
\dot{r}^{(0)} = \phid^{(0)}.
\end{aligned} \label{EQareaRADIUSlinearizationSCHWARZSCHILD}
\end{align}
\end{remark}

\subsection{Linearized transport equations for $(\mathbf{E},\mathbf{P},\mathbf{L},\mathbf{G})$ at Schwarzschild} \label{SEClinearizedTRANSPORTEPLGSCHWARZSCHILD} In this section we linearize the charges $(\dot{\mathbf{E}},\dot{\mathbf{P}},\dot{\mathbf{L}},\dot{\mathbf{G}})$ at Schwarzschild of mass $M\geq0$, and analyze their transport equations along $\HH$.

First, from Definition \ref{DEFnonlinearcharges6} and \eqref{EQGaussEquation} and \eqref{EQgausscodazzinonlinear1}, we recall that

\begin{align} 
\begin{aligned} 
-\frac{8\pi}{\sqrt{4\pi}} \mathbf{E} :=& \lrpar{r^3 \lrpar{\rh+r\Divd \be}}^{(0)} \\
=&- \lrpar{r^3 \lrpar{K+\frac{1}{4} \trchi\trchib -\half (\chih,\chibh)}}^{(0)}\\
&- \lrpar{r^4 \lrpar{\Divd \Divd \chih - \half \Ld \trchi + \Divd \lrpar{\chih \cdot \zeta- \half \trchi \zeta}}}^{(0)},
\end{aligned} \label{EQlinEPLGKerr1}
\end{align}
and
\begin{align} 
\begin{aligned} 
- \frac{8\pi}{\sqrt{\frac{4\pi}{3}}} \mathbf{P}^m :=& \lrpar{r^3 \lrpar{\rh+r\Divd \be}}^{(1m)}\\
=& - \lrpar{r^3 \lrpar{K+\frac{1}{4} \trchi\trchib -\half (\chih,\chibh)}}^{(1m)}\\
&- \lrpar{r^4\lrpar{\Divd \Divd \chih - \half \Ld \trchi + \Divd \lrpar{\chih \cdot \zeta- \half \trchi \zeta}}}^{(1m)},\\
\frac{16\pi}{\sqrt{\frac{8\pi}{3}}}\mathbf{L}^m :=& \lrpar{r^3 \lrpar{\di\trchi+ \trchi \lrpar{\eta-\di \log \Om}}}^{(1m)}_H, \\
\frac{16\pi}{\sqrt{\frac{8\pi}{3}}} \mathbf{G}^m :=& \lrpar{ r^3 \lrpar{\di\trchi+ \trchi \lrpar{\eta-\di \log \Om}}}^{(1m)}_E. 
\end{aligned} \label{EQlinEPLGKerr2}
\end{align}

\ni Linearizing these expressions at Schwarzschild of mass $M\geq0$, see \eqref{EQspheredataSSM111222}, and using \eqref{EQgaussDHR} and \eqref{EQareaRADIUSlinearization}, we get the explicit expressions

\begin{align*}
\begin{aligned} 
-\frac{8\pi}{\sqrt{4\pi}}\dot{\mathbf{E}} =& -\frac{6M\phid^{(0)}}{ r_M} +2 \phid^{(0)}-2 r_M \Om_M \Omd^{(0)} + \frac{r_M^2}{2} \omtrchid^{(0)}-\frac{r_M^2}{2} \omtrchibd^{(0)}, \\
- \frac{8\pi}{\sqrt{\frac{4\pi}{3}}} \dot{\mathbf{P}}^m =& 2 r_M\lrpar{2-\Om_M} \Omd^{(1m)} + \frac{r_M^2}{2}\lrpar{1-\frac{2}{\Om_M}}\omtrchid^{(1m)} \\
&- \frac{r_M^2}{2}\omtrchibd^{(1m)} + r_M\Om_M \Divdo \etad^{(1m)} ,\\
\frac{16\pi}{\sqrt{\frac{8\pi}{3}}}\dot{\mathbf{L}}^m =&  2 r_M^2 \Om_M \etad^{(1m)}_H, \\
\frac{16\pi}{\sqrt{\frac{8\pi}{3}}}\dot{\mathbf{G}}^m =& \frac{r_M^3}{\Om_M} \di \omtrchid^{(1m)}_E +2 r_M^2 \Om_M \etad^{(1m)}_E - 4 r_M^2 \di \Omd^{(1m)}_E,
\end{aligned} 
\end{align*}
where we used that for scalar functions $f$, $(\Ldo f)^{[1]} = -2 f^{[1]}$, see Appendix \ref{SECellEstimatesSpheres}.

By applying the homogeneous linearized null structure at Schwarzschild, see Lemma \ref{LEMlinearizedConstraintsSCHWARZSCHILD}, together with \eqref{EQgaussDHRSCHWARZSCHILD} and \eqref{EQareaRADIUSlinearizationSCHWARZSCHILD}, it is straight-forward to derive transport equations for these linearized charges at Schwarzschild. The resulting equations are summarized in the following lemma.

\begin{lemma}[Linearized transport equations for charges at Schwarzschild] \label{LEMlinearizedtransportChargesSS} The following hold for $M\geq0$ and $m=-1,0,1$,
\begin{align*} 
\begin{aligned} 
-\frac{8\pi}{\sqrt{4\pi}} D\lrpar{\dot{\mathbf{E}}} =& M \omtrchid^{(0)}, \\
- \frac{8\pi}{\sqrt{\frac{4\pi}{3}}} D\lrpar{\dot{\mathbf{P}}^m} =& -2 \lrpar{2(1-\Om_M) -\frac{6M}{r_M}} \Omd^{(1m)}\\
&-  \lrpar{M\lrpar{\frac{1}{\Om_M}-3}+r_M (1-\Om_M)} \omtrchid^{(1m)} \\
&-  \lrpar{\frac{M}{r_M}\lrpar{2-3\Om_M} + (\Om_M -1)} (\Divdo \etad)^{(1m)}, \\
\frac{16\pi}{\sqrt{\frac{8\pi}{3}}}D\lrpar{\dot{\mathbf{L}}^m} =& 2\Om_M M \etad_H^{(1m)}, \\
\frac{16\pi}{\sqrt{\frac{8\pi}{3}}}D\lrpar{\dot{\mathbf{G}}^m} =& -\frac{M r_M}{\Om_M} (\di \omtrchid)_E^{(1m)}+2 \Om_M M \etad_E^{(1m)} \\
&- \frac{4M \Om_M}{r_M} (\di\phid)_E^{(1m)}-4M \lrpar{\di \Omd}_E^{(1m)}
\end{aligned} 
\end{align*}

\end{lemma}

\ni By definition of $\Om_M$ in \eqref{EQspheredataSSM111222} it holds that for $M$ small,
\begin{align} 
\begin{aligned} 
\vert \Om_M-1 \vert \les M,
\end{aligned} \label{EQexpansionOMEGAc}
\end{align}
so that we can write the equations of Lemma \ref{LEMlinearizedtransportChargesSS} for $M\geq0$ small as
\begin{align} 
\begin{aligned} 
D\lrpar{\dot{\mathbf{E}}} =& \OO(M) \omtrchid^{(0)}, \\
D\lrpar{\dot{\mathbf{P}}^m} =& \OO(M)\Omd^{(1m)}+ \OO(M) \omtrchid^{(1m)} + \OO(M) (\Divdo \etad)^{(1m)}, \\
D\lrpar{\dot{\mathbf{L}}^m} =& \OO(M) \etad_H^{(1m)}, \\
D\lrpar{\dot{\mathbf{G}}^m} =& \OO(M) (\di \omtrchid)_E^{(1m)}+\OO(M) \etad_E^{(1m)} \\
&+ \OO(M)(\di\phid)_E^{(1m)}+\OO(M) \lrpar{\di \Omd}_E^{(1m)},
\end{aligned} \label{EQMsmalllinearizedEPLGSS99}
\end{align}
where $\OO(M)$ denotes terms that are bounded by $M$.

\section{Hodge systems and Fourier theory on $2$-spheres} \label{SECellEstimatesSpheres}
\ni In this Section \ref{SEChodgesystems} we recall the theory of $2$-dimensional Hodge systems, see also \cite{ChrKl93}. In Section \ref{SECfourierSpheres} we recapitulate the definition and properties of tensor spherical harmonics, following the notation of \cite{Czimek1}. In Section \ref{SECspecAnalysis} we use tensor spherical harmonics to analyse differential operators which appear in this paper.
\subsection{Hodge systems on Riemannian $2$-spheres}\label{SEChodgesystems}

\begin{definition}[Hodge operators] Let $(S, \gd)$ be a Riemannian $2$-sphere. Define
\begin{enumerate}
\item for a $1$-form $X_A$, 
\begin{align*}
\DDd_1(X) := (\Divd X, \Curld X).
\end{align*}
\item for a $2$-tensor $W_{AB}$, 
\begin{align*}
\DDd_2(W)_C := (\Divd W)_C.
\end{align*}
\item for a pair of functions $(f_1,f_2)$, 
\begin{align*}
\DDd_1^{\ast}(f_1,f_2):= -\di f_1 + {}^\ast\di f_2.
\end{align*}
\item for a $1$-form $X_A$,
\begin{align*}
\DDd_2^{\ast}(X)_{AB} := -\half \left( \Nd_A X_B + \Nd_B X_A - (\Divd X) \gd_{AB} \right).
\end{align*}
\end{enumerate}
Throughout the paper we abuse notation by denoting $\DD_2$ as $\Divd$. In the following, we use on the round sphere $(S_{u,v},(v-u)^2\gac)$ the notation
\begin{align*}
\DDdo_1 := (v-u)^2 \DDd_1, \,\, \DDdo_2 := (v-u)^2 \DDd_2.
\end{align*}
\end{definition}

\ni The following lemma is a paraphrase of the material in \cite{ChrKl93}.
\begin{lemma}The following holds.
\begin{enumerate}
\item The kernels of $\DDd_1$ and $\DDd_2$ are trivial.
\item The kernel of $\DDd_1^\ast$ consists of pairs of constant functions $(f_1,f_2)=(c_1,c_2)$.
\item The kernel of $\DDd_2^\ast$ consists of the set of conformal Killing vectorfields (a $6$-dimensional space on the round sphere).
\item The $L^2$-range of $\DDd_1$ consists of all pairs of functions $(f_1,f_2)$ on $S$ with vanishing mean.
\item The $L^2$-range of $\DDd_2$ consists of all $L^2$-integrable $1$-forms on $S$ which are orthogonal to the conformal Killing vectorfields.

\item The operators $\DDd_1^\ast$ and $\DDd_2^\ast$ are conformally invariant. 
\end{enumerate}
\end{lemma}

\subsection{Tensor spherical harmonics} \label{SECfourierSpheres} Tensor spherical harmonics are defined on the standard round unit sphere as follows.
\begin{definition}[Tensor spherical harmonics] Introduce the following spherical harmonics functions, vectorfields and tracefree symmetric $2$-tensors.
\begin{enumerate}
\item For integers $l\geq0$, $-l \leq m \leq l$, let $Y^{(lm)}$ be the standard (real-valued) spherical harmonics on the round unit sphere $S_1$. 
\item For $l\geq1$, $-l \leq m \leq l$, define the vectorfields
\begin{align*} 
\begin{aligned} 
E^{(lm)} := \frac{1}{\sqrt{l(l+1)}} \DDd_1^\ast (Y^{(lm)},0), \,\, H^{(lm)} :=\frac{1}{\sqrt{l(l+1)}} \DDd_1^\ast (0, Y^{(lm)}).
\end{aligned} 
\end{align*}
The vectorfields $E^{(lm)}$ and $H^{(lm)}$ are called \emph{electric} and \emph{magnetic}, respectively.

\item For $l\geq2$, $-l \leq m \leq l$, define the tracefree symmetric $2$-tensors
\begin{align*} 
\begin{aligned} 
\psi^{(lm)}:= \frac{1}{\sqrt{\half l(l+1)-1}} \DDd_2^\ast \lrpar{E^{(lm)}}, \,\, \phi^{(lm)} := \frac{1}{\sqrt{\half l(l+1)-1}} \DDd_2^\ast \lrpar{H^{(lm)}}.
\end{aligned} 
\end{align*}
The tensors $\psi^{(lm)}$ and $\phi^{(lm)}$ are called \emph{electric} and \emph{magnetic}, respectively.

\end{enumerate}

\end{definition}

\ni The following lemma is a summary of properties of spherical harmonics, see for example \cite{Czimek1} for more details and proofs.
\begin{lemma} The following holds.
\begin{enumerate}
\item On the round unit sphere $S_1$, $L^2$-integrable functions $f$, vectorfields $X$ and tracefree symmetric $2$-tensors $V$ can be decomposed as follows,
\begin{align*}
f =& \sum\limits_{l\geq0} \sum\limits_{-l\leq m \leq l} f^{lm} Y^{(lm)}, \\
X=&  \sum\limits_{l\geq1} \sum\limits_{-l\leq m \leq l} X^{lm}_E E^{(lm)} + X^{lm}_H H^{(lm)},\\
V =&  \sum\limits_{l\geq2} \sum\limits_{-l\leq m \leq l} V^{lm}_\psi \psi^{(lm)} +V^{lm}_\phi \phi^{(lm)},
\end{align*}
where 
\begin{align*} 
\begin{aligned} 
f^{(lm)} :=& \int\limits_{S_1} f Y^{(lm)} d\mu_{\gac}, &&\\
X_E^{(lm)} :=& \int\limits_{S_1} X \cdot E^{(lm)} d\mu_{\gac}, &
X_H^{(lm)} :=& \int\limits_{S_1} X \cdot H^{(lm)} d\mu_{\gac}, \\
V_\psi^{(lm)} :=& \int\limits_{S_1} V \cdot \psi^{(lm)} d\mu_{\gac}, &
V_\phi^{(lm)} :=& \int\limits_{S_1} V \cdot \phi^{(lm)} d\mu_{\gac}, 
\end{aligned} 
\end{align*}
where $d\mu_\gac$ denotes the volume element of the standard round unit metric on $S_1$ and $\cdot$ denotes the product with respect to $\gac$.
\item It holds that for $l\geq1$,
\begin{align} \begin{aligned}
(\di f)_E^{(lm)}  =& - \sqrt{l(l+1)} f^{(lm)}, &(\di f)_H^{(lm)}  =& 0,\\
(\DDd_1^\ast(0,f))_E^{(lm)}  =& 0, &
(\DDd_1^\ast(0,f))_H^{(lm)}  =& \sqrt{l(l+1)} f^{(lm)},\\
(\Divdo X)^{(lm)} =& \sqrt{l(l+1)} X_E^{(lm)}, &&
\end{aligned}\label{EQFourierMultiplier7778}\end{align}
and for $l\geq2$,
\begin{align} \begin{aligned}
\DDd_2^\ast(X)_\psi^{(lm)} =& \sqrt{\half l(l+1)-1} \, X_E^{(lm)}, & \DDd_2^\ast(X)_\phi^{(lm)} =& \sqrt{\half l(l+1)-1} \, X_H^{(lm)}, \\
(\Divdo V)_E^{(lm)}=& \sqrt{\half l(l+1)-1} \, V_\psi^{(lm)}, &
(\Divdo V)_H^{(lm)}=& \sqrt{\half l(l+1)-1} \, V_\phi^{(lm)}.
\end{aligned}\label{EQFouriermultipliers333}\end{align}
\item The operator $\DDd_1$ is a bijection between vectorfields and pairs of functions $(f,g)$ with vanishing means,
\begin{align*}
\DDd_1: X^{[l\geq1]} \to (Y^{[l\geq1]}, Y^{[l\geq1]}).
\end{align*}
Moreover, the following restrictions are bijections:
\begin{align*}
\DDd_1:& E^{[l\geq1]} \to (Y^{[l\geq1]}, 0), \\
\DDd_1:& H^{[l\geq1]} \to (0, Y^{[l\geq1]}).
\end{align*}
The spherical harmonics vectorfields of mode $l=1$ form the space of conformal Killing vectorfields of the unit round sphere.
\item The operator $\DDd_2$ is a bijection between tracefree symmetric $2$-tensors and vectorfields of modes $l\geq2$, 
\begin{align*}
\DDd_2: V^{[l\geq2]} \to X^{[l\geq2]}.
\end{align*}
 Moreover, the following mappings are bijections:
\begin{align*}
\DDd_2: \psi^{[l\geq2]} \to E^{[l\geq2]}, \,\,
\DDd_2: \phi^{[l\geq2]} \to H^{[l\geq2]}.
\end{align*}
\item Let $k\geq0$ be an integer. There exists a constant $C_k>0$, depending only on $k$, such that for scalar functions $f$, vectorfields $X$ and symmetric tracefree $2$-tensors $V$ on $S_1$, we have the following equivalence of norms,
\begin{align*} 
\begin{aligned} 
\sum\limits_{0 \leq k' \leq k} \Vert \Nd^{k'} f \Vert^2_{L^2(S_1)} \sim& \sum\limits_{l\geq0} \sum\limits_{-l\leq m \leq l} (l+1)^{2k} \lrpar{ f^{(lm)}}^2, \\
\sum\limits_{0 \leq k' \leq k} \Vert \Nd^{k'} X \Vert^2_{L^2(S_1)} \sim& \sum\limits_{l\geq1} \sum\limits_{-l\leq m \leq l} (l+1)^{2k} \lrpar{ \lrpar{ X_E^{(lm)}}^2+ \lrpar{ X_H^{(lm)}}^2}, \\
\sum\limits_{0 \leq k' \leq k} \Vert \Nd^{k'} V \Vert^2_{L^2(S_1)} \sim& \sum\limits_{l\geq2} \sum\limits_{-l\leq m \leq l} (l+1)^{2k} \lrpar{ \lrpar{ V_\psi^{(lm)}}^2+ \lrpar{ V_\phi^{(lm)}}^2}.
\end{aligned} 
\end{align*}

\end{enumerate}
\end{lemma}

\ni \textbf{Notation.} Given a scalar function
\begin{align*} 
\begin{aligned} 
f= \sum\limits_{l\geq0} \sum\limits_{-l\leq m \leq l} f^{lm} Y^{(lm)},
\end{aligned} 
\end{align*}
we denote, for integers $l' \geq0$,
\begin{align*} 
\begin{aligned} 
f^{[l']} = \sum\limits_{l=l'} \sum\limits_{-l\leq m \leq l} f^{lm} Y^{(lm)}, \,\, f^{[\geq l']} = \sum\limits_{l\geq l'} \sum\limits_{-l\leq m \leq l} f^{lm} Y^{(lm)}, \,\,f^{[\leq l']} = \sum\limits_{0\leq l\leq l'} \sum\limits_{-l\leq m \leq l} f^{lm} Y^{(lm)}.
\end{aligned} 
\end{align*}
Similarly for vectorfields $X$ and symmetric tracefree $2$-tensors $V$. Moreover, denote the electric part and the magnetic part of a vectorfield $X$ by $X_E$ and $X_H$, respectively, and similarly for symmetric tracefree $2$-tensors $V$ by $V_\psi$ and $V_\phi$, respectively.
\subsection{Spectral analysis of differential operators} \label{SECspecAnalysis} In this section, we discuss the differential operators that appeared in Section \ref{SEClinearizedProblem}. \\ 

\ni \textbf{Analysis of $\DDd_2^\ast \Divdo +1$.} Let $V$ be a tracefree symmetric $2$-tensor,
\begin{align*}
V= \sum\limits_{l\geq2}\sum\limits_{-l\leq m \leq l} V_\psi^{(lm)} \psi^{lm} + V_\phi^{(lm)} \phi^{lm}.
\end{align*}
Then 
\begin{align*}
(\DDd_2^\ast \Divdo +1)V = \sum\limits_{l\geq2}\sum\limits_{-l\leq m \leq l} \underbrace{\lrpar{\half l(l+1) -1 +1}}_{>0 \text{ for } l\geq2}\lrpar{V_\psi^{(lm)} \psi^{lm} + V_\phi^{(lm)} \phi^{lm}}.
\end{align*}
Hence the operator has no kernel and we have the following elliptic estimate. Let $W$ be a given tracefree symmetric $2$-tensor. Then there exists a unique solutions $V$ to
\begin{align*} 
\begin{aligned} 
\lrpar{\DDd_2^\ast \Divdo +1}V = W,
\end{aligned} 
\end{align*}
and for integers $k\geq0$ we have the estimate,
\begin{align} 
\begin{aligned} 
\Vert V \Vert_{H^{k+2}(S_1)} \les \Vert W \Vert_{H^{k}(S_1)}.
\end{aligned}\label{EQestimate3app}
\end{align}

\ni \textbf{Analysis of $(\Divdo \DDd_2^\ast +1+ \di\Divdo)$.} Let $X$ be a vectorfield,
\begin{align*}
X= \sum\limits_{l\geq1}\sum\limits_{-l\leq m \leq l} X_E^{(lm)} E^{(lm)} + X_H^{(lm)} H^{(lm)}.
\end{align*}
Then it holds that
\begin{align*}
\Divdo \DDd_2^\ast X=& \sum\limits_{l\geq1}\lrpar{\half l(l+1)-1} \lrpar{X_E^{(lm)} E^{(lm)} + X_H^{(lm)} H^{(lm)}},
\end{align*}
and
\begin{align*}
\di \Divdo X=& \sum\limits_{l\geq1}\sum\limits_{-l\leq m \leq l}(-l(l+1)) X_E^{(lm)}E^{(lm)}.
\end{align*}
Therefore
\begin{align*}
&\lrpar{\Divdo \DDd_2^\ast  + 1 + \di \Divdo } X \\
=&\sum\limits_{l\geq1}\sum\limits_{-l\leq m \leq l} X_E^{(lm)} E^{(lm)} \lrpar{\lrpar{\half l(l+1)-1}  + 1 - l(l+1) } \\
&+\sum\limits_{l\geq1}\sum\limits_{-l\leq m \leq l} X_H^{(lm)} H^{lm} \lrpar{ \lrpar{\half l(l+1)-1} +1} \\
=&\sum\limits_{l\geq1}\sum\limits_{-l\leq m \leq l} X_E^{(lm)} E^{lm} \underbrace{\lrpar{-\half l(l+1) }}_{<0 \text{ for } l\geq1} +\sum\limits_{l\geq1} X_H^{(lm)} H^{lm} \underbrace{\lrpar{\half l(l+1) } }_{>0 \text{ for } l\geq1}.
\end{align*}
Hence the operator has no kernel and we have the following elliptic estimate. Let $Y$ be a given vectorfield. Then there exists a unique solution $X$ to
\begin{align*} 
\begin{aligned} 
(\Divdo \DDd_2^\ast +1+ \di\Divdo)X = Y,
\end{aligned} 
\end{align*}
and for integers $k\geq0$ we have the estimate,
\begin{align} 
\begin{aligned} 
\Vert X \Vert_{H^{k+2}(S_1)} \les \Vert Y \Vert_{H^{k}(S_1)}.
\end{aligned} \label{EQestimate4app}
\end{align}
By \eqref{EQFouriermultipliers333} and \eqref{EQestimate4app}, it follows in particular that the operator
\begin{align*} 
\begin{aligned} 
\DDd_2^\ast (\Divdo \DDd_2^\ast +1+ \di\Divdo) \Divdo
\end{aligned} 
\end{align*}
has no kernel and admits the following estimate. For any given symmetric tracefree $2$-tensor $W$, there exists a unique solution $V$ to 
\begin{align*} 
\begin{aligned} 
\DDd_2^\ast (\Divdo \DDd_2^\ast +1+ \di\Divdo) \Divdo V =W
\end{aligned} 
\end{align*}
satisfying
\begin{align} 
\begin{aligned} 
\Vert V \Vert_{H^{k+4}(S_1)} \les \Vert W \Vert_{H^{k}(S_1)}.
\end{aligned} \label{EQestimate42app}
\end{align}

\ni \textbf{Analysis of $(\Divdo\DDd_2^\ast+1+\half\di \Divdo)\di$.} Let $f$ be a scalar function,
\begin{align*}
f= \sum\limits_{l\geq0} \sum\limits_{-l\leq m \leq l} f^{(lm)} Y^{lm}.
\end{align*}
Then
\begin{align*}
\Divdo \DDd_2^\ast \di f = \sum\limits_{l\geq0} \sum\limits_{-l\leq m \leq l} \lrpar{\half l(l+1)-1} \lrpar{-\sqrt{l(l+1)}} E^{(lm)},
\end{align*}
and
\begin{align*}
\di \Divdo \di f = \sum\limits_{l\geq0} \sum\limits_{-l\leq m \leq l} \lrpar{-\sqrt{l(l+1)}} \lrpar{-l(l+1)} f^{(lm)} E^{(lm)}.
\end{align*}
Therefore,
\begin{align*}
&\lrpar{\Divdo \DDd_2^\ast + 1 + \half \di\Divdo } \di f \\
=& \sum\limits_{l\geq0} \sum\limits_{-l\leq m \leq l} \underbrace{\lrpar{\lrpar{\half l(l+1)-1}+1 +\half \lrpar{- l(l+1)}}}_{=0} \lrpar{-\sqrt{l(l+1)}} E^{lm}.
\end{align*}
This shows that 
\begin{align} \label{EQidentity5app}
\lrpar{ \Divdo \DDd_2^\ast + 1 + \half \di\Divdo } \di f = 0.
\end{align}

\ni \textbf{Analysis of the operator $ \lrpar{2-\Divdo \DDd_2^\ast} $.} Let $X$ be a vectorfield,
\begin{align*}
X= \sum\limits_{l\geq1}\sum\limits_{-l\leq m \leq l} X_E^{(lm)} E^{(lm)} + X_H^{(lm)} H^{(lm)}.
\end{align*}
Then it holds that
\begin{align*}
(2-\Divdo \DDd_2^\ast) X=& \sum\limits_{l\geq1}\lrpar{2- \lrpar{\half l(l+1)-1}} \lrpar{X_E^{(lm)} E^{(lm)} + X_H^{(lm)} H^{(lm)}} \\
=& \sum\limits_{l\geq1}\half \lrpar{6- l(l+1)} \lrpar{X_E^{(lm)} E^{(lm)} + X_H^{(lm)} H^{(lm)}}
\end{align*}
We conclude that the kernel of the operator $ \lrpar{2-\Divdo \DDd_2^\ast} $ is given by the set of vectorfields
\begin{align} 
\begin{aligned} 
\{ X: X = X^{[2]} \}.
\end{aligned} \label{EQkernelofDUomboperator}
\end{align}
Further, let $Y$ be a vectorfield such that $Y=Y^{[\geq3]}$. Then there exists a unique vectorfield $X$ such that $X=X^{[\geq3]}$ and
\begin{align*} 
\begin{aligned} 
(2-\Divdo \DDd_2^\ast) X=Y,
\end{aligned} 
\end{align*}
with the estimate
\begin{align*} 
\begin{aligned} 
\Vert X \Vert_{H^{k+2}(S_1)} \les \Vert Y \Vert_{H^{k}(S_1)}.
\end{aligned} 
\end{align*}

\ni In particular, it follows moreover that for any given function $f$ with $f=f^{[\geq3]}$, there is a unique solution $V$ with $V=V^{[\geq3]}$ of
\begin{align*} 
\begin{aligned} 
\Divdo \lrpar{2-\Divdo \DDd_2^\ast} \Divdo V = f.
\end{aligned} 
\end{align*}
with the estimate
\begin{align} 
\begin{aligned} 
\Vert X \Vert_{H^{k+4}(S_1)} \les \Vert f \Vert_{H^{k}(S_1)}.
\end{aligned} \label{EQappEstimateHodge6}
\end{align}



\begin{thebibliography}{99}

\bibitem{CitePriceLaw} 
Y.~Angelopoulos, S.~Aretakis, D.~Gajic. 
\newblock {\em Price's law and precise late-time asymptotics for subextremal Reissner--Nordstr\"om black holes}. 
\newblock arXiv:2102.11888, 65 pp.

\bibitem{CiteKerr}
Y.~Angelopoulos, S.~Aretakis, D.~Gajic. 
\newblock {\em Late-time tails and mode coupling of linear waves on Kerr spacetimes}. 
\newblock arXiv:2102.11884, 100 pp.

\bibitem{Easymptotics} Y.~Angelopoulos, S.~Aretakis, D.~Gajic. 
\newblock {\em Late-time asymptotics for the wave equation on extremal Reissner-Nordstr\"om backgrounds}. 
\newblock Adv. Math. 375 (2020), 107363, 139 pp.

\bibitem{CiteSS} Y.~Angelopoulos, S.~Aretakis, D.~Gajic. 
\newblock {\em Late-time asymptotics for the wave equation on spherically symmetric, stationary spacetimes}.
\newblock Adv. Math. 323 (2018), 529-621.

\bibitem{PRL} Y.~Angelopoulos, S.~Aretakis, D.~Gajic. 
\newblock {\em Horizon hair of extremal black holes and measurements at null infinity}. 
\newblock Physical Review Letters, 121(13):131102, 2018.

\bibitem{CiteGluing} S.~Aretakis.
\newblock {\em The characteristic gluing problem and conservation laws for the wave equation on null hypersurfaces}. 
\newblock Ann. PDE 3 (2017), no. 1, Paper no. 3, 56 pp.

\bibitem{CiteElliptic} S.~Aretakis.
\newblock {\em On a foliation-covariant elliptic operator on null hypersurface}. 
\newblock Int. Math. Res. Not. 2015, no. 15, 6433-6469.

\bibitem{CiteExtremal1} S.~Aretakis.
\newblock {\em Stability and instability of extreme Reissner-Nordstr\"om black hole spacetimes for linear scalar perturbations I}. 
\newblock Commun. Math. Phys. 307 (2011), no. 1, 17-63.

\bibitem{CiteExtremal2} S.~Aretakis.
\newblock {\em Stability and instability of extreme Reissner-Nordstr\"om black hole spacetimes for linear scalar perturbations II}. 
\newblock Ann. Henri Poincar\'e, 12 (2011), no. 8, 1491-1538.

\bibitem{CiteExtremal3} S.~Aretakis.
\newblock {\em Horizon instability of extremal black holes}. 
\newblock Adv. Theor. Math. Phys. 19 (2015), no. 3, 507-530.

\bibitem{ACR3}
S.~Aretakis, S.~Czimek, I.~Rodnianski.
\newblock {\em Characteristic gluing for the Einstein equations and applications}. 
\newblock arXiv, 31 pp.

\bibitem{ACR2}
S.~Aretakis, S.~Czimek, I.~Rodnianski.
\newblock {\em Characteristic gluing to the Kerr family and application to spacelike gluing}. 
\newblock arXiv, 88 pp.

\bibitem{Khanna} 
L.~Burko, G.~Khanna, S.~Sabharwal. 
\newblock {\em Scalar and gravitational hair for extreme Kerr black holes}. 
\newblock Phys. Rev. D 103, 021502 (2021).

\bibitem{CarlottoSchoen}
A.~Carlotto, R.~Schoen. 
\newblock {\em Localizing solutions of the Einstein constraint equations}. 
\newblock Invent. Math. 205 (2016), no. 3, 559-615.

\bibitem{ChrFormationBlackHoles}
D.~Christodoulou.
\newblock {\em The formation of black holes in general relativity}. 
\newblock European Mathematical Society (EMS), Z\"urich, 2009. x+589 pp.

\bibitem{ChrKl93}
D.~Christodoulou, S.~Klainerman.
\newblock {\em The global nonlinear stability of the Minkowski space}. 
\newblock Princeton Mathematical Series, 41. Princeton University Press, Princeton, NJ, 1993. x+514 pp.

\bibitem{ChruscielDelay1}
P.~Chru\'sciel, E.~Delay.
\newblock {\em Existence of non-trivial, vacuum, asymptotically simple spacetimes}. 
\newblock Classical and Quantum Gravity, 19(9):L71, 2002.

\bibitem{ChruscielDelay2}
P.~Chru\'sciel, E.~Delay.
\newblock {\em On mapping properties of the general relativistic constraints operator in weighted function spaces, with applications}. 
\newblock M\'em. Soc. Math. Fr. (N.S.), no. 94 (2003), vi+103 pp.

\bibitem{CIP1}
P.~Chru\'sciel, J.~Isenberg, D.~Pollack.
\newblock {\em Gluing initial data sets for general relativity}. 
\newblock Physical review letters, 93(8):081101, 2004.

\bibitem{CIP2}
P.~Chru\'sciel, J.~Isenberg, D.~Pollack.
\newblock {\em Initial data engineering}. 
\newblock Comm. Math. Phys. 257 (2005), no. 1, 29-42.

\bibitem{ChruscielMazzeo}
P.~Chru\'sciel, R.~Mazzeo.
\newblock {\em On 'many-black-hole' vacuum spacetimes}. 
\newblock Classical and Quantum Gravity 20 (2003), no. 4, 729-754.

\bibitem{ChruscielPollack}
P.~Chru\'sciel, D.~Pollack.
\newblock {\em Singular Yamabe metrics and initial data with exactly Kottler-Schwarzschild-de Sitter ends}. 
\newblock Ann. Henri Poincar\'e, 9 (2008), no. 4, 639-654.

\bibitem{Cortier}
J.~Cortier.
\newblock {\em Gluing construction of initial data with Kerr-de Sitter ends}. 
\newblock Ann. Henri Poincar\'e 14 (2013), no. 5, 1109-1134.

\bibitem{Corvino}
J.~Corvino.
\newblock {\em  Scalar curvature deformation and a gluing construction for the Einstein constraint equations}. 
\newblock Comm. Math. Phys., 214(1):137-189, 2000.

\bibitem{CorvinoSchoen}
J.~Corvino, R.~Schoen.
\newblock {\em On the asymptotics for the vacuum Einstein constraint equations}. 
\newblock J. Differential Geom., 73(2):185-217, 2006.

\bibitem{Czimek1}
S.~Czimek.
\newblock {\em An extension procedure for the constraint equations}.
\newblock Ann. PDE 4 (2018), no. 1, Paper No. 2, 130 pp.

\bibitem{DHR}
M.~Dafermos, G.~Holzegel, I.~Rodnianski.
\newblock {\em The linear stability of the Schwarzschild solution to gravitational perturbations}. 
\newblock Acta Math. 222 (2019), no. 1, 1-214.

\bibitem{GromovL}
M.~Gromov, H.~Lawson.
\newblock {\em The classification of simply connected manifolds of positive scalar curvature}.
\newblock Ann. of Math. (2) 111 (1980), no. 3, 423-434.

\bibitem{HawkingEllis}
S.~Hawking, G.~Ellis.
\newblock {\em The large scale structure of space-time}.
\newblock Cambridge Monographs on Mathematical Physics, No. 1. Cambridge University Press, London-New York, 1973. xi+391 pp.

\bibitem{Hintz}
P.~Hintz.
\newblock {\em Black hole gluing in de Sitter space}.
\newblock Communications in Partial Differential Equations (2021), DOI: 10.1080/03605302.2020.1871368.

\bibitem{IMP3}
J.~Isenberg, D.~Maxwell, D.~Pollack.
\newblock {\em A gluing construction for non-vacuum solutions of the Einstein-constraint equations}. 
\newblock Adv. Theor. Math. Phys., 9(1):129-172, 2005.

\bibitem{IMP1}
J.~Isenberg, R.~Mazzeo, D.~Pollack.
\newblock {\em Gluing and wormholes for the Einstein constraint equations}. 
\newblock Comm. Math. Phys., 231(3):529-568, 2002.

\bibitem{IMP2}
J.~Isenberg, R.~Mazzeo, D.~Pollack.
\newblock {\em On the topology of vacuum spacetimes}. 
\newblock Ann. Henri Poincar\'e, volume 4, pages 369-383. Springer, 2003.

\bibitem{LukChar}
J.~Luk.
\newblock {\em On the local existence for the characteristic initial value problem in general relativity}. 
\newblock Int. Math. Res. Not. IMRN 2012, no. 20, 4625-4678.

\bibitem{GravImpulsesLukRod1}
J.~Luk, I.~Rodnianski.
\newblock {\em Local propagation of impulsive gravitational waves}. 
\newblock Comm. Pure Appl. Math. 68 (2015), no. 4, 511-624.

\bibitem{Ma}
S.~Ma, L.~Zhang. 
\newblock {\em Sharp decay estimates for massless Dirac fields on a Schwarzschild background}. 
\newblock arXiv:2008.11429, 2020.

\bibitem{MarsdenImplicit}
T.~Ratiu, R.~Abraham, J. E.~Marsden.
\newblock {\em Manifolds, tensor analysis, and applications}.
\newblock Third edition. Applied Mathematical Sciences, 75. Springer-Verlag, New York, 1988. x+654 pp.

\bibitem{Rendall}
A.~Rendall.
\newblock {\em Reduction of the Characteristic Initial Value Problem to the Cauchy Problem and Its Applications to the Einstein Equations.}
\newblock Proc. Roy. Soc. London Ser. A 427 (1990), no. 1872, 221-239.

\bibitem{SchoenYauPSCM}
R.~Schoen, S.~Yau.
\newblock {\em On the structure of manifolds with positive scalar curvature}.
\newblock Manuscripta Math. 28 (1979), no. 1-3, 159-183.

\bibitem{J3}
J.~Szeftel.
\newblock {\em Parametrix for wave equations on a rough background III: space-time regularity of the phase}.
\newblock Ast\'erisque 2018, no. 401, viii+321 pp.

\bibitem{TaylorPDE3}
M.~Taylor.
\newblock {\em Partial Differential Equations III: Nonlinear equations}.
\newblock  Applied Mathematical Sciences, 117. Springer-Verlag, New York, 1997. xxii+608 pp.

\end{thebibliography}
\end{document}